\documentclass[10Ishpt,a4paper,twoside]{amsart}
\usepackage{bm}
\usepackage{amsmath, latexsym, amsfonts, amssymb, amsthm}
\usepackage{graphicx, color,hyperref,dsfont,caption,subcaption,wrapfig}
\usepackage{datetime}
\hypersetup{
    linktoc=page,
    linkcolor=red,          
    citecolor=blue,        
    filecolor=blue,      
     urlcolor=cyan,
    colorlinks=true           
}

\numberwithin{equation}{section}

\newtheorem{theorem}{Theorem}[section]
\newtheorem{lemma}[theorem]{Lemma}
\newtheorem{proposition}[theorem]{Proposition}
\newtheorem{corollary}[theorem]{Corollary}

\newtheorem{remark}[theorem]{Remark}
\newtheorem{definition}[theorem]{Definition}

\newtheorem{assumption}[theorem]{Assumption}

\theoremstyle{definition}

\definecolor{remi}{rgb}{0,0,0}

\renewcommand{\tilde}{\widetilde}          
\DeclareMathSymbol{\leqslant}{\mathalpha}{AMSa}{"36} 
\DeclareMathSymbol{\geqslant}{\mathalpha}{AMSa}{"3E} 
\DeclareMathSymbol{\eset}{\mathalpha}{AMSb}{"3F}     
\renewcommand{\leq}{\;\leqslant\;}                   
\renewcommand{\geq}{\;\geqslant\;}                   
\newcommand{\dd}{\text{\rm d}}             

\newcommand{\mc}{\mathcal}

\newcommand{\cc}{\mathbb{C}}

\newcommand{\la}{\lambda}
\newcommand{\eps}{\epsilon}

\newcommand{\pl}{\partial}

\newcommand{\bbar}{\overline}

\newcommand{\cjd}{\rangle}
\newcommand{\cjg}{\langle}

\newcommand{\demi}{\tfrac{1}{2}}

\DeclareMathOperator{\supp}{supp}

\newcommand{\D}{\mathbb{D}}
\newcommand{\R}{\mathbb{R}}
\newcommand{\Z}{\mathbb{Z}}
\newcommand{\N}{\mathbb{N}}

\newcommand{\E}{\mathds{E}}
\renewcommand{\P}{\mathds{P}}

\newcommand{\caA}{{\mathcal A}}

\newcommand{\caC}{{\mathcal C}}

\newcommand{\caH}{{\mathcal H}}

\newcommand{\C}{\mathbb{C}}

\renewcommand{\P}{\mathds{P}}

\newcommand{\hf}{\frac{_1}{^2}}

\usepackage{xparse}

\def\eps{\varepsilon}
\def\S{\mathbb{S}}
\def\T{\mathbb{T}}

\def\bi{\begin{itemize}}
\def\ei{\end{itemize}}

\def\bnum{\begin{enumerate}}
\def\enum{\end{enumerate}}
\def\<#1{\langle #1 \rangle}

\begin{document}

\title[Compactified Imaginary Liouville Theory]{Compactified Imaginary Liouville Theory}

\author[C. Guillarmou]{Colin Guillarmou$^{1}$}
\address{Universit\'e Paris-Saclay, CNRS,  Laboratoire de math\'ematiques d'Orsay, 91405, Orsay, France.}
\email{colin.guillarmou@universite-paris-saclay.fr}

 \author[A. Kupiainen]{ Antti Kupiainen$^{2}$}
\address{University of Helsinki, Department of Mathematics and Statistics}
\email{antti.kupiainen@helsinki.fi}

\author[R. Rhodes]{R\'emi Rhodes$^{3}$}
\address{Aix Marseille Univ, CNRS, I2M, Marseille, France}
\email{remi.rhodes@univ-amu.fr }

\footnotetext[1]{}


\keywords{  Liouville Quantum Gravity, quantum field theory, Gaussian multiplicative chaos, Ward identities, BPZ equations, DOZZ formula  }

 \begin{abstract}
On a given Riemann surface, we construct a path integral based on the Liouville action functional with imaginary parameters. 
The construction relies on the compactified Gaussian Free Field (GFF), which we perturb   with a curvature term and an exponential potential.  In physics this path integral is conjectured to describe the scaling limit of critical loop models such as Potts and O(n) models.   
The potential term is defined by means of imaginary Gaussian Multiplicative Chaos theory. The curvature term  involves
 integrated 1-forms, which are multivalued on the manifold, and requires a delicate regularisation in order to preserve diffeomorphism invariance.
We prove that the probabilistic
 path integral 
 satisfies the axioms of Conformal Field Theory (CFT)  including Segal's gluing axioms. 
 We construct the correlation functions for this CFT, involving  electro-magnetic operators. This CFT has several exotic features: most importantly, it is non unitary and has the structure of a  logarithmic CFT. This is the first mathematical construction of a logarithmic CFT and therefore the present paper  provides a concrete mathematical setup for this concept.
 \end{abstract}

\maketitle
\setcounter{tocdepth}{2}
\tableofcontents


\begin{center}
\end{center}
\footnotesize



\normalsize

\section{Introduction}
Conformal Field Theory (CFT), in its Euclidean formulation,  is a framework in physics for describing the various possible universality classes for statistical physics systems undergoing a second order phase transition at their critical point. In particle physics Quantum Field Theories are expected to tend to CFTs in the limit of small scales and 
 two dimensional CFTs are also  building blocks for string theory.  The success of CFTs in physics owes to the fact that their study gave birth to the theory of conformal bootstrap \cite{P,FGG}. In the groundbreaking work by Belavin-Polyakov-Zamolodchikov \cite{BPZ}, this method coupled with the Virasoro algebra symmetry was used to solve important two dimensional CFTs by   expressing their correlation functions in terms of representation theoretical special functions.
Ever since this work CFT has also had a deep influence on mathematics with applications  in modular forms, representation theories of  infinite-dimensional Lie algebras and vertex algebras, Monstrous Moonshine, geometric Langlands theory and knot theory to mention a few. Laying the rigorous foundations  of CFTs also  triggered interest among mathematicians. Several axiomatic setups  have been proposed, each of them relying mainly on one of the various aspects of CFTs. The treatment of CFTs as Vertex algebras by Kac \cite{kac} grew out of the representation theoretical structure of CFTs whereas Frenkel and Ben-Zvi \cite{frenkel} formalized CFTs within the context of algebraic geometry, and this was further expanded by Beilinson and Drinfeld \cite{drinfeld}.  The lecture notes \cite{gawedzki} emphasized the probabilistic approach based on the Feynman path integral, closer in spirit to statistical physics. Segal \cite{segal} proposed an axiomatic definition to CFTs  inspired by the path integral approach and designed to capture geometrically the conformal bootstrap for CFTs. 

In spite of various axiomatic definitions of CFTs, finding concrete examples obeying any given set of axioms has proven more challenging and is, actually, a rare event in mathematics. For long, the preferential approach has been mostly algebraic but, even though advances were made,   even the most basic examples of CFTs, the minimal models, are still not fully constructed mathematically (see \cite{bingui} for recent advances in the context of Vertex Operator Algebras, in particular concerning their chiral part). More recently, the probabilistic formulation has led to the first proof of conformal bootstrap in a non-trivial CFT: the path integral for the Liouville CFT, the prototype of unitary noncompact CFT, was constructed in \cite{DKRV,GRV} and it was later shown to satisfy Segal's axioms for CFT and obey the conformal bootstrap  \cite{GKRV,GKRV2}.  

Whereas  the connection of Liouville CFT to statistical physics is indirect, via the the so called KPZ conjecture \cite{KPZ} that relates a statistical physics model on a random lattice to the same model on a non random lattice coupled with Liouville CFT,  an imaginary version of Liouville CFT is expected in physics to describe a major class of two dimensional statistical physics models, the so-called loop models.
 Given a    compact Riemann surface $\Sigma$ with Riemannian metric $g$, this CFT is formulated in terms of an action functional
\begin{equation}\label{defintroaction}
S_{\rm L}(\phi,g):=\frac{1}{4\pi}\int_\Sigma\big(|\dd \phi|_g^2+iQK_g\phi+\mu e^{i\beta \phi}\big)\,\dd {\rm v}_g,
\end{equation}
where $\mu\in \C$, $Q,\beta\in\R$ are parameters, ${\rm v}_g$ and $K_g$ are respectively the volume form and the scalar curvature in the metric $g$, and  $\phi:\Sigma\to \T_R$  is a map taking values into the circle $\T_R:=\R/(2\pi R\Z)$ of radius $R$. Statistical averages of functionals $F(\phi)$ of the map $\phi$ are  then formally given by a  path integral  over some space  of maps $\phi:\Sigma\to\T_R$
\begin{equation}\label{defintroPI}
\langle F(\phi)\rangle_{\Sigma,g}:=\int_{\phi:\Sigma\to\T_R}F(\phi)e^{-S_{\rm L}(\phi,g)}\,D\phi.
\end{equation}
If the parameters $\mu$ and $Q$ are set to $0$ then this path integral gives rise to the compactified Gaussian Free Field (or compactified GFF for short, see \cite[Lecture 1.4]{gawedzki} or \cite[subsections 6.3.5 and 10.4.1]{DMS} for physics references, or  \cite[subsection 2.1.3]{dubedat} in mathematics), a CFT with central charge 1. The   case $\mu=0$ goes in physics under the name of Coulomb gas, in which case the so-called background charge (the curvature term in the action) shifts the central charge of the compactified GFF to the value ${\bf c}=1-6Q^2$. The  exponential term gives rise to a nontrivial interacting CFT provided the value of $Q$ is fixed to
\begin{equation}\label{defQ}
Q=\frac{\beta}{2}-\frac{2}{\beta}.
\end{equation}
In this paper we give a probabilistic construction of the path integral \eqref{defintroPI} and prove that the resulting theory satisfies Segal's axioms of conformal field theory. We also give an explicit integral representation for the structure constants (i.e. 3 point correlation functions) of this CFT  in terms of (a generalization of) the famous Dotsenko-Fateev integrals. These results provide the foundation for proving conformal bootstrap for the theory to be discussed in later works. Let us now briefly discuss the mathematical and physical motivations for studying this theory.

\subsection{Logarithmic CFT}
There are two major differences with  the action functional \eqref{defintroaction} and  that of   Liouville CFT. In the latter one takes $i\beta\in\R$ and the field $\phi$ takes values in $\R$ in contrast to the circle in \eqref{defintroaction}. These two changes have a drastic 
effect on the algebraic, geometric and probabilistic aspect of this CFT.
First it is a  \textbf{non-unitary} (i.e. non-reflection positive) CFT.  This means its Hamiltonian, obtained as a generator of dilations in a canonical Hilbert space of the theory which we construct in this paper, is \textbf{non self-adjoint}. Second
and most importantly, it gives rise to a  \textbf{logarithmic CFT}, topic which has been under active research   in physics \cite{feigin,fuchs,gurarie,mathieu,nives,raoul,read,ruelle,zuber}  and in mathematics \cite{liu,kytola1,kytola2}.  Concretely, this means 
that  the Hamiltonian  should be diagonalizable in Jordan blocks. Important questions in this context are the classification of the reducible but indecomposable representations of the Virasoro algebra and the structure of the conformal bootstrap, for which the  holomorphic factorization is no more expected at least under its standard form.
Third, we prove this CFT possesses   \textbf{non scalar primary fields}, namely they  possess a spin. Fourth  it  is    \textbf{non-rational and non diagonal}: its spectrum  is countable  and different representations of the Virasoro algebra are involved in the chiral and anti-chiral parts entering the correlation functions. Hence our path integral provides the foundations to the mathematical study of all these concepts.


Beyond the CFT aspects, we also expect this CFT will have a strong interplay with   Schramm Loewner Evolution and Conformal Loop Ensembles. Indeed, this was advocated first in physics (see for instance \cite{IJK} or the review \cite{reviewCC})  and a similar interplay occurred  recently in the case of  the Liouville CFT, in which case  these two objects were bridged by the mating-of-trees formalism\footnote{A  framework to study the coupling between SLE or CLE  and Liouville CFT in terms of more classical probabilistic objects such as Brownian motion, Levy process, and Bessel process.}  \cite{dms}. This has led to mutual better understanding of these objects \cite{ahs21,arsz23,as21} as well as  unexpected results for various statistical physics models \cite{ars22,pierre} and the same is expected in case of the CFT studied in this paper.


\subsection{Physics: from loop models to the Coulomb gas and imaginary Liouville theory}
 
 In the approach of \cite{BPZ} 
 the critical exponents of a statistical physics model were related to the representation theory of the Virasoro algebra of the  CFT conjectured to describe its scaling limit. 
Progress in the physical understanding of two-dimensional critical phenomena has therefore been obtained on the one hand by algebraically 
 classifying CFTs and on the other hand by associating them with putative scaling limits of concrete lattice models \cite{BPZ,cardy}. It has been known since the 70's that the critical points of many models of statistical physics can be described in terms of the GFF, dubbed the Coulomb gas representation  (see \cite{cardy,FSZ,nienhuis} for instance). This picture turned out to be particularly suited to describing the scaling limit of so called loop models. These are models for a random ensemble of loops (closed paths) on a two dimensional lattice and they could be mapped to models of fluctuating surfaces described by a height function $h$ mapping the lattice points to $\Z$ or $\R$, the basic idea being to think of loops as contour lines of the random surface.   The scaling limit of $h$ was then argued to be related to GFF. The variance of the GFF (here rescaled to the parameter  $\beta$ in   \eqref{defintroaction}) was then fixed by a priori knowledge of  an exact value of some critical exponent, see \cite{cardy} for instance.

However, for the full description of the loop models  the naive GFF theory needs to be modified  in several ways: first the Gaussian height function has to be taken periodic i.e. to take values on a circle instead of the real line, giving rise this way to the compactified GFF described before. Secondly, to obtain the expected central charge $c$ of the CFT one needs to include the curvature term in \eqref{defintroaction} which shifts $c$ from the GFF value $c=1$ to $c=1-6Q^2$. Finally, later on it was argued in \cite{kondev,JK,KGN}\footnote{See also \cite{KH} for multicomponent compactified bosons with Kac-Moody symmetry.} (see also the review paper \cite{reviewCC}) that the action should also include the  Liouville potential $e^{i\beta\phi}$ and requiring it to be a marginal perturbation  (and thus potentially giving rise to a CFT with the same $c$) then fixes the value of $\beta$ in terms of $Q$ as prescribed by \eqref{defQ}. 
  This is how physicists ended up with the path integral \eqref{defintroPI} with action functional \eqref{defintroaction} to describe the long-distance properties of  self-avoiding loop models. For instance, the critical q-state Potts model on a connected planar graph $G=(V,E)$ where $V$ is a set of vertices and $E$ the set of edges has  the  following loop gas representation of its partition function \cite{baxter} 
 $$
 Z=q^{|V|/2} \sum_{\rm loops } x^{|E'|}q^{\ell(E')/2},$$
where the sum runs over subgraphs $(V,E')\subset (V,E)$ and  $\ell(E')$ being the number of loops in the loop configurations on the medial graph. Here $x>0$, and $q^{1/2}$ appears as a formal parameter and can thus be assigned complex values. At the critical point $x=x_c$, the  conjectural relation  with the path integral is
$$ q^{1/2} =-2\cos(\pi \beta^2/4)$$
for $0<\beta<2$, with corresponding central charge 
$c=1-6(\frac{\beta}{2}-\frac{2}{\beta})^2.$
A similar correspondence exists with the $O(n)$-model,  see \cite{cardy2,reviewCC,rychkov}\footnote{A further  orbifold of our path integral might be involved: this question is under investigation in physics but the mathematical construction of the path integral is quite similar.}.  These conjectures are  important but still far from being understood in mathematics. Let us stress that in the different context of the dimer model,   yet of the same flavour, there has been a long series of works studying the convergence towards the compactified GFF, see \cite{K1,K2, dubedat,BDT,BLR1,BLR2}. We will not develop any further this aspect in the paper. Our main goal is  here to construct mathematically the path integral \eqref{defintroPI}, \eqref{defintroaction} and establish the basic CFT axioms.

Let us close this introduction by stressing that one important byproduct of our work is a mathematical construction of the Coulomb gas and we feel that some foundational aspects of this fundamental model have been a bit overlooked in the physics literature, which has some impact on the conclusions drawn for the path integral \eqref{defintroPI}. Our concerns come from two modifications needed  to define   the curvature term: one is topological and the other one comes from the windings of the electro-magnetic operators. The topological modification for the curvature was already considered in \cite{verlinde} but to our knowledge the other problem has not been addressed before. However, even in \cite{verlinde}, the main properties of the  modification are not addressed in detail and it turns out that the resulting path integral is not diffeomorphism invariant unless $Q\in \frac{1}{R}\Z$ (in the case considered in \cite{verlinde} this holds). In the case of the Coulomb gas (hence with no potential), for any $Q$ we can choose the radius $R$ so that $Q\in \frac{1}{R}\Z$, and this  condition is always assumed in physics.  But in the case of the imaginary Liouville path integral, the presence of the potential forces $\beta\in \frac{1}{R}\Z$ as well. If we further impose $Q\in \frac{1}{R}\Z$, as we do in the present paper, then this forces the central charge to be of the form ${\bf c}=1-6\frac{(p-q)^2}{pq}$ for some (relatively prime)  integers $p,q$  which matches the possible set of central charges for minimal models. Yet, in physics, the path integral is also considered for irrational values of the central charge, in particular $Q\notin \frac{1}{R}\Z$,  and then  the path integral violates the most basic requirement (diffeomorphism invariance) of a Quantum Field Theory. We haven't found any reference to this fact in the literature except for a remark by  A.B. Zamolodchikov in  \cite{Za}   for a somewhat related model (Generalized Minimal Models), quoting: `` It remains a question if such infinite algebra is
consistent with general requirements of quantum field theory. In particular, the construction of a modular invariant partition function of Generalized Minimal Models  obviously encounters severe problems." Yet the path integral can be constructed in this case, see Section \ref{irrational}, even though it exhibits unstandard features.
  
\section{Outline and main results}

 Because the construction of the path integral requires some geometric background, we sketch here the construction, as a guideline for the paper.  First we discuss informally some aspects of the construction  at a general level but we will then give two concrete   examples with simple topologies to introduce pedagogically two important notions of the construction: the defect graph and the topological instantons.
 
 To make sense of the measure \eqref{defintroPI} with action \eqref{defintroaction}, it is first necessary to understand the path integral for the compactified GFF (i.e. with $Q=\mu=0$ in  \eqref{defintroaction}). For this, one observes that the differential $\dd \phi$ of a  (say smooth) map $\phi:\Sigma\to \T_R$ defines a real valued closed $1$-form $\omega$ on $\Sigma$, with periods $\int_\gamma \omega\in 2\pi R\Z$ if $\gamma$ is any closed curve on $\Sigma$. The Hodge decomposition then allows to uniquely decompose  $\omega=\omega_h+\dd f$, where   $f:\Sigma\to \R$ is a function
and $\omega_h$ is a harmonic 1-form (the De Rham cohomology group is represented by harmonic 1-forms). In other words the $\T_R$ valued map $\phi$ can be viewed as a sum of $f$ and a   multivalued harmonic function $\int_{\gamma_{x_0,x}} \omega:=I_{x_0}(\omega_h)$ on $\Sigma$ where $\gamma_{x_0,x}$ is a path from $x_0$ to $x$. 
This decomposition is orthogonal in the space $\Omega^1(\Sigma)$ of real valued 1-forms on $\Sigma$ equipped with the inner product 
\begin{equation}
\langle \omega,\omega'\rangle_2:=\int_\Sigma  \omega\wedge *\,\omega',
\end{equation}
where $*$ is the Hodge operator, in such a way that
$$ \int_\Sigma|\dd \phi|_g^2 \,\dd {\rm v}_g=\int_\Sigma|\dd f|_g^2 \,\dd {\rm v}_g+ \int_\Sigma  \omega_h\wedge *\, \omega_h.$$
Therefore the formal measure in \eqref{defintroPI} for the compactified GFF can be understood as the measure
\begin{align}\label{boson}
e^{-\frac{1}{4\pi} \int_\Sigma|\dd \phi|_g^2 \,\dd {\rm v}_g}D\phi
=
e^{-\frac{1}{4\pi}\int_\Sigma|\dd f|_g^2 \,\dd {\rm v}_g}Df\times \,e^{-\frac{1}{4\pi}     \int_\Sigma  \omega_h\wedge \star\, \omega_h     }\dd \mu(\omega_h)\times\,\dd c
\end{align}
where  
$\dd \mu(\omega_h)$ is  the counting measure on the De Rham cohomology group (isomorphic to $\Z^{2{\mathfrak{g}}}$ where ${\mathfrak{g}}$ is the genus of $\Sigma$), $dc$ is the Lebesgue measure on $\R/2\pi R\Z$ ($c$ plays the role of the zero mode)  and the formal measure $e^{-\frac{1}{4\pi}\int_\Sigma|\dd f|_g^2 \,\dd {\rm v}_g}Df$ is interpreted in our work as the distribution law of the usual  Gaussian Free Field on $\Sigma$. The functional $F(\phi)$ in \eqref{defintroPI} then  has to be well chosen so that  the result does not depend on the ambiguity related to the fact that $I_{x_0}(\omega_h)$ is multivalued on $\Sigma$ (see Section \ref{secPI}). This aspect is quite subtle and  it manifests, for instance,  as soon as one wishes to make sense of the curvature term in \eqref{defintroaction}. We expand this point a bit further now.

There are actually two main difficulties with the curvature term. First, since the  zero mode  $c$ lives on the circle $\R/2\pi R\Z$, the curvature term $\frac{Q}{4\pi}\int_\Sigma \phi K_g\dd {\rm v}_g$ has to be $2\pi R$-periodic in the variable $c$. By Gauss-Bonnet $\frac{1}{4\pi}\int_\Sigma  K_g\dd {\rm v}_g=2-2{\mathfrak{g}}$ 
so this forces $2Q$ to be an integer multiple of $1/R$\footnote{Actually, this point could be circumvented by adding artificial (electric) operators as is usually done in the physics literature.}. The same argument holds for the Liouville potential and this forces $\beta$ to be an integer  multiple of $1/R$. Second, and more  subtly, the curvature term contains the primitive in  $\phi=c+f+I_{x_0}(\omega_h)$, which is a multivalued function on $\Sigma$. Therefore the integral  $\int_\Sigma \phi K_g \dd {\rm v}_g$ does not make sense  unambiguously; even worse, 
if we define the primitive on the universal cover, and then descend  it on a fundamental domain,  the quantity $\int_\Sigma \phi K_g \dd {\rm v}_g$ will depend on the choice of the fundamental domain. That may sound like mathematical nonsense, but concretely, this means that a naive definition of the curvature term produces a theory that is not diffeomorphism invariant. A considerable part of our work consists in regularizing this integral via branch-cuts on $\Sigma$ and in proving that the result does not depend on the choice of the branch-cuts.  The result is that a change of $\sigma$ changes the regularised  curvature integral by an integer multiple of $2\pi R$ and thus one has to impose $QR\in\Z$. Together with $\beta R\in\Z$ and 
 the expression for $Q$ given by \eqref{defQ} we conclude $\beta^2$ has to be rational and hence also that the central charge $\bf c$ will be rational. Hence a diffeomorphism invariant theory exists only for rational values of  $\beta^2$ and   $\bf c$.

Concerning the Liouville potential term, the construction  involves some technology related to the renormalization of imaginary Gaussian Multiplicative Chaos (GMC for short), in particular some tools developed in \cite{LRV15,LRV19}, and is restricted to $0<\beta^2<2$. The full picture of this path integral is expected to cover the values $0<\beta^2<4$ in which case further renormalisations  (in addition to the Wick ordering here) are needed, reminiscent of the Sine-Gordon model. We stress that these renormalizations need to cover not only the construction of the partition function but also all the correlation functions; a task that has not been completed yet for   the Sine-Gordon model. Yet, this is an important future direction of development.  
The construction  for $0<\beta^2<2$ rational of the path integral and the correlation functions is carried out in Section \ref{secPI}. We give now more details in two simple cases.
 
 \subsection{ Construction on the Riemann sphere.}
 Now we give a more concrete  construction on the simplest possible topology, the Riemann sphere, because the compactified GFF coincides then with the standard GFF. We first give the construction of the path integral without electro-magnetic operators. Consider a smooth Riemannian metric $g$ on the Riemann sphere $\hat\C$, with volume form ${\rm v}_g$ and Ricci curvature $K_g$. The GFF in the metric $g$ is denoted by $X_g$ (see section \ref{IGMC} for details). We consider 
 $\beta,Q\in \frac{1}{R}\Z$ and $Q=\frac{\beta}{2}-\frac{2}{\beta}$. 
 The construction involves imaginary GMC, namely the limit
 $$M_\beta(X_g,\dd x):=\lim_{\epsilon\to 0}\epsilon^{-\frac{\beta^2}{2}}e^{i\beta X_{g,\epsilon}(x)}\dd x$$ where $X_{g,\epsilon}$ is a  reasonable regularization of the GFF at scale $\epsilon$ in the metric $g$. The convergence is non trivial for $\beta^2\in (0,2)$ and the limit is a distribution of order 2 with exponential moments. More generally, we will consider imaginary GMC $M_\beta(\phi_g,\dd x)$ where we replace the GFF $X_g$ by a more general functional of the GFF, denoted by $\phi_g$ and called the \emph{Liouville field}, which will be made precise depending on the context. The definition of the path integral is then (up to some  metric dependent terms encoded in the constant $C(g)$, the precise meaning of which we skip for now)
 \begin{equation}\label{defPIintro}
   \langle F\rangle_{\hat\C,g } 
   :=
C(g) \int_{\R/2\pi R\Z}\E\Big[ F(\phi_g)e^{-\frac{i   Q}{4\pi}\int_{\hat\C }  K_g\phi_g\,{\rm dv}_g -\mu  M^g_\beta(\phi_g,\hat\C)}\Big]\,\dd c
\end{equation}
 for all reasonable test functions $F$, periodic with period $2\pi R\Z$ in the variable $c$, and here the Liouville field is $\phi_g=c+X_g$. Note the integration $\dd c$ of the constant mode $c$ of the field $\phi_g$  over the   circle of radius $R$. Also, the curvature term here perfectly makes sense because, at this stage, the Liouville field is a well defined distribution on $\hat\C$ and can be integrated over $\hat\C$.
 
\subsubsection{Electric operators}  
Next we introduce {electric operators}: given a point $x\in\hat\C$ on the sphere and a weight $\alpha\in R^{-1}\Z$, they are formally defined by 
  \[V_{(\alpha,0)}(x)=\lim_{\epsilon\to 0}\epsilon^{-\frac{\alpha^2}{2}}e^{i\alpha\phi_{g,\epsilon}(x)}.\] 
  The condition $\alpha\in R^{-1}\Z$ makes sure that the electric operator, seen as a function of $c$, is well defined on the circle. Products of such operators for distinct points ${\bf x}=(x_1,\dots,x_n)\in\hat\C^n$ with respective weights $\boldsymbol{\alpha}=(\alpha_1,\dots,\alpha_n)\in (R^{-1}\Z)^n$ will be denoted by $V_{(\boldsymbol{\alpha},0)}({\bf x}):=\prod_{j=1}^nV_{(\alpha_j,0)}(x_j)$. Correlation functions for electric operators are formally given by evaluating the path integral above with $F(\phi_g)=V_{(\boldsymbol{\alpha},0)}({\bf x})$, namely   $\langle V_{(\boldsymbol{\alpha},0)}({\bf x})\rangle_{\hat{\C},g}$.  Of course, the limit $\epsilon\to 0$ is ill defined because, making sense only as a distribution, the GFF cannot be evaluated pointwise. To remedy this, the usual trick is to apply the Girsanov transform; special care has to be taken here because of the imaginary weights so that rigorously implementing the Girsanov transform has to go through analytic continuation arguments (see Section \ref{sec:correls}). The outcome is that the path integral with electric operators is (again, up to explicit trivial factors  $C(g,{\bf x},\boldsymbol{\alpha})$)
  \begin{equation}\label{defPIelec}
   \langle FV_{(\boldsymbol{\alpha},0)}({\bf x})\rangle_{\hat\C,g } 
   :=
 C(g,{\bf x},\boldsymbol{\alpha}) \int_{\R/2\pi R\Z}e^{i c\sum_{j=1}^n\alpha_j}\E\Big[ F(\phi_g)e^{-\frac{i   Q}{4\pi}\int_{\hat\C }  K_g\phi_g\,\dd v_g -\mu  M^g_\beta(\phi_g,\hat\C)}\Big]\,\dd c
\end{equation}
 where the Liouville field is now $\phi_g=c+X_g+u_{\bf x}$, and the function $u_{\bf x}(x):=\sum_{j=1}^ni\alpha_jG_g(x,x_j)$ (with $G_g$ the Green function on $\hat\C$ in the metric $g$) stands for the shift resulting from the Girsanov transform. Observe that this shift creates singularities in the potential $M^g_\beta(\phi_g,\Sigma)$, which can be formally written as 
 $e^{i\beta c}\int_{\hat\C} e^{-\beta\sum_j\alpha_jG_g(x,x_j)}M_\beta(X_g,\dd x)$. Well-posedness  and existence of exponential moments of this random variable are part of our work.
 
\subsubsection{Magnetic operators.} The next step is to introduce {magnetic operators}. Magnetic operators have the effect to make the Liouville field $\phi_g$ acquire a winding when turning  around some prescribed points. More precisely, consider distinct points ${\bf z}=(z_1,\dots,z_n)\in\hat\C^n$ with respective magnetic charges $\boldsymbol{m}=(m_1,\dots,m_n)\in\Z^n$. Then we consider a closed 1-form $\nu_{{\bf z},\boldsymbol{m}}$ on $\hat\C\setminus \{{\bf z}\}$ (with $ \{{\bf z}\}:=\cup_{j=1}^n \{z_j\}$) such that $\nu_{{\bf z},\boldsymbol{m}}$ is of the form $m_jR \dd\theta$  in local radial coordinates $z-z_j=re^{i\theta}$ near the point $z_j$, for $j=1,\dots,n$. Such a 1-form exists if and only if the charges satisfy $\sum_{j=1}^nm_j=0$. Also, note that the choice of this closed 1-form is not unique but the path integral we will construct will not depend on this choice. Then we want to offset the GFF with a primitive of this 1-form, namely we want to consider the path integral \eqref{defPIelec} where the Liouville field is now 
 \[\phi_g=c+X_g+u_{\bf x}+I_{x_0}(\nu_{{\bf z},\boldsymbol{m}}),\] 
 and $I_{x_0}(\nu_{{\bf z},\boldsymbol{m}}):=\int_{x_0}^x\nu_{{\bf z},\boldsymbol{m}}$ is a primitive of the 1-form $\nu_{{\bf z},\boldsymbol{m}}$ with base point $x_0$. The point is that $\nu_{{\bf z},\boldsymbol{m}}$ is not exact (and $\hat\C\setminus \{{\bf z}\}$ is of course not simply connected) so that the primitive is multivalued on $\hat\C\setminus \{{\bf z}\}$: it has a monodromy $2\pi m_jR$ when turning once around a given point $z_j$.  The first important remark is that any function of the form $e^{i \frac{k}{R}I_{x_0}(\nu_{{\bf z},\boldsymbol{m}}) }$, for $k\in\Z$, is unambiguously defined as a smooth function on $\hat\C\setminus \{{\bf z}\}$. In particular,   since $\beta$ is an integer multiple of $1/R$, the potential term in \eqref{defPIelec} is well defined as it reads 
 $$e^{i\beta c}\int_{\hat\C} e^{-\beta\sum_j\alpha_jG_g(x,x_j)}e^{i\beta I_{x_0}(\nu_{{\bf z},\boldsymbol{m}})(x)}M_\beta(X_g,\dd x). $$
 Also, and at the level of electric operators, offsetting the field by $I_{x_0}(\nu_{{\bf z},\boldsymbol{m}})$ produces a further factor of the form $\prod_{j=1}^ne^{i\alpha_j I_{x_0}(\nu_{{\bf z},\boldsymbol{m}})(x_j)}$; again, the conditions $\alpha_j\in R^{-1}\Z$ make this product unambiguously defined. Actually the main problem comes from the curvature term in \eqref{defPIelec}: the monodromy of the field $\phi_g$ makes this term being ill-defined and depending on the choice of the primitive. This term must be regularized. For this, we first introduce a system of branch cuts for $I_{x_0}(\nu_{{\bf z},\boldsymbol{m}})$, which we call \emph{defect graph}. It consists in a family of smooth   simple curves $(\xi_p)_{p=1,\dots,m-1}$, which do not intersect except eventually at their endpoints. Each arc $\xi_p:[0,1]\to \hat\C$ is a smooth oriented curve parametrized by arclength with endpoints $\xi_p(0)=z_j$ and $\xi_p(1)=z_{j'}$ for $j\not= j'$ with orientation in the sense of increasing charges, meaning $m_j\leq m_{j'}$. We must further impose a direction along which the arcs reach the endpoints. So we fix a family of unit tangent vectors $v_j\in T_{z_j}\hat \C$ ($j=1,\dots,n$) and we further require these arcs to obey: if $\xi_p(a)=z_j$ with $a\in \{0,1\}$, then  $\xi_p'(a)$ is positively colinear to $v_j$.    Finally, consider the  oriented  graph with vertices $\{\mathbf{z}\}$ and edges $(z_j,z_{j'})$ when there is an arc connecting $z_j$ to $z_{j'}$. This graph must be connected and without cycle. What we call defect graph is the set $\mc{D}:=\bigcup_{p=1}^{m-1}\xi_p([0,1])$, see Figure \ref{fig:defectintro}. 

\begin{figure}
\includegraphics[scale=0.5]{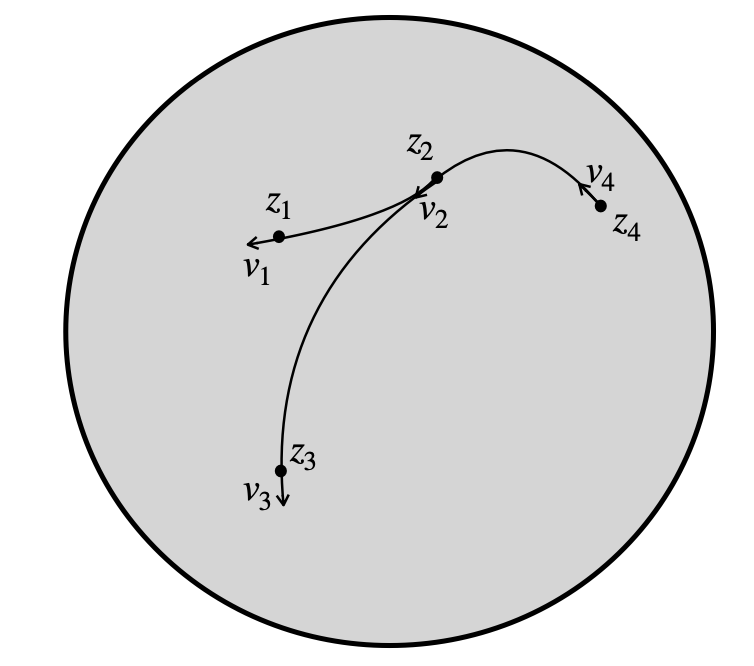}
\caption{A defect graph with 4 points.} 
\label{fig:defectintro}
\end{figure}

The defect graph is a proper branch cut for $\nu_{{\bf z},\boldsymbol{m}}$ in the sense that this 1-form is exact on the complement $\hat\C\setminus \mc{D}$. Therefore it admits an unambiguously defined primitive denoted by $I^{ \boldsymbol{\xi}}_{x_0}(\nu_{{\bf z},\boldsymbol{m}})$. Then we introduce the regularized curvature term as
 $$\int_{\hat\C }^{\rm reg}  K_g\phi_g\,{\rm dv}_g:=\int_{\hat\C}  K_g(c+X_g+u_{\bf x})\,{\rm dv}_g+ \int_{\Sigma}^{\rm reg}  I^{ \boldsymbol{\xi}}_{x_0}( \nu^{\rm h}_{\mathbf{z},\mathbf{m}}) K_g \,{\rm dv}_g$$
 where
 \begin{equation}\label{regintro}
\int_{\hat\C}^{\rm reg}  I^{ \boldsymbol{\xi}}_{x_0}( \nu_{\mathbf{z},\mathbf{m}}) K_g \,{\rm dv}_g:=\int_{\Sigma}  I^{ \boldsymbol{\xi}}_{x_0}( \nu_{\mathbf{z},\mathbf{m}}) K_g \,{\rm dv}_g-2\sum_{p=1}^{n}\kappa(\xi_p)\int_{\xi_p}k_g\dd \ell_g
\end{equation}
where $k_g$ denotes the geodesic curvature of an oriented curve, $\dd \ell_g$ the Riemannian measure on $\xi_p$ induced by $g$ and $\kappa(\xi_p)$ are coefficients that have explicit expressions in terms of the magnetic charges; their value is actually imposed by an argument relying on the Gauss-Bonnet theorem (see Subsection \ref{curvature_mp}). The definition of the path integral with both electric and magnetic operators is then
  \begin{align}\label{defPIboth}
   &\langle FV_{(\boldsymbol{\alpha},0)}({\bf x})V_{(0,\boldsymbol{m})}({\bf v})\rangle_{\hat\C,g } \\
 &  :=
 C(g,{\bf x},{\bf z},\boldsymbol{\alpha},\boldsymbol{m}) \int_{\R/2\pi R\Z}e^{i c\sum_{j=1}^n\alpha_j}\E\Big[e^{-\frac{1}{2\pi}\langle \dd X_g, \nu_{\mathbf{z},\mathbf{m}}\rangle_2} F(\phi_g)e^{-\frac{i   Q}{4\pi}\int^{\rm reg}_{\Sigma }  K_g\phi_g\,\dd v_g -\mu  M^g_\beta(\phi_g,\Sigma)}\Big]\,\dd c\nonumber
\end{align}
with $\phi_g=c+X_g+u_{\bf x}+I_{x_0}(\nu_{{\bf z},\boldsymbol{m}})$ and ${\bf v}=((z_1,v_1),\dots,(z_{n_{\mathfrak{m}}},v_{n_{\mathfrak{m}}}))\in (T\hat\C)^{n}$, and $ C(g,{\bf x},{\bf z},\boldsymbol{\alpha},\boldsymbol{m})$ is a constant encoding trivial factors (in particular it contains the product $\prod_{j=1}^ne^{i\alpha_j I_{x_0}(\nu_{{\bf z},\boldsymbol{m}})(x_j)}$). What is non trivial is to show that this regularized curvature does not depend on the choice of the defect graph: this is proved in subsection  \ref{curvature_mp}.

Electro-magnetic operators, denoted by $ V_{(\boldsymbol{\alpha},\boldsymbol{m})}({\bf v})$ are then defined by merging both electric and magnetic operators, namely by taking the limit ${\bf x}\to {\bf z}$ in \eqref{defPIboth}. Yet, this limit is ill-defined because of the windings around the points ${\bf z}$. This is why we must further impose the direction along which each $x_j$ reaches $z_j$, and we require this limit to be taken in the direction $v_j$, in which case we prove that the limit makes sense and defines this way correlation functions of electro-magnetic operators  $\langle V_{(\boldsymbol{\alpha},\boldsymbol{m})}({\bf v})\rangle_{\Sigma,g } $. Furthermore, we obtain an explicit description how these correlation functions  are affected by a change of the choice of the unit vectors $v_j$, which translates the notion of spin for primary fields in the physics literature.
 
\begin{figure}[h]
         \begin{subfigure}[b]{0.45\textwidth}
         \includegraphics[width=1\textwidth]{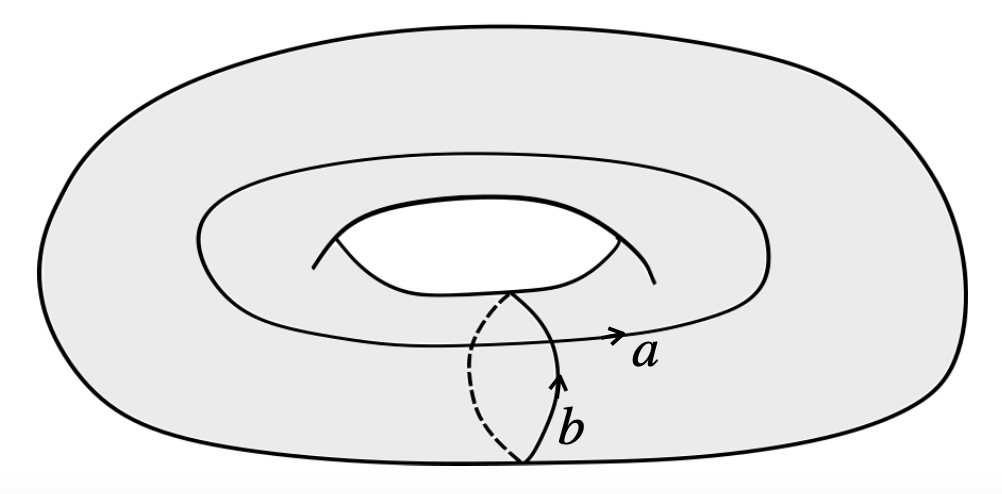}
         \caption{Torus with its two homology cycles.}
         \label{fig:three sin x}
     \end{subfigure}
    \hfill
     \begin{subfigure}[b]{0.54\textwidth}
         \includegraphics[width=1\textwidth]{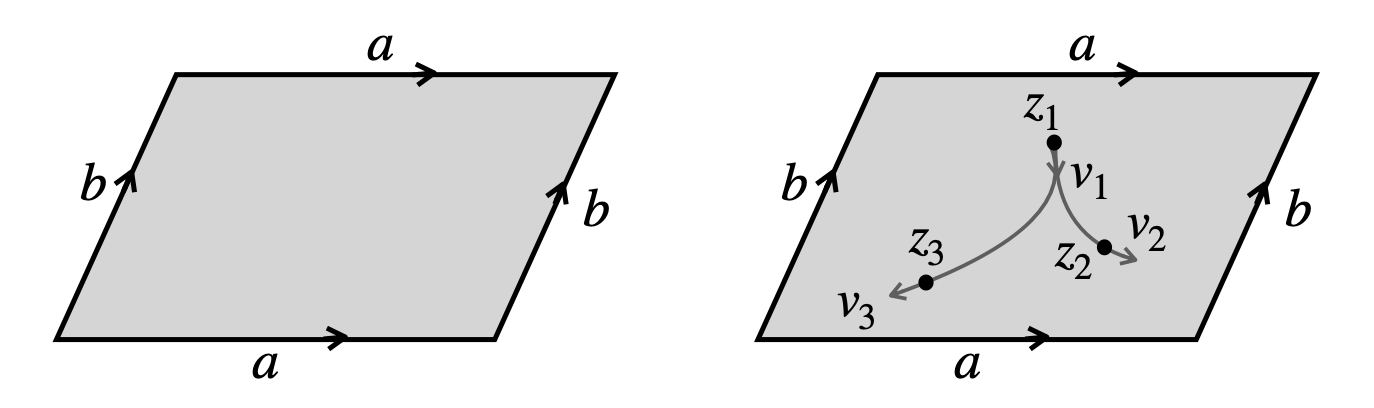}
         \caption{Planar representation of $\T_\tau$. On the right-side, the torus $\T_\tau$ with a defect graph.}
     \end{subfigure}
        \caption{Homology cycles and defect graph on the torus.}
        \label{fig:torus}
\end{figure}
 

\subsection{Other topologies.}
Now we sketch the philosophy of the construction for more complex topologies, and we focus here on the case of complex tori. As soon as the genus of the surface is non zero, the compactified GFF integration measure \eqref{boson} has a so called  instanton component  corresponding to the summation over the De Rham cohomology.  Concretely, the torus $\T_\tau=\C/(\Z+\tau\Z)$ with modulus $\tau=\tau_1+i\tau_2\in\C$ (with $\tau_2>0$) 
has a basis of homology given by the cycles $a(t)=t$ and $b(t)=t\tau$ for $t\in[0,1]$. A dual basis of cohomology is given by 
$\omega_1=\dd x-\frac{\tau_1}{\tau_2}\dd y$ and $\omega_2= \frac{1}{\tau_2} \dd y$. For ${\bf k}=(k_1,k_2)\in\Z^2$ we set $\omega_{\bf k}=k_1\omega_1+k_2\omega_2$. Then, the path integral basically corresponds to \eqref{defPIintro} with the Liouville field given by 
\[\phi_g=c+X_g+I_{x_0}(\omega_{\bf k}),\] 
and a further summation over ${\bf k}\in \Z^2$. Again, the problem here is that the curvature term is ill-defined\footnote{Note that we do not assume the curvature is uniformized with $K_g=0$.} as there is an arbitrary choice to be made for the primitive $I_{x_0}(\omega_{\bf k})$, which is multivalued on $\Sigma$, and the resulting integral $\int_{\Sigma } K_gI_{x_0}(\omega_{\bf k})\,{\rm dv}_g $ does depend on this choice. So the curvature term has to be regularized: the cycles $a,b$ are chosen as branch cuts and then the 1-form $\omega_{\bf k}$ is exact on $\Sigma\setminus (a\cup b)$. One needs then to introduce counterterms in a way similar to \eqref{regintro} to finally get the path integral
\begin{equation}
  \langle F\rangle_{\T_\tau,g } 
   :=
C(g)\sum_{{\bf k}\in\Z^{2}} e^{-\frac{1}{4\pi}\| \omega_{\bf k}\|_2^2}\int_{\R/2\pi R\Z}\E\Big[e^{-\frac{1}{2\pi}\langle \dd X_g, \omega_{\bf k}\rangle_2} F(\phi_g)e^{-\frac{i   Q}{4\pi}\int_{\T_\tau }^{\rm reg}  K_g\phi_g\,\dd v_g -\mu  M^g_\beta(\phi_g,\T_\tau)}\Big]\,\dd c.
\end{equation}
On general Riemann surfaces, a considerable part of our work consists in proving that the path integral does not depend on these branch cuts, with the invariance under diffeomorphisms as aconsequence.  One can then define electro-magnetic operators on $\T_\tau$ as was done for the Riemann sphere, introducing  further branch cuts for the magnetic operators (see Figure \ref{fig:torus}). This concludes our short overview of the construction.

\subsection{Main results}
Let us expand now in further details the properties of this CFT. First we stress that the correlation functions of interest are expectation  values  of the electro-magnetic fields, the construction of which is summarized above. In the Coulomb gas picture, an electric operator with weight $\alpha=\frac{e}{R}$ creates   an electric charge $e\in \Z$ 
at the insertion point $x\in\Sigma$. A magnetic operator in turn creates  a magnetic charge $m\in \Z$ at its insertion point $z$. Its effect in the loop model is to create a discontinuity $2\pi R m$ in the height field $\phi$ along a line emanating from $z$. This means the field $\phi$ has a winding $2\pi Rm$ around the point $z$. The discontinuity curve has to end to the location of another magnetic charge and if we have charges $m_i$ at points $z_i$ for $i=1,\dots, n,$ we must impose neutrality condition $\sum_im_i=0$. Electro-magnetic operators $V_{\alpha,m}({z, v})$ mix the two effects.
They are the primary fields of the compactified imanginary Liouville theory; they are not scalar as they possess a spin, details can  be found in Section  \ref{sec:correls}.

 In CFTs, the 3 point correlation functions on the Riemann sphere, or structure constants,  play a special role as building blocks of the conformal bootstrap formulae.  In Section \ref{sec:3point} we compute the three-point functions of the electro-magnetic operators  and show them to be given by a generalization of the well known Dotsenko-Fateev integrals. In the case when the magnetic charges are set to $0$, these integrals are given by the imaginary DOZZ formula \cite{Za}. A similar explicit expression for these integrals in the presence of magnetic charges seems not to be known and is an interesting open question.

Finally we complete the CFT axioms by establishing   Segal's gluing axioms  \cite{segal} for the path integral. Segal's axioms  were designed to capture  the conformal bootstrap approach to CFTs using a geometrical perspective (for a nice introduction to mathematicians see \cite{gawedzki}). In practical terms we extend the definition of the path integral \eqref{defintroPI} to surfaces $\Sigma$ with boundary $\partial\Sigma=\cup_{a=1}^mC_a$ consisting of $m$ analytic circles $C_a$. The result is a function $\caA_\Sigma$ (called the amplitude) of the boundary values $\{\phi|_{C_a}\}_{a=1}^m$ and can be viewed as an element of   $\caH^{\otimes m}$ where $\caH$ is a Hilbert space realised as an $L^2$ space on the space of the field configurations $\phi|_C$. The main theorems to prove are the gluing axioms, Propositions \ref{glueampli} and \ref{selfglueampli}, stating that the amplitudes pair in a natural way upon gluing surfaces along boundary circles. This result in turn allows  to evaluate the correlation functions  $\langle \prod_j V_{e_j,m_j}(z_j)\rangle_{\Sigma,g}$ by cutting the surface $\Sigma$ with $n$ marked points $z_j$ along  simple analytic curves $C_a$, $a=1,\dots, 3{\mathfrak{g}}-3+n$, to  building blocks  $B_a$, $a=1,\dots,2{\mathfrak{g}}-2+n$ with each $B_a$ topologically a sphere with $i\in\{0,1,2\}$ marked points and $3-i$ holes thereby reducing the correlation function to the pairing of the amplitudes of the simple building blocks.  The Hilbert space $\caH$  is here concretely given as
\begin{align*}
\caH=L^2(H^{-s}(\T)\times\Z,d\mu)
\end{align*}
where $H^{-s}(\T)$ with $s>0$ is a Sobolev space of distributions on the circle coming from 
 the $c$ variable and the restriction of the full plane GFF to a circle and $m\in\Z$ parametrises the winding of $\phi|_C$ around the circle. $m$ is the additional "magnetic quantum number" of the compactified GFF. 
 
 The following theorem summarises the main results of this paper for the   correlation functions of the electro-magnetic operators; we refer to Theorem \ref{limitcorel} and to Propositions \ref{glueampli} and \ref{selfglueampli} for more detailed statements. As explained we suppose $\beta=\frac{m}{R}$ and $Q=\frac{n}{R}$ for $m\in\Z_+$ and $n\in\Z$. Letting $k\in\Z_+$ be the greatest common divisor of $m$ and $m-2n$ set $p=\frac{m}{k}$ and $q=\frac{m-2n}{k}$. Then some calculation gives for the compactification radius $R=\frac{k}{2}\sqrt{pq}$ and the central charge  becomes
 \begin{align}\label{cpq}
{\bf c}=1-6\frac{(p-q)^2}{pq}.
\end{align}
 \begin{theorem}
Let $\beta^2=4\frac{p}{q}$  where $p,q\in\N$ are coprime and let the compactification radius $R=\frac{k}{2}\sqrt{pq}$ with $k\in\N$. Let $(\Sigma,g)$ be a closed Riemannian surface and $(z_j,v_j)\in T\Sigma^{n}$. Assume the electric charges satisfy $\alpha_j:=e_j/R>Q$ and that $m_j\in \Z$ satisfy 
$\sum_{j=1}^{n}m_j=0$. Then
\vskip 1mm
\noindent (i) 
 The correlation functions  $\langle  V_{(\boldsymbol{\alpha},\boldsymbol{m})}({\bf v})\rangle_{\Sigma,g}$ 
exist as limits of regularised objects and do not depend on the choice of the cohomology basis nor on the magnetic discontinuity curves. 
\vskip 1mm
\noindent (ii)  $\langle  V_{(\boldsymbol{\alpha},\boldsymbol{m})}({\bf v})\rangle_{\Sigma,g}$
satisfy the axioms of a Conformal Field Theory with central charge \eqref{cpq} and conformal weights $\Delta_{e,m}$ and spin $s$ given by
\begin{align*}
\Delta_{e,m}=\frac{e}{2R}(\frac{e}{2R}-Q)+\frac{m^2R^2}{4},\ \ \ s=QRm.
\end{align*}
in the sense that 
they transform in a covariant way under the action of 
diffeomorphisms $g\to\psi^\ast g$, see \eqref{diffinvariance},  
 Weyl scaling  $g\to e^\sigma g$, $g\in C^\infty(\Sigma)$, see \eqref{confan} and 
rotations of the vectors $v_j\to O_jv_j$, $O_j\in SO(2)$, see \eqref{spinrelation}.
\vskip 1mm

\noindent (iii)  $\langle  V_{(\boldsymbol{\alpha},\boldsymbol{m})}({\bf v})\rangle_{\Sigma,g}$ satisfy the gluing axioms  of Segal for conformal field theory under cutting of the surface along analytically parametrised simple curves, see Proposition \ref{glueampli} and \ref{selfglueampli}.
\vskip 1mm

\noindent (iv) If $(e_1,e_2)$ is the canonical basis of $\R^2=T\hat{\C}$ and $g_0=\frac{|dz|^2}{\max (|z|,1)^{4}}$ the choice of metric on the Riemann sphere, 
then the $3$-point function on the Riemann sphere is given by 
\begin{align*}
& \langle   V_{(\alpha_1,m_1)}(0,e_1)V_{(\alpha_2,m_2)}(1,e_1)V_{(\alpha_3,m_3)}(\infty,e_1)    \rangle_{\hat \C,g_0}=\\
& 2\pi R  \frac{(-\mu)^\ell}{\ell!}\int_{\C^\ell}\prod_{j=1}^\ell x_j^{\Delta_1}\bar{x}_j^{\bar\Delta_1}(1-x_j)^{ \Delta_2}
(1-\bar x_j)^{\bar\Delta_2}\prod_{j<j'}|x_j-x_{j'}|^{\beta^2}\dd x_1\dots\dd x_\ell.
\end{align*}
where 
\[\ell:=\frac{ 2 Q-\sum_j\alpha_j}{\beta}\in \N,\quad \Delta_j=\beta\alpha_j+\frac{kpm_j}{2}, \quad \bar\Delta_j=\beta\alpha_j-\frac{kpm_j}{2}.\]
\end{theorem}

\noindent{\bf Remarks.} 

\vskip 2mm
\noindent{\bf 1.}
Segal's axioms give  access to the spectrum of the CFT by considering the generator of the semigroup of annuli (see \cite{segal} or \cite{GKRV2} in the case of Liouville theory) i.e. the Hamiltonian operator of the CFT.  This operator is not self-adjoint but it has discrete spectrum and non-trivial Jordan blocks. This will be studied in a forthcoming work. 
\vskip 2mm

\noindent{\bf 2.}  Let $p=2$. Then the central charge of our construction coincides with that of the $(2,q)$- minimal models \cite{BPZ} and the set of  its spectral weights contains  the weights of the degenerate Virasoro representations, namely the weights    $\Delta_{e,0}$ with $e/R\in \sqrt{\frac{p}{q}}\Z$.  For the minimal radius $R=\frac{1}{2}\sqrt{pq}$ only the degenerate weights are present in our case. It has been argued in physics by \cite{KM} that the minimal model CFT whose spectrum consists of the degenerate weights with  $e/R=r \sqrt{\frac{p}{q}}$ with $r=1,\dots,q-1$ can be recovered from this compactified imaginary Liouville theory by a so-called BRST reduction, providing this way a path integral construction of the minimal models. This question, raised to us by N. Seiberg, was the original motivation for this work, and we plan to understand this aspect in future work.
\vskip 2mm

\noindent{\bf 3.}  There is an important open question in physics to construct a CFT dubbed Timelike Liouville Theory or Imaginary Liouville Theory \cite{Schomerus, Rib1,Za,HaMaWi}. This is a non compact CFT with continuous spectrum and structure constants given by the inverse of the standard DOZZ formula. Even though there may be some links with such a CFT, our path integral does not coincide with it (among other reasons, it has discrete spectrum). We choose to call our path integral CILT because our construction is based on the Liouville action with imaginary parameters, yet on the compactified boson so that it is not mistaken with the putative  Timelike Liouville Theory.
 \vskip 2mm

\noindent{\bf 4.} 
 The compactified boson satisfies the electro-magnetic duality, meaning that it is invariant under $R\leftrightarrow 1/R$ up to swapping the roles of electric/magnetic operators. If we think of the relation of this model to the scaling limit of loop models, then this duality echoes the duality of SLE curves \cite{dubSLE}. It is not clear to us how to interpret this duality at the level of the path integral, mainly because it requires to construct the path integral for $\beta^2>2$. 
 
 \noindent{\bf 5}
We can also define correlation function in the case where the central charge is not rational; the main condition is that $\sum_j\alpha_j\in \chi(\Sigma)Q+\frac{1}{R}\Z$. The invariance by diffeomorphisms is not as general as in the rational case, so it is not completely clear that this model fits in the full  CFT picture. On the other hand the spectral analysis and representation theory for this model is possible and seems very interesting. This is the reason we mention this case in Theorem \ref{irrational_case}. 
 

 \vskip 2mm
\noindent\textbf{Acknowledgements.}  A. Kupiainen acknowledges the support of the ERC Advanced Grant 741487 and of Academy of Finland. R. Rhodes is partially supported by the Institut Universitaire de France (IUF). R. Rhodes  acknowledges the support of the ANR-21-CE40-0003. The authors wish to thank  L. Eberhardt, J. Jacobsen, E. Peltola, S. Ribault, R. Santachiara, J. Teschner, B. Wu  for fruitful discussions about this work. Special thanks are addressed to Nathan Seiberg; this work originates from discussions with him about the importance of constructing a path integral for minimal models.

\section{Geometric preliminaries}
 
\textbf{Notations:} We shall denote $\dot{u}(t)=\pl_tu(t)$ the derivative of $C^1$ curves $u:[0,1]\to \Sigma$ with values in a surface. If $(\Sigma,g)$ is a closed, or compact with boundary, oriented surface, we denote $\cjg f,f'\cjd_2=\int f\bar{f}'
{\rm dv}_g$ the $L^2$ pairing with respect to the Riemannian measure ${\rm v}_g$, we denote $\cjg u,f\cjd$ the distribuional pairing if $u\in \mc{D}'(\Sigma)$ is a distribution and $f\in C_c^\infty(\Sigma^\circ,\R)$ compactly supported in the interior $\Sigma^\circ$ 
of $\Sigma$, with the convention that if $u\in C^\infty(\Sigma)$ then $\cjg u,f\cjd=\cjg u,f\cjd_2$. 
 
\subsection{Closed surfaces}
Let $\Sigma$ be an oriented compact surface of genus ${\mathfrak{g}}$ and let $g$ be a smooth Riemannian metric. 
The metric $g$ induces a conformal class $[g]:=\{ e^{\varphi}g; \varphi\in C^\infty(\Sigma)\}$, which is equivalent to a complex structure, i.e. a field $J\in C^\infty(\Sigma,{\rm End}(T\Sigma))$ of endomorphisms of the tangent bundle such that $J^2=-{\rm Id}$. There are local charts $\omega_j:U_j\to \D$ such that $\omega_j\circ \omega_k^{-1}$ are holomorphic functions and $\omega_j^*g=e^{\rho_j}|dz|^2$ on $\D$, where $z=x+iy$ is the usual complex coordinate on $\C$ and $\rho_j\in C^\infty(\D)$. The complex structure in the holomorphic charts $\omega_j$ is given by $J\pl_x=\pl_y$ and $J\pl_y=-\pl_x$. The Hodge operator $*:T^*\Sigma\to T^*\Sigma$ is the dual to $J$, it satisfies $*dx=dy$ and $*dy=-dx$ in local holomorphic charts. 
The pair $(\Sigma,J)$ (or equivalently $(\Sigma,[g])$) is called a \emph{closed Riemann surface}. The orientation is given by any non-vanishing $2$-form $w_\Sigma\in C^\infty(\Sigma;\Lambda^2T^*\Sigma)$ so that $(\omega_j^{-1})^*w_\Sigma=e^{f_j} dx\wedge dy$ in $\D$ for some function $f_j$.

The Gauss-Bonnet formula reads  
\begin{equation}\label{GB} 
\int_{\Sigma}K_g{\rm dv}_g=4\pi\chi(\Sigma)
\end{equation}
where $\chi(\Sigma)=(2-2{\mathfrak{g}})$ is the Euler characteristic, $K_g$ the scalar curvature of $g$ and ${\rm dv}_g$ the Riemannian measure.  The uniformisation theorem says that in the conformal class of $g$, there exists a unique metric $g_0=e^{\rho_0}g$ of scalar curvature $K_{g_0}=  -2$ if ${\mathfrak{g}}\geq 2$, $K_{g_0}=  0$ if ${\mathfrak{g}}=1$ or $K_{g_0}=  2$ if ${\mathfrak{g}}=0$. For a metric $\hat{g}=e^{\rho}g$, one has the relation 
\[ K_{\hat{g}}=e^{-\rho}(\Delta_g \rho+K_{g})\]
where $\Delta_g=d^*d$ is the non-negative Laplacian (here $d$ is exterior derivative and $d^*$ its adjoint). 
Let us recall a result of Aubin \cite{Aubin} on prescribing the curvature (it will be useful later).
\begin{lemma}[Aubin]\label{Aubin}
Let $(\Sigma,g)$ be  a closed Riemannian surface with genus ${\mathfrak{g}}\geq 2$. Let $f\in C^\infty(\Sigma)$ be a non-positive function such that $\int_\Sigma f \dd {\rm v}_{g}<0$. Then there exists a conformal metric $\hat{g}:=e^{\rho}g$ for some $\rho\in C^\infty(\Sigma)$ such that $K_{\hat{g}}=f$.
\end{lemma} 
 
\subsection{Surfaces with analytic parametrized boundary.}\label{subsub:boundary} Let $\T=\{e^{i\theta}\in \C\,|\, \theta\in [0,2\pi]\}$ be the unit circle.  A compact Riemann surface $\Sigma$ with parametrized analytic boundary $\partial\Sigma=\sqcup_{j=1}^{{\mathfrak{b}}}\pl_j\Sigma$ is a 
compact oriented surface with smooth boundary with a family of charts 
 $\omega_j:U_j\to \omega_j(U_j)\subset \C$ for $j=1,\dots,j_0$  and smooth diffeomorphisms $\zeta_j:\T \to \pl_j\Sigma$ where 
 \begin{itemize}
 \item $\cup_j U_j$ is an open covering of $\Sigma$ with $U_j\cap \pl_j\Sigma\not=\emptyset$ if and only if $j\in [1,{\mathfrak{b}}]$
 \item $\omega_j(U_j)=\D$ if $j>{\mathfrak{b}}$ and 
 $\omega_j(U_j)=\mathbb{A}_{\delta}:=\{z\in \C\,|\, |z|\in (\delta,1] \}$ for some $\delta<1$ if $j\leq {\mathfrak{b}}$ and  $\omega_j(\pl_j\Sigma)=\T$
 \item $\omega_k \circ \omega_j^{-1}$ are holomorphic maps in the interior of the domain  $\omega_j(U_j\cap U_k)$ 
 \item $\omega_j\circ \zeta_j:\T \to \T$ is real analytic, i.e. it extends holomorphically in a neighborhood of $\T$.
 \end{itemize}
The charts induce  a complex structure $J$ as for the closed case. A Riemannian metric $g$ is \emph{compatible} with 
$J$ if $(\omega_j^{-1})^*g=e^{\rho_j}|dz|^2$ for some smooth function $\rho_j$ on $\omega_j(U_j)$.
The boundary circles $\pl_j\Sigma$ inherit an orientation from the orientation of $\Sigma$, simply by taking 
the $1$-form $-\iota_{\pl_j\Sigma}^*(i_{\nu}\omega_\Sigma)$  where $i_\nu$ is the interior product with non-vanishing interior pointing vector field $\nu$ to $\Sigma$ and $\iota_{\pl_j\Sigma}:\pl_j\Sigma \to \Sigma$ is the natural inclusion. 
In the chart given by the annulus $\mathbb{A}_\delta$, the orientation is then given by $\dd\theta$ (i.e. the counterclockwise orientation) on the unit circle parametrized by $(e^{i\theta})_{\theta\in [0,2\pi]}$. 
We say that the boundary $\pl_j\Sigma$ is \emph{outgoing} if the orientation $(\zeta_j)_*(\dd\theta)$ is the orientation of $\pl_j\Sigma$ induced by that of $\Sigma$ (we also say that the orientation of $\pl_j\Sigma$ is positive) otherwise the parametrized boundary $\pl_j\Sigma$ is called \emph{incoming}. We define  $\varsigma_j\in \{\pm 1\}$ with $\varsigma_j=-1$ if $\pl_j\Sigma$ outgoing and $\varsigma_j=1$ if $\pl_j\Sigma$ is incoming. Composing $\zeta_j$ by the inversion $o(z):=1/z$ reverses orientation and transform an outgoing boundary into an incoming one and conversely.
Notice that there is $\delta<1$ such that $\zeta_j
$ is holomorphic from an annular neighborhood $\mathbb{A}_{\delta}$  of $\T$ (resp.  $\mathbb{A}_{\delta}^{-1}:=\delta^{-1}\mathbb{A}_{\delta}$)   
 to a neighborhood $U'_j$ of $\pl_j\Sigma$ if $\omega_j\circ \zeta_j|_{\T}$ preserves orientation, i.e. $\pl_j\Sigma$ is outgoing (resp. reverses orientation, i.e. $\pl_j\Sigma$ is incoming). Up to adding $U'_j$ to the set of charts and replacing $U_j$ by $U_j\cap \Sigma^\circ$ for $j=1,\dots,{\mathfrak{b}}$, the new set of charts 
 $((U'_j, \omega'_j), (U_j,\omega_j))$ with $\omega_j':=\zeta_j^{-1}\circ o^{(1+\varsigma_j)/2}$ produces the same complex structure on $\Sigma$ as the original one, and  the parametrisation of the boundary component $\pl_j\Sigma$ is given by $(\omega'_j)^{-1}|_\T$. Without loss of generality, we can and will thus assume from now that, when choosing our set of charts $(U_j,\omega_j)$ above, the boundary parametrisations for $j=1,\dots,\mathfrak{b}$ are given by 
 \[\zeta_j :\T \to \pl_j \Sigma ,\quad \zeta_j(e^{i\theta}):=\left\{\begin{array}{ll}
 \omega_{j}^{-1}(e^{i\theta}) & \textrm{ if }\varsigma_j=-1\\
\omega_{j}^{-1}(e^{-i\theta}) &  \textrm{ if }\varsigma_j=1
 \end{array}\right..
 \]
The metric $g$ is said \emph{admissible} if  it is compatible with the complex structure and for $(U_j,\omega_j)$ with $j=1,\dots,\mathfrak{b}$, we have 
\[ (\omega_j^{-1})^*g=|dz|^2/|z|^2.\]
on $\omega_j(U_j)$. In that case the boundary is geodesic for $g$, with length $2\pi$, and the metric $g$ has curvature $K_g=0$ near $\pl \Sigma$.

\subsection{Gluing and cutting surfaces}\label{sub:gluing}
Let $\Sigma$ be a Riemann surface, not necessarily connected, with $\mathfrak{b}$ parametrized analytic boundary connected components $\zeta_j:\T\to \pl_j\Sigma$. Let $\omega_j:U_j\to \mathbb{A}_\delta$ be the charts near $\pl_j\Sigma$ with $\omega_j^{-1}|_{\T}=\zeta_j$ as above.
If $\pl_j\Sigma$ is outgoing and $\pl_k\Sigma$ is incoming, we can glue $\pl_j\Sigma$ to $\pl_k\Sigma$ to obtain a new Riemann surface $\Sigma^{\#}$ with $({\mathfrak{b}}-2)$-boundary components: this is done by identifying  
$\zeta_j(e^{i\theta})\simeq \zeta_k(e^{i\theta})$. A neighborhood in $\Sigma^{\#}$ of the identified circle $\pl_j\Sigma \simeq \mc{C}_k$ is given by $(U_j\cup U_k)/\sim $ where $\sim$ means the identification $\pl_j\Sigma \sim\pl_k\Sigma$,  and  
\[\omega_{jk}: z\in U_j \mapsto \omega_j(z) \in \mathbb{A}_\delta, \quad z\in U_k\mapsto \frac{1}{\omega_k(z)}\in \mathbb{A}^{-1}_\delta.\]
produces a chart. 
 If $\pl_j\Sigma$ and $\pl_k\Sigma$ belong to different connected components of $\Sigma$, $b_0(\Sigma^{\#})=b_0(\Sigma)-1$  if $b_0$ denotes the $0$-th Betti number; otherwise $b_0(\Sigma^{\#})=b_0(\Sigma)$.
If $\Sigma$ is equipped with an admissible metric, then $g$ induces a smooth metric on the glued Riemann surface $\Sigma^{\#}$ compatible with the complex structure. 

Conversely, starting from a Riemann surface $(\Sigma,J)$ with ${\mathfrak{b}}\geq 0$ boundary circles $\pl_1\Sigma,\dots,\pl_{{\mathfrak{b}}}\Sigma$ and choosing a simple analytic curve $\mc{C}$ in the interior $\Sigma^\circ$ of $\Sigma$
and a chart $\omega:V\mapsto \{|z|\in [\delta,\delta^{-1}]\}\subset \C$ for some $\delta<1$ with $\omega(\mc{C})=\T$, there is a natural Riemann surface $\bbar{\Sigma}_{\mc{C}}$ obtained by compactifying 
 $\Sigma_{\mc{C}}:=\Sigma\setminus \mc{C}$ into a Riemann surface with  ${\mathfrak{b}}+2$ boundary circles by adding two copies 
 $\pl_{{\mathfrak{b}}+1}\bbar{\Sigma}_{\mc{C}}= \mc{C}$, $\pl_{{\mathfrak{b}}+1}\bbar{\Sigma}_{\mc{C}}=\mc{C}$ of $\mc{C}$ to $\Sigma_{\mc{C}}$ using the chart $\omega$ on respectively $\omega^{-1}(\{|z|\leq 1\})$ and $\omega^{-1}(\{|z|\geq 1\})$ (one incoming, one outgoing). Using the gluing procedure of $\pl_{{\mathfrak{b}}+1}\bbar{\Sigma}_{\mc{C}}$ with $\pl_{{\mathfrak{b}}+2}\bbar{\Sigma}_{\mc{C}}$ in $\bbar{\Sigma}_{\mc{C}}$ just described above using the parametrizations  $\zeta_1:=\omega^{-1}|_{\T}$ and $\zeta_2:=\omega^{-1}|_{\T}$, we recover $(\Sigma,J)$.
    
%


\subsection{Determinant of Laplacians} \label{detoflap}

For a Riemannian metric $g$ on a connected oriented compact surface $\Sigma$, the non-negative Laplacian $\Delta_g=\dd^*\dd$ has discrete spectrum
${\rm Sp}(\Delta_g)=(\la_j)_{j\in \N_0}$ with $\la_0=0$ and $\la_j\to +\infty$. We can define the determinant of $\Delta_g$ by 
\[ {\det} '(\Delta_g)=\exp(-\pl_s\zeta(s)|_{s=0})\]
where $\zeta(s):=\sum_{j=1}^\infty \la_j^{-s}$ is the spectral zeta function of $\Delta_g$, which admits a meromorphic continuation from ${\rm Re}(s)\gg 1$ to $s\in \cc$ and is holomorphic at $s=0$. We recall that if $\hat{g}=e^{\varphi}g$ for some $\varphi\in C^\infty(\Sigma)$, one has the so-called Polyakov formula (see \cite[eq. (1.13)]{OPS}) 
\begin{equation}\label{detpolyakov} 
\log \frac{{\det}'(\Delta_{\hat{g}})}{{\rm v}_{\hat{g}}(\Sigma)}= \log \frac{{\det}'(\Delta_g)}{{\rm v}_{g}(\Sigma)} -\frac{1}{48\pi}\int_\Sigma(  |d\varphi|_g^2+2K_g\varphi){\rm dv}_g\end{equation}
where $K_g$ is the scalar curvature of $g$ as above.   

When $(\Sigma,g)$ is a connected oriented compact surface $\Sigma$ with boundary, the Laplacian $\Delta_{g}$ with Dirichlet boundary conditions has discrete spectrum $(\la_{j,D})_{j \geq 1}$ and the determinant is defined in a similar way as $\det (\Delta_{g,D}): = e^{- \zeta'_{g,D}(0)}$ where $\zeta_{g,D}$   is the spectral zeta function of $\Delta_{g}$ with Dirichlet boundary conditions defined for ${\rm Re}(s)\gg 1$ by $\zeta_{g,D}(s):=\sum_{j=1}^\infty \la_{j,D}^{-s}$. The function $\zeta_{g,D}(s)$ admits a meromorphic extension to $s\in \C$ and is holomorphic at $s=0$.

\subsection{Green's function and resolvent of Laplacian}\label{sec:Green}
Each compact Riemannian surface $(\Sigma,g)$ has a (non-negative) Laplace operator $\Delta_g=\dd^*\dd$ where the adjoint is taken with respect to ${\rm v}_g$.  The Green function $G_g$ on a surface $\Sigma$ without boundary  is defined 
to be the integral kernel of the resolvent operator $R_g:L^2(\Sigma)\to L^2(\Sigma)$ satisfying $\Delta_{g}R_g=2\pi ({\rm Id}-\Pi_0)$, $R_g^*=R_g$ and $R_g1=0$, where $\Pi_0$ is the orthogonal projection in $L^2(\Sigma,{\rm dv}_g)$ on $\ker \Delta_{g}$ (the constants). By integral kernel, we mean that for each $f\in L^2(\Sigma,{\rm dv}_g)$
\[ R_gf(x)=\int_{\Sigma} G_g(x,x')f(x'){\rm dv}_g(x').\] 
The Laplacian $\Delta_{g}$ has an orthonormal basis of real valued eigenfunctions $(e_j)_{j\in \N_0}$ in $L^2(\Sigma,{\rm dv}_g)$ with associated eigenvalues $\la_j\geq 0$; we set $\la_0=0$ and $\varphi_0=({\rm v}_g(\Sigma))^{-1/2}$.  The Green function then admits the following Mercer's representation in $L^2(\Sigma\times\Sigma, {\rm dv}_g \otimes {\rm dv}_g)$
\begin{equation}
G_g(x,x')=2\pi \sum_{j\geq 1}\frac{1}{\lambda_j}e_j(x)e_j(x').
\end{equation}
Similarly, on a surface with smooth boundary $\Sigma$, we will consider the Green function with Dirichlet boundary conditions $G_{g,D}$ associated to the Laplacian $\Delta_{g}$ with Dirichlet condition on $\pl \Sigma$. In this case, the associated resolvent operator
\[ R_{g,D}f(x)=\int_{\Sigma} G_{g,D}(x,x')f(x'){\rm dv}_g(x')\] 
solves $\Delta_{g}R_{g,D}=2\pi {\rm Id}$. 



\subsection{First homology group and symplectic basis on closed surfaces} \label{detoflap}

 Let $\Sigma$ be a closed oriented surface of genus ${\mathfrak{g}}$. We denote by $\mc{H}_1(\Sigma)$ the first homology group of $\Sigma$ with value in $\Z$. It is an abelian group isomorphic to $\Z^{2{\mathfrak{g}}}$, which can be obtained as the abelianization of the fundamental group $\pi_1(\Sigma,x_0)$ for some fixed $x_0\in \Sigma$.
Recall that elements in $\pi_1(\Sigma,x_0)$ are equivalence classes of closed $C^1$ curves on the surfaces (equivalence been given by homotopy fixing $x_0$), and for a closed curve $c$ on $\Sigma$, we denote by $[c]$ its class in $\mc{H}_1(\Sigma)$.

For two transverse oriented $C^1$ curves $a:\T\to \Sigma$, $b: \T\to \Sigma$ on $\Sigma$, the algebraic intersection number  $\iota(a,b)\in\Z$ is the number of intersection points $p=a(t_1)=b(t_2)$ weighted by $+1$ (resp. $-1$) if the orientation of the basis $(\dot{a}(t_1),\dot{b}(t_2))$ of $T_p\Sigma$ at $p$ is positive (resp. negative) with respect to the orientation of $\Sigma$. 
This number only depends on the homology class $[a],[b]$ of $a,b$, it defines a bilinear skew-symmetric map  
\[ \iota : \mc{H}_1(\Sigma)\wedge \mc{H}_1(\Sigma)\to \Z.\]
This map is symplectic and there exists a symplectic basis $([a_i],[b_i])_{i=1,\dot, {\mathfrak{g}}}$ of $\mc{H}_1(\Sigma)$ represented by closed simple curves $(a_i,b_i)_i$ such that
\begin{equation}\label{int_numbers} 
\iota(a_i,b_j)=\delta_{ij}, \quad \iota(a_i,a_j)=0, \quad \iota(b_i,b_j)=0.
\end{equation}
A   symplectic basis $([a_i],[b_i])_{i=1,\dots,{\mathfrak{g}}}$ of $\mc{H}_1(\Sigma)$ represented by simple closed curves $(a_i,b_i)_i$ will be called a \emph{geometric symplectic basis}. 
See Figure \ref{fig:homology_closed}.
 If $c=\sum_{j=1}^{\mathfrak{g}}\alpha_j a_j+
\beta_jb_j$ and $c'=\sum_{j=1}^{\mathfrak{g}}\alpha'_j a_j+
\beta'_{j}b_j$, we have 
\[ \iota(c,c')= \sum_{j=1}^{{\mathfrak{g}}}\alpha_j\beta_{j}'-\alpha'_{j}\beta_j=  (\alpha,\beta){\rm J}(\alpha',\beta')^\top\]
where ${\rm J}$ is the canonical symplectic matrix on $\R^{2{\mathfrak{g}}}$. We denote by ${\rm Sp}(2{\mathfrak{g}},\R)$ the symplectic group, defined by 
$A\in {\rm Sp}(2{\mathfrak{g}},\R)$ if $A{\rm J}A^\top={\rm J}$, and let ${\rm SP}(2{\mathfrak{g}},\Z):={\rm SP}(2{\mathfrak{g}},\R)\cap {\rm GL}(2{\mathfrak{g}},\Z)$ the subgroup whose coefficients are integers. 
Any other  basis $(a',b')=(a_1',\dots,a_{{\mathfrak{g}}}',b_1',\dots b'_{\mathfrak{g}})$  (not necessarily symplectic)
of $\mc{H}_1(\Sigma)$ is related to $(a,b)=(a_1,\dots,a_{\mathfrak{g}},b_1,\dots, b_{\mathfrak{g}})$ by a matrix $A\in {\rm SL}(2{\mathfrak{g}},\Z)$ with integer coefficients and of determinant $1$ 
\[(a',b')=A(a,b).\]
The basis $(a',b')$ is symplectic if and only if $A\in {\rm Sp}(2{\mathfrak{g}},\Z)$. We shall typically use the notation $\boldsymbol{\sigma}=(\sigma_1,\dots,\sigma_{2{\mathfrak{g}}})$ for a basis of $\mc{H}_1(\Sigma)$ and $(a_1,b_1,\dots,a_{\mathfrak{g}},b_{\mathfrak{g}})$ when the basis is a geometric symplectic basis. We refer to \cite[Chapter 6.1]{Farb-Margalit} for a detailed discussion on the symplectic structure of $\mc{H}_1(\Sigma)$. 
It follows from the proof of \cite[Theorem 6.4]{Farb-Margalit} that each symplectic basis of $\mc{H}_1(\Sigma)$ 
can be realized by a geometric symplectic basis. Finally, as explained in \cite[Section 6.3]{Farb-Margalit},
if $\psi:\Sigma\to \Sigma$ is an orientation preserving diffeomorphism, it induces an automorphism $\psi_*: \mc{H}_1(\Sigma)\to \mc{H}_1(\Sigma)$ which belongs to ${\rm Sp}(2{\mathfrak{g}},\Z)$ (as it preserves the intersection form).
\begin{figure}
\includegraphics[width=0.5\textwidth]{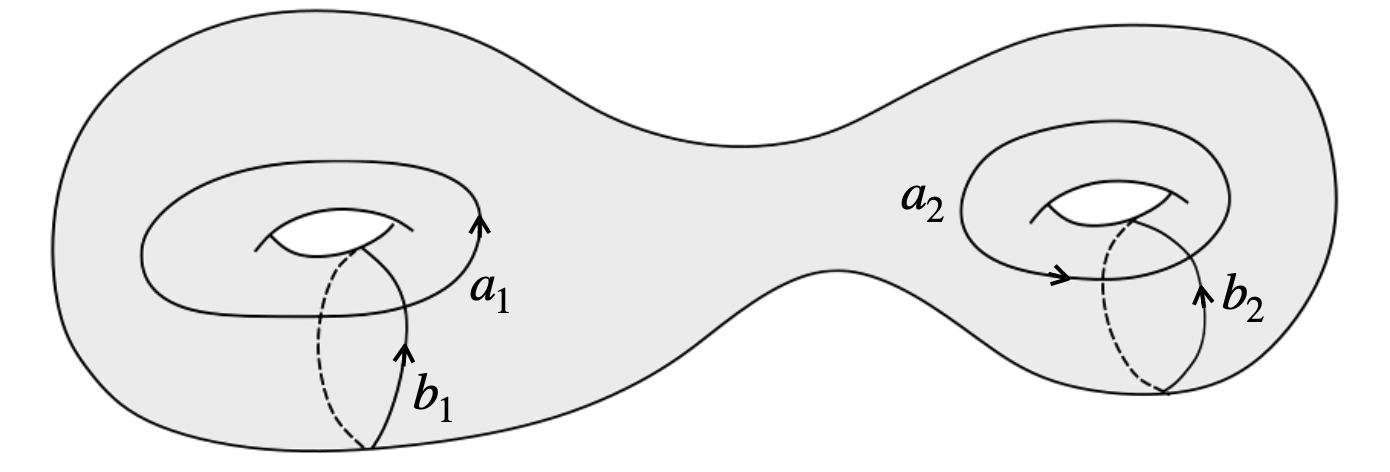}
\caption{A geometric symplectic basis of $\mc{H}_1(\Sigma)$}
\label{fig:homology_closed}
\end{figure}

 \subsection{De Rham first cohomology group on closed Riemann surfaces} 
 We equip $\Sigma$ with a complex structure $J\in {\rm End}(T\Sigma)$, i.e. satisfying $J^2=-{\rm Id}$. Denote by $\Lambda^p\Sigma:=\Lambda^p T^*\Sigma$ the bundle of differential $p$ forms. Let 
$\dd: C^\infty(\Sigma,\Lambda^k\Sigma)\to C^\infty(\Sigma,\Lambda^{k+1}\Sigma)$ be the exterior derivative and $*:T^*\Sigma\to T^*\Sigma$ be the Hodge operator,  dual to $-J$. The Hodge operator induces a scalar product on the space $C^\infty(\Sigma,\Lambda^1\Sigma)$ of real valued $1$-forms 
\[ \cjg \omega,\omega'\cjd_2 := \int_{\Sigma} \omega\wedge *\omega'.\]
The formal adjoint of $\dd: C^\infty(\Sigma)\to C^\infty(\Sigma,\Lambda^1\Sigma)$ with respect to this scalar product on $1$-forms 
and the Riemannian $L^2$ scalar product on functions is given by $\dd^*=-*\, \dd\, *$. 
The first de Rham cohomology space is defined by 
\[ \mc{H}^1(\Sigma):= \ker \dd|_{C^\infty(\Sigma,\Lambda^1\Sigma)}/ {\rm Im}\, \dd|_{C^\infty(\Sigma)}.\]
It is isomorphic to the real vector space of real-valued closed and co-closed forms, or equivalently real harmonic $1$-forms,
\[{\rm Harm}^1(\Sigma)=\{\omega \in C^\infty(\Sigma,\Lambda^1)\,|\, \dd\omega=0,\dd*\omega=0\}.\] 

The space $\mc{H}^1(\Sigma)$ is dual to $\mc{H}_1(\Sigma)$ with the duality map given by 
\[ \mc{H}^1(\Sigma) \times \mc{H}_1(\Sigma) \to \R ,  \quad  \cjg \omega, \sigma\cjd:= \int_\sigma \omega\] 
and called the period of $\omega$ over $\sigma$. If we fix a basis $(\sigma_{1},\dots,\sigma_{2{\mathfrak{g}}})$ of $\mc{H}_1(\Sigma)$, one can then find a unique basis $(\omega_1,\dots,\omega_{2{\mathfrak{g}}})$ dual to $(\sigma_{1},\dots,\sigma_{2{\mathfrak{g}}})$.
For $R>0$, we define the $\Z$-module consisting of cohomology classes with periods in $2\pi R\Z$:
\begin{equation}\label{H^1RSigma}
\mc{H}^1_R(\Sigma) :=\{ \omega \in \mc{H}^1(\Sigma)\,|\, \forall \sigma\in \mc{H}_1(\Sigma),  \cjg \omega,\sigma\cjd \in 2\pi R\Z\}.
\end{equation}
The discussion above implies the following:
\begin{lemma}\label{basisH^1}
Let $R>0$ and let $\boldsymbol{\sigma}=(\sigma_1,\dots,\sigma_{2{\mathfrak{g}}})$ be a basis of $\mc{H}_1(\Sigma)$. Then 
there exists $2{\mathfrak{g}}$ independent closed smooth $1$-forms $\omega_1,\dots,\omega_{2{\mathfrak{g}}}$ forming a basis of $\mc{H}^1_R(\Sigma)$ dual to $(\sigma_1,\dots,\sigma_{2{\mathfrak{g}}})$ in the sense that 
\[ \forall i,j,\quad  \frac{1}{2\pi R}\int_{\sigma_i}\omega_j=\delta_{ij}.\]
If $\omega_1',\dots,\omega'_{2{\mathfrak{g}}}$ is another such family, then for each $j=1,\dots,2{\mathfrak{g}}$ there exists $\varphi_j\in C^\infty(\Sigma)$ such that $\omega'_j=\omega_j-\dd\varphi_j$. There is a unique such basis of $\mc{H}_R^1$ so that in addition $\dd *\omega_j=0$ for all $j=1,\dots,2{\mathfrak{g}}$.
\end{lemma} 
   
We deduce that if $\omega_j$ are some closed forms as in Lemma \ref{basisH^1}, for each ${\bf k}=(k_1,\dots,k_{2{\mathfrak{g}}})\in \Z^{2{\mathfrak{g}}}$, the form 
\begin{equation}\label{omega_k}
\omega_{\bf k}:=\sum_{j=1}^{2{\mathfrak{g}}}k_j \omega_j
\end{equation} 
satisfies $\int_{\sigma_i}\omega_{\bf k}=k_i 2\pi R$ and the map ${\bf k}\mapsto \omega_{\bf k}$ identifies $\Z^{{2\mathfrak{g}}}$ with $\mc{H}_R^1(\Sigma)$.\\

We pursue with a technical lemma that will be useful later:
\begin{lemma}\label{dX_gomega} 
Let $\Sigma$ be a closed Riemann surface and let $\omega\in C^\infty(\Sigma,\Lambda^1\Sigma)$ be a closed $1$-form. Then 
\[\cjg \dd R_g\dd^*\omega,\omega\cjd_2=2\pi \|(1-\Pi_1)\omega\|_2^2\]
where $\Pi_1:L^2(\Sigma,\Lambda^1\Sigma)\to {\rm Harm}^1(\Sigma)$ is the orthogonal projection.
\end{lemma}
\begin{proof}
We observe the following about the Green's function: $G_g$ is the integral kernel of the operator $R_g=2\pi \Delta_g^{-1}$ satisfying 
\[ \Delta_g \Delta_g^{-1}=({\rm Id}-\Pi_0)\]
where $\Pi_0$ is the projector on constants with respect to the volume form ${\rm v}_g$. Let $\Delta_{g,1}=\dd^*\dd+\dd\dd^*$ be the Laplacian on $1$-forms and 
$R_{g,1}$ be the operator so that 
\[  \Delta_{g,1} R_{g,1}=R_{g,1} \Delta_{g,1}=2\pi({\rm Id}-\Pi_1)\]
where $\Pi_1$ is the orthogonal projector on ${\rm Harm}_1(\Sigma)=\ker \dd\cap \ker \dd^*|_{C^\infty(\Sigma,\Lambda^1\Sigma)}$. Then 
\[ \Delta_{g,1} \dd R_g=\dd \Delta_g R_g=2\pi \dd ({\rm Id}-\Pi_0)=2\pi \dd,\]
thus applying $R_{g,1}$ on the left, 
\begin{equation}\label{formule_link_green} 
2\pi ({\rm Id}-\Pi_1) \dd R_g =2\pi  \dd R_g=2\pi R_{g,1}\dd
 \end{equation}
This proves, using that $\dd\omega=0$, that 
\[\cjg dR_g\dd^*\omega,\omega\cjd_2=\cjg R_{g,1}\dd \dd^* \omega,\omega\cjd_2= \cjg R_{g,1}\Delta_{g,1} \omega,\omega\cjd=2\pi \|(1-\Pi_1)\omega\|_2^2.\qedhere\]
\end{proof}

\subsection{Homology and cohomology on  surfaces with boundary} \label{sub:cohomoboundary}

Now, let us consider the case of a Riemann surface $(\Sigma,J)$ with boundary (note that $J$ induces an orientation). Assume that there are ${\mathfrak{b}}$ boundary connected components $\pl_1\Sigma,\dots,\pl_{{\mathfrak{b}}}\Sigma$, oriented positively with respect to the orientation of $\Sigma$, and the genus is denoted by ${\mathfrak{g}}\geq 0$. 

\begin{figure}[!h]
     \centering
     \begin{subfigure}[b]{0.45\textwidth}
         \centering
         \includegraphics[width=0.5\textwidth]{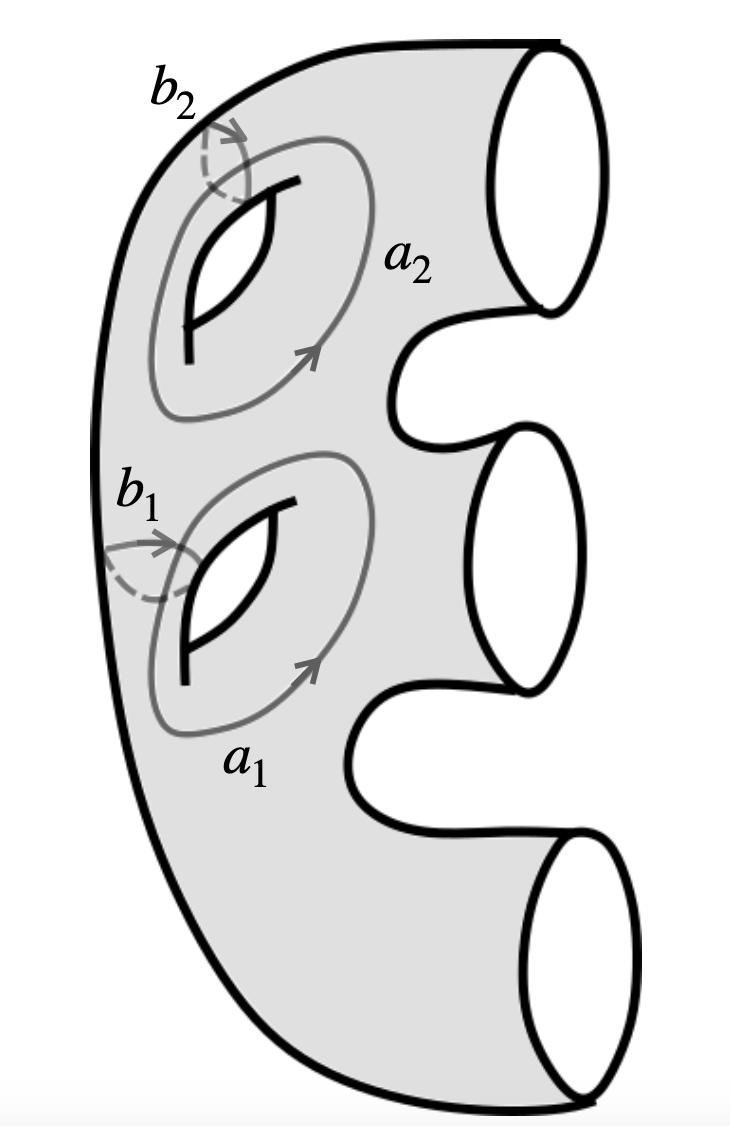}
         \caption{Surface $\Sigma$ with a boundary and a basis of interior cycles. \vspace{1cm}}
     \end{subfigure}
     \hfill
     \begin{subfigure}[b]{0.45\textwidth}
         \centering
         \includegraphics[width=0.7\textwidth]{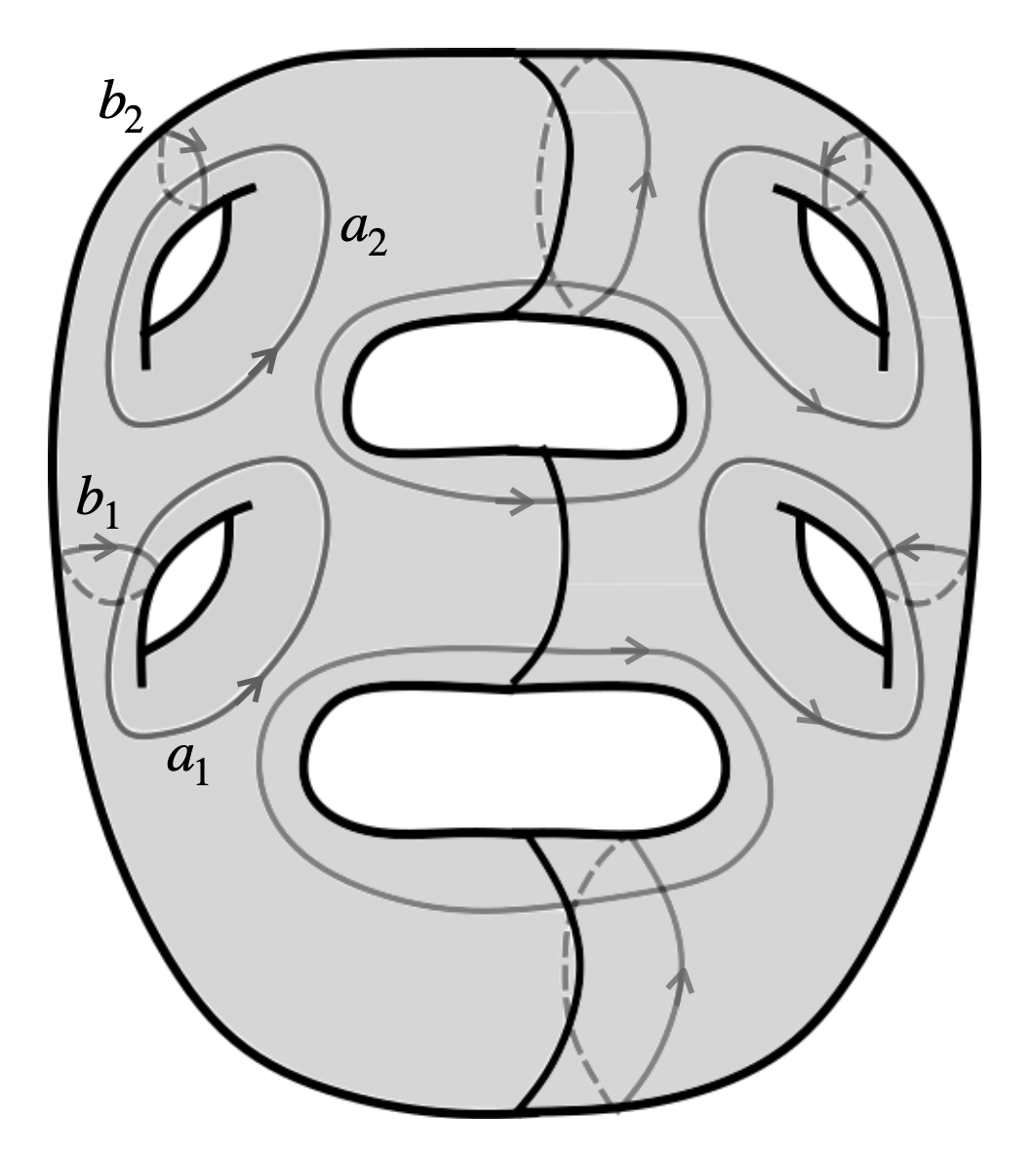}
         \caption{The double $\Sigma^{\#2}$ with a choice of geometric symplectic basis. The cycles that are part of the cohomology basis on $\Sigma$, completed to a full cohomology basis in grey. In black, the boundary of $\Sigma$ as a subset of   $\Sigma^{\#2}$.}
     \end{subfigure}
        \caption{Homology on the doubled surface.}
        \label{fig:double}
\end{figure}

As for closed surfaces, the first homology group $\mc{H}_1(\Sigma)$ is represented by oriented closed curves in $\Sigma$. It is isomorphic to $\Z^{2{\mathfrak{g}}+{\mathfrak{b}}-1}$. The positively oriented boundary circles $c_j=\pl_j\Sigma$ for $j=1,\dots,{\mathfrak{b}}$ are 
elements in $\mc{H}_1(\Sigma)$ and $\sum_{j=1}^{\mathfrak{b}}[c_j]=[\pl \Sigma]=0$ is the class of the boundary, which is trivial. 
If $\Sigma^{\# 2}=\Sigma\#\Sigma$ is the double of $\Sigma$ 
($(\Sigma,J)$ glued with $(\Sigma,-J)$ along the boundary), the inclusion map of the right copy of $\Sigma$ into $\Sigma^{\#2}$ induces  a linear injection 
\[ i_\Sigma: \mc{H}_1(\Sigma)\to \mc{H}_1(\Sigma^{\#2})\]
and the intersection pairing  $\iota$ defined above for closed surfaces also makes sense on $\Sigma$ using this injection. The involution $\tau_{\Sigma}$ exchanging the right copy of $\Sigma$ with the left copy induces a linear map on $\mc{H}_1(\Sigma)$ fixing the boundary curves $[c_j]$.
The group $\mc{H}_1(\Sigma)$ is generated by $([\sigma_1],\dots,[\sigma_{2{\mathfrak{g}}}], [c_{i_1}],\dots,[c_{i_{{\mathfrak{b}}-1}}])$ where $\sigma_j$ are simple curves that do not intersect $\pl \Sigma$ and $i_1,\dots,i_{{\mathfrak{b}}-1}\in \{1,\dots,{\mathfrak{b}}\}$ are distinct.  Notice that these cycles $\sigma$ can also be viewed as element in the relative homology $\mc{H}_1(\Sigma,\pl \Sigma)$, and we can assume that $[\sigma_j]\not=0$ in $\mc{H}_1(\Sigma,\pl \Sigma)$. In fact, one can choose $(\sigma_j)_j$ to be $(a_1,b_1,\dots,a_{{\mathfrak{g}}},b_{{\mathfrak{g}}})$ with 
\[ \iota(a_j,a_i)=0= \iota(b_j,b_i), \quad \iota(a_j,b_i)=\delta_{ij},\]
that is, $(i_\Sigma(a_j),i_{\Sigma}(b_j))_{j=1,\dots,{\mathfrak{g}}}$ is a subset of a geometric symplectic basis in $\mc{H}_1(\Sigma^{\#2})$; see Figure \ref{fig:double}. We call such a basis 
\begin{equation}\label{baseabsolute}
(a_1,b_1,\dots,a_{{\mathfrak{g}}},b_{{\mathfrak{g}}},c_{i_1},\dots c_{i_{{\mathfrak{b}}-1}})
\end{equation} 
a \emph{canonical geometric basis} of $\mc{H}_1(\Sigma)$. Up to renumbering the boundary components, we will assume that $i_j=j$.\\


Elements in the relative homology $\mc{H}_1(\Sigma,\pl \Sigma)$ are represented by either oriented closed curves or oriented curves with endpoints on the boundary $\pl \Sigma$. The relative homology $\mc{H}_1(\Sigma,\pl \Sigma)$  also has dimension $2{\mathfrak{g}}+{\mathfrak{b}}-1$. A basis of cycles of $\mc{H}_1(\Sigma,\pl \Sigma)$ can be obtained by taking 
\begin{equation}\label{baserelative}
(a_1,b_1,\dots,a_{\mathfrak{g}},b_{\mathfrak{g}}, d_{i_1},\dots,d_{i_{{\mathfrak{b}}-1}})
\end{equation} 
where $(a_j,b_j)$ are chosen  inside $\Sigma^\circ$ such that the intersection numbers are given as in \eqref{int_numbers} as before, and the $d_{i_j}$ for $(i_1,\dots,i_{{\mathfrak{b}}-1})\in \{1,\dots,{\mathfrak{b}}\}$ all disjoint are non-intersecting oriented simple curves  with the initial endpoint on 
$\pl_{i_j}\Sigma$ and the final endpoint on $\pl_{i_{j+1}}\Sigma$, and each $d_{i_j}$ is not intersecting any other curve of the basis (see Figure \ref{fig:homologybdry}). Such a basis will be called a \emph{canonical geometric basis} of $\mc{H}_1(\Sigma,\pl\Sigma)$
As above,  we will assume that $i_j=j$.

\begin{figure}[!h]
     \centering
     \begin{subfigure}[b]{0.45\textwidth}
         \centering
         \includegraphics[width=0.5\textwidth]{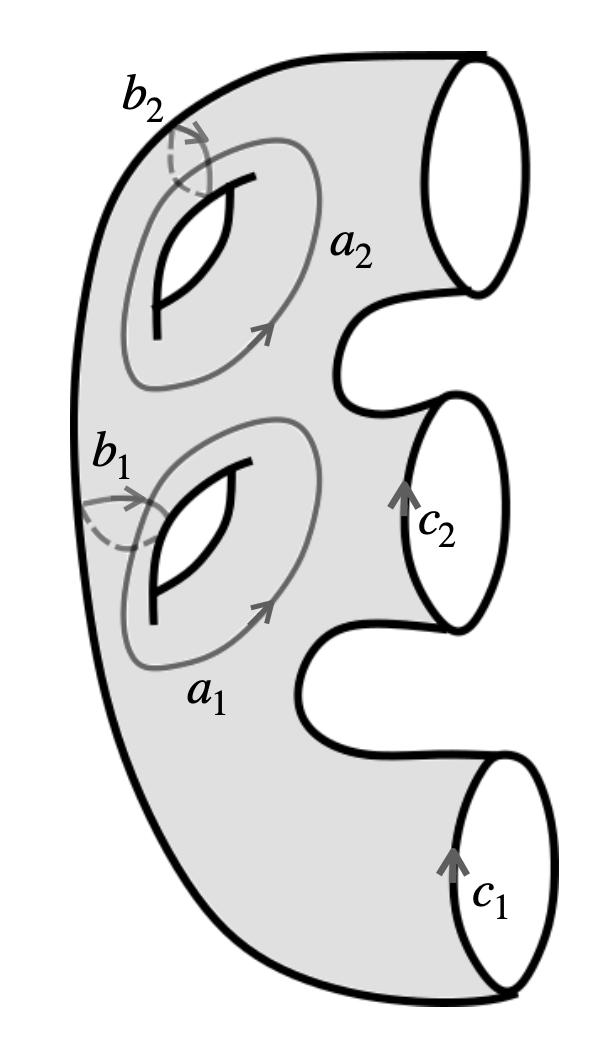}
         \caption{A canonical geometric  basis of $\mc{H}_1(\Sigma)$.}
     \end{subfigure}
     \hfill
     \begin{subfigure}[b]{0.45\textwidth}
         \centering
         \includegraphics[width=0.5\textwidth]{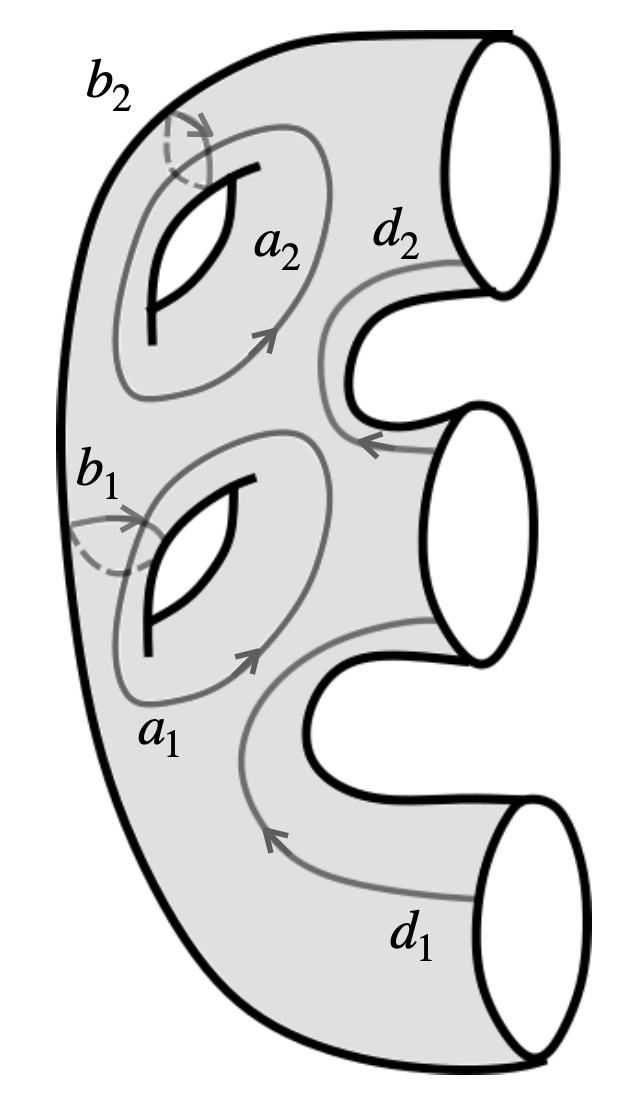}
         \caption{A canonical geometric  basis of  $\mc{H}_1(\Sigma,\pl\Sigma)$.}
     \end{subfigure}
        \caption{Absolute and relative homology on $\Sigma$.}
        \label{fig:homologybdry}
\end{figure}


The \emph{absolute cohomology} $\mc{H}^k(\Sigma)$ of degree $k\in \{0,1,2\}$ is defined by 
 \[ \mc{H}^k(\Sigma):= \ker \dd|_{C_{\rm abs}^\infty(\Sigma,\Lambda^k\Sigma)}/{\rm Im}\, \dd|_{C_{\rm abs}^\infty(\Sigma,\Lambda^{k-1}\Sigma)}\]
where, if $\nu$ is the interior pointing unit normal vector to the boundary and  $\iota_{\pl\Sigma}:\pl\Sigma \to \Sigma$ denotes the inclusion map, we define the set of real-valued absolute forms ($i_\nu$ denotes interior product)
\[C^\infty_{\rm abs}(\Sigma,\Lambda^k\Sigma):= \{ u\in C^\infty (\Sigma,\Lambda^k\Sigma)\,|\, \iota_{\pl \Sigma}^*(i_\nu u)=0, \iota_{\pl \Sigma}^*(i_\nu du)=0\}.\]
The space  $\mc{H}^{k}(\Sigma)$ is isomorphic to the real vector space of real-valued absolute closed and co-closed forms
\[ {\rm Harm}^k(\Sigma):=\{ u\in C^\infty (\Sigma,\Lambda^k\Sigma)\,|\, \dd u=0, \, \dd*u=0, \, \iota_{\pl\Sigma}^*(i_{\nu} u)=0\}.\] 

The \emph{relative cohomology} $\mc{H}^k(\Sigma,\pl \Sigma)$ of degree $k\in \{0,1,2\}$ is defined by
 \[ \mc{H}^k(\Sigma,\pl \Sigma):= \ker \dd|_{C_{\rm rel}^\infty(\Sigma,\Lambda^k\Sigma)}/{\rm Im}\, \dd|_{C_{\rm rel}^\infty(\Sigma,\Lambda^{k-1}\Sigma)}\]
where
\[C^\infty_{\rm rel}(\Sigma,\Lambda^k\Sigma):= \{ u\in C^\infty (\Sigma,\Lambda^k\Sigma)\,|\, \iota_{\pl \Sigma}^*u=0\}.\]
The space $\mc{H}^k(\Sigma,\pl \Sigma)$ is isomorphic to the real vector space of real-valued closed and co-closed relative forms
\[ {\rm Harm}^k(\Sigma,\pl \Sigma):=\{ u\in C^\infty (\Sigma,\Lambda^k\Sigma)\,|\, \dd u=0, \, \dd*u=0, \, \iota_{\pl \Sigma}^*u=0\}.\] 
The Poincar\'e duality says  that the Hodge star operator
\[  *:  {\rm Harm}^1(\Sigma)\to {\rm Harm}^1(\Sigma,\pl \Sigma) \]
is an isomorphism.  

Using the double $\Sigma^{\#2}$ of $\Sigma$ obtained by gluing $\Sigma$ with itself at the boundary and the natural involution $\tau$ on $\Sigma^{\#2}$, one has that $\mc{H}^1(\Sigma,\pl\Sigma)\simeq \{u \in \mc{H}^1(\Sigma^{\#2})\,|\, \tau^*u=-u\}$ and 
$\mc{H}^1(\Sigma)\simeq \{u \in \mc{H}^1(\Sigma^{\#2})\,|\, \tau^*u=u\}$. We recover that
\[ 4{\mathfrak{g}}+2{\mathfrak{b}}-2=\dim \mc{H}^1(\Sigma^{\#2})=\dim \mc{H}^1(\Sigma)+\dim \mc{H}^1(\Sigma,\pl\Sigma)=2\dim \mc{H}^1(\Sigma)\]  
and there are $2{\mathfrak{g}}+\mathfrak{b}-1$ independent forms in $\mc{H}^1(\Sigma)$ (resp. $\mc{H}^1(\Sigma,\pl\Sigma)$).\\
The duality between $\mc{H}^1(\Sigma)$  and $\mc{H}_1(\Sigma)$ is given by the pairing
\begin{equation}\label{duality1}
\mc{H}^1(\Sigma) \times \mc{H}_1(\Sigma) \to \R , \quad \cjg u,\sigma\cjd:= \int_{\sigma} u.
\end{equation} 
The fact that $\sum_{j=1}^{\mathfrak{b}}[c_j]=0$ in $\mc{H}_1(\Sigma)$ becomes simply Stokes formula: for any closed $1$-form $\omega$ on $\Sigma$
\[ \sum_{j=1}^{\mathfrak{b}} \int_{c_j}\omega=\int_{\pl \Sigma}\omega=\int_{\Sigma} \dd \omega=0.\]
As in \eqref{H^1RSigma} we define for $R>0$
\[\mc{H}^1_R(\Sigma) :=\{ \omega \in \mc{H}^1(\Sigma)\,|\, \forall \sigma\in \mc{H}_1(\Sigma),  \cjg \omega,\sigma\cjd \in 2\pi R\Z\}.\]

For the relative homology $\mc{H}_1(\Sigma,\pl \Sigma)$, the relative condition ensures that the integral of a closed $1$-form on a curve $\sigma$ representing an element $[\sigma]\in \mc{H}_1(\Sigma,\pl \Sigma)$
only depends on the homotopy class of $\sigma$ if the homotopy is such that  the endpoints  of the family of curves stay on $\pl \Sigma$.  The relative homology $\mc{H}_1(\Sigma,\pl \Sigma)$  is dual to $\mc{H}^1(\Sigma,\pl\Sigma)$ by the pairing
\begin{equation}\label{duality2}
\mc{H}^1(\Sigma,\pl\Sigma) \times \mc{H}_1(\Sigma,\pl \Sigma) \to \R , \quad \cjg u,\sigma\cjd:= \int_{\sigma} u  
\end{equation}
and we define 
\[\mc{H}^1_R(\Sigma,\pl\Sigma) :=\{ \omega \in \mc{H}^1(\Sigma,\pl \Sigma)\,|\, \forall \sigma\in \mc{H}_1(\Sigma,\pl \Sigma),  \cjg \omega,\sigma\cjd \in 2\pi R\Z\}.\]

The duality isomorphisms \eqref{duality1} and \eqref{duality2} imply the following:
\begin{lemma}\label{existence_primitive}
1) Fix a basis of $\mc{H}_1(\Sigma)$. Let $\omega_1,\omega_2\in C_{\rm abs}^\infty(\Sigma,\Lambda^1\Sigma)$ such that 
$\int_{\sigma}\omega_1=\int_{\sigma}\omega_2$ for all $\sigma$ in the basis of $\mc{H}_1(\Sigma)$. Then there is $f\in C^\infty(\Sigma)$ with $\pl_\nu f|_{\pl \Sigma}=0$ such that 
\[ \omega_1=\omega_2+\dd f.\]
2) Fix a basis of $\mc{H}_1(\Sigma,\pl \Sigma)$. Let $\omega_1,\omega_2\in C_{\rm rel}^\infty(\Sigma,\Lambda^1\Sigma)$ such that $\int_{\sigma}\omega_1=\int_{\sigma}\omega_2$ for all $\sigma$ in the basis of $\mc{H}_1(\Sigma,\pl \Sigma)$. Then there is $f\in C^\infty(\Sigma)$ with $f|_{\pl \Sigma}=0$ such that 
\[ \omega_1=\omega_2+\dd f.\]
\end{lemma}
Now, we will construct particular bases of $\mc{H}^1_R(\Sigma,\pl\Sigma)$, which have the nice property of being compactly supported. This will be particularly useful for gluing Segal's amplitudes later.
\begin{lemma}\label{compactsupp}
Let $R>0$ and let $\Sigma$ be an oriented compact surface with genus ${\mathfrak{g}}$ and ${\mathfrak{b}}$ boundary connected components $\pl_1\Sigma,\dots,\pl_{\mathfrak{b}}\Sigma$ 
and $\sigma_1,\dots,\sigma_{2{\mathfrak{g}}+{\mathfrak{b}}-1}$ be a basis of $\mc{H}_1(\Sigma,\pl \Sigma)$. Then there 
is a basis $\omega^{c}_1,\dots, \omega^{c}_{2{\mathfrak{g}}+{\mathfrak{b}}-1}$ of $\mc{H}^1_R(\Sigma,\pl\Sigma)$ made of closed forms that are compactly supported inside the interior 
$\Sigma^\circ$ of $\Sigma$ and dual to $\sigma_1,\dots,\sigma_{2{\mathfrak{g}}+{\mathfrak{b}}-1}$ in the sense that 
\[  \frac{1}{2\pi R}\int_{\sigma_k}\omega_j^c =\delta_{jk}.\]
Moreover, $\omega_{2{\mathfrak{g}}+j}^c=df_j$ for some smooth function $f_j$ on 
$\Sigma$ that is locally constant in a neighborhood of $\pl\Sigma$. 
\end{lemma}
\begin{proof} Let us start with a basis $\omega^{r}_1,\dots, \omega^{r}_{2{\mathfrak{g}}+{\mathfrak{b}}-1}$ of $\mc{H}^1(\Sigma,\pl\Sigma)$, dual to  
$\sigma_1,\dots,\sigma_{2{\mathfrak{g}}+{\mathfrak{b}}-1}$ of $\mc{H}_1(\Sigma,\pl \Sigma)$ in the sense above. 
Since $\iota^*_{\pl\Sigma}\omega_j^r=0$, we see that $\int_{\pl_i\Sigma}\omega_j^r=0$ for all $i,j$. In particular, in a collar neighborhood $U_i$ of $\pl_i\Sigma$, we can define for $z_i\in \pl_i\Sigma$
\[ f_{ij}(z):=\int_{\alpha_{z_i,z}} \omega_j^r\]
which is a well-defined smooth function in $U_i$ if $\alpha_{z_i,z}\subset U_i$ is a smooth curve with initial 
endpoint $z_i$ and final endpoint $z$. Notice also that $f_{ij}|_{\pl_i \Sigma}=0$.
Let $\chi_i\in C^\infty(\Sigma)$ with support in $U_i$ and equal to $1$ near $\pl_i\Sigma$. Then $\omega^{c}_j:=\omega_j^r-\sum_{i=1}^{\mathfrak{b}} d(\chi_i f_{ij})$ belongs to $\mc{H}^1(\Sigma,\pl\Sigma)$ and is compactly supported inside $\Sigma^\circ$, and it provides a basis of  $\mc{H}^1(\Sigma,\pl\Sigma)$ since the integrals of $\omega_j^c$ on the cycles forming a basis of $\mc{H}_1(\Sigma,\pl \Sigma)$ are equal to the integrals of $\omega_j^r$ on these cycles.
The fact that $\omega_{2{\mathfrak{g}}+j}^c=df_j$ for some $f_j$ can be checked by writing $f_j(x)=\int_{\alpha_{x_0,x}}\omega_{2{\mathfrak{g}}+j}^c$ for some $x_0\in \pl\Sigma$ where $\alpha_{x_0,x}$ is a smooth curve with endpoint at $x_0$ and $x$ (depending smoothly on $x$), and noting that the result does not depend on the curve by our assumptions on $\omega_{2{\mathfrak{g}}+j}^c$. It is also locally constant near $\pl \Sigma$ since $\omega_{2{\mathfrak{g}}+j}^c$ is compactly supported in $\Sigma^\circ$.
\end{proof}

\begin{lemma}\label{boundary_forms_absolute}
Let $R>0$ and let $\Sigma$ be an oriented compact surface with genus ${\mathfrak{g}}$ and ${\mathfrak{b}}$ boundary connected components $\pl_1\Sigma,\dots,\pl_{\mathfrak{b}}\Sigma$, oriented positively with respect to $\Sigma$, let $c_i$ be the cycle corresponding to $\pl_i\Sigma$. 
Let $(\sigma_j)_{j=1,\dots,2{\mathfrak{g}}}$ be independent cycles of  $\mc{H}_1(\Sigma)$ so that $((\sigma_j)_{j\leq 2{\mathfrak{g}}},(c_i)_{i\leq {\mathfrak{b}}-1})$ form a basis of $\mc{H}_1(\Sigma)$.
Then there are ${\mathfrak{b}}-1$ independent closed forms $\omega_2^{\rm a},\dots,\omega^{\rm a}_{{\mathfrak{b}}}$ in $\mc{H}^1_R(\Sigma)$ such that  
\begin{align*} 
& \forall \ell=2,\dots, {\mathfrak{b}},\,   \frac{1}{2\pi R}\int_{\pl_1\Sigma} \omega_\ell^{\rm a}=-1, \quad \forall i \in [2, {\mathfrak{b}}], \,  \frac{1}{2\pi R} \int_{\pl_i\Sigma} \omega_\ell^{\rm a}=\delta_{\ell i}, \\
& \forall j=1,\dots,2{\mathfrak{g}}, \, \forall \ell=2,\dots,{\mathfrak{b}}, \, 
\int_{\sigma_j}\omega_\ell^{\rm a}=0.
\end{align*}
Moreover, $\omega_\ell^{\rm a}$ can be chosen so that $\omega^{\rm a}_\ell|_{U_i}=0$ for some open neighborhood $U_i$ of $\pl_i\Sigma$ if $i\not\in \{\ell,1\}$, 
while in $U_\ell$ and $U_1$, there are biholomorphisms $\psi_\ell: \{z\in \C\,| \, |z|\in (\delta,1]\}\to U_\ell$ and $\psi_1: \{z\in \C\,|\, |z|\in (\delta,1]\}\to U_1$ for some $\delta<1$ such that $\psi_\ell(\T)=\pl_\ell\Sigma$,  $\psi_1(\T)=\pl_1\Sigma$ and if $z=re^{i\theta}$ are radial coordinates 
\[ \psi_\ell^*\omega_\ell ^{\rm a}=R d\theta, \quad  \psi_1^*\omega_\ell^{\rm a}=-R d\theta.\]  
\end{lemma}
\begin{proof} Let us take
$\omega_2,\dots, \omega_{{\mathfrak{b}}}\in \mc{H}^1(\Sigma)$ such that $(2\pi R)^{-1}\int_{c_i}\omega_\ell=\delta_{i\ell}$ for $i\in [2, {\mathfrak{b}}]$ and $\int_{\sigma_j}\omega_\ell=0$ for all $j\leq 2{\mathfrak{g}}$. We necessarily have  $(2\pi R)^{-1}\int_{c_1}\omega_\ell =-1$ since $\sum_{i=1}^{\mathfrak{b}}c_i=0$ in $\mc{H}_1(\Sigma)$. With the arguments of Lemma \ref{compactsupp}, we can replace $\omega_\ell$ by $\omega_\ell':=\omega_\ell-\sum_{i\not\in \{1,\ell\}}df_{\ell i}$ for some $f_{\ell i}\in C^\infty(\Sigma)$ 
supported in $U_i$ so that $\supp(\omega_\ell')\cap \pl_i\Sigma=\emptyset$ for $i\notin\{\ell,1\}$. Now, in $U_\ell$, we see that 
$\omega_\ell'-R (\psi_\ell)_*d\theta$ integrates to $0$ on $c_\ell=\pl_\ell\Sigma$. Therefore there is $f_{\ell\ell}\in C^\infty(U_\ell)$ with compact 
support in $U_\ell$ such that $\omega_\ell'=R (\psi_\ell)_*d\theta+df_{\ell\ell}$ in $U_\ell$. The same argument shows that there is $f_{\ell {\mathfrak{b}}}\in C^\infty(U_{\mathfrak{b}})$ with compact support in $U_b$ such that $\omega_\ell'=-R (\psi_{\mathfrak{b}})_*d\theta+df_{\ell {\mathfrak{b}}}$ in $U_{\mathfrak{b}}$, and we can just choose
\[ \omega_\ell^{\rm a}:=\omega_\ell-\sum_{i=1}^{\mathfrak{b}} df_{\ell i}\]
which satisfies the desired properties.
\end{proof}

\subsection{Gluing of surfaces and homology/cohomology} \label{sub:opencohomo}

Consider two surfaces $\Sigma_1,\Sigma_2$ with parametrized boundary $\boldsymbol{\zeta}_i=(\zeta_{i1},\dots \zeta_{i\mathfrak{b}_i})$ for $i=1,2$ where $\zeta_{ij}: \T\to \pl_j \Sigma_i$ by identifying 
$\pl_1\Sigma_1\sim \pl_1\Sigma_2$ using the parametrizations of the boundary. 
We can construct a basis of homology on the glued surface  $\Sigma:=\Sigma_1\#\Sigma_2$ from bases of relative/absolute homology/cohomology on $\Sigma_j$.

\begin{figure}[!h]
     \centering
     \begin{subfigure}[b]{0.54\textwidth}
         \centering
         \includegraphics[width=1\textwidth]{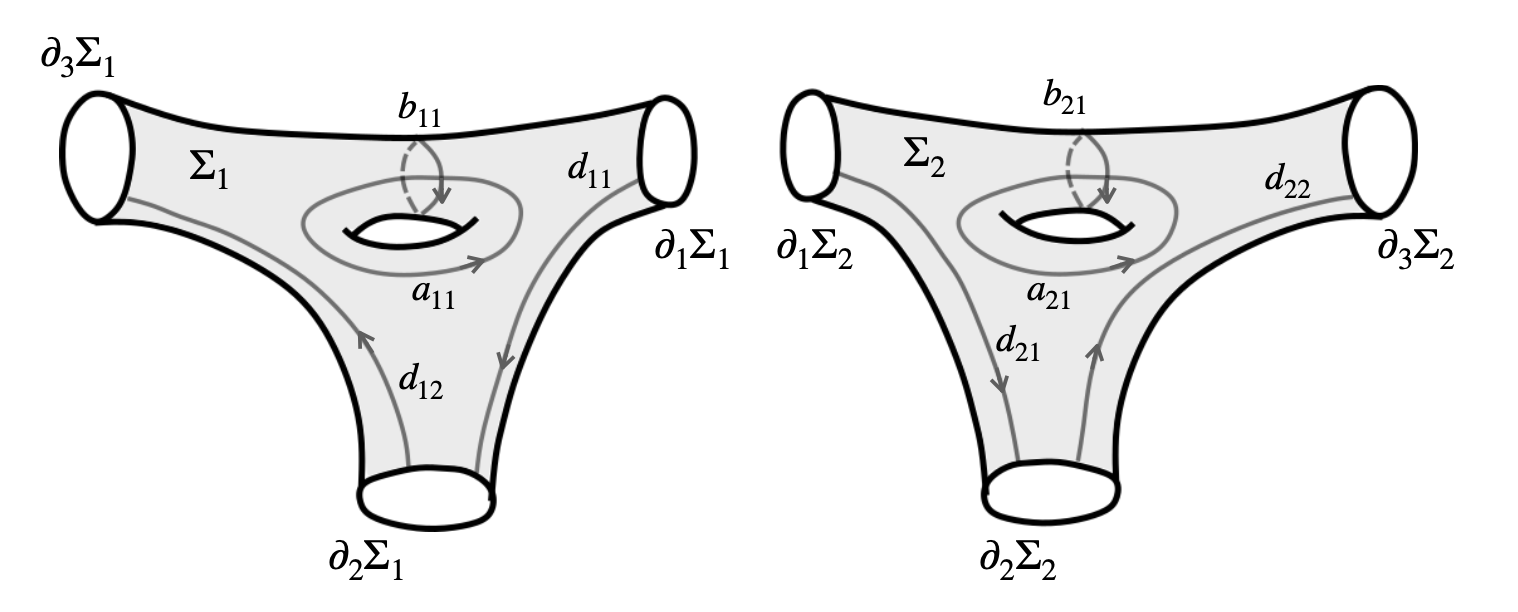}
         \caption{Relative homology on $\Sigma_1$ and $\Sigma_2$.}
     \end{subfigure}
     \hfill
     \begin{subfigure}[b]{0.45\textwidth}
         \centering
         \includegraphics[width=1\textwidth]{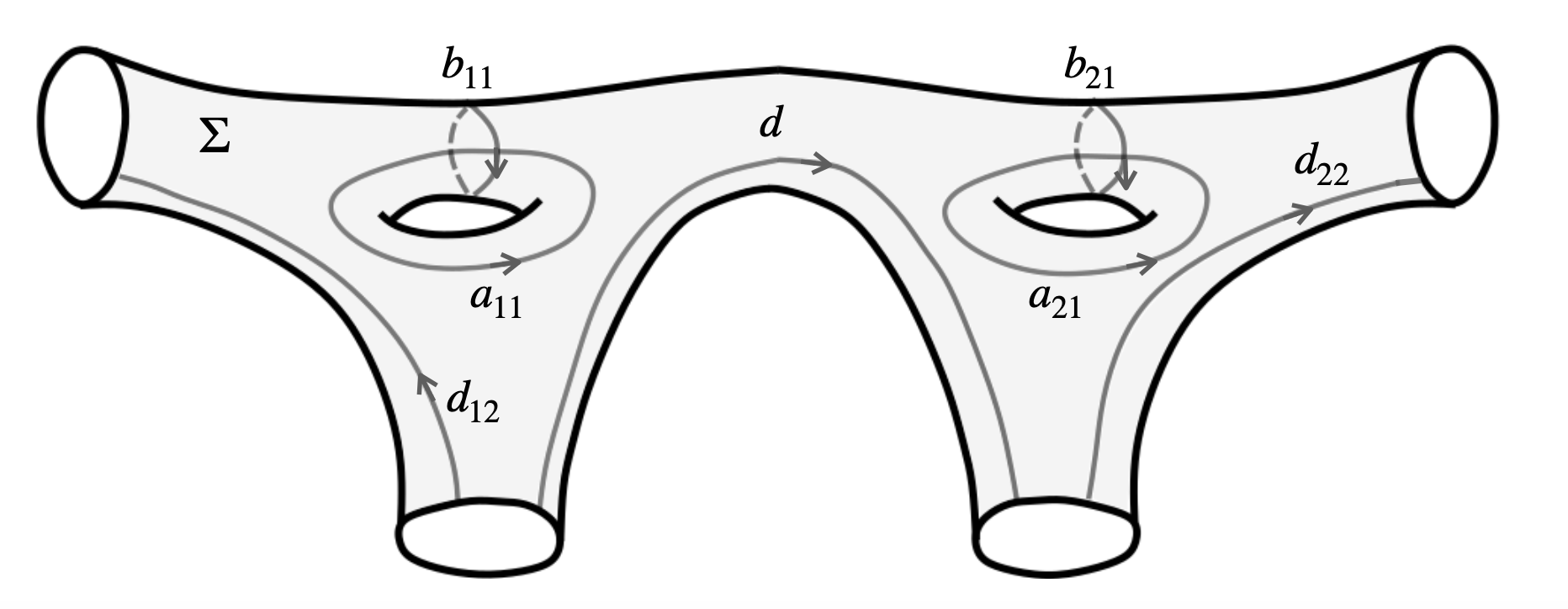}
         \caption{A canonical geometric  basis of  $\mc{H}_1(\Sigma,\pl\Sigma)$.}
     \end{subfigure}
        \caption{Gluing relative homology bases on $\Sigma=\Sigma_1\#\Sigma_2$.}
        \label{fig:glue_homology}
\end{figure}


\begin{lemma}\label{baseglue}
For $i=1,2$, let $\Sigma_i$  be two oriented surfaces with  genus ${\mathfrak{g}}_i$ and ${\mathfrak{b}}_i$ boundary connected components parametrized by $\boldsymbol{\zeta}_i=(\zeta_{i1},\dots,\zeta_{i\mathfrak{b}_i})$ with $\zeta_{ij}: \T\to \pl_j\Sigma_i$ all oriented positively. 
Let $\Sigma$ be the oriented surface  obtained by gluing $\Sigma_1$ to $\Sigma_2$ using the identifications $\pl_1 \Sigma_1\sim -\pl_1 \Sigma_2$ given by $\zeta_{11}(e^{i\theta})=\zeta_{21}(e^{-i\theta})$. The surface $\Sigma$ has genus ${\mathfrak{g}}_1+{\mathfrak{g}}_2$ and ${\mathfrak{b}}_1+{\mathfrak{b}}_2-1$ boundary connected components.\\
1) Let $\boldsymbol{\sigma}_i:=S_i\cup D_i$ be a canonical geometric basis of the relative homology $\mc{H}_1(\Sigma_j,\pl\Sigma_i)$  as described in \eqref{baserelative} for $j=1,2$, with 
\[S_i=( a_{i1}, b_{i1},\dots,a_{i{\mathfrak{g}}_i}, b_{i{\mathfrak{g}}_i}), \quad  D_i:=( d_{i1},\dots,d_{i{\mathfrak{b}}_i-1})  \]
and where $d_{i1}$ are chosen so that the endpoint of $d_{11}$ on $\pl_1\Sigma_1$ coincide with the endpoint of $d_{21}$ on $\pl_1\Sigma_2$. Let 
\[\begin{gathered}
D:=d_{11}-d_{21}\subset \mc{H}_1(\Sigma,\pl \Sigma), \quad  \hat{D}_i=(d_{i2},\dots,d_{i{\mathfrak{b}}_i-1}) \in \mc{H}_1(\Sigma,\pl\Sigma)
\end{gathered}\] 
where we identify cycles in $\Sigma_i$ with their image in $\Sigma$ induced by the inclusion $\Sigma_i\to \Sigma$.
Then the set 
\begin{equation}\label{basisH_1rel} 
\boldsymbol{\sigma}:= D \cup S_1 \cup S_2\cup \hat{D}_1\cup \hat{D}_2  \subset \mc{H}_1(\Sigma,\pl\Sigma)
\end{equation}
forms a canonical geometric basis of $\mc{H}_1(\Sigma,\pl \Sigma)$. See Figure \ref{fig:glue_homology}. We say that $\boldsymbol{\sigma}$ is the gluing of the relative homology bases $\boldsymbol{\sigma}_1$ and $\boldsymbol{\sigma}_2$ and we use the notation 
\[\boldsymbol{\sigma}=\boldsymbol{\sigma}_1\# \boldsymbol{\sigma}_2.\] 
2) Let $\boldsymbol{\sigma}_i= S_i\cup C_i$ be a canonical geometric basis of the absolute homology $\mc{H}_1(\Sigma_i)$  as described in \eqref{baseabsolute} for $i=1,2$, with 
\[S_i=( a_{i1},b_{i1},\dots,a_{i{\mathfrak{g}}_i}, b_{i{\mathfrak{g}}_i}), \quad C_i=(c_{i2}\dots,c_{i{\mathfrak{b}}_i})\subset \mc{H}_1(\Sigma_i)\]  
and $c_{ij}=\pl_j\Sigma_i$ the positively oriented boundary cycles associated to $\pl_1\Sigma_i,\dots,\pl_{{\mathfrak{b}}_i-1}\Sigma_i$.  Let 
\[
\hat{C}_i=(c_{i2},\dots,c_{i({\mathfrak{b}}_i-1)})\subset \mc{H}_1(\Sigma)
\]
where we identify cycles in $\Sigma_i$ with their image in $\Sigma$ induced by the inclusion $\Sigma_i\to \Sigma$. Then the set 
\begin{equation}\label{basisH_1abs}  
\boldsymbol{\sigma}:=S_1\cup S_2\cup C_1\cup \hat{C}_2 \subset \mc{H}_1(\Sigma) 
\end{equation}
is a canonical geometric basis of $\mc{H}_1(\Sigma)$. See Figure \ref{fig:gluing_homology_a.png}. We say that $\boldsymbol{\sigma}$ is the gluing of the absolute homology bases $\boldsymbol{\sigma}_1$ and $\boldsymbol{\sigma}_2$ and we use the notation 
\[\boldsymbol{\sigma}=\boldsymbol{\sigma}_1\# \boldsymbol{\sigma}_2.\] 
\\
3) For $i=1,2$, let $\omega^{i,{\rm c}}_1,\dots, \omega^{i,{\rm c}}_{2{\mathfrak{g}}_i+{\mathfrak{b}}_i-1}\in \mc{H}_R^1(\Sigma_i,\pl \Sigma_i)$ 
be the compactly supported basis from Lemma \ref{compactsupp} associated to $S_i\cup D_i$, with $\omega^{i,{\rm c}}_{2{\mathfrak{g}}_i+j}$ being the form dual to $d_{ij}$
for $j=1,\dots,{\mathfrak{b}}_i-1$. Let $\omega^{1,{\rm a}}_2,\dots, \omega^{1,{\rm a}}_{{\mathfrak{b}}_1}\in \mc{H}^1_R(\Sigma_1)$ and $\omega^{2,{\rm a}}_1,\dots, \omega^{2,{\rm a}}_{{\mathfrak{b}}_2-1}\in \mc{H}^1_R(\Sigma_2)$ be the closed forms from Lemma  \ref{boundary_forms_absolute} chosen such that, 
\[\begin{gathered}
\forall j,\ell \in [2,\mathfrak{b}_1],\,  \frac{1}{2\pi R}\int_{c_{11}}\omega^{1,{\rm a}}_j =-1,\, \quad  \frac{1}{2\pi R}\int_{c_{1\ell}}\omega^{1,{\rm a}}_j =\delta_{j\ell} \\
\forall j,\ell \in [1,\mathfrak{b}_2-1], \, \frac{1}{2\pi R}\int_{c_{2{\mathfrak{b}}_2}}\omega^{2,{\rm a}}_j =-1, \quad  \frac{1}{2\pi R}\int_{c_{2\ell}}\omega^{2,{\rm a}}_j =\delta_{j\ell}.
\end{gathered}\]
Then $\omega^{i,{\rm c}}_1,\dots, \omega^{i,{\rm c}}_{2{\mathfrak{g}}_i+{\mathfrak{b}}_i-1}\subset \mc{H}^1_R(\Sigma_i,\pl \Sigma_i)$ for $i=1,2$, and $\omega^{2,{\rm a}}_2,\dots,\omega^{2,{\rm a}}_{{\mathfrak{b}}_2-1}\subset \mc{H}^1_R(\Sigma_2)$, can all be considered as smooth closed forms on $\Sigma$, extending them by $0$ outside $\Sigma_j$ and $\Sigma_2$ respectively --- recall that $\omega^{2,{\rm a}}_j=0$ near $\pl_{1}\Sigma_2$. The forms 
\[\omega_j^{\rm a}:= \textbf{1}_{\Sigma_1}\omega^{1,{\rm a}}_j+\textbf{1}_{\Sigma\setminus \Sigma_1}\omega^{2,{\rm a}}_{1} \, \textrm{ for }j \in [2,{\mathfrak{b}}_1]\] 
are also smooth and closed on $\Sigma$. The set 
\begin{equation}\label{basisH^1_abs}
 \omega^{1,{\rm c}}_1,\dots, \omega^{1,{\rm c}}_{2{\mathfrak{g}}_1}, \omega^{2,{\rm c}}_1,
\dots, \omega^{2,{\rm c}}_{2{\mathfrak{g}}_2}, \omega_{2}^{1,{\rm a}},\dots,  \omega_{{\mathfrak{b}}_1}^{1,{\rm a}}, 
\omega_{2}^{2,{\rm a}},\dots  \omega_{{\mathfrak{b}}_2-1}^{2,{\rm a}}
\end{equation}
is a basis of $\mc{H}^1_R(\Sigma)$, dual to \eqref{basisH_1abs}. 
The set 
\begin{equation}\label{basisH^1_rel} 
\omega^{1,{\rm c}}_1,\dots, \omega^{1,{\rm c}}_{2{\mathfrak{g}}_1+{\mathfrak{b}}_1-1}, \omega^{2,{\rm c}}_1,
\dots, \omega^{2,{\rm c}}_{2{\mathfrak{g}}_2},  \omega^{2,{\rm c}}_{2{\mathfrak{g}}_2+2},\dots, \omega^{2,{\rm c}}_{2{\mathfrak{g}}_2+{\mathfrak{b}}_2-1}
\end{equation}
is a basis of $\mc{H}^1_R(\Sigma,\pl \Sigma)$, dual to \eqref{basisH_1rel}.
\end{lemma}
\begin{proof} Part 1) is a straightforward exercise of topology using Mayer-Vietoris. For 2), first consider the relative cohomology case. It is readily checked that the forms $\omega^{i,c}_j$ have vanishing integrals on all cycles of \eqref{basisH_1rel} except 
\[\begin{gathered}
\frac{1}{2\pi R}\int_{d_{11}-d_{21}}\omega^{1,c}_{2{\mathfrak{g}}_1+1}=1,\quad \frac{1}{2\pi R}\int_{d_{1j}}\omega^{1,c}_{2{\mathfrak{g}}_1+j}=1\, 
\textrm{ if } j\in [2,{\mathfrak{b}}_1-1], \\
\frac{1}{2\pi R}\int_{d_{2j}}\omega^{2,c}_{2{\mathfrak{g}}_2+j}=1 , \, \textrm{ if } j\in [2,{\mathfrak{b}}_2-1],  \quad \frac{1}{2\pi R}\int_{\sigma_{ij}}\omega^{i,c}_j=1 \textrm{ if } j\in [1,2{\mathfrak{g}}_i].
\end{gathered}\] 
where $(\sigma_{i1},\dots,\sigma_{i2\mathfrak{g}_i})=(a_{i1},b_{i1},\dots,a_{i\mathfrak{g}_i},b_{i\mathfrak{g}_i})$. This shows that \eqref{basisH^1_rel} is a dual basis to \eqref{basisH_1rel}.

Similarly,  the forms $\omega^{2,{\rm a}}_j$ and $\omega^{\rm a}_j$ have vanishing integrals on all cycles of \eqref{basisH_1rel} except for the following
\[
\frac{1}{2\pi R}\int_{c_{1j}}\omega^{a}_j=1 \,  \textrm{ if } j\in [2,{\mathfrak{b}}_1] ,\quad \frac{1}{2\pi R}\int_{c_{2j}}\omega^{2,a}_{j}=1\, \textrm{ if } j\in [2,{\mathfrak{b}}_2-1] 
\] 
which implies, using the pairing we already did in the relative case, that \eqref{basisH^1_abs} is a dual basis to \eqref{basisH_1abs}.
\end{proof}

\begin{figure}[!h]
     \centering
     \begin{subfigure}[b]{0.51\textwidth}
         \centering
         \includegraphics[width=1\textwidth]{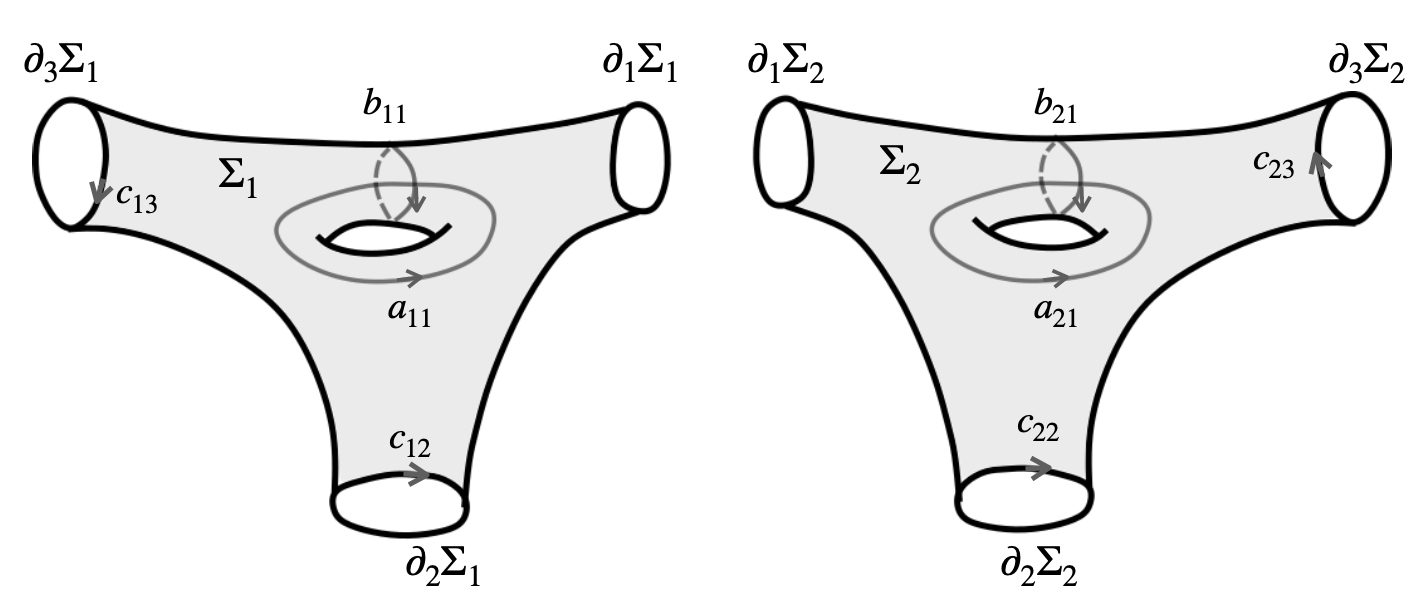}
         \caption{Homology bases on $\Sigma_1$ and $\Sigma_2$.}
     \end{subfigure}
     \hfill
     \begin{subfigure}[b]{0.47\textwidth}
         \centering
         \includegraphics[width=1\textwidth]{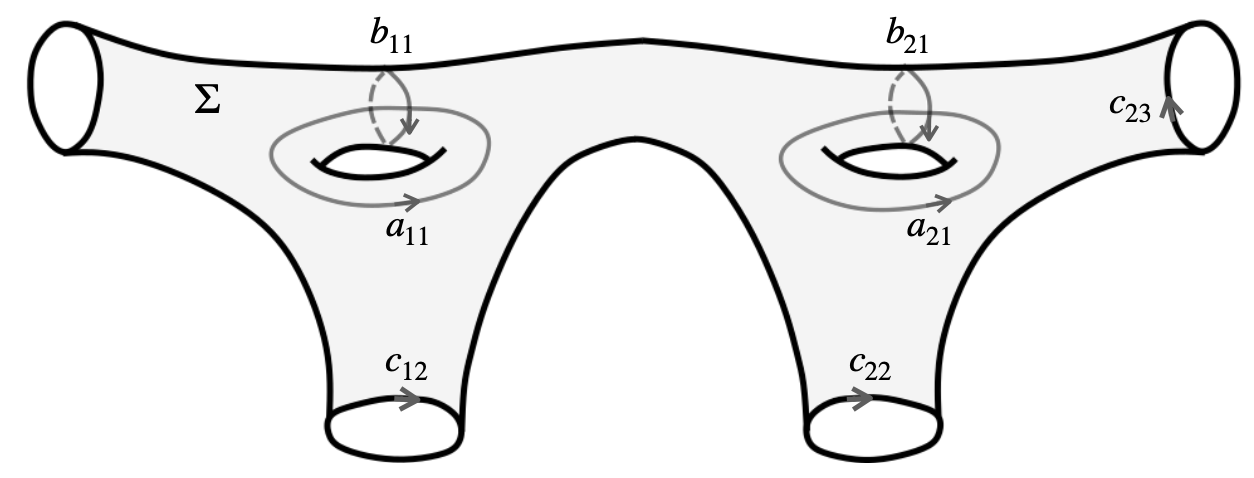}
         \caption{Homology  basis of  $\mc{H}_1(\Sigma)$.}
     \end{subfigure}
        \caption{Gluing  homology bases on $\Sigma=\Sigma_1\#\Sigma_2$.}
        \label{fig:gluing_homology_a.png}
\end{figure}

%

Similarly, with the same argument, we can glue two boundary components on a connected oriented surface.
\begin{lemma}\label{baseselfglue}
Let $\Sigma$ be an oriented surface with genus ${\mathfrak{g}}$ and ${\mathfrak{b}}$ boundary connected 
components $\zeta_\ell: \T\to \pl_\ell\Sigma$ for $\ell\in [1, {\mathfrak{b}}]$, all oriented positively. 
Let $\Sigma^{\#}$ be the oriented surface  obtained by identifying $\pl_1 \Sigma\sim - \pl_2 \Sigma$ in $\Sigma$ via $\zeta_1(e^{i\theta})=\zeta_2(e^{-i\theta})$, where the minus sign $-\pl_2\Sigma$ denotes  $\pl_2\Sigma$  with the reverse orientation. The surface $\Sigma^{\#}$ has genus ${\mathfrak{g}}+1$ and ${\mathfrak{b}}-2$ boundary connected components.\\
1) Let  $\boldsymbol{\sigma}=S\cup D$ be a canonical geometric basis of the relative homology $\mc{H}_1(\Sigma,\pl\Sigma)$  as described in \eqref{baserelative}, with 
\[S=( a_1,\dots,a_{{\mathfrak{g}}_j}, b_1,\dots,b_{{\mathfrak{g}}}), \quad  D:=( d_1,\dots,d_{{\mathfrak{b}}-1})  \]
and where $d_1$ is chosen so that the endpoint of $d_1$ on $\pl_1\Sigma$ coincide with the endpoint of 
$d_1$ on $\pl_2\Sigma$ after gluing $\pl_1\Sigma$ with $\pl_2\Sigma$. 
Let 
\[  C=(c_1,\dots,c_{\mathfrak{b}})\subset \mc{H}_1(\Sigma)\]  
be the boundary cycles associated to $\pl_1\Sigma,\dots,\pl_{{\mathfrak{b}}}\Sigma$ and let $c^\#_\ell$ and $d^\#_\ell$ the image of $c_\ell$ and $d_\ell$ in $\Sigma^{\#}$ after gluing; in particular $c_1^\#=-c_2^\#$ and $d_1^\#$ is a closed curve. 
The set    
\begin{equation}\label{basisH_1rel'} 
\boldsymbol{\sigma}^\#:= S \cup d_1^\#\cup c_1^\#\cup \cup_{\ell=3}^{{\mathfrak{b}}-1}d_j^\# \subset \mc{H}_1(\Sigma,\pl\Sigma)
\end{equation}
is a basis of $\mc{H}_1(\Sigma^{\#},\pl \Sigma^{\#})$. \\
2) With the same notation as in 1), the set
\begin{equation}\label{basisH_1abs'}  
 \boldsymbol{\sigma}^\#:= S \cup d_1^\# \cup c_1^\#\cup \cup_{\ell=3}^{{\mathfrak{b}}-1}c^\#_\ell \subset \mc{H}_1(\Sigma^{\#}) 
\end{equation}
is a canonical geometric basis of $\mc{H}_1(\Sigma^{\#})$.\\
3) Let $\omega^{{\rm c}}_1,\dots, \omega^{{\rm c}}_{2{\mathfrak{g}}+{\mathfrak{b}}-1}\in \mc{H}^1(\Sigma,\pl \Sigma)$ 
be the compactly supported basis from Lemma \ref{compactsupp} associated to $S\cup D$, with $\omega^{{\rm c}}_{2{\mathfrak{g}}+\ell}$ being the form dual to $d_\ell$
for $\ell=1,\dots,{\mathfrak{b}}-1$. Let $\omega^{{\rm a}}_1,\dots, \omega^{{\rm a}}_{{\mathfrak{b}}-1}\in \mc{H}^1(\Sigma)$ be the closed forms from Lemma  \ref{boundary_forms_absolute}. 
Then the form $\omega_1^{{\rm a}}-\omega_2^{\rm a}$ and $\omega_{2\mathfrak{g}+1}^{\rm c}$ both induce smooth closed $1$-forms $\omega_{c_1^\#}$ and $\omega_{d_1^\#}$ 
on the glued surface $\Sigma^{\#}$. 
The set 
\begin{equation}\label{basisH^1_abs'}
 \omega^{{\rm c}}_1,\dots, \omega^{{\rm c}}_{2{\mathfrak{g}}}, \omega_{d_1^\#}, \omega_{c_1^\#}, \omega^{\rm a}_3,\dots,\omega^{\rm a}_{{\mathfrak{b}}-1}
\end{equation}
is a basis of $\mc{H}^1(\Sigma^{\#})$, dual to \eqref{basisH_1abs'}. 
The set 
\begin{equation}\label{basisH^1_rel'} 
\omega^{{\rm c}}_1,\dots, \omega^{{\rm c}}_{2{\mathfrak{g}}}, \omega_{d_1^\#}, \omega_{c_1^\#}, \omega^{\rm c}_{2{\mathfrak{g}}+3},\dots,\omega^{\rm c}_{2{\mathfrak{g}}+{\mathfrak{b}}-1}
\end{equation}
is a basis of $\mc{H}^1(\Sigma^{\#},\pl \Sigma^{\#})$, dual to \eqref{basisH_1rel'}.
\end{lemma}

\subsection{Equivariant functions and distributions} \label{UnivCover}
 
Let $\Sigma$ be a compact Riemann surface, with or without boundary. 
We let ${\mathfrak{b}}\geq 0$ the number of boundary connected components and ${\mathfrak{g}}$ the genus, and 
$\beta_1=\dim \mc{H}^1(\Sigma)$ the first Betti number ($\beta_1=2{\mathfrak{g}}$ if $\pl \Sigma=\emptyset$ and 
$\beta_1=2{\mathfrak{g}}+{\mathfrak{b}}-1$ if $\pl \Sigma\not=\emptyset$).
 The universal cover $\pi: \tilde{\Sigma}_{x_0}\to \Sigma$ of $\Sigma$  can be constructed from a fixed base point $x_0\in \Sigma$ as the set of 
continuous paths $c:[0,1]\to \Sigma$ with $c(0)=x_0$ up to homotopy. It has a distinguished  point 
$\tilde{x}_0$  projecting to $x_0$ (given by the curve $c(t)=c(0)=x_0$),  
and we can view the fundamental group $\pi_1(\Sigma,x_0)$ as a group of deck transformations on $\tilde{\Sigma}_{x_0}$.

For each closed $1$-form $\omega$ on $\Sigma$ (eventually with a boundary), we can construct a function (a primitive)  of $\omega$ on $\tilde{\Sigma}_{x_0}$ by 
\begin{equation}\label{primitive}
I_{x_0}(\omega)(\tilde{x})=\int_{\alpha_{x_0,x}}\omega
 \end{equation}
where the integral is along any $C^1$ path $\alpha_{x_0,x}:[0,1]\to \Sigma_{x_0}$ such that
$\alpha_{x_0,x}(0)=x_0$, $\alpha_{x_0,x}(1)=x:=\pi(\tilde{x})$ and 
$\alpha_{x_0,x}$ lifts to $\tilde{\Sigma}=\tilde{\Sigma}_{x_0}$ as a curve with initial point $\tilde{x}_0$ and endpoint $\tilde{x}$.
 Notice that $\dd I_{x_0}(\omega)=\pi^*\omega$, and if $\omega$ is harmonic on $\Sigma$ then 
 $I_{x_0}(\omega)$ is a harmonic function on $\tilde{\Sigma}$. 
\begin{lemma}\label{phidescend}
Let $\omega\in \mc{H}^1_R(\Sigma)$ be a closed $1$-form on $\Sigma$ with integrals in $2\pi R\Z$ along the homology curves. Then 
$e^{\frac{i}{R}I_{x_0}(\omega)}$ descends to a  well-defined smooth function on $\Sigma$. 
\end{lemma}
 \begin{proof} For $x\in \tilde{\Sigma}_{x_0}$ and $\gamma\in \pi_1(\Sigma,x_0)$, one has 
 \[\begin{split} 
I_{x_0}(\omega)(\gamma.x)-I_{x_0}(\omega)(x)=& \int_{\alpha_{x,\gamma.x}}\pi^*\omega \in 2\pi R\Z 
  \end{split}\]
where $\alpha_{x,y}$ is any $C^1(\tilde{\Sigma}_{x_0})$ curve with $\alpha(0)=x$ and $\alpha(1)=y$ and 
$\alpha_{x,\gamma.x}$ descends to a closed curve on $\Sigma$, thus an element in $\mc{H}_1(\Sigma)$. This shows that 
$e^{\frac{i}{R}  I_{x_0}(\omega)}$ is invariant by $\pi_{1}(\Sigma,x_0)$, and thus descends on $\Sigma$ as a smooth function.
 \end{proof}
 
We can now define the space of equivariant functions on $\tilde{\Sigma}=\tilde{\Sigma}_{x_0}$: let $\Gamma:=\pi_1(\Sigma,x_0)$ be the fundamental group of $\Sigma$, then we define the $\Z$-module
\[ C^\infty_{\Gamma}(\tilde{\Sigma}):=\{ u\in C^\infty(\tilde{\Sigma})\,|\, \forall \gamma\in \Gamma, \gamma^*u-u \in 2\pi R\Z \}.\] 
This space of functions can also be defined by the property that $e^{i\frac{u}{R}}$ is a smooth function that descends on $\Sigma$. Now, we observe that each $u\in C^\infty_{\Gamma}(\tilde{\Sigma})$ induces a map 
\[\chi_u: \Gamma \to 2\pi R\Z ,\quad \chi_u(\gamma):=\gamma^*u-u.\]
One easily checks that $\chi_u({\rm Id})=0$ and $\chi_u(\gamma_1\gamma_2)=\chi_u(\gamma_1)+\chi_u(\gamma_2)$, so that $\chi_u$ is a group morphism. Now,
each  group morphism $\chi: \Gamma\to 2\pi R\Z$ is equivalent to an element in $\mc{H}^1_R(\Sigma)$. Fixing a basis 
$\omega_1,\dots,\omega_{\beta_1}$ of $\mc{H}^1_R(\Sigma)$  there is ${\bf k}\in \Z^{\beta_1}$ such that $\chi(\gamma)= \int_{\gamma}\omega_{{\bf k}}$ for all $\gamma\in \Gamma$ if $\omega_{\bf k}=\sum_{j=1}^{\beta_1}k_j \omega_j$.
We denote $\chi_{\bf k}$ this morphism associated to $\omega_{\bf k}$. Next we define 
\[ C^\infty_{\chi}(\tilde{\Sigma}):=\{ u\in C^\infty(\tilde{\Sigma})\,|\,\forall \gamma\in \Gamma, \gamma^*u-u=\chi(\gamma)\}\subset C^\infty_{\Gamma}(\tilde{\Sigma}).\] 
For $u\in C^\infty_{\chi_{\bf k}}(\tilde{\Sigma})$, we see that $u-I_{x_0}(\omega_{{\bf k}})$ is $\Gamma$-invariant, thus descends to a smooth function on $\Sigma$. We can thus rewrite 
\[ C^\infty_{\chi_{\bf k}}(\tilde{\Sigma})=\{ \pi^*f +I_{x_0}(\omega_{{\bf k}}) \, |\, f\in C^\infty(\Sigma)\}.\]
The discussion above shows that 
\[ C^\infty_{\Gamma}(\tilde{\Sigma})= \cup_{{\bf k}\in \Z^{\beta_1}} C^\infty_{\chi_{\bf k}}(\tilde{\Sigma}).\]
In particular each $u\in C^\infty_{\Gamma}(\tilde{\Sigma})$ can be written in a unique way as  $u= \pi^*f+I_{x_0}(\omega_{{\bf k}})$
for some $f\in C^\infty(\Sigma)$ and some ${\bf k}\in \Z^{\beta_1}$. We can also consider, for $s\in \R$, the Sobolev $\Z$-module 
\begin{equation}\label{HsGamma} 
H^s_{\Gamma}(\tilde{\Sigma})=:\{ f\in H^{s}_{\rm loc}(\tilde{\Sigma})\,|\, \forall \gamma\in \Gamma,\, \gamma^*f-f \in 2\pi R\Z \}
\end{equation}
where $H^s_{\rm loc}(\tilde{\Sigma})$ is the space of distributions that are in the Sobolev space $H^s(U)$ 
on each relatively compact open set $U\subset \tilde{\Sigma}$.
As for smooth functions, each $u\in H^s_{\Gamma}(\tilde{\Sigma})$ can be written as $u= \pi^*f+I_{x_0}(\omega_{{\bf k}})$ for some $f\in H^s(\Sigma)$ and some ${\bf k}\in \Z^{\beta_1}$. In other words, we get an identification 
\[  \Z^{\beta_1} \times H^{s}(\Sigma) \to H^s_{\Gamma}(\tilde{\Sigma}), \quad  ({\bf k},f)\mapsto  \pi^*f+I_{x_0}(\omega_{{\bf k}}).\]

\subsection{Harmonic $1$-forms with poles at prescribed marked points }
 Let $\Sigma$ be a Riemann surface of genus ${\mathfrak{g}}$, with or without boundary. In case there is a boundary, we denote $\pl_j\Sigma$ for $j=1,\dots,{\mathfrak{b}}$ the oriented boundary connected components, where $\varsigma_j=-1$ if the boundary $\pl_j\Sigma$ is outgoing (i.e. positive) and $\varsigma_j=1$ if it is incoming.
 We write $\beta_1=\dim \mc{H}_1(\Sigma)$ the first Betti number. 
 Consider some distinct marked points ${\bf z}=(z_1,\dots,z_{n})\in\Sigma^{n}$ in the interior of $\Sigma$. We attach some winding numbers ${\bf m}:=(m_1,\dots,m_{n})\in \Z^{n}$ to these points, these will be called \emph{magnetic charges}. Denote by $U_1,\dots, U_{n}$ neighborhoods of $z_1,\dots,z_{n}$ and biholomorphic maps $\psi_j:
 \D\to U_j$ such that $\psi_j(0)=z_j$. Let $\hat{\Sigma}:=\Sigma\setminus \cup_{j=1}^n\psi_j(\D^\circ)$. In the disc $\D$, we can use the variable $z=re^{i\theta}$ and we denote by $\dd\theta$ the $1$-form in $\D\setminus \{0\}$ that is closed and coclosed in that pointed disk.
 
  \begin{proposition} \label{harmpoles}
 Let $\boldsymbol{\sigma}=(\sigma_1,\dots,\sigma_{\beta_1})\subset \mc{H}_1(\Sigma)$ 
 be a basis realized by closed curves not intersecting the disks $U_j$ around $z_j$, and such that $\sigma_{2{\mathfrak{g}}+\ell}=\pl_\ell\Sigma$ for $\ell=1,\dots,{\mathfrak{b}}-1$ and $\sigma_\ell\cap \pl \Sigma=\emptyset$ if $\ell\leq 2{\mathfrak{g}}$. Let ${\bf k}=(k_1,\dots,k_{{\mathfrak{b}}})\in \Z^{{\mathfrak{b}}}$ and ${\bf m}:=(m_1,\dots,m_{n})\in \Z^{n}$  and assume that  $\sum_{\ell=1}^{{\mathfrak{b}}}\varsigma_\ell k_\ell +\sum_{j=1}^{n}m_j=0$. 
 Then there exists a smooth real valued closed 1-form on $\Sigma\setminus \{z_1,\dots,z_{n}\}$, denoted by $\nu_{{\bf z},{\bf m}}$ if $\pl \Sigma=\emptyset$, resp. $\nu_{{\bf z},{\bf m},{\bf k}}$ if $\pl\Sigma\not=\emptyset$, such that $\iota_{\pl \Sigma}(i_\nu \nu_{{\bf z},{\bf m},{\bf k}})=0$ and 
\[\left\{  \begin{array}{ll}
\psi_j^*(\nu_{{\bf z},{\bf m}}|_{U_j})=m_jR\dd\theta, & \pl\Sigma=\emptyset \\
 \psi_j^*(\nu_{{\bf z},{\bf m},{\bf k}}|_{U_j})=m_j R\dd\theta, &   \pl\Sigma\not=\emptyset 
\end{array}\right.\]
 and such that for all $j=1,\dots, 2\mathfrak{g}$ and $\ell=1,\dots, {\mathfrak{b}}$,
 \[\left\{\begin{array}{ll}
 \int_{\sigma_j}\nu_{{\bf z},{\bf m}}=0, & \pl\Sigma=\emptyset\\
 \int_{\sigma_j}\nu_{{\bf z},{\bf m},{\bf k}}=0 \,\, \textrm{and }\,\,     \frac{1}{2\pi R} \int_{\pl_\ell\Sigma} \nu_{\mathbf{z},\mathbf{m},\mathbf{k}}=\varsigma_\ell k_\ell, &  \pl\Sigma\not=\emptyset .
\end{array}\right.\]
The form $\nu_{{\bf z},{\bf m}}$ satisfies $\dd^*\nu_{{\bf z},{\bf m}}\in C^\infty(\Sigma)\cap C_c^\infty(\Sigma\setminus \{{\bf z}\})$ when $\pl \Sigma=\emptyset$ and $\dd^*\nu_{{\bf z},{\bf m},{\bf k}}\in C^\infty(\Sigma)\cap C_c^\infty(\Sigma\setminus \{{\bf z}\})$ when $\pl \Sigma=\emptyset$.
Moreover we have, in the distribution sense, 
 \[\left\{\begin{array}{ll}
 \dd\nu_{{\bf z},{\bf m}}=-  2\pi R \sum_{j=1}^{n}m_j \delta_{z_j} & \pl\Sigma=\emptyset \\
 \dd\nu_{{\bf z},{\bf m},{\bf k}}=-  2\pi R \sum_{j=1}^{n}m_j \delta_{z_j} & \pl\Sigma\not=\emptyset 
 \end{array}\right.\]
 where $\delta_{z_j}$ is the Dirac measure at $z_j$. 
There is a unique real-valued closed and coclosed $1$-form $\nu^{{\rm h}}_{{\bf z},{\bf m}}$ if $\pl \Sigma=\emptyset$, resp. 
$\nu^{{\rm h}}_{{\bf z},{\bf m},{\bf k}}$ if $\pl \Sigma\not=\emptyset$,  on $\Sigma\setminus \{{\bf z}\}$ such that 
\begin{equation}\label{nuh-nu}
\left\{ \begin{array}{ll}
 \nu^{{\rm h}}_{{\bf z},{\bf m}}-\nu_{{\bf z},{\bf m}}=\dd f_{{\bf m}}, & \pl \Sigma=\emptyset\\
   \nu^{{\rm h}}_{{\bf z},{\bf m},{\bf k}}-\nu_{{\bf z},{\bf m},{\bf k}}=\dd f_{{\bf m},{\bf k}}, & \pl \Sigma\not=\emptyset
\end{array}  \right.
 \end{equation}
for some $f_{{\bf m}}\in C^\infty(\Sigma)$ if $\pl \Sigma=\emptyset$, resp. $f_{{\bf m},{\bf k}}\in C^\infty(\Sigma)$ with 
$f_{{\bf m},{\bf k}}|_{\pl \Sigma}=0$ if $\pl \Sigma\not=\emptyset$.
\end{proposition}   

 \begin{proof}
For the first claim, when $\pl \Sigma=\emptyset$ it suffices to take a linear combination $\nu_{{\bf z},{\bf m}}=-\sum_{j=2}^{n}m_j \omega_j^{{\rm a}}$ of the forms $\omega_{j}^{\rm a}$ 
of Lemma \ref{boundary_forms_absolute} applied to the surface with boundary $\hat{\Sigma}$, and extend smoothly this form by setting  $\nu_{{\bf z},{\bf m}}|_{U_j}=(\psi_j)_*(m_jR \dd\theta)$ for all $j=1,\dots,n$. When now 
$\pl \Sigma\not=\emptyset$, let us assume the boundary components are all outgoing to simplify the exposition. We define $\pl_j\hat{\Sigma}:=\pl_j\Sigma$ when $j=1,\dots,{\mathfrak{b}}$ and 
$\pl_{{\mathfrak{b}}+j} \hat{\Sigma}:=\pl_j U_j$ for $j=1,\dots,n$ which are all supposed to be oriented positively inside $\hat{\Sigma}$. We then set 
$\nu_{{\bf z},{\bf m},{\bf k}}=-\sum_{j=2}^{n}m_j \omega_{{\mathfrak{b}}+j}^{{\rm a}}+\sum_{j=2}^{{\mathfrak{b}}}k_j\omega_j^a$ where the forms $\omega_{j}^{\rm a}$ are the forms of Lemma \ref{boundary_forms_absolute} applied to the surface with boundary $\hat{\Sigma}$, and we extend smoothly  $\nu_{{\bf z},{\bf m},{\bf k}}$ on $U_j$ 
by setting  $\nu_{{\bf z},{\bf m},{\bf k}}|_{U_j}=(\psi_j)_*(m_j R \dd\theta)$ for all $j=1,\dots,n$. In both cases, $\pl\Sigma=\emptyset$ or $\pl \Sigma\not=\emptyset$, the forms $\nu_{{\bf z},{\bf m}}$ 
and $\nu_{{\bf z},{\bf m},{\bf k}}$ satisfy the desired properties.

Next, we compute $\dd\nu_{{\bf z},{\bf m}}$ and $\dd^*\nu_{{\bf z},{\bf m}}$ in the distribution sense. 
Using that $\dd*\dd\theta=0$ in $\D\cap \{|z|>\eps\}$ and that $*\dd\theta=\dd r/r$, 
if $f\in C_c^\infty(\D^\circ)$ is real valued, 
\[\cjg \dd^*\nu_{{\bf z},{\bf m}},f\cjd= \int_{\Sigma} \nu_{{\bf z},{\bf m}}\wedge * \dd f=\lim_{\eps\to 0}\int_{|z|>\eps}\nu_{{\bf z},{\bf m}}\wedge *\dd f=-\lim_{\eps\to 0}\int_{|z|=\eps}f*\nu_{{\bf z},{\bf m}}=0\]
which shows that $\dd^*\nu_{{\bf z},{\bf m}}=0$ in the distribution sense near $x_j$, and $\dd^*\nu_{{\bf z},{\bf m}}$ is thus smooth on $\Sigma$ since 
$\nu_{{\bf z},{\bf m}}$ is smooth outside ${\bf z}$. 
Now for $\dd\nu_{{\bf z},{\bf m}}$, we already know this is $0$ outside ${\bf z}$. We check that near $z_j$
\[\int_{\Sigma} \nu_{{\bf z},{\bf m}}\wedge *\dd^*(f{\rm v}_g)=\lim_{\eps\to 0}\int_{|z|>\eps}\nu_{{\bf z},{\bf m}}\wedge \dd*(f{\rm v}_g)=-\lim_{\eps\to 0}\int_{|z|=\eps}f\nu_{{\bf x},{\bf m}}=- m_j2\pi R f(0).\]
The same argument applies for $\nu_{{\bf z},{\bf m},{\bf k}}$ in the case $\pl \Sigma\not=\emptyset$.

To find the harmonic form, let $\omega:=\nu_{{\bf z},{\bf m}}$ if $\pl \Sigma=\emptyset$ and 
$\omega:=\nu_{{\bf z},{\bf m},{\bf k}}$ if $\pl \Sigma\not=\emptyset$. Let $H_0^1(\Sigma)=\{f\in H^1(\Sigma)\,|\, f|_{\pl \Sigma}=0\}$ where $H^1(\Sigma)$ is the Sobolev space of functions  $f\in L^2(\Sigma)$ such that $\dd f\in L^2$.
We search a minimizer of the functional on $H^1(\Sigma)$ (resp. $H_0^1(\Sigma)$ if $\pl \Sigma\not=\emptyset$)
\[ E(f):= \|\dd f\|_2^2 +2 \cjg \dd f,\omega\cjd_2=\|\dd f\|_2^2 +2 \cjg f,\dd^*\omega\cjd_2\]
Since $\dd^*\omega$ is smooth on $\Sigma$, we see that 
\[C^{-1}\|f\|^2_{H^1(\Sigma)}-C\|f\|_2^2 \leq E(f)\leq C\|f\|_{H^1(\Sigma)}^2\]
for some $C>0$. By Sobolev compact embedding $H^1(\Sigma)\to L^2(\Sigma)$, this implies that there is a minimizer $f_0$ of $E(f)$ that belongs to $H^1(\Sigma)$ (resp. $H_0^1(\Sigma)$ if $\pl \Sigma\not=\emptyset$). 
It must be a critical point that thus solves
\[\Delta_g f_0+\dd^*\omega=0.\]
The solution $f_0$ is unique up to constant if $\pl \Sigma=\emptyset$ and unique if $\pl \Sigma\not=\emptyset$. We then define $\omega^{\rm h}:=\omega+\dd f_0$: it satisfies  $\dd \omega^{\rm h}=0$ outside ${\bf z}$ and $\dd^*\omega^{\rm h}=
\Delta_g f_0+\dd^*\omega=0$.
 \end{proof} 

The form $\nu_{\mathbf{z},\mathbf{m}}\in L^1(\Sigma)$ (resp. $\nu_{\mathbf{z},\mathbf{m},{\bf k}}\in L^1(\Sigma)$) does not belong to $L^2(\Sigma)$, but we can define a renormalized $L^2$-norm: first, for $\omega\in C^\infty(\Sigma\setminus \{{\bf z}\})$, let 
\[\|\omega\|^2_{g,\epsilon}:=\int_{\Sigma_{{\bf z},\epsilon,g}} \omega \wedge *\bbar{\omega}\]
where $\Sigma_{{\bf z},\epsilon,g}:=\Sigma\setminus \bigcup_{j=1}^{n}B_g(z_j,\epsilon)$ with $B_g(z_j,\epsilon)$ the geodesic ball centered at $z_j$ with radius $\epsilon$ with respect to $g$.

\begin{lemma}\label{renorm_L^2}
Let $\omega=\nu_{{\bf z},{\bf m}}$ if $\pl \Sigma=\emptyset$ or $\omega=\nu_{{\bf z},{\bf m},{\bf k}}$ if $\pl \Sigma\not=\emptyset$. As $\epsilon\to 0$, the following limit exists 
\begin{equation}\label{limit_renorm_norm}
\|\omega\|^2_{g,0}:=\lim_{\epsilon\to 0}\Big(\|\omega\|^2_{g,\epsilon}+  2\pi R^2 (\log \epsilon) \sum_{j=1}^{n} m_j^2\Big).\end{equation}
Furthermore, if $g'=e^\rho g$ is conformal to $ g$ then
\[\|\omega\|^2_{g',0} =\|\omega\|^2_{g,0}+\pi R^2\sum_{j=1}^n m_j^2\rho(z_j).\]
The same holds with $\omega=\nu^{\rm h}_{\mathbf{z},\mathbf{m}}$ or $\omega=\nu^{\rm h}_{\mathbf{z},\mathbf{m},{\bf k}}$.
\end{lemma}
\begin{proof} We consider the case $\pl \Sigma=\emptyset$, the proof in the case with boundary being exactly the same.
Let us write the metric in a small geodesic ball $B_g(z_j,\eps)$ in complex coordinates (using $\psi_j:\D\to U_j$) under the form $g_j:=e^{\rho_j}|\dd z|^2$. One has for $\eps>0$ small that 
\[ \int_{U_j\setminus B_{g}(z_j,\eps)}\nu_{\mathbf{z},\mathbf{m}}\wedge *\bbar{\nu_{\mathbf{z},\mathbf{m}}}=
R^2 m_j^2 \int_{\D\setminus B_{g_j}(0,\eps)} \frac{1}{r^2} (\dd r\wedge r\dd\theta) 
\]
where $z=re^{i\theta}$. Note that $\dd r\wedge r\dd\theta={\rm v}_{|dz|^2}$ is the Euclidean volume form, which is smooth. 
Using the exponential map at $z_j$ for the metric $g_j$, it is direct to check that the Riemannian 
distance $d_j(z):=d_{g_j}(0,z)$ to $0$ in $\mathbb{D}$ satisfies
$d_j(z)=re^{\rho_j(0)}(1+F_j(z))$ with $F_j$ smooth and $F_j(0)=0$. Moreover, one also checks that for $\delta>0$ fixed small, 
as $\eps\to 0$
\[\int_{\delta \D\setminus B_{g_j}(0,\eps)} \frac{1}{r^2} (\dd r\wedge r\dd\theta) =
 2\pi \int_{\eps e^{-\frac{\rho_j(0)}{2}}}^\delta \frac{\dd r}{r}+\mc{O}(1)=-2\pi \log(\eps)+\mc{O}(1).\]
 Next, a direct computation also shows that one can use the Riesz regularization (right hand side) to compute the Hadamard regularization (left hand side):
 \[\lim_{\eps\to 0}\Big(\int_{\delta \D\setminus \{d_j\geq \eps\}} \frac{1}{r^2} {\rm v}_{|dz|^2}+2\pi\log(\eps)\Big)={\rm FP}_{s=0}\int_{\delta \D}d_j(z)^{s}\frac{1}{r^2} {\rm v}_{|dz|^2}\]
where $s\in \C$ and ${\rm FP}_{s=0}$ denotes the finite part (note that the RHS is meromorphic in $s$ with a simple pole at $s=0$). 
If $g'=e^{\rho}g$ is a metric conformal to $g$ and $g'_j=\psi_j^*g'$, 
let $d'_j(z):=d_{g'_j}(0,z)$: then $d'_j(z)=e^{\frac{\rho(z_j)}{2}}d_j'(z)(1+G_j(z))$ for some smooth $G_j(z)$ satisfying $G_j(0)=0$. We then obtain for ${\rm Re}(s)>0$
\[\int_{\delta \D}d'_j(z)^{s}\frac{1}{r^2} {\rm v}_{|dz|^2}=e^{s\frac{\rho(z_j)}{2}}\int_{\delta \D}d_j(z)^{s}(1+G_j(z))^{s}\frac{1}{r^2} {\rm v}_{|dz|^2}.\]
Using that $G_j(z)=\mc{O}(d_j(z))$, we then obtain
\[\begin{split}
e^{s\frac{\rho(z_j)}{2}}\int_{\delta \D}d_j(z)^{s}(1+G_j(z))^{s}\frac{1}{r^2} {\rm v}_{|dz|^2}=& e^{s\frac{\rho(x_j)}{2}}\int_{\delta \D}d_j(z)^{s}(1+s\log(1+G_j(z))+s^2K_j(s,z))
\frac{1}{r^2} {\rm v}_{|dz|^2}\\
=& (1+s\frac{\rho(z_j)}{2})\int_{\delta \D}d_j(z)^{s}\frac{1}{r^2} {\rm v}_{|dz|^2}+ H(s)
\end{split}\]
where $K_j(s,z)$ is holomoprhic in $s$ and smooth in $d_j$, while 
$H(s)$ is holomorphic near $s=0$ and $H(0)=0$ (here we used that $\log(1+G_j(z))=\mc{O}(d_j(z))$, so 
$\int_{\delta \D}d_j(z)^{s}\log(1+G_j(z))r^{-2} {\rm v}_{|dz|^2}$ is analytic near $s=0$). Since as $s\to 0$
\[ \int_{\delta \D}d_j(z)^{s}\frac{1}{r^2} {\rm v}_{|dz|^2}=2\pi \int_0^\delta r^{s-1}dr +\mc{O}(1)= \frac{2\pi}{s}+\mc{O}(1) \]
we deduce that
\[{\rm FP}_{s=0}\int_{\delta \D}d'_j(z)^{s}\frac{1}{r^2} {\rm v}_{|dz|^2}={\rm FP}_{s=0}\int_{\delta \D}d_j(z)^{s}\frac{1}{r^2} {\rm v}_{|dz|^2}+\pi \rho(x_j),\]
which shows the result. The same proof works with $\nu^{\rm h}_{\mathbf{x},\mathbf{m}}$ by using \eqref{nuh-nu}.
\end{proof}

\subsection{Equivariant functions and Sobolev distributions. Case with marked points.}\label{sec:equivariant}
Let $(\Sigma,g)$ be a closed Riemannian surface and ${\bf z}=(z_1,\dots,z_n)$ disjoint marked points on $\Sigma$.
As in Section \ref{UnivCover}, let $\Gamma:= \pi_1(\Sigma\setminus\{{\bf z}\},x_0)$ be the fundamental group of the punctured surface $\Sigma_{\bf z}:=\Sigma\setminus\{{\bf z}\}$, $\tilde{\Sigma}_{\bf z}$ be the universal cover of $\Sigma_{\bf z}$ with $\pi: \tilde{\Sigma}_{\bf z}\to \Sigma_{\bf z}$ the projection, and $\tilde{x}_0\in \tilde{\Sigma}_{\bf z}$ is a fixed preimage of $x_0$ used to define $\tilde{\Sigma}_{\bf z}$. The metric $g$ lifts to $\tilde{\Sigma}_{\bf z}$ and provides a Riemannian measure. We say that $u\in L^2_{\rm loc}(\tilde{\Sigma}_{\bf z})$ if on each fundamental domain $\mc{F}$ of $\Gamma$, $u\in L^2(\mc{F},{\rm dv}_g)$.
Then we can consider the space 
\[  L^2_{\Gamma}(\tilde{\Sigma}_{\bf z}):=\{ u\in L^2_{\rm loc}(\tilde{\Sigma}_{\bf z})\,|\, \forall \gamma\in \Gamma, \gamma^*u-u \in 2\pi R\Z \}.\]
We have $L^2_{\Gamma}(\tilde{\Sigma}_{\bf z})=\bigcup_{\chi} L^2_{\chi}(\tilde{\Sigma}_{\bf z})$
where the union is over the set of group morphisms $\chi: \Gamma_{{\bf z}}\to 2\pi R\Z$ and 
\[ L^2_{\chi}(\tilde{\Sigma}_{\bf z})=\{ u\in L^2_{\rm loc}(\tilde{\Sigma}_{\bf z})\,|\, \forall \gamma\in \Gamma, \gamma^*u-u=\chi(\gamma)\}.\]
Each such group morphism is represented by an element $\omega_{\bf k}+\nu_{{\bf z},{\bf m}}$ 
for $({\bf k},{\bf m})\in \Z^{2{\mathfrak{g}}}\times \Z^n$ via 
\[\chi_{{\bf k},{\bf m}}(\gamma)=\int_{\gamma} (\nu_{{\bf z},{\bf m}}+\omega_{\bf k})\] 
where $\omega_{\bf k}\in \mc{H}^1_R(\Sigma)$ for ${\bf k}\in \Z^{2{\mathfrak{g}}}$ as in \eqref{omega_k} using a basis of $\mc{H}^1_R(\Sigma)$.
This means that 
\[L^2_{\chi_{{\bf k},{\bf m}}}(\tilde{\Sigma}_{\bf z})=\{ \pi^*f +I_{x_0}(\omega_{\bf k})+I_{x_0}(\nu_{{\bf z},{\bf m}}) \, |\, f\in L^2(\Sigma)\}\]
and each element $u\in L^2_{\chi_{{\bf k},{\bf m}}}(\tilde{\Sigma}_{\bf z})$ has a unique decomposition under the form $u=\pi^*f +I_{x_0}(\omega_{\bf k})+I_{x_0}(\nu_{{\bf z},{\bf m}})$.
Here we have set $I_{x_0}(\nu_{{\bf z},{\bf m}})(\tilde{x})=\int_{\alpha_{x_0,x}}\nu_{{\bf z},{\bf m}}$ if $\alpha_{x_0,x}$ lifts to a curve with initial point $\tilde{x}_0$ and endpoint $\tilde{x}$.
We also consider, for $s\in (-1/2,0)$, 
\[ H^s_\Gamma (\tilde{\Sigma}_{\bf z}):=L^2_{\Gamma}(\tilde{\Sigma}_{\bf z})+\pi^*(H^s(\Sigma))\]
where $H^s(\Sigma)$ is the Sobolev space or order $s$ on $\Sigma$.  
There is a one-to-one correspondance 
\[  \Z^{2{\mathfrak{g}}+{\bf m}} \times H^{s}(\Sigma) \to H^s_{\Gamma}(\tilde{\Sigma}_{\bf z}), \quad  ({\bf k},{\bf m},f)\mapsto  \pi^*f+I_{x_0}(\omega_{{\bf k}})+I_{x_0}(\nu_{{\bf z},{\bf m}}).\]

 \section{Curvature term} \label{sub:fund}

Let $(\Sigma,g)$ be a closed oriented Riemannian surface. For $\omega_{\bf k}\in \mc{H}^1_{R}(\Sigma)$ a closed $1$-form, the construction of the path integral will require to make sense of the integrals 
\[\int_\Sigma K_g I_{x_0}(\omega_{\bf k})\dd {\rm v}_g, \quad \int_\Sigma K_g I_{x_0}(\nu_{{\bf z},{\bf m}})\dd {\rm v}_g \] 
 the problem being that $ I_{x_0}(\omega_{\bf k})$ and $I_{x_0}(\nu_{{\bf z},{\bf m}})$ are multivalued on $\Sigma$, i.e. they live on the universal cover $\tilde{\Sigma}$ of $\Sigma$ and $\tilde{\Sigma}_{\bf z}$ of $\Sigma\setminus \{\bf z\}$. 
 We will thus consider $I_{x_0}(\omega_{\bf k})$ and $I_{x_0}(\nu_{{\bf z},{\bf m}})$ as well-defined 
 functions on a dense open set of $\Sigma$ and $\Sigma\setminus \{\bf z\}$ by removing curves. 
 To obtain an invariant definition, it will be required to remove a curvature term coming from these curves.

\subsection{Curvature term associated to $\mc{H}^1_{R}(\Sigma)$}

Before giving the definition of the regularized curvature term, let us introduce the following notation:  if $\boldsymbol{\sigma}=(a_j,b_j)_{j=1,\dots,{\mathfrak{g}}}$ is a geometric symplectic basis of $\mc{H}_1(\Sigma)$, we let 
\begin{equation}\label{sigma^2}
\Sigma_{\boldsymbol{\sigma}}:= \Sigma \setminus \boldsymbol{\sigma}=\Sigma \setminus \cup_{j=1}^{{\mathfrak{g}}}(a_j\cup b_j).
\end{equation}
We observe that any closed form $\omega$ on $\Sigma_{\boldsymbol{\sigma}}$ is exact (see below) and we denote, for 
 $x_0\in \Sigma_{\boldsymbol{\sigma}}$ a fixed base point, 
\[ I^{\boldsymbol{\sigma}}_{x_0}(\omega)(x):=\int_{\alpha_{x_0,x}}\omega\]
defined using the integral of $\omega$ along an oriented path $\alpha_{x_0,x}\subset \Sigma_{\boldsymbol{\sigma}}$ with  $x_0,x$ as initial and final points, depending smoothly on $x$. This
is a well-defined smooth function on $\Sigma_{\boldsymbol{\sigma}}$, not depending on the choice of path $\alpha_{x_0,x}$ in $\Sigma_{\boldsymbol{\sigma}}$, and $\dd I^{\boldsymbol{\sigma}}_{x_0}(\omega)=\omega$.

\begin{definition}\label{curvature_integral}
For $\boldsymbol{\sigma}=(a_j,b_j)_{j=1,\dots,{\mathfrak{g}}}$ a geometric symplectic basis of $\mc{H}_1(\Sigma)$ and $\omega\in \mc{H}^1_R(\Sigma)$
with associated morphism $\chi_\omega :\mc{H}_1(\Sigma) \to 2\pi R\Z$ given by $\chi_\omega (\gamma):=\int_\gamma \omega$, we define the regularized integral 
\begin{equation}\label{def_reg_integtral}
\int_{\Sigma_{\boldsymbol{\sigma}}}^{\rm reg} K_g I^{\boldsymbol{\sigma}}_{x_0}(\omega)\dd {\rm v}_g:=  \int_{\Sigma_{\boldsymbol{\sigma}}}  I^{\boldsymbol{\sigma}}_{x_0}(\omega) K_g\dd {\rm v}_g+2\sum_{j=1}^{{\mathfrak{g}}}
\Big(\chi_\omega (a_j)\int_{b_j}k_{g}\dd \ell_{g}-\chi_\omega (b_j)\int_{a_j}k_{g}\dd \ell_{g}\Big)
\end{equation}
where we use the convention that the geodesic curvature of an oriented curve $c(t)\subset \Sigma$ parametrized by arclength is defined by 
\[ k_{g}(c(t))= \cjg \nabla_{\dot{c}(t)}\dot{c}(t),\nu(t)\cjd_g\]
where $\nu(t)$ is the unit normal to the curve $c$ at $c(t)$ such that $(\dot{c}(t),\nu)$ is positively oriented in $\Sigma$.
\end{definition}
We notice that in the physics literature, a related renormalization was proposed by Verlinde-Verlinde 
\cite[Section 6.3]{verlinde} in the context of bozonization of fermions. 
Our renormalization only uses a symplectic homology basis and not a fundamental domain for $\Sigma$ in the universal cover $\tilde{\Sigma}$. This seems to us more adapted to study the invariance of the curvature integral under change of diffeomorphisms.
 
We state the important invariance properties of the regularized integral in $4$ Lemmas: local invariance with respect to the geometric symplectic basis, invariance by conformal change of metric, diffeomorphism invariance 
and global invariance by change of symplectic basis of $\mc{H}_1(\Sigma)$. The proofs will be done below.

\begin{lemma}[\textbf{Local invariance}]\label{lem:local_invariance}
Let $\boldsymbol{\sigma}'=(a_j',b_j')_{j=1,\dots, {\mathfrak{g}}}$ be another geometric symplectic basis of $\mc{H}_1(\Sigma)$ representing the symplectic homology basis $[\sigma]=([a_j],[b_j])_j$. Then 
\[\int_{\Sigma_{\boldsymbol{\sigma}}}^{\rm reg} K_gI^{\boldsymbol{\sigma}}_{x_0}(\omega)\dd {\rm v}_g=\int_{\Sigma_{\boldsymbol{\sigma}'}}^{\rm reg} K_g I^{\boldsymbol{\sigma}'}_{x_0}(\omega)\dd {\rm v}_g.\]
\end{lemma}

\begin{lemma}[\textbf{Conformal change of metrics}]\label{conformal_change_reg_int}
Let $x_0\in \Sigma$ and $\sigma=(a_j,b_j)_{j=1,\dots,{\mathfrak{g}}}$ be a geometric symplectic basis of $\mc{H}_1(\Sigma)$. Let $\rho\in C^\infty(\Sigma)$ and $\hat{g}=e^{\rho}g$ be two conformally related metrics on $\Sigma$ and 
$\omega\in \mc{H}^1_{R}(\Sigma)$ a closed $1$-form. Then the following identity holds true
\[ 
\int^{\rm reg}_{\Sigma_{\boldsymbol{\sigma}}} I^{\boldsymbol{\sigma}}_{x_0}(\omega)K_{\hat{g}}\dd {\rm v}_{\hat{g}}
=\int^{\rm reg}_{\Sigma_{\boldsymbol{\sigma}}} I^{\boldsymbol{\sigma}}_{x_0}(\omega)K_{g}\dd {\rm v}_{g}+\cjg  d\rho,\omega\cjd_2.\]
 \end{lemma} 

\begin{lemma}[\textbf{Invariance by diffeomorphism}]\label{lemcurvdiff}
 For $\psi: \Sigma\to \Sigma$ an orientation preserving diffeomorphism and $\boldsymbol{\sigma}=(a_j,b_j)_{j=1,\dots,{\mathfrak{g}}}$ a geometric symplectic basis  of $\mc{H}_1(\Sigma)$, let $\psi(\boldsymbol{\sigma})=(\psi(a_j),\psi(b_j))_{j=1,\dots,{\mathfrak{g}}}$ be  the image geometric symplectic basis representing the basis
 $(\psi_*[a_j],\psi_*[b_j])_{j=1,\dots,{\mathfrak{g}}}$ of $\mc{H}_1(\Sigma)$. 
 Let $x_0\in \Sigma$, then, the following identity holds true 
\[ \int_{\Sigma_{\boldsymbol{\sigma}}}^{\rm reg} K_g I^{\boldsymbol{\sigma}}_{x_0}(\omega)\dd {\rm v}_g =\int_{\Sigma_{\psi(\sigma)}}^{\rm reg} K_{\psi_*g} I^{\psi(\boldsymbol{\sigma})}_{\psi(x_0)}(\psi_*\omega)\dd {\rm v}_{\psi_*g}.\]
 \end{lemma} 
\begin{lemma}[\textbf{Invariance by change of symplectic basis of $\mc{H}_1(\Sigma)$}]\label{independence_basis}
 Let $\sigma=(a_j,b_j)_{j=1,\dots,{\mathfrak{g}}}$ and $\sigma'=(a_j',b'_j)_{j=1,\dots,{\mathfrak{g}}}$ be two geometric symplectic bases  
 of $\mc{H}_1(\Sigma)$.   Let $x_0\in \Sigma$, then the following identity holds true 
\[ \int_{\Sigma_{\boldsymbol{\sigma}}}^{\rm reg} K_g I^{\boldsymbol{\sigma}}_{x_0}(\omega)\dd {\rm v}_g -\int_{\Sigma_{\sigma'}}^{\rm reg} K_{g}I^{\boldsymbol{\sigma}'}_{x_0}(\omega)\dd {\rm v}_{g}\in 8\pi^2 R\Z.\]
 \end{lemma}  
 
In order to write the proofs of these Lemmas, we first need to introduce a few notations and geometric 
decomposition of $\Sigma_{\boldsymbol{\sigma}}$.
Choose a geometric symplectic basis $\boldsymbol{\sigma}:=(a_j,b_j)_{j=1,\dots, {\mathfrak{g}}}$ of $\mc{H}_1(\Sigma)$ with intersection points  $x_j:=a_j\cap b_j$.
The surface $\Sigma$ can be decomposed under the form 
\[ \Sigma = S_{\boldsymbol{\sigma}} \cup\cup_{j=1}^{{\mathfrak{g}}} \mc{T}_j\]
where $S_{\boldsymbol{\sigma}}$ is a sphere with ${\mathfrak{g}}$ disks $\mc{D}_1,\dots,\mc{D}_{\mathfrak{g}}$  removed, thus having ${\mathfrak{g}}$ boundary circles $(c_j)_j$, and each 
$\mc{T}_j$ is a torus with a disk $\mc{D}'_j$ removed and whose  non-trivial homology cycles are $(a_j,b_j)$. The boundary of $\mc{D}'_j$ is 
glued to $c_j$.  Thus
\begin{equation}\label{sigma^2}
\Sigma_{\boldsymbol{\sigma}}:= \Sigma \setminus \cup_{i=1}^{{\mathfrak{g}}}(a_j\cup b_j)
\end{equation}
is an open surface which decomposes as  
\[ \Sigma_{\boldsymbol{\sigma}}= S_\sigma\cup \cup_{i=1}^{\mathfrak{g}} K_j \]
where 
\[K_j = \{ z\in \C\,|\, {\rm Re}(z)\in (0,1),{\rm Im}(z)\in (0,1)\} \setminus D(e,\eps) , \quad e=\frac{1}{2}(1+i),\]   
is a square with a small disk $\mc{D}'_j=D(e,\eps)$ centered at $e$ and of radius $\eps<1/4$ removed. The circle $c_j$ is glued to the circle 
$\pl\mc{D}'_j\subset K_j$, and the closure $\bbar{K}_j$ is a surface with a boundary circle $\pl D(e,\eps)$ and with 
$4$ oriented boundary curves 
\[\sigma_{a_j}=\{t\in \C\,| t\in [0,1]\}, \, \bar{\sigma}_{a_j}=\{i+t\in \C\,|, t\in [0,1]\},\, 
\sigma_{b_j}=\{ it\in \C\,|\, t\in [0,1]\}, \, \bar{\sigma}_{b_j}=\{1+it\in \C\,|\, t\in [0,1]\}\]
forming $4$ corners. 
The torus $\mc{T}_j$ is realized as a quotient  $\mc{T}_j = \tilde{\mc{T}}_j/ (\Z+i\Z)$ where 
$\tilde{\mc{T}}_j=\C \setminus \cup_{k\in \Z+i\Z}D(e+k,\eps)$ of $\mc{T}_j$, with the action of the abelian group $\Z+i\Z$ being by translation. 
The generators $1,i$ of $\Z+i\Z\simeq \Z^2$ are identified to the cycles $a_j,b_j$ and we write $\gamma_{a_j}(z)=z+1$ and $\gamma_{b_j}(z)=z+i$.
The set $K_j$ is a fundamental domain for the quotient map $\pi_j:\tilde{\mc{T}}_j\to \tilde{\mc{T}}_j/(\Z+i\Z)$, and  
$\gamma_{a_j}(\sigma_{a_j})=\bar{\sigma}_{a_j}$ and  $\gamma_{b_j}(\sigma_{b_j})= \bar{\sigma}_{b_j}$. 
The curves $\sigma_{a_j},\bar{\sigma}_{a_j}$ (resp. $\sigma_{b_j},\bar{\sigma}_{b_j}$) are lifts of $a_j\cap \mc{T}_j$ (resp. $b_j\cap \mc{T}_j$) to 
$\bbar{K}_j$.

\begin{figure}[h] 
\includegraphics[width=0.7\textwidth]{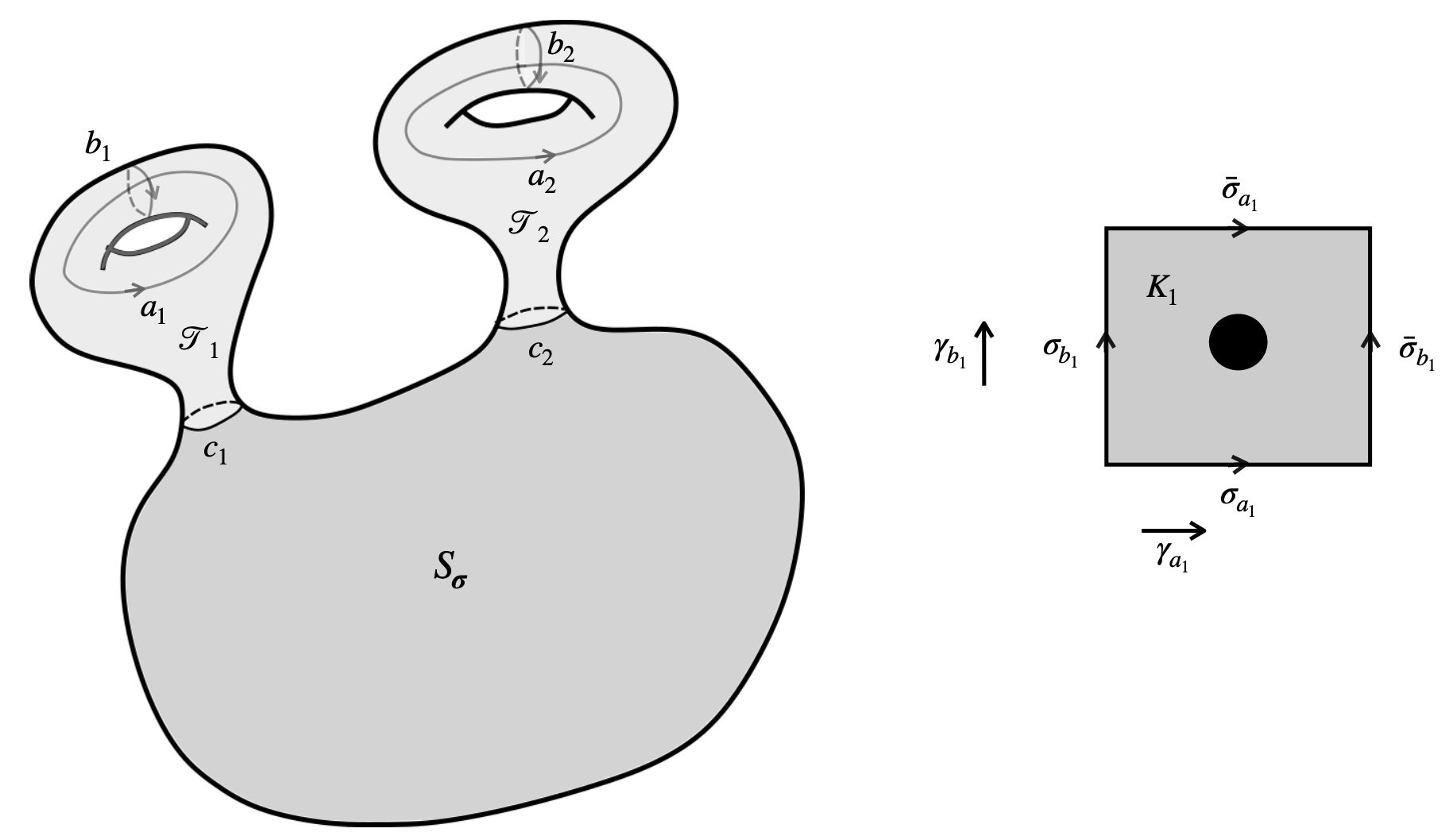} 
\caption{The surface decomposition and a geometric symplectic basis of $\mc{H}_1(\Sigma)$.}
\label{fig:domaine}
\end{figure}

 Since $c_j$ is the boundary of $\mc{T}_j$ (viewed as embedded in $\Sigma$),
we observe that $\int_{c_j}\omega=0$ for all closed form $\omega$. Since the only non contractible closed curves in $\Sigma_{\boldsymbol{\sigma}}$ are generated by the $(c_j)_j$, this shows that all closed form $\omega$ on $\Sigma_{\boldsymbol{\sigma}}$ are exact, as was claimed above, and $I^{\boldsymbol{\sigma}}_{x_0(\omega)}$ is well-defined on $\Sigma_{\boldsymbol{\sigma}}$.
This function, restricted to $\mc{T}_j\setminus \{a_j,b_j\}$, pulls-back by $\pi_j$ to a smooth function on $K_j$ that extends smoothy to $\bbar{K}_j$ satisfying 
\[ \forall z\in \sigma_{a_j}, \, I^{\boldsymbol{\sigma}}_{x_0}(\omega)(\gamma_{b_j}(z))=I^{\boldsymbol{\sigma}}_{x_0}(\omega)(z)+\int_{b_j}\omega, \quad  \forall z\in \sigma_{b_j}, \, I^{\boldsymbol{\sigma}}_{x_0}(\omega)(\gamma_{a_j}(z))=I^{\boldsymbol{\sigma}}_{x_0}(\omega)(z)+\int_{a_j}\omega.\]
If we glue $\tilde{\mc{T}}_j$  to $S_{\boldsymbol{\sigma}}$ by identifying $\pl \mc{D}'_j \subset K_j$ with $c_j\subset \pl S_{\boldsymbol{\sigma}}$, then $I^{\boldsymbol{\sigma}}_{x_0}(\omega)|_{K_j}$ extends smoothly from $K_j$ to $\tilde{\mc{T}}_j$ satisfying 
\[I^{\boldsymbol{\sigma}}_{x_0}(\omega)(\gamma_{b_j}(z))=I^{\boldsymbol{\sigma}}_{x_0}(\omega)(z)+\int_{b_j}\omega,\quad   
I^{\boldsymbol{\sigma}}_{x_0}(\omega)(\gamma_{a_j}(z))=I^{\boldsymbol{\sigma}}_{x_0}(\omega)(z)+\int_{a_j}\omega.\] 
We will call it the \emph{equivariant extension} of $I^{\boldsymbol{\sigma}}_{x_0}(\omega)$ to $\tilde{\mc{T}}_j$
and denote it in the same way $I^{\boldsymbol{\sigma}}_{x_0}(\omega)$. 

We also make the following observation: let $U^\sigma_{x_0}\subset \tilde{\Sigma}$ be the connected component, in the universal cover of $\Sigma$, of the open set $\pi^{-1}(\Sigma_{\boldsymbol{\sigma}})$ that contains the point $\tilde{x}_0$. Then $\pi^*I^{\boldsymbol{\sigma}}_{x_0}(\omega)=I_{x_0}(\omega)$ on $U^{\boldsymbol{\sigma}}_{x_0}$ and 
therefore $\pi^*I^{\boldsymbol{\sigma}}_{x_0}(\omega)$ extends on $\tilde{\Sigma}$ as a smooth function in $C^\infty_\Gamma(\tilde{\Sigma})$.

\begin{proof}[Proof of Lemma \ref{lem:local_invariance}] It suffices to prove that, if $a_j^s, b_j^s$ are families of simple curves representing $[a_j],[b_j]$ with $a_j^0=a_j,b_j^0=b_j$,
then for $\boldsymbol{\sigma}^s=(a_j^s,b_j^s)$ the derivative 
\[ \pl_{s}\Big(\int_{\Sigma_{\boldsymbol{\sigma}^s}}^{\rm reg} K_g I_{x_0}^{\boldsymbol{\sigma}^s}(\omega)\dd {\rm v}_g\Big)|_{s=0}=0.\]
Without loss of generality, we prove this when only one $a_j^s$ is depending on $s$, and for small $s$ the curve $a_j^s\subset \mc{T}^\circ_j$.
Let $a_j^s(t)=a_j(t)+sv_j(t)+\mc{O}(s^2)$ and we consider its lifts to the cover $\tilde{\mc{T}}_j$: this defines two families of curves 
$\sigma_{a_j^s}$ and $\bar{\sigma}_{a_j^s}=\gamma_{b_j}(\sigma_{a_j^s})$ joining the vertical lines $i\R$ and $i\R+1$ 
(the lifts of $\sigma_{b_j}$ and $\bar{\sigma}_{b_j}$) and the domain $K_j^s\subset \{{\rm Re}(z)\in [0,1]\}$ 
enclosed between these two   curves produces a new fundamental domain for $\mc{T}_j$. If $I^{\boldsymbol{\sigma}}_{x_0}(\omega)$ is the smooth equivariant extension of $I^\sigma(\omega)|_{K_j}$ to $\tilde{\mc{T}}_j$, one has 
\[\int_{K_j^s} I^{\boldsymbol{\sigma}^s}_{x_0}(\omega)K_g\dd {\rm v}_g=\int_{K_j^s} I^{\boldsymbol{\sigma}}_{x_0}(\omega)K_g\dd {\rm v}_g\]
and it thus suffices to compute the variation
\[\begin{split}
\pl_s\Big(\int_{K_j^s} I^{\boldsymbol{\sigma}}_{x_0}(\omega)K_g\dd {\rm v}_g\Big)|_{s=0}=& -\int_{\sigma_{a_j}}g(v_j,\nu) I^{\boldsymbol{\sigma}}_{x_0}(\omega)K_g\dd \ell_g+
\int_{\bar{\sigma}_{a_j}}g(d\gamma_{b_j}.v_j,d\gamma_j.\nu) I^{\boldsymbol{\sigma}}_{x_0}(\omega)K_g\dd \ell_g\\
=& \chi(\gamma_{b_j})\int_{\sigma_{a_j}}g(v_j,\nu)K_g\dd \ell_g
\end{split}\]
where $\nu$ is the incoming unit normal vector field to $\sigma_{a_j}$ in $K_j$. 
 \begin{figure}[h] 
\includegraphics[width=0.35\textwidth]{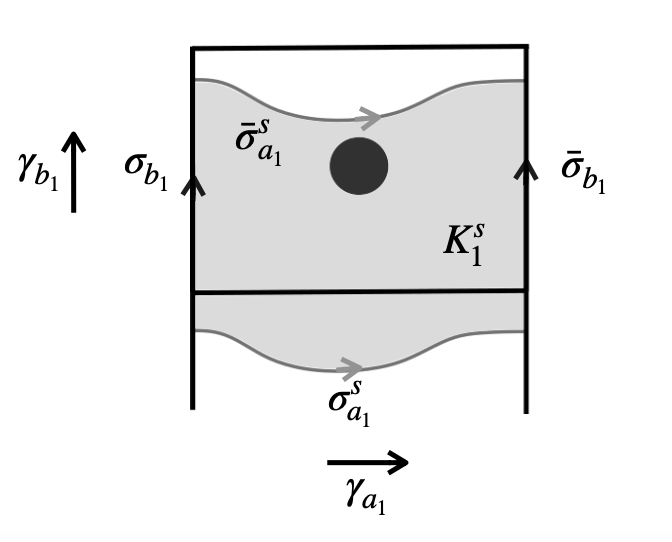} 
\caption{Moving the curve $a_1^s$} 
\label{fig:moving_domaine}
\end{figure} 
We can now differentiate the Gauss-Bonnet formula in the polygonal 
domain $\hat{K}_j^s$ bounded by $\sigma_{a_j^s}\cup \bar{\sigma}_{a_j}\cup i\R\cup (1+i\R)$: since the sum of the interior angles is constant (equal to 
$2\pi$), we get
\[ \int_{\sigma_{a_j}}g(v_j,\nu)K_g\dd \ell_g=-\pl_s\Big(\int_{\hat{K}_{j}^s} K_g\dd {\rm v}_g\Big)|_{s=0}=
2\pl_s\Big(\int_{\sigma_{a_j^s}}k_g d\ell_g\Big)|_{s=0}.\]
This implies that 
\[\pl_s\Big(\int_{K_j^s} I^{\boldsymbol{\sigma}}_{x_0}(\omega)K_g\dd {\rm v}_g\Big)|_{s=0}=2\chi(\gamma_{b_j})\pl_s\Big(\int_{\sigma_{a_j^s}}k_g d\ell_g\Big)|_{s=0}\]
and thus 
\[\pl_s\Big(\int^{\rm reg}_{\Sigma_{\sigma^s}}I^{\boldsymbol{\sigma}^s}_{x_0}(\omega)K_g\dd {\rm v}_g\Big)|_{s=0}=0.\qedhere\]
\end{proof}

Next, we check the conformal covariance of the regularized integral.
 \begin{proof}[Proof of Lemma \ref{conformal_change_reg_int}] We use the relation $K_{\hat{g}}=e^{-\rho}(K_g+\Delta_{g}\rho)$ and $\dd{\rm v}_{\hat{g}}=e^{\rho}\dd{\rm v}_{g}$, 
 then by integration by parts 
 \[\begin{split} 
 \int_{\Sigma_{\sigma}} I^{\boldsymbol{\sigma}}_{x_0}(\omega)K_{\hat{g}}\dd {\rm v}_{\hat{g}}=&\int_{\Sigma_{\boldsymbol{\sigma}}} I^{\boldsymbol{\sigma}}_{x_0}(\omega)K_{g}\dd {\rm v}_{g}+ 
 \cjg  d\rho,\omega\cjd_2+ \sum_{j=1}^{{\mathfrak{g}}}\int_{\sigma_{a_j}}\pl_{\nu}\rho I^{\boldsymbol{\sigma}}_{x_0}(\omega)\dd \ell_g\\
& +  \sum_{j=1}^{{\mathfrak{g}}}\int_{\bar{\sigma}_{a_j}}\pl_{\nu}\rho I^{\boldsymbol{\sigma}}_{x_0}(\omega)\dd \ell_g + 
\int_{\sigma_{b_j}}\pl_{\nu}\rho I^{\boldsymbol{\sigma}}_{x_0}(\omega)\dd \ell_g +\int_{\bar{\sigma}_{b_j}}\pl_{\nu}\rho I^{\boldsymbol{\sigma}}_{x_0}(\omega)\dd \ell_g 
\end{split}
\]
 where $\pl_\nu$ is the  interior boundary unit normal pointing vector in $K_j$.
We can use that $\bar{\sigma}_{a_j}=\gamma_{b_j}(\sigma_{a_j})$ and $\bar{\sigma}_{b_j}=\gamma_{a_j}(\sigma_{b_j})$  
that $\rho$ is a well define smooth function on $\Sigma$: since 
 $I^{\boldsymbol{\sigma}}_{x_0}(\omega)(\gamma_{\sigma_j} x)=I^{\boldsymbol{\sigma}}_{x_0}(\omega)(x)+\chi(\gamma_{\sigma_j})$ for $\sigma\in \{a,b\}$ and
  $(\gamma_{b_j})_*\pl_{\nu}=-\pl_\nu$ on $\sigma_{a_j}$, 
\[\int_{\sigma_{a_j}}\pl_{\nu}\rho I^{\boldsymbol{\sigma}}_{x_0}(\omega)\dd \ell_g+\int_{\bar{\sigma}_{a_j}}\pl_{\nu}\rho I^{\boldsymbol{\sigma}}_{x_0}(\omega)\dd \ell_g=-\chi(\gamma_{b_j}) \int_{\sigma_{a_j}}\pl_{\nu}\rho \,\dd \ell_g, \]
\[\int_{\sigma_{b_j}}\pl_{\nu}\rho I^{\boldsymbol{\sigma}}_{x_0}(\omega)\dd \ell_g+\int_{\bar{\sigma}_{b_j}}\pl_{\nu}\rho I^{\boldsymbol{\sigma}}_{x_0}(\omega)\dd \ell_g=\chi(\gamma_{a_j}) \int_{\sigma_{b_j}}\pl_{\nu}\rho \,\dd \ell_g,\]
and we compute that $k_{\hat{g}}\dd \ell_{\hat{g}}=k_{g}\dd \ell_{g}-\demi\pl_{\nu}\rho \, \dd \ell_{g}$ on $\sigma_j$. 
Combining these facts, we obtain the desired formula for the conformal change of $\int_{\Sigma_{\boldsymbol{\sigma}}}^{\rm reg} K_g I^{\boldsymbol{\sigma}}_{x_0}(\omega)\dd {\rm v}_g$.
 \end{proof}
 
 The next step is to check the invariance of the regularized integral with respect to diffeomorphism.

\begin{proof}[Proof of Lemma \ref{lemcurvdiff}]
First we observe that for $\omega \in \mc{H}^1_R(\Sigma)$ and $y\in \Sigma_{\psi(\boldsymbol{\sigma})}$ 
\[ I_{x_0}(\omega)(\psi^{-1}(y))=\int_{x_0}^{\psi^{-1}(y)}\omega=\int_{\psi(x_0)}^{y}\psi_*\omega =I^{\psi(\boldsymbol{\sigma})}_{\psi(x_0)}(\psi_*\omega)(y)\]
thus we get 
\[ \int_{\Sigma_{\boldsymbol{\sigma}}}  I^{\boldsymbol{\sigma}}_{x_0}(\omega) K_g\dd {\rm v}_g=\int_{\Sigma_{\psi(\boldsymbol{\sigma})}} I^{\psi(\boldsymbol{\sigma})}_{\psi(x_0)}(\psi_*\omega)K_{\psi_*g}\dd {\rm v}_{\psi_*g}.\]
On the other hand, we also have for $\sigma\in \{a,b\}$
\[ \chi(\gamma_{\sigma_j})=\int_{\sigma_j}\omega=\int_{\psi(\sigma_j)}\psi_*\omega,\]
and thus we obtain
\[ \sum_{j=1}^{{\mathfrak{g}}}\int_{b_j}\omega \int_{a_j}k_{g}\dd {\rm v}_{g}=\sum_{j=1}^{{\mathfrak{g}}}\int_{\psi(b_j)}\psi_*\omega \int_{a_j}k_{\psi_*g}\dd {\rm v}_{\psi_*g}\]
and similarly when we exchange $a_j$ and $b_j$. This ends the proof.
\end{proof}  

The final step consists in proving Lemma \ref{independence_basis} that the regularized curvature integral of $I^{\boldsymbol{\sigma}}_{x_0}(\omega)$ on $\Sigma_{\boldsymbol{\sigma}}$ 
does not depend on the choice of canonical basis of $\mc{H}_1(\Sigma)$. Since the proof is slightly more technical and longer, we defer it to Appendix \ref{Lemma_indep} for readability.

 \subsection{Curvature terms associated to magnetic points.}\label{curvature_mp}

 Let ${\bf z}=(z_1,\dots,z_{n})$ be disjoint marked points and ${\bf m}=(m_1,\dots,m_n)\in \Z^n$ some associated magnetic charges, and let $v_j \in T_{z_j}\Sigma$ some unit tangent vectors (with respect to the metric $g$ on $\Sigma$), and denote ${\bf v}=((z_1,v_1),\dots,(z_{n},v_n))\in (T\Sigma)^n$. Consider the closed 1 form $\nu_{{\bf z},{\bf m}} $ of Proposition \ref{harmpoles}.
 As above, we will need to define $\int_\Sigma K_g I_{x_0}(\nu_{{\bf z},{\bf m}}){\rm dv}_g$, but $ I_{x_0}(\nu_{{\bf z},{\bf m}})$ being multivalued, we have to remove some curves and curvature terms along these curves to obtain a natural quantity.
 
We need a family of arcs, which we call \emph{defect lines} and form a defect graph, constructed as follows:
\begin{definition}{\bf Defect graph:}
We consider a family of $n-1$ arcs as follows:
\begin{itemize}
\item these arcs are indexed by $p\in\{1,\dots, n-1\}$, are simple and do not intersect except possibly at their endpoints,
\item each arc is a smooth oriented curve $\xi_p:[0,1]\to  \Sigma $ parametrized by arclength with endpoints $\xi_p(0)=z_j$ and $\xi_{p}(1)=z_{j'}$ for some $j\not=j'$, with orientation in the direction of increasing charges, meaning $m_j\leq m_{j'}$.
\item these arcs reach the endpoints in the directions prescribed by $\mathbf{v}$, meaning  $\xi_{p}'(0)=\lambda_{p,j} v_{j}$  and $\xi_{p}'(1)=\lambda_{p,j'} v_{j'}$ for some  $\lambda_{p,j},\lambda_{p,j'} >0$. 
\item consider the  oriented  graph with vertices $\mathbf{z}$ and edges $(z_j,z_{j'})$ when there is an arc connecting $z_j$ to $z_{j'}$. This graph must be connected and without cycle, i.e. there is no sequence of edges $(z_{j_1},z_{j_2}),\dots,(z_{j_k},z_{j_{k+1}})$ with $j_1=j_{k+1}$.  
\end{itemize}
In what follows, the union $\mc{D}_{\mathbf{v},\boldsymbol{\xi}}:=\bigcup_{p\in\{1,\dots, n-1\}}\xi_p([0,1])$ will be called the defect graph associated to $\mathbf{v}$ and the collection of arcs $\boldsymbol{\xi}:=(\xi_{1},\dots, \xi_{n-1})$.
\end{definition}
We notice that the graph $\mc{D}_{\mathbf{v},\boldsymbol{\xi}}$, viewed as a subset of $\Sigma$, is homotopic to a point.  Let us first state the following lemma, the proof of which will be done below.  
\begin{lemma}\label{exactform}
On $\Sigma\setminus \mc{D}_{\mathbf{v},\boldsymbol{\xi}}$, the 1-form $\nu_{\mathbf{z},\mathbf{m}}$ is exact.
\end{lemma}

If $\boldsymbol{\xi}$ is a  defect graph for ${\bf z}$, we denote by 
\[I^{ \boldsymbol{\xi}}_{x_0}( \nu_{\mathbf{z},\mathbf{m}})(x)=\int_{\alpha_{x_0,x}} \nu_{\mathbf{z},\mathbf{m}}\] 
the primitive of $ \nu_{\mathbf{z},\mathbf{m}}$ on $\Sigma\setminus \mc{D}_{\mathbf{v},\boldsymbol{\xi}}$ vanishing at $x_0\in\Sigma$, where $\alpha_{x_0,x}$ is any smooth curve in $\Sigma\setminus \mc{D}_{\mathbf{v},\boldsymbol{\xi}}$, the result being independent of the curve by Lemma \ref{exactform}. 
Note that the mapping 
\[z\mapsto e^{\frac{i }{R} I^{ \boldsymbol{\xi}}_{x_0}( \nu_{\mathbf{z},\mathbf{m}})}\] 
 is single valued  on $\Sigma \setminus \{\mathbf{z}\}$. 

\begin{definition}{\bf Regularized curvature:}
We assign to each arc $\xi_p$ in the defect graph a value $\kappa(\xi_p)\in 2\pi R\Z$, corresponding to the difference of the values of $I^{ \boldsymbol{\xi}}_{x_0}( \nu^{\rm h}_{\mathbf{z},\mathbf{m}})$ on both sides of the arc,  as follows: take a small neigborhood $\mc{D}_{\mathbf{v},\boldsymbol{\xi}}(\eps)$ of $\mc{D}_{\mathbf{v},\boldsymbol{\xi}}$ homeomorphic to a disk for some small $\eps>0$, and for $x\in \xi_p(]0,1[)$, consider a $C^1$ simple curve $\alpha_x:[0,1]\to \mc{D}_{\mathbf{v},\boldsymbol{\xi}}(\eps)$ with endpoints $\alpha_x(0)=\alpha_x(1)=x$, with $\alpha_x(]0,1[)\cap 
\mc{D}_{\mathbf{v},\boldsymbol{\xi}}=\emptyset$, such that the disk bounded by $\alpha_x([0,1])$ contains at least one point of ${\bf z}$, $(\dot{\alpha}_x(t),\nu(t))$ is a positive basis of the tangent space of $\Sigma$ at $\alpha_x(t)$ if $\nu(t)$ is the unit inward normal to the disk enclosed by $\alpha_x$ and the angle between the curve $\xi_p$ at $x$ and $\dot{\alpha}(0)$ is $\pi/2$ (i.e. we start from the left face of the arc $\xi_p$, see Figure \ref{fig:kappa}). 
Then we set 
\begin{equation}\label{defkappaxip}
\kappa(\xi_p)=\int_{\alpha_x} \nu_{\mathbf{z},\mathbf{m}}.
\end{equation} 

We define the regularized curvature term  similarly to Definition \ref{curvature_integral} by 
\begin{equation}\label{curv:mag2}
\int_{\Sigma}^{\rm reg}  I^{ \boldsymbol{\xi}}_{x_0}( \nu_{\mathbf{z},\mathbf{m}}) K_g \,\dd v_g:=\int_{\Sigma}  I^{ \boldsymbol{\xi}}_{x_0}( \nu_{\mathbf{z},\mathbf{m}}) K_g \,\dd v_g-2\sum_{p=1}^{n_\mathfrak{m}-1}\kappa(\xi_p)\int_{\xi_p}k_g\dd \ell_g.
\end{equation}
\end{definition}

  \begin{figure}[h!] 
\centering
\includegraphics[width=.7\textwidth]{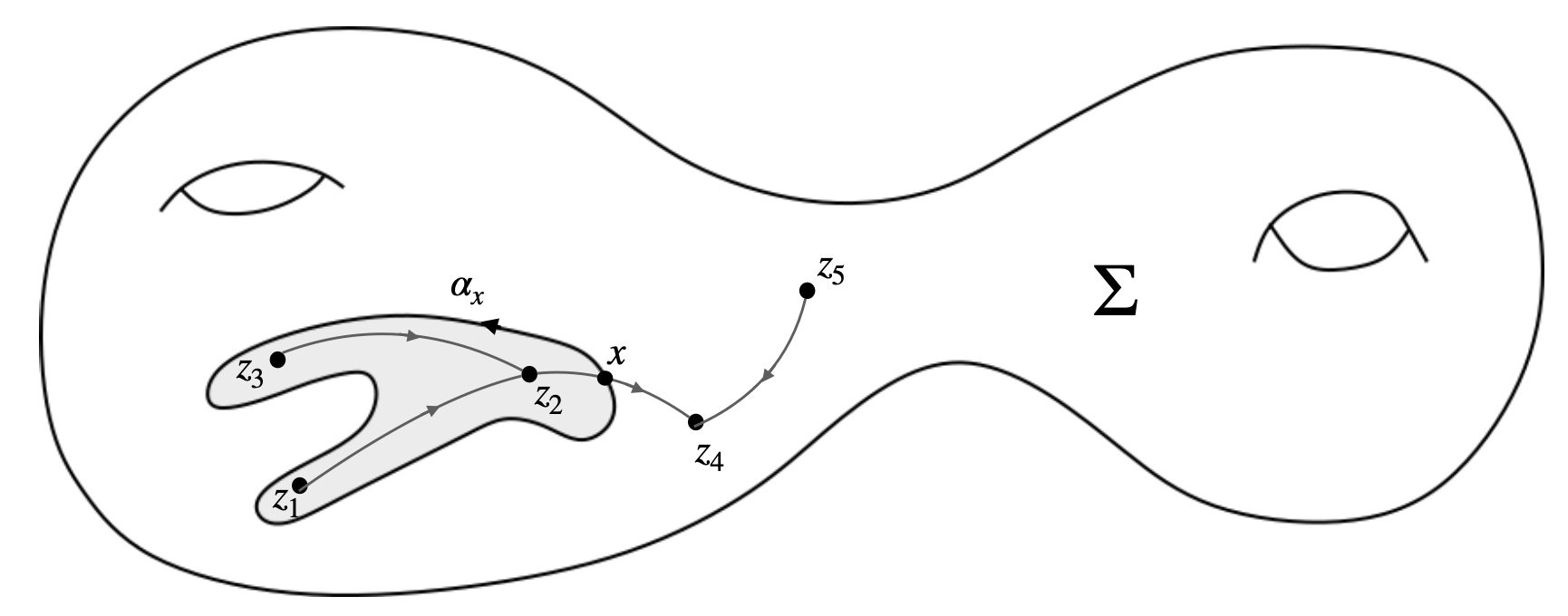} 
\caption{Curve $\alpha_x$: in gray the defect graph. The curve $\alpha_x$ starting at the point $x\in \xi_p$ is the boundary of the shaded red area corresponding to a topological disk. On this example $\kappa(\xi_p)=m_1+m_2+m_3$.}
\label{fig:kappa}
\end{figure}

This quantity can also be written as $ \kappa(\xi_p):=2\pi R\sum_{j\in J_p}m_j$ where $J_p$ is set of points $z_j$ enclosed by the curve $\alpha_{x}$. Here we assume that the number of turns of $\alpha_x$ around  $\mc{D}_{\mathbf{v},\boldsymbol{\xi}}$ is $1$, but taking curves which turn $k\geq 1$ times (with positive orientation) would lead to the same result by using that $\sum_{j=1}^{n}m_j=0$. As before, 
let $\pi: \tilde{\Sigma}_{\bf z}\to \Sigma_{\bf z}$ be the covering map on the universal cover of $\Sigma_{\bf z}=\Sigma\setminus \{z\}$,  $\Gamma=\pi_1(\Sigma_{\bf z},x_0)$ the fundamental group of $\Sigma_{\bf z}$ with $x_0\in \Sigma_{\bf z}$, $\tilde{x}_0\in \tilde{\Sigma}_{\bf z}$ a point so that $\pi(\tilde{x}_0)=x_0$ 
and $U_{x_0}\subset  \tilde{\Sigma}_{\bf z}$ the connected component of $\pi^{-1}(\Sigma_{\bf z})$ containing $\tilde{x}_0$. Then the function $\pi^*I^{ \boldsymbol{\xi}}_{x_0}( \nu_{\mathbf{z},\mathbf{m}})|_{U_{x_0}}$ 
is equal to 
\[I_{x_0}(\nu_{\mathbf{z},\mathbf{m}})(\tilde{x})=\int_{\alpha_{x_0,x}} \nu_{\mathbf{z},\mathbf{m}}\] 
where $\alpha_{x_0,x}\subset \Sigma_{\bf z}$ is any smooth curve with initial point $x_0$ and endpoint $x$ so that its lift to $\tilde{\Sigma}_{\bf z}$ is a curve with initial point $\tilde{x}_0$ and endpoint $\tilde{x}$. The function $\pi^*I^{ \boldsymbol{\xi}}_{x_0}( \nu^{\rm h}_{\mathbf{z},\mathbf{m}})|_{U_{x_0}}$ thus  extends smoothly as an equivariant function in $C^\infty_{\Gamma}(\tilde{\Sigma}_{\bf z})$.

We state now the main properties of the regularized curvature and prove all the lemmas in this subsection after this.

\begin{lemma}[\textbf{Invariance with respect to defect graph}]\label{inv_par_graphe}
The value of $\int_{\Sigma}^{\rm reg}  I^{ \boldsymbol{\xi}}_{x_0}( \nu_{\mathbf{z},\mathbf{m}}) K_g \,\dd  {\rm v}_g$ is independent of the choice of defect graph $\boldsymbol{\xi}$. 
\end{lemma}

\begin{lemma}[{\bf Conformal change of metrics}]\label{change_conf_magnetic}
Consider two conformal metrics $g'=e^{\rho}g$. The regularized magnetic curvature term defined by \eqref{curv:mag2}  satisfies
\[\int_{\Sigma}^{\rm reg}  I^{\boldsymbol{\xi}}_{x_0}( \nu^{\rm h}_{\mathbf{z},\mathbf{m}}) K_{g'} \, {\rm dv}_{g'}=\int_{\Sigma}^{\rm reg}  I^{\boldsymbol{\xi}} _{x_0}( \nu^{\rm h}_{\mathbf{z},\mathbf{m}}) K_g \, {\rm dv}_g.\]
\end{lemma}

\begin{proof}[Proof of Lemma \ref{exactform}] 
By construction $\int_{\sigma_j}\nu_{\mathbf{z},\mathbf{m}}=0$ for all cycles $\sigma_j$ in $\mc{H}_1(\Sigma)$. Let $u(x):=\int_{\alpha_{x_0,x}} \nu_{\mathbf{z},\mathbf{m}}$ where $\alpha_{x_0,x}$ is a $C^1$ curve with image in 
$\Sigma\setminus \mc{D}_{\mathbf{v},\boldsymbol{\xi}}$ and endpoints at $x_0$ and $x$. The function $u$ is a priori multivalued and $\dd u=\nu_{\mathbf{z},\mathbf{m}}$. To prove it is singled valued, it suffices to check that the value of $u$ does not depend on the choice of curve $\alpha_{x_0,x}$.  Taking two such curves, we get two (a priori multivalued) primitives $u$ and $u'$ of $\nu_{\mathbf{z},\mathbf{m}}$ with $u(x)-u'(x)=\int_{\beta_{x_0}}\nu_{\mathbf{z},\mathbf{m}}$ for some closed curve $\beta_{x_0}\in \pi_1(\Sigma\setminus \mc{D}_{\mathbf{v},\boldsymbol{\xi}},x_0)$ in the fundamental group of $\Sigma\setminus \mc{D}_{\mathbf{v},\boldsymbol{\xi}}$. By assumption, 
the graph  $\mc{D}_{\mathbf{v},\boldsymbol{\xi}}$, viewed as a subset of $\Sigma$, is homotopic to a point, thus the first absolute
homology group of $\Sigma \setminus \mc{D}_{\mathbf{v},\boldsymbol{\xi}}$ is that of $\Sigma$ with a small disk removed, thus isomorphic to $\mc{H}_1(\Sigma)$. This means that 
the homology class $[\beta_{x_0}]$ of $\beta_{x_0}$ in $\Sigma\setminus \mc{D}_{\mathbf{v},\boldsymbol{\xi}}$ is a linear combination 
of the basis elements $[\sigma_j]\in \mc{H}_1(\Sigma)\simeq \mc{H}_1(\Sigma\setminus \mc{D}_{\mathbf{v},\boldsymbol{\xi}})$, and therefore $\int_{\beta_{x_0}}\nu_{\mathbf{z},\mathbf{m}}=0$.
\end{proof}

\begin{proof}[Proof of Lemma \ref{inv_par_graphe}] We will describe elementary deformations (S for smooth, A for arrival, D for departure) of the defect graph (for fixed $\mathbf{v}$) that produce the same correlation functions. In that aim, we need to first introduce the basis moves:
\begin{definition}{{\bf basic moves}}
\begin{itemize}
\item  An S-move on the defect graph $\mc{D}_{\mathbf{v},\boldsymbol{\xi}}$ consists in picking $p\in\{1,\dots, n-1\}$ and in replacing $\xi_p$ by another smooth simple oriented curve 
$\tilde{\xi}_p:[0,1]\to \Sigma$, homotopic to $\xi_p$, parametrized by arclength with the same endpoints $\zeta_j(0)=z_j$ and $\zeta_{p}(0)=z_{j'}$ for $j\not=j'$, and still reaching the endpoints in the directions prescribed by $\mathbf{v}$, meaning  $\zeta_{p}'(0)=\lambda_{p,j} v_{j}$  and $\zeta_{p}'(1)=\lambda_{p,j'} v_{j'}$ for some  $\lambda_{p,j},\lambda_{p,j'} >0$. An S-move thus consists in changing the shape of the curve between the endpoints of an edge.
 \begin{figure}[h!] 
\centering
\includegraphics[width=.5\textwidth]{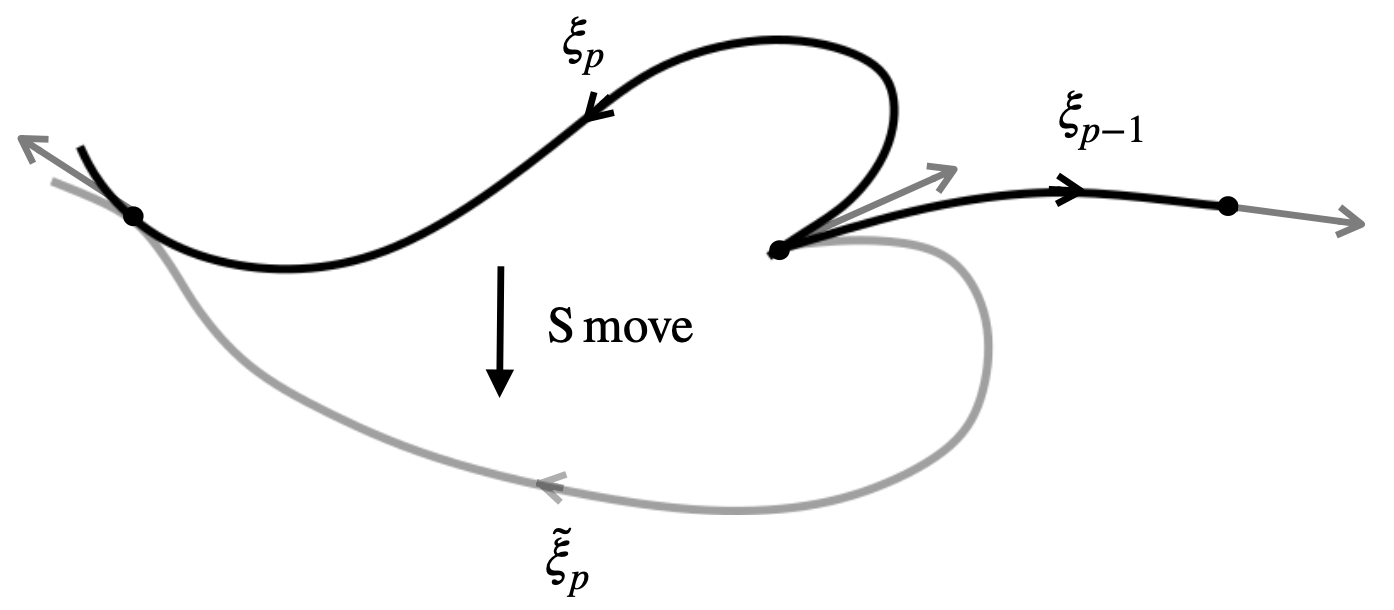} 
\label{fig:amplitude}
\end{figure}
\item An A-move changes the structure of the graph. Assume we are given distinct $p,p'\in\{1,\dots, n-1\}$ such that $\xi_{p}(1)=\xi_{p'}(1)=z_j$ and $\xi_{p}(0)=z_{j'}$, $\xi_{p'}(0)=z_{j''}$ with $z_{j'}\not=z_{j''}$. Assume $m_{j''}\geq m_{j'}$. Choose a smooth oriented curve $\tilde{\xi}:[0,1]\to  \Sigma $ homotopic to $\xi_{p'}^{-1}\circ \xi_p$ and parametrized by arclength with endpoints $\tilde{\xi}(0)=z_{j'}$ and $\tilde{\xi}(1)=z_{j''}$, and $\tilde\xi'(0)=\lambda_{j'} v_{j'}$  and $\tilde{\xi}'(1)=\lambda_{j''} v_{j''}$ for some  $\lambda_{j'},\lambda_{j''} >0$.   The A-move then consists in removing the edge $\xi_p$ and in replacing it by $\tilde{\xi}$.
 \begin{figure}[h] 
\centering
\includegraphics[width=.8\textwidth]{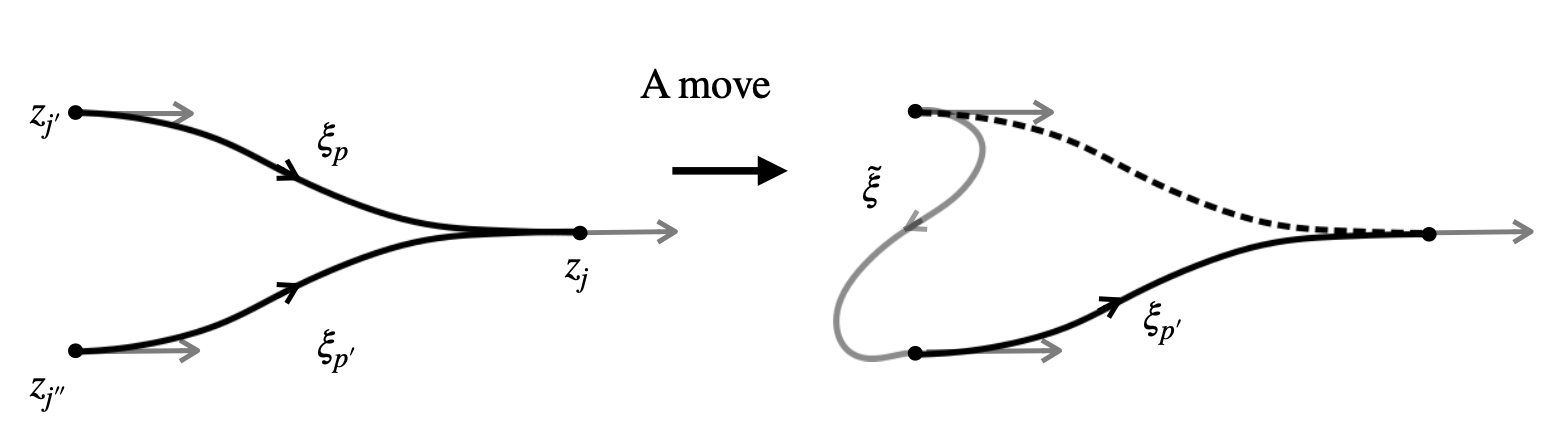} 
\label{fig:amplitude}
\end{figure}
\item A D-move changes the structure of the graph too. Assume we are given distinct $p,p'\in\{1,\dots, n-1\}$ such that $\xi_{p}(0)=\xi_{p'}(0)=z_j$ and $\xi_{p}(1)=z_{j'}$, $\xi_{p'}(1)=z_{j''}$ with $z_{j'}\not=z_{j''}$. Assume $m_{j''}\geq m_{j'}$. Choose a smooth oriented curve $\tilde{\xi}:[0,1]\to  \Sigma $ homotopic to $\xi_{p'}\circ \xi_p^{-1}$ and parametrized by arclength with endpoints $\tilde{\xi}(0)=z_{j'}$ and $\tilde{\xi}(1)=z_{j''}$, and $\tilde\xi'(0)=\lambda_{j'} v_{j'}$  and $\xi_{p}'(1)=\lambda_{j''} v_{j''}$ for some  $\lambda_{j'},\lambda_{j''} >0$.   The D-move then consists in removing the edge $\xi_p$ and in replacing it by $\tilde{\xi}$.
 \begin{figure}[h] 
\centering
\includegraphics[width=.7\textwidth]{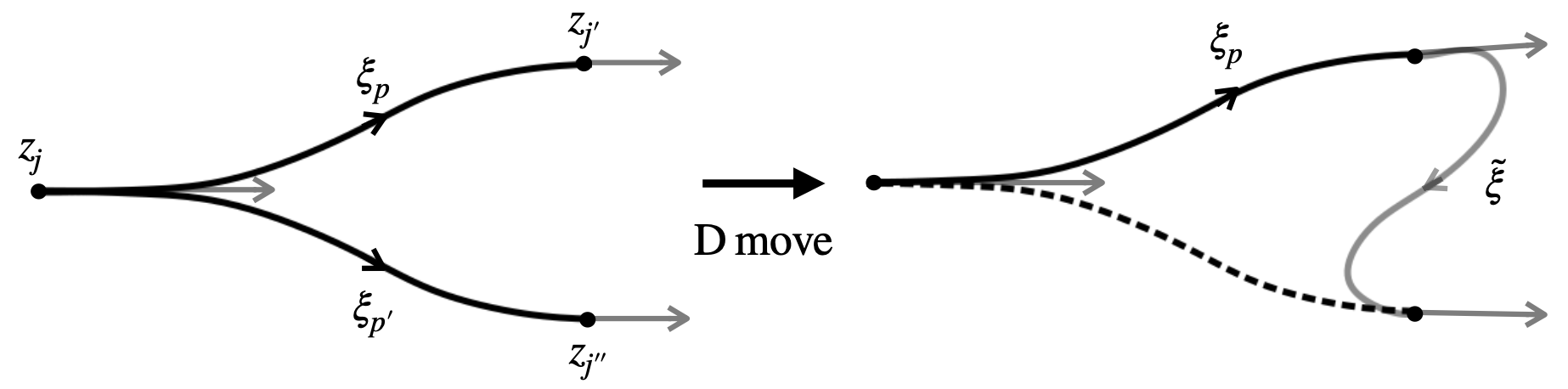} 
\label{fig:amplitude}
\end{figure}
\end{itemize}
\end{definition}
Now we claim that the value of $\int_{\Sigma}^{\rm reg}  I^{ \boldsymbol{\xi}}_{x_0}( \nu_{\mathbf{z},\mathbf{m}}) K_g \,\dd  {\rm v}_g$ remains unchanged if we perform S-moves, A-moves and D-moves to the defect graph.  Let us focus first on the case of $S$-move. We consider first the case when the curves $\xi_p$ and $\tilde{\xi}_p$ do not intersect (except at their endpoints). 
We can always reduce to this case since, if two curves intersect, by choosing a third one that does not intersect these two curves and comparing the integrals for the two curves with this third one, we get the desired result. 
Let us call $\tilde{\mc{D}}_{\mathbf{v},\boldsymbol{\xi}}$ the defect graph after the S-move. The curves $\xi_p$ and $\tilde{\xi}_p$ bound a domain denoted by $D$ (homeomorphic to the  disk (see Figure \ref{fig:proofS}). Furthermore, the boundary of $D$ inherits an orientation from $\Sigma$ (with the positive orientation given 
by $-J\nu$ if $\nu$ is the inward unit normal to $\pl D$ and $J$ the rotation of $+\pi/2$), which 
coincides coincide with the orientation of either $\xi_p$ or $\tilde{\xi}_p$.
\begin{figure}[h] 
\centering
\includegraphics[width=.6\textwidth]{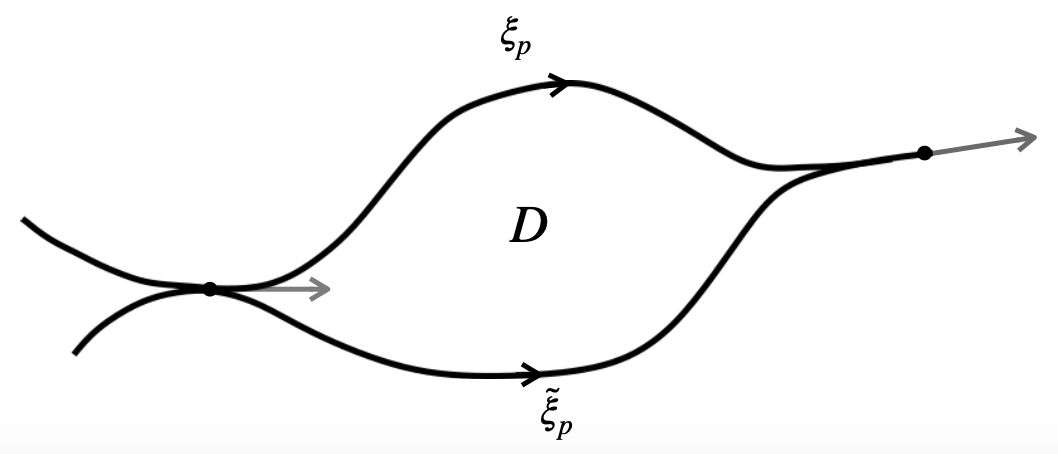} 
\caption{}
\label{fig:proofS}
\end{figure}
 Without loss of generality, we may assume this coincides with the orientation of $\tilde{\xi}_p$. The two defect graphs give rise to two different primitives $ I^{ \boldsymbol{\xi}}_{x_0}( \nu_{\mathbf{z},\mathbf{m}}) $ and $ I^{\tilde {\boldsymbol{\xi}}}_{x_0}( \nu_{\mathbf{z},\mathbf{m}}) $ and, on $D$, we have $  I^{\tilde {\boldsymbol{\xi}}}_{x_0}( \nu_{\mathbf{z},\mathbf{m}})+\kappa(\xi_p)= I^{ {\boldsymbol{\xi}}}_{x_0}( \nu_{\mathbf{z},\mathbf{m}})$. 
The difference of the two regularized integrals is then
\begin{align*}
\int_{\Sigma}^{\rm reg}&  I^{ \boldsymbol{\xi}}_{x_0}( \nu_{\mathbf{z},\mathbf{m}}) K_g \,\dd {\rm v}_g-\int_{\Sigma}^{\rm reg}   I^{\tilde {\boldsymbol{\xi}}}_{x_0}( \nu_{\mathbf{z},\mathbf{m}})K_g \,\dd {\rm v}_g
\\
&=\int_D( I^{\boldsymbol{\xi}}_{x_0}( \nu_{\mathbf{z},\mathbf{m}})-   I^{\tilde {\boldsymbol{\xi}}}_{x_0}( \nu_{\mathbf{z},\mathbf{m}}))K_g\dd {\rm v}_g-2\kappa(\xi_p)\int_{\xi_p}k_g\dd \ell_g+2\kappa(\xi_p)\int_{\tilde\xi_p}k_g\dd \ell_g.
\end{align*}
Now we apply the Gauss-Bonnet theorem  on $D$ to get
\[\begin{split}
\int_D( I_{x_0}^{\boldsymbol{\xi}}( \nu_{\mathbf{z},\mathbf{m}})-  I^{\tilde {\boldsymbol{\xi}}}_{x_0}( \nu_{\mathbf{z},\mathbf{m}}))K_g \dd {\rm v}_g=
\int_D k(\xi_p)K_g \dd {\rm v}_g 
=- 2 \kappa(\xi_p)(\int_{\tilde{\xi}_p}k_g\dd \ell_g -  \int_{\xi_p}k_g\dd \ell_g) 
\end{split}\]
and we deduce that the difference of the two regularized integrals vanishes.
This means that the regularized curvature term does not change if we perform an S-move to the defect graph. Furthermore, since $F\in  \mc{E}_R(\Sigma,g)$, we have $F(c+X_g+I^{\boldsymbol{\sigma}}_{x_0}( \omega_{\mathbf{k}})+ I^{  {\boldsymbol{\xi}}}_{x_0}( \nu_{\mathbf{z},\mathbf{m}}))=F(c+X_g+I^{\boldsymbol{\sigma}}_{x_0}( \omega_{\mathbf{k}})+ I^{\tilde {\boldsymbol{\xi}}}_{x_0}( \nu_{\mathbf{z},\mathbf{m}}))$ so that this 
establishes our claim in the case of S-moves. The cases of A-moves and D-moves can be treated similarly by applying the Gauss-Bonnet theorem in the triangle with vertices $z_j,z_{j'},z_{j''}$. Now it can be checked that every defect graph can be deformed, via S-A-D-moves, to the following defect graph: let $\sigma$ be a permutation of $\{1,\dots,n_{\mathfrak{m}}\}$ reordering the charges, i.e. $m_{\sigma(1)}\leq \dots\leq m_{\sigma(n_{\mathfrak{m}})}$ and $\xi_p$ the curve joining $z_{\sigma(p)}$ to $z_{\sigma(p+1)}$. Also, note that this graph is not unique if some charges are equal but then   two such graphs are still related by S-A-D-moves (see Figure \ref{fig:masses}). This proves that the magnetic correlation functions dont depend on the defect graph.
\begin{figure}[h] 
\centering
\includegraphics[width=.5\textwidth]{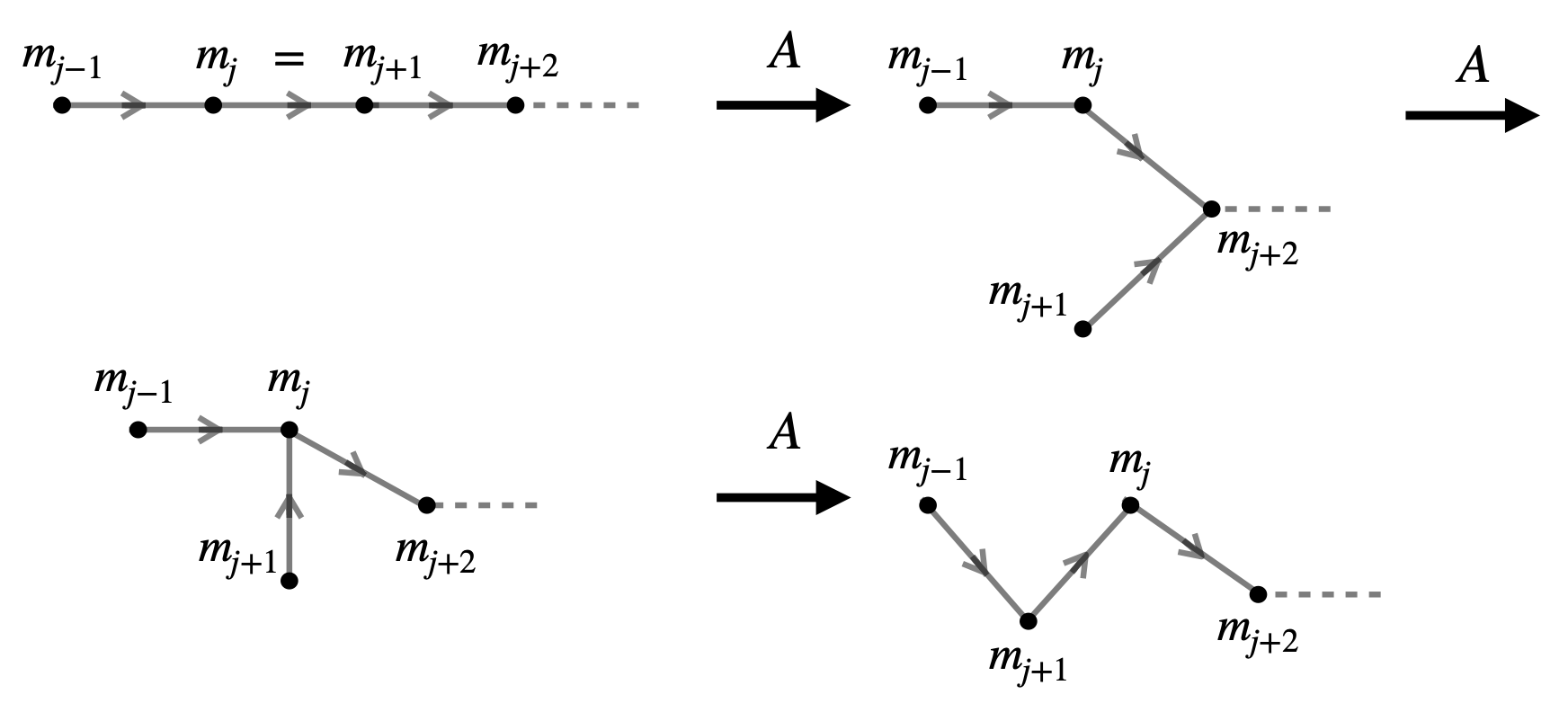} 
\caption{Sequence of A-moves to switch two points with equal masses }
\label{fig:masses}
\end{figure}
\end{proof}
 
 \begin{proof}[Proof of Lemma \ref{change_conf_magnetic}]
 We use the relation $K_{g'}=e^{-\rho}(K_g+ \Delta_{g}\rho)$ and $\dd{\rm v}_{g'}=e^{\rho}\dd{\rm v}_{g}$,
 then by integration by parts 
 \[\begin{split} 
 \int_{\Sigma} I^{\boldsymbol{\xi}} _{x_0}( \nu^{\rm h}_{\mathbf{z},\mathbf{m}}) K_{g'} \,\dd v_{g'}
 =&
 \int_{\Sigma} I^{\boldsymbol{\xi}} _{x_0}( \nu^{\rm h}_{\mathbf{z},\mathbf{m}})K_g\dd v_{g}
 +  \cjg  \dd\rho, \nu^{\rm h}_{\mathbf{z},\mathbf{m}}\cjd_2   \\
&- \sum_{p=1}^{n_{\mathfrak{m}}-1}\int_{\xi_p}\pl_{\nu}\rho (I^{\boldsymbol{\xi}} _{x_0}( \nu^{\rm h}_{\mathbf{z},\mathbf{m}})^+-I^{\boldsymbol{\xi}} _{x_0}( \nu^{\rm h}_{\mathbf{z},\mathbf{m}})^-)\dd \ell_g \end{split} \]
 where $\nu=J\dot{\xi}_p$ is the left normal pointing vector with respect to the orientation of $\xi_p$ ($J$ being the rotation of $+\pi/2$ in the tangent space), $I^{\boldsymbol{\xi}} _{x_0}( \nu^{\rm h}_{\mathbf{z},\mathbf{m}})^+$ is the limit of $I^{\boldsymbol{\xi}} _{x_0}( \nu^{\rm h}_{\mathbf{z},\mathbf{m}})$ on $\xi_p$ from the side given by $-\nu$ (the right side)
  and $I^{\boldsymbol{\xi}} _{x_0}( \nu^{\rm h}_{\mathbf{z},\mathbf{m}})^-$   the value from the side given by $\nu$ (the left side). Also $I^{\boldsymbol{\xi}} _{x_0}( \nu^{\rm h}_{\mathbf{z},\mathbf{m}})^+-I^{\boldsymbol{\xi}} _{x_0}( \nu^{\rm h}_{\mathbf{z},\mathbf{m}})^-=2\pi R\kappa(\xi_p)$.
  Next we use  that $k_{g'}\dd \ell_{g'}=k_{g}\dd \ell_{g}-\demi\pl_{\nu}\rho \, \dd \ell_{g}$ on $\xi_p$. Finally  $\cjg  \dd\rho, \nu^{\rm h}_{\mathbf{z},\mathbf{m}}\cjd_2 =0$ because $\nu^{\rm h}_{\mathbf{z},\mathbf{m}}$ is harmonic.
Combining these facts, we obtain our claim.
\end{proof}

\section{Imaginary Gaussian multiplicative chaos}\label{IGMC}

In this section, we first recall some facts about the Gaussian Free Field $X_g$ (resp. $X_{g,D}$) on closed surfaces (resp. surfaces with boundary). This is a Gaussian random distribution on the surface, living in a negative Sobolev space $H^{s}(\Sigma)$ for $s<0$. In order to make sense of certain functionals, we need to regularize it at a small scale $\eps>0$. This will be done either  in a geometric fashion or using white noise, as we explain below. Finally, we recall the construction and properties of the Imaginary Gaussian multiplicative chaos $e^{i\beta X_g}\dd {\rm v}_g$, which is a random distribution on $\Sigma$. We shall need estimates on its exponential moments. 

 \subsection{Gaussian Free Fields} 
On a Riemann surface without boundary, the Gaussian Free Field (GFF in short) is defined as follows. Let $(a_j)_j$ be a sequence of i.i.d. real Gaussians   $\mc{N}(0,1)$ with mean $0$ and variance $1$, defined on some probability space   $(\Omega,\mc{F},\mathbb{P})$,  and define  the Gaussian Free Field with vanishing mean in the metric $g$ by the random distribution
\begin{equation}\label{GFFong}
X_g:= \sqrt{2\pi}\sum_{j\geq 1}a_j\frac{e_j}{\sqrt{\la_j}} 
\end{equation}
 where  the sum converges almost surely in the Sobolev space  $H^{s}(\Sigma)$ for $s<0$ defined by
\begin{equation}\label{sobolev}
 H^{s}(\Sigma):=\{f=\sum_{j\geq 0}f_je_j\,|\, \|f\|_{s}^2:=|f_0|^2+\sum_{j\geq 1}\lambda_j^{s}|f_j|^2<+\infty\}.
 \end{equation}
This Hilbert space is independent of $g$, only its norm depends on a choice of $g$.
The covariance $X_g$ is then the Green function when viewed as a distribution, which we will write with a slight abuse of notation
\[\mathbb{E}[X_g(x)X_g(x')]= \,G_g(x,x').\]

In the case of a surface with boundary $\Sigma$, the Dirichlet Gaussian free field (with covariance $G_{g,D}$) will be denoted $X_{\Sigma,D}$. It is defined similarly to the sum \eqref{GFFong} with the $(e_j)_j$  and $(\lambda_j)_j$ replaced by the normalized eigenfunctions $(e_{j,D})_j$ and ordered eigenvalues $(\lambda_{j,D})_j$ of the Laplacian with Dirichlet boundary conditions, the sum being convergent   almost surely  in the Sobolev space  $H^{s}(\Sigma)$ (for all $s\in (-1/2,0)$) defined by
\begin{equation}\label{sobolevD}
H^{s}(\Sigma):=\{f=\sum_{j\geq 0}f_je_{j,D}\,|\, \|f\|_{s}^2:=\sum_{j\geq 0}\lambda_{j,D}^{s}|f_j|^2<+\infty\}.
\end{equation}
In both cases (closed or open surfaces), we will denote by $(\cdot, \cdot)_s$ the inner product in $H^{s}(\Sigma)$, and by extension also the duality bracket on $H^{s}(\Sigma)\times H^{-s}(\Sigma)$. 
 
 \subsection{$g$-Regularisations and white noise regularisation} 
 As Gaussian Free Fields are rather badly behaved (they are distributions), we will need to consider  regularizations, and we will mainly consider two of them. First we introduce a regularization procedure, which we will call \emph{$g$-regularization}. Let $\Sigma$ be a surface with or without boundary equipped with a Riemannian metric $g$ and associated distance $d_g$. For a random distribution $h$ on $\Sigma$ and  for $\eps>0$ small, we define a regularization $h_{\eps}$ of $h$ by averaging on geodesic circles of radius $\eps>0$: let $x\in \Sigma$ and let $\mc{C}_g(x,\eps)$ be the geodesic circle of center $x$ and radius $\eps>0$, and let $(f^n_{x,\eps})_{n\in \N} \in C^\infty(\Sigma)$ be a sequence with $||f^n_{x,\eps}||_{L^1}=1$ 
which is given by $f_{x,\eps}^n=\theta^n(d_g(x,\cdot)/\eps)$ where $\theta^n(r)\in C_c^\infty((0,2))$ is non-negative, 
supported near $r=1$ and such that $f^n_{x,\eps}{\rm dv}_g$ 
converges in $\mc{D}'(\Sigma)$ to the uniform probability measure 
$\mu_{x,\eps}$
on $\mc{C}_g(x,\eps)$ as $n\to \infty$ (for $\eps$ small enough, the geodesic circles form a submanifold and the restriction of $g$ along this manifold gives rise to a finite measure, which corresponds to the uniform measure after renormalization so as to have mass $1$). If the pairing $\langle h, f_{x,\eps}^n\rangle$ converges almost surely towards a random variable $h_\epsilon(x)$ that has a modification which is continuous in the parameters $(x,\epsilon)$, we will say that $h$ admits a $g$-regularization $(h_\epsilon)_\epsilon$. This is the case for the GFF $X_g,X_D$ and we denote by $X_{g,\epsilon},X_{D,\epsilon}$ their  $g$-regularization.

Second, and in the case we consider the Dirichlet GFF over a surface $ \Sigma$ with boundary, we introduce another regularization, dubbed white-noise regularization. The Green function $G_{g,D}$ can be written as
$$G_{g,D}(x,x')=2\pi\int_0^\infty p_t(x,x')\,\dd t $$
where $p_t(x,x')$   denotes the transition densities of the Brownian motion on  $\Sigma$ killed upon touching the boundary $\partial \Sigma$ (i.e. the heat kernel of the Laplacian with Dirichlet condition).
Let $W$ be a white noise on $\R_+\times \Sigma$ and define for $\delta>0$
$$X_{g,D,\delta}(x):=(2\pi)^{1/2}\int_{\delta^2}^{\infty}\int_{\Sigma}p_{t/2}(x,y)W(\dd t,\dd y).$$
Then the covariance kernel of these processes is given by
$$\E[X_{g,D,\delta}(x)X_{g,D,\delta'}(x')]=2\pi\int_{\delta^2\vee\delta'^2}^{\infty} p_t(x,x')\,\dd t.$$

\subsection{Gaussian multiplicative chaos}
For $\beta\in\R$ and $h$ a random distribution admitting a  $g$-regularization $(h_\epsilon)_\epsilon$, we define the complex measure 
\begin{equation}\label{GMCg}
M^{g,\eps}_{\beta}(h,\dd x):= \eps^{-\frac{\beta^2}{2}}e^{i\beta h_{ \eps}(x)}{\rm dv}_g(x).
\end{equation}
Of particular interest for us is the case when $h=X_g$ or  $h=X_{g,D}$. In that case, for $\beta^2<2$,  the random measures above converge  as $\eps\to 0$ in $L^2(\Omega)$ and weakly in the space of distributions \cite[Theorem 3.1]{LRV15}  to a non trivial distribution of order $2$  denoted by $M^g_\beta(X_g,\dd x)$ or $M^g_\beta(X_{g,D},\dd x)$; this means that there exists a random variable $D_\Sigma\in L^2(\Omega)$,  such that
\begin{equation}\label{order2}
\forall \varphi\in C^\infty(\Sigma),\quad \Big|\int_\Sigma \varphi (x) M^g_\beta(X_g,\dd x)\Big|\leq D_\Sigma \|d^2\varphi\|_\infty.
\end{equation}
For notational readability and if no risk of confusion, we will use the notation $M^{g,\eps}_{\beta}( \dd x)$ for $M^{g,\eps}_{\beta}(X_g,\dd x)$ (i.e. we skip the field dependence). Also we stress that the condition $\beta^2<2$ will be in force throughout the paper.

Also, from  \cite[Lemma 3.2]{GRV},  we recall that there exist $W_g$ (resp. $W_{g,D}$) such that uniformly on $\Sigma$ 
\begin{equation}\label{varYg}
\lim_{\eps \to 0}\E[X^2_{g,\eps}(x)]+\ln\eps=W_g(x)\quad \text{ and }\quad \lim_{\delta \to 0}\E[X^2_{g,D,\delta}(x)]+ \ln\delta=W_{g,D}(x) .
\end{equation}
 In particular, considering a metric $g'=e^{\omega}g$ conformal to the  metric $g$, we obtain the relation
\begin{equation}\label{relationentrenorm} 
M^{  g'}_{\beta}(X_g,\dd x)=e^{(1-\frac{\beta^2}{4})\omega(x)} M^{g}_{\beta}(X_g,\dd x).
\end{equation}

We state the following elementary, though important, equivalence between the GMC construction via white-noise or g-regularization.
\begin{proposition}\label{equiv}
Assume $\Sigma$ is a surface without boundary and $\mc{C}$ is a family of smooth simple curves splitting $\Sigma$ into two connected components $\Sigma_1$ and $\Sigma_2$. Consider the equality in law stated in Proposition \ref{decompGFF}
$$X_g=Y_1+Y_2+P-c_g$$
where $Y_1,Y_2,P$ are three independent Gaussian fields with $Y_i$ the Dirichlet GFF on $\Sigma_i$ for $i=1,2$, $P$ the harmonic extension of the boundary values $X_{g|\mc{C}}$ and $c_g=\frac{1}{\rm v_g(\Sigma)}\int_\Sigma (Y_1+Y_2+P)\,\dd {\rm v}_g$. Then, for $\beta^2<2$, the following limits agree in law
$$M_\beta^g(X_g,\dd x)\stackrel{law}{=}M_\beta^g(Y_1+Y_2+P-c_g,\dd x)$$
where the limit in the lhs is taken via $g$-regularization and the limit in the rhs can be either white noise or g-regularization.
\end{proposition}  

\begin{proof}
This result is standard and left to the reader as an exercise: it results from straightforward $L^2$-computations. Recall that $\beta^2<2$ and GMC is therefore in $L^2$.
\end{proof}

\subsection{Exponential moments}
In this section, we assume that $\beta^2<2$ and we recall the following result, originally proved in  \cite{LRV19} but here adapted to our context, which will be fundamental for the existence of the path integral and correlation functions:
\begin{proposition}\label{expmoment}
Assume $\Sigma$ is a Riemann surface with or without boundary. If $\Sigma$ has a non empty boundary, we consider the Dirichlet GFF $X_{g,D}$ on $\Sigma$ and we set $\mc{D}'=\Sigma$. If $\Sigma$ has empty boundary,  we consider an open subset    $\mc{D}'$  of $\Sigma$, with closure $\overline{\mc{D}'}\not=\Sigma$. Let $Z$ a  real valued  random variable (not necessarily assumed to be independent of $X_{g,D}$). Let   $\mc{D}$ be an open subset of $\mc{D}'$.  Finally we consider a measurable function $f:\mc{D}\to \C$. Then for $\mu\in\R$
$$ \E\Big[ \exp\Big(\mu\Big| \int_{\mc{D}} f(x)M^{g }_{\beta}(Z+X_{g,D},\dd x) \Big|\Big) \Big]
\leq
e^{C\mu v}(1+C\mu ue^{C\mu^2u^2})$$
with
$$  u^2:=  \iint_{\mc{D}^2}|f(x)||f(y)|e^{\beta^2G_{g,D}(x,y)}{\rm v}_g(\dd x){\rm v}_g(\dd y)   ,\quad \text{ and }\quad  v:=\int_{\mc{D}} |f(x)|{\rm v}_g(\dd x),$$
for some constant $C$ only depending on $\beta$.
\end{proposition}  

\begin{proof} Before proceeding to the proof, let us stress that the crucial input for the proof is an integral representation of the Green function in terms of a positive heat kernel. This is valid for the Dirichlet GFF and that is the reason why we focus on this case. From this result, we will deduce later exponential moment estimates for the GFF on closed Riemann surfaces (for which the heat kernel is not positive). Of course, on closed Riemann surfaces there is no globally defined Dirichlet GFF and this is the reason why we restrict to a strict subset in this case.

Note that adding $Z$ to the GFF is harmless because it multiplies the GMC by a number with modulus $1$, so we may as well assume that $Z  =0$.
Let us first   prove exponential moments for the real part of $   \int_{\mc{D}} f(x)M^{g}_{\beta}(X_D,\dd x) $; the argument for the imaginary part is the same. The real part can be obtained 
 as the limit as $\delta\to 0$ of the process
$$ \int_{\mc{D}}  \delta^{-\frac{\beta^2}{2}}{\rm Re}(  f(x))\cos(\beta  X_{g,D,\delta}(x) ){\rm dv}_g(x)-\int_{\mc{D}}  \delta^{-\frac{\beta^2}{2}}{\rm Im}(  f(x))\sin(\beta  X_{g,D,\delta}(x) ){\rm dv}_g(x) .$$
For readability we will write the proof with details in the case when ${\rm Im}( f(x))=0$ but the argument is similar for the second term in the relation above.

Recall the asymptotic of $\E[( X_{g,D,\delta}(x))^2]$  in \eqref{varYg}. Therefore the limit above can be obtained as the limit $t\to\infty$ of the following martingale
$$\mc{M}_{t}= \int_{\mc{D}}k(x)\cos(\beta  X_{g,D,e^{-t}}(x) )e^{\frac{\beta^2}{2}\E[ X_{g,D,e^{-t}}(x) ^2]}  {\rm dv}_g(x)$$
where we have set $k(x)={\rm Re}(  f(x)) e^{-\frac{\beta^2}{2}W_{g,D}(x)}$. The argument will follow from the boundedness of the quadratic variations of this martingale for large times. Yet  controlling negatively large times  would involve controlling the long-time behaviour of the heat kernel $p_u(x,y)$ (with $u=e^{-2t}$ ), which would be unnecessarily complicated. Instead, we first remove from the martingale the contribution from negative times and we write for $t\geq 0$
$$\mc{M}_{t}=\mc{M}^0_{t}+\mc{M}_{0}, \quad \text{with }\mc{M}^0_{t}:=\mc{M}_{t}-\mc{M}_{0}.$$
Then $(\mc{M}^0_{t})_{t\geq 0}$ is a continuous square integrable martingale starting from $\mc{M}^0_{0}=0$.  In what follows, we will be mainly focused on controlling this martingale, the case of $\mc{M}_{0}$ being trivial.

The quadratic variations of $\mc{M}^0_{t}$ are easily computed with basic It\^o calculus. Indeed, we have
$$\dd \mc{M}^0_{t}=- \int_{\mc{D}}k(x) \beta\sin(\beta  X_{g,D,e^{-t}}(x) )e^{\frac{\beta^2}{2}\E[ X_{g,D,e^{-t}}(x) ^2]}  {\rm dv}_g(x)\,\dd X_{g,D,e^{-t}}(x), $$
so that
\begin{align*}
\langle \mc{M}^0  \rangle_t =& \int_{0}^{t}\iint_{\mc{D}^2}k(x)k(y)\beta^2\sin(\beta  X_{g,D,e^{-u}}(x) )\sin(\beta  X_{g,D,e^{-u}}(y) ) \\
&
e^{\frac{\beta^2}{2}\E[ X_{g,D,e^{-u}}(x) ^2]} e^{\frac{\beta^2}{2}\E[ X_{g,D,e^{-u}}(y) ^2]}  {\rm dv}_g(x) {\rm dv}_g(y)\,\dd \langle   X_{g,D,e^{-u}}(x),   X_{g,D,e^{-u}}(y)\rangle_u.
\end{align*}
The bracket  is given by (using the Markov property of the transition kernels)
 $$\dd \langle   X_{g,D,e^{-u}}(x),   X_{g,D,e^{-u}}(y)\rangle_u=  4\pi e^{-2u}p_{e^{-2u}}(x,y)\,\dd u.$$
 Now we bound the quadratic variation, by estimating both sines by $1$, by using the standard  inequality for the Dirichlet heat kernel (it follows from \cite{CLY} and the monotonicity of the Dirichlet heat kernel with respect to the domain) 
 $$  e^{-2u}p_{e^{-2u}}(x,y)\leq Ce^{-d_g(x,y)^2e^{2u}/C},$$
 and by using \eqref{varYg} again, to get  that (for some constant $C$ that may change along lines)
\begin{align*}
\sup_{t\geq 0} \langle \mc{M}^0  \rangle_t  \leq &C\int_{0}^\infty \iint_{\mc{D}^2}|k(x)||k(y)|e^{\frac{\beta^2}{2}(W_{g,D}(x)+W_{g,D}(y))}\beta^2e^{\beta^2 t}e^{-d_g(x,y)^2e^{2t}/2}\,{\rm dv}_g(x){\rm dv}_g(y)\dd t\\
\leq &C  \iint_{\mc{D}^2}|k(x)||k(y)|e^{\frac{\beta^2}{2}(W_{g,D}(x)+W_{g,D}(y))} d_g(x,y)^{-\beta^2} \,{\rm dv}_g(x){\rm dv}_g(y) ,
\end{align*}
where we have performed the change of variables $s=d_g(x,y)e^{t}$ to get the last line.
Recall next that, as  a continuous martingale starting from $0$, the martingale $ \mc{M}^0 $ satisfies 
$$\E[e^{\mu  \mc{M}^0_t-\frac{\mu^2}{2} \langle \mc{M}^0  \rangle_t}]=1$$
and, since the bracket is bounded, we deduce for $\mu\in\R$
$$\E[e^{\mu  \mc{M}^0_\infty}]\leq e^{ \frac{\mu^2}{2} \sup_t\langle \mc{M}^0  \rangle_t}.$$
This implies, with the notations above, $\P(\mc{M}^0_\infty>x)\leq \exp(-\frac{x^2}{4Cu^2})$. The same argument with $\E[e^{-\mu  \mc{M}^0_\infty}]$ leads to the estimate $\P(\mc{M}^0_\infty<-x)\leq \exp(-\frac{x^2}{4Cu^2})$ for $x>0$. Therefore
$$\P(|\mc{M}^0_\infty|>x)\leq 2\exp(-\frac{x^2}{4Cu^2}).$$
Finally, we can use the standard trick
$$\E[e^{\mu|\mc{M}^0_\infty|}]=1+\int_0^\infty \P(|\mc{M}^0_\infty|>\frac{x}{\mu})e^x\,\dd x$$
to get that $$\E[e^{\mu|\mc{M}^0_\infty|}]\leq 1+Cu\mu e^{Cu^2\mu^2},$$
up to multiplying $C$ by an irrelevant factor. And we can then conclude by coupling this estimate to the fact that $\mc{M}_{0}$  is obviously bounded by $Cv$, with $v:=\int_\Sigma |f(x)|{\rm v}_g(\dd x)$, to get 
$$\E[e^{\mu|\mc{M}^0_\infty|}]\leq e^{C\mu v}(1+Cu\mu e^{Cu^2\mu^2}).$$
as claimed.
\end{proof}
Note that the proof relies on a white noise decomposition of the GFF and therefore works only for positive heat kernels; that is why we focused on the Dirichlet GFF. Such estimates for the GFF on a closed manifold will be later established using the Markov decomposition of the GFF.
\section{The   path integral and correlation functions}\label{secPI}
 Throughout this section we will work under the constraint $\beta^2<2$, in order to ensure convergence of the imaginary GMC.  We also assume that $\beta>0$ since the case $\beta<0$ would be symmetric to $\beta>0$. 
  We further impose the compactification radius $R>0$ to obey 
\begin{equation}
 \beta\Z \subset\frac{1}{R}\Z.
\end{equation}
Let us introduce  a further parameter    $\mu\in \C\setminus\{0\}$ and set
 \begin{equation}\label{valueQ}
  Q= \frac{\beta}{2}- \frac{2}{\beta}.
 \end{equation} 
 Finally we introduce the central charge   ${\bf c}=1-6Q^2$. We focus in this section in the case that ew call rational, meaning we assume
 \begin{equation}\label{valueQrat}
 Q\in \frac{1}{R}\Z.
  \end{equation}

\subsection{Path integral }\label{subsecPI}

 Consider now a closed Riemann surface $\Sigma$  equipped with a metric $g$.  To construct the path integral, we need:
 
\begin{assumption}\label{ass}
We fix:
 \begin{itemize}
 \item a geometric symplectic basis $\sigma=(\sigma_j)_{j=1,\dots,2{\mathfrak{g}}}$ of  $\mc{H}^1(\Sigma)$ and $2{\mathfrak{g}}$ independent closed smooth $1$-forms $\omega_1,\dots,\omega_{2{\mathfrak{g}}}$ forming a basis of the cohomology $\mc{H}^1_R(\Sigma)$ dual to $\sigma$ (see Lemma \ref{basisH^1}). Let $\omega_{\bf k}=\sum_{j=1}^{2{\mathfrak{g}}}k_j\omega_j$ if ${\bf k}=(k_1,\dots,k_{2{\mathfrak{g}}})$, as defined in \eqref{omega_k}. 
\item a base point $x_0\in \Sigma_{\boldsymbol{\sigma}}=\Sigma \setminus \cup_{j=1}^{2{\mathfrak{g}}}\sigma_j$, and we define  $I_{x_0}^\sigma(\omega_{\bf k})$ the   function \eqref{primitive} on $\Sigma_{\boldsymbol{\sigma}}$ obtained from the  closed form $\omega_{\bf k}$. 
\end{itemize}
\end{assumption}

The first step  is to introduce  a space of reasonable test functions for our path integral.    Recall  that the family $(e_j)_{j\geq 0}$ stands for an orthonormal basis of $L^2(\Sigma,{\rm v}_g)$ of eigenfunctions of the Laplacian. Write any function $f\in H^s(\Sigma)$ (for $s<0$) as
\begin{equation}\label{f_Fourier}
f=f_0+\sqrt{2\pi}\sum_{j\geq1}f_je_j
\end{equation}
 and notice that the zero mode $f_0$ is unnormalized in the sense that it is not multiplied by ${\rm v}_g(\Sigma)^{-1/2}$, which corresponds to the constant eigenfunction $e_0$\footnote{This term can be absorbed in the $c$-integral up to changing the compactification radius and this is why we choose this normalization.}. We equip $H^s ( \Sigma)$  with the pushforward of the measure $\dd c\otimes\P$ on $(\R/2\pi R\Z)\times \Omega$ through the map $(c,\omega)\mapsto c+X_g(\omega)$.

 Recall from Section \ref{UnivCover} that equivariant distributions $u\in H^s_{\Gamma}(\tilde\Sigma)$ 
 can be uniquely  decomposed as
 \[u=\pi^*(f_0+\sqrt{2\pi}\sum_{j\geq1}f_j e_j) +I_{x_0}(\omega_{\bf k}) \]
 for some ${\bf k}\in\Z^{2\mathfrak{g}}$, where $\pi:=\tilde{\Sigma}\to \Sigma$ is the projection on the base. Identifying this way $H^s_{\Gamma}(\tilde\Sigma)$ with $\Z^{2\mathfrak{g}}\times H^s(\Sigma)$, we consider the space $\mc{E}_R(\Sigma)$ of functionals $F$, defined on $H^s_{\Gamma}(\tilde\Sigma)$, of the form
\begin{equation}\label{polytrigo}
F(u)=\sum_{n=-N}^Ne^{\frac{i}{R} n f_0}P_n(f-f_0)  G( e^{\frac{i}{R} I_{x_0}(\omega_{\bf k})} ) 
\end{equation}
 for arbitrary $N\in\N$, where $P_n$ are polynomials of the form $P_n\big(( f-f_0,g_1)_s ,\dots, ( f-f_0,g_{m_n})_s \big)$ where $g_1,\dots,g_{m_n}$  belong to $H^{-s}(\Sigma)$, and $G$ is   continuous and bounded on $C^0(\Sigma;\S^1)$. Notice in particular that these functionals are $2\pi R$-periodic in the zero mode $f_0$. Next we define   the space:
 \begin{itemize}
\item $\mc{L}^{\infty,p}(H^s(\Sigma))$ as the closure of $\mc{E}_R(\Sigma)$ with respect to the  norm 
\begin{equation}\label{defnorm}
\|F\|_{\mc{L}^{\infty,p}}:=\sup_{\bf k}\Big(\int_{\R/2\pi R\Z}\E\Big[e^{-\frac{1}{2\pi}\langle \dd X_g,\omega_{\bf k} \rangle_2-\frac{1}{4\pi}\|\dd f_{\bf k}\|_2^2}|F(c+\pi^*X_g+I_{x_0}(\omega_{\bf k}))|^p\Big]\,\dd c\Big)^{1/p}
\end{equation}
where $(1-\Pi_1)\omega_{\bf k}=\dd f_{\bf k}$ with $\Pi_1$ is the projection on harmonic forms (recall Lemma \ref{dX_gomega}).
\end{itemize}

\begin{lemma}\label{invariance_of_norm} The norm $\|F\|_{\mc{L}^{\infty,p}}$ does not depend on the choice of cohomology basis.
 \end{lemma}

 \begin{proof}
First we show that this norm does not depend on the choice of cohomology basis dual to $\sigma$. Let $\omega_{\bf k}^{\rm h}:=\Pi_1\omega_{\bf k}$ be the harmonic 1 form with in the same cohomology class as $\omega_{\bf k}$, so that $\omega_{\bf k}=\omega_{\bf k}^{\rm h}+\dd f_{\bf k}$. Using the Girsanov transform and a shift in the $c$-variable we have
 \begin{align*} 
 \int_{\R/2\pi R\Z}\E\Big[&e^{-\frac{1}{2\pi}\langle \dd X_g,\omega_{\bf k} \rangle_2-\frac{1}{4\pi}\|\dd f_{\bf k}\|_2^2}|F(c+\pi^*X_g+I_{x_0}(\omega_{\bf k}))|^p\Big]\,\dd c\\
 =&\int_{\R/2\pi R\Z}\E\Big[e^{-\frac{1}{2\pi}\langle \dd X_g,\omega^{\rm h}_{\bf k} \rangle_2 }|F(c+\pi^*X_g+I_{x_0}(\omega^{\rm h}_{\bf k}))|^p\Big]\,\dd c\\
 =&\int_{\R/2\pi R\Z}\E\Big[ |F(c+\pi^*X_g+I_{x_0}(\omega^{\rm h}_{\bf k}))|^p\Big]\,\dd c
  \end{align*}
  where we have used that $\langle \dd X_g,\omega_{\bf k}^{\rm h} \rangle_2=\langle  X_g,\dd^*\omega_{\bf k}^{\rm h} \rangle_2=0$ because $\omega_{\bf k} ^{\rm h}$ is harmonic. This proves our first claim. Next we study a change of homology basis $\tilde{\sigma}$, and therefore a change of cohomology basis $\tilde{\omega}^{\rm h}_1,\dots,\tilde{\omega}^{\rm h}_{2\mathfrak{g}}$ made of harmonic 1-forms dual to $\tilde{\sigma}$. Then there is $A\in {\rm SL}(2{\mathfrak{g}},\Z)$ such that ${\omega}^{\rm h}_{\bf k}=\tilde{\omega}^{\rm h}_{A\bf k}$, and it is clear that
 \begin{align*} 
  \sup_{\bf k}\Big(\int_{\R/2\pi R\Z}\E\Big[ |F(c+\pi^*X_g+I_{x_0}(\omega^{\rm h}_{\bf k}))|^p\Big]\,\dd c\Big)^{1/p}=& \sup_{\bf k}\Big(\int_{\R/2\pi R\Z}\E\Big[ |F(c+\pi^*X_g+I_{x_0}(\omega^{\rm h}_{A\bf k}))|^p\Big]\,\dd c\Big)^{1/p}\\
  =& \sup_{\bf k}\Big(\int_{\R/2\pi R\Z}\E\Big[ |F(c+\pi^*X_g+I_{x_0}(\tilde{\omega}^{\rm h}_{\bf k}))|^p\Big]\,\dd c\Big)^{1/p}.
   \end{align*}
  Hence our claim.
 \end{proof}

 \medskip

On closed surfaces  $\Sigma$, we will denote the Liouville field by 
\[\phi_g:=c+X_g+I^{\boldsymbol{\sigma}}_{x_0}(\omega_{\bf k}).\] 
This field belongs to $H^s(\Sigma_{\boldsymbol{\sigma}})$ but can also be considered as an element in $H^{s}_{\Gamma}(\tilde{\Sigma})$. Indeed, recall from Section \ref{sub:fund} that $I^{\boldsymbol{\sigma}}_{x_0}(\omega_{\bf k})$ is a smooth function on $\Sigma_{\boldsymbol{\sigma}}$ such that 
there is an open set $U^{\boldsymbol{\sigma}}_{x_0}\subset \tilde{\Sigma}$ containing $\tilde{x}_0$ for which 
$I_{x_0}(\omega_{\bf k})|_{U^\sigma_{x_0}}=\pi^*I_{x_0}^{\boldsymbol{\sigma}}(\omega_{{\bf k}})$ and $\pi:U_{x_0}^{\boldsymbol{\sigma}} \to \Sigma_{\boldsymbol{\sigma}}$ a surjective local diffeomorphism. 
This means that the lift $\pi^*\phi_g|_{U^{\boldsymbol{\sigma}}_{x_0}}$ has a unique  extension to $\tilde{\Sigma}$ 
as an equivariant element in $H^{s}_{\Gamma}(\tilde{\Sigma})$, and we shall therefore freely  identify $\phi_g$ with this extension when considering $F(\phi_g)$ 
with $F$ defined on $H^{s}_{\Gamma}(\tilde{\Sigma})$.
 
The Liouville field is thus a function of the zero mode   $c\in \R/ 2\pi R\Z$, the free field $X_g$,  ${\bf k}\in\Z^{2\mathfrak{g}}$,  the base point $x_0$, the canonical basis $\sigma$  and the choice of cohomology basis $\omega_1,\dots,\omega_{2{\mathfrak{g}}}$.

\begin{definition}[Path integral] 
We consider the path integral, defined for all $F\in  \mc{E}_R(\Sigma),$  
\begin{equation}\label{defPI}
   \langle F\rangle_{\Sigma,g } 
   :=
 \big(\frac{{\rm v}_{g}(\Sigma)}{{\det}'(\Delta_{g})}\big)^\hf\sum_{{\bf k}\in \Z^{2\mathfrak{g}}}e^{-\frac{1}{4\pi}\|\omega _{\bf k}\|_2^2}\int_{\R/2\pi R\Z}\E\Big[e^{-\frac{1}{2\pi}\langle \dd X_g,\omega_{\bf k} \rangle_2}F(\phi_g)e^{-\frac{i   Q}{4\pi}\int_{\Sigma_{\boldsymbol{\sigma}}}^{\rm reg} K_g\phi_g\,\dd v_g -\mu  M^g_\beta(\phi_g,\Sigma)}\Big]\,\dd c
\end{equation}
where the curvature term is defined following \eqref{def_reg_integtral}, namely
\begin{equation}
\int_\Sigma^{\rm reg} K_g\phi_g\, {\rm dv}_g :=\int_\Sigma   (c+X_g)K_g\,{\rm dv}_g +\int_{\Sigma_{\boldsymbol{\sigma}}}^{\rm reg}I^{\boldsymbol{\sigma}}_{x_0}(\omega_{\bf k})K_g \,{\rm dv}_g 
\end{equation}
\end{definition}
Note   that  Lemma \ref{phidescend} entails that  
the potential term $M^g_\beta(\phi_g,\Sigma)$ is well defined since the (regularized) integrand  descends to a function on $\Sigma$.
Note also that the above definition a priori depends on the marking $\sigma$ (i.e. the basis of $\mc{H}_1(\Sigma)$), the choice of closed forms representing a basis of cohomology (used to define the $(\omega_{\bf k})_{\bf k}$) as well as the base point $x_0$. We will show that actually  it does not and this is why we dropped all of these dependences from the notations. 

\begin{proposition}\label{propdefpath} 
The path integral \eqref{defPI} satisfies the following basic properties:
\begin{enumerate}
\item the quantity $ \langle F\rangle_{\Sigma,g }$  is well defined and finite for $F\in  \mc{E}_R(\Sigma)$, and extends to $F\in  \mc{L}^{\infty,p}(H^s(\Sigma))$ for $p>1$. 
\item the quantity $ \langle F\rangle_{\Sigma,g }$    depends neither on the base point $x_0\in \Sigma$,  nor on the choice of the homology basis $\sigma$, nor on the closed forms representing the cohomology basis dual to $\sigma$.
\end{enumerate}
\end{proposition}

\begin{proof}
We first prove (1). For this  we observe that $\E\Big[\Big|\exp\Big(-\mu  M^g_\beta(c+X_g+I^{\boldsymbol{\sigma}}_{x_0}(\omega_{\bf k}),\Sigma)\Big)\Big|\Big]\leq C$ for some constant $C$ depending only on $\Sigma$ and  $\beta$ (and thus not on ${\bf k}$). To see this, we want to use Proposition \ref{expmoment}. Let us consider an analytic closed simple curve $\mc{C}$ disconnecting the surface $\Sigma$ into two connected components $\Sigma_1$ and $\Sigma_2$ (for example $\mc{C}$ bound a small disk). Now we use the Markov decomposition in Proposition \ref{decompGFF} (item 2) to write the GFF as the sum
$$X_g=X_1+X_2+P-c_g$$ where $X_1,X_2,P$ are independent Gaussian processes, $X_{1},X_2$ are Dirichlet GFF respectively on $\Sigma_1,\Sigma_2$,   $P$ is the harmonic extension of the boundary values $X_{g}|_{\mc{C}}$ (which we also write  $X_{\mc{C}}$) and $c_g$ is a Gaussian random variable. Conditioning on the values $X_{\mc{C}}$, we can then bound our expectation as
\begin{align*}
\E\Big[ \Big|&\exp\Big(-\mu  M^g_\beta(c+X_g+I^{\boldsymbol{\sigma}}_{x_0}(\omega_{\bf k}),\Sigma)\Big)\Big|\Big]
\\
\leq &\E\Big[  \E\Big[\prod_{j=1,2} \Big| \exp\Big(-\mu  M^g_\beta(c-c_g+X_j+P+I^{\boldsymbol{\sigma}}_{x_0}(\omega_{\bf k}),\Sigma_j)\Big)\Big|\,|\,X_{\mc{C}}\Big]               \Big]
\\
\leq &\E\Big[ \prod_{j=1,2} \E\Big[  \exp\Big(  \Big| \mu  M^g_\beta(c+X_j+P+I^{\boldsymbol{\sigma}}_{x_0}(\omega_{\bf k}),\Sigma_j) \Big| \Big)\,|\,X_{\mc{C}}\Big]               \Big].
\end{align*}
Proposition \ref{expmoment} (applied with $Z=c$ and $f=e^{i\beta(P+I^{\boldsymbol{\sigma}}_{x_0}(\omega_{\bf k}))}$) then ensures that the following estimate holds true for the conditional expectation given $X_{\mc{C}}$
$$\E\Big[   \exp\Big( \big| \mu  M^g_\beta(c+X_j+P+I^{\boldsymbol{\sigma}}_{x_0}(\omega_{\bf k}),\Sigma_j) \big|\Big) |\,X_{\mc{C}}\Big]  \leq  e^{C|\mu| v}(1+C|\mu| ue^{C\mu^2u^2})$$
with
$$  u^2:=  \iint_{\Sigma_j^2}|f(x)||f(y)|e^{\beta^2G_{g,D}(x,y)}{\rm v}_g(\dd x){\rm v}_g(\dd y)   ,\quad \text{ and }\quad  v:=\int_{\Sigma_j} |f(x)|{\rm v}_g(\dd x),$$
for some constant $C$ only depending on $\beta$. Since $|f(x)|=1$  the above conditional expectation is uniformly bounded by a deterministic constant (independent of $c,{\bf k}$). We deduce
$$\E\Big[ \Big| \exp\Big(-\mu  M^g_\beta(c+X_g+I^{\boldsymbol{\sigma}}_{x_0}(\omega_{\bf k}),\Sigma)\Big)\Big|\Big]<+\infty.$$

As a consequence,   our claim (1) follows easily: indeed,  the term  $e^{-\frac{i   Q}{4\pi}\int_{\Sigma_{\boldsymbol{\sigma}}}^{\rm reg} K_g\phi_g\,\dd {\rm v}_g}$ has absolute value bounded by $1$. Using H\"older, the integrand in \eqref{defPI}   is thus bounded by
\begin{multline*}
  \big(\frac{{\rm v}_{g}(\Sigma)}{{\det}'(\Delta_{g})}\big)^\hf  \sum_{{\bf k}\in \Z^{2\mathfrak{g}}}e^{-\frac{1}{4\pi}\|\omega _{\bf k}\|_2^2} \Big(\int_{\R/2\pi R\Z}\E\Big[e^{-\frac{1}{2\pi}\langle \dd X_g,\omega_{\bf k} \rangle_2}|F(\phi_g)|^{p_1}\Big]\dd c \Big)^{1/p_1}\\
   \Big(\int_{\R/2\pi R\Z}\E\Big[e^{-\frac{1}{2\pi}\langle \dd X_g,\omega_{\bf k} \rangle_2}|e^{  -\mu    M^g_\beta(\phi_g,\Sigma)}|^{p_2}\Big]\dd c \Big)^{1/p_2}   
\end{multline*}
for some $p_1,p_2>1$ with $\frac{1}{p_1}+\frac{1}{p_2} =1$. Using \eqref{defnorm}
\begin{align*}
   \big(\frac{{\rm v}_{g}(\Sigma)}{{\det}'(\Delta_{g})}\big)^\hf  \sum_{{\bf k}\in \Z^{2\mathfrak{g}}}e^{-\frac{1}{4\pi}\|\omega _{\bf k}\|_2^2}
  \|F\|_{\mc{L}^{\infty,p}}e^{\frac{1}{4\pi p_1}\|(1-\Pi_1)\omega_{\bf k}  \|_2^2} \Big(\int_{\R/2\pi R\Z}\E\Big[e^{-\frac{1}{2\pi}\langle \dd X_g,\omega_{\bf k} \rangle_2}|e^{  -\mu   M^g_\beta(\phi_g,\Sigma)}|^{p_2}\Big]\dd c \Big)^{1/p_2}.
   \end{align*}

 We can apply the Girsanov transform to the last expectation above to get (by Lemma \ref{dX_gomega})
\begin{align*}
 \Big(\int_{\R/2\pi R\Z}\E\Big[&e^{-\frac{1}{2\pi}\langle \dd X_g,\omega_{\bf k} \rangle_2}|e^{  -\mu   M^g_\beta(\phi_g,\Sigma)}|^{p_2}\Big]\dd c \Big)^{1/p_2} \\
  \leq & e^{\frac{1}{4\pi p_2}\|(1-\Pi_1)\omega_{\bf k}  \|_2^2} \Big(\int_{\R/2\pi R\Z}\E\Big[ |e^{  -\mu   M^g_\beta(c+X_g+I_{x_0}^\sigma (\omega_{\bf k}^{\rm h}),\Sigma)}|^{p_2}\Big]\dd c \Big)^{1/p_2}  .
\end{align*}  
The last expectation is bounded by some constant $C$ independent of ${\bf k}$ as shown above. Summarizing,
\begin{align*}
  \langle F\rangle_{\Sigma,g } \leq  & C \big(\frac{{\rm v}_{g}(\Sigma)}{{\det}'(\Delta_{g})}\big)^\hf \sum_{{\bf k}\in \Z^{2\mathfrak{g}}}e^{-\frac{1}{4\pi}\|\omega _{\bf k}\|_2^2} \|F\|_{\mc{L}^{\infty,p}}e^{\frac{1}{4\pi  }\|(1-\Pi_1)\omega_{\bf k}  \|_2^2}
  \\
  =&
  C \big(\frac{{\rm v}_{g}(\Sigma)}{{\det}'(\Delta_{g})}\big)^\hf  \|F\|_{\mc{L}^{\infty,p}}\sum_{{\bf k}\in \Z^{2\mathfrak{g}}}e^{-\frac{1}{4\pi}\|\omega _{\bf k}^{\rm h}\|_2^2} .
  \end{align*}  
We conclude about integrability and well posedness of the path integral, as well as to its extension to $\mc{L}^{\infty,p}(H^s(\Sigma))$.

For (2), observe that changing the base point $x_0$ amounts to shifting the zero mode $c$ by some constant, and this is absorbed by  a change of variables in the $\dd c$-integral, since the integrand is periodic in $c$ (the assumption $F\in \mc{L}^{\infty,p}(\Sigma,g)$, hence periodic in $c$, is crucial). Next we show that the path integral is invariant under  change of cohomology basis (even if not dual to $\sigma$). Let $\hat{\omega}_{j} $, for $j=1,\dots, 2\mathfrak{g}$,  be another basis of cohomology.  For  ${\bf k}\in \Z^{2\mathfrak{g}}$, we set  $\hat{\omega}_{\bf k} :=\sum_{j=1}^{2\mathfrak{g}}k_j\hat{\omega}_{j} $. Then there is $A\in {\rm GL}_{2\mathfrak{g}}(\Z)$ such that $\omega _{\bf k}=\hat{\omega}_{A\bf k} +\dd f_{A\bf k}$ for all ${\bf k}$ and for some exact form $\dd f_{A\bf k} $ (see Lemma \ref{basisH^1}).
  We can then replace $\omega _{\bf k}$ by $\hat{\omega}_{A\bf k} +\dd f_{A\bf k}$ in the expression for the path integral. By making a change of variables in the summation over ${\bf k}$, we can get rid of the change of basis matrix $A$, i.e. we get
  \begin{equation} 
   \langle F\rangle_{\Sigma,g } 
   =
 \big(\frac{{\rm v}_{g}(\Sigma)}{{\det}'(\Delta_{g})}\big)^\hf\sum_{{\bf k}\in \Z^{2\mathfrak{g}}}e^{-\frac{1}{4\pi}\|\hat{\omega}_{\bf k}+\dd f_{\bf k}\|_2^2}\int_{\R/2\pi R\Z}\E\Big[e^{-\frac{1}{2\pi}\langle \dd X_g,\hat{\omega}_{\bf k}+\dd f_{\bf k}\rangle_2}F(\phi_g)e^{ -\frac{i   Q}{4\pi}\int_{\Sigma_{\boldsymbol{\sigma}}}^{\rm reg} K_g \phi_g\,\dd v_g -\mu  M^g_\beta(\phi_g,\Sigma)}\Big]\,\dd c
\end{equation}
where the Liouville field is now $\phi_g=c+X_g+I_{x_0}^\sigma(\hat{\omega}_{\bf k})+f_{\bf k}(x)-f_{\bf k}(x_0)$. Next we apply the Girsanov transform to the term $e^{-\frac{1}{2\pi}\langle \dd X_g,\dd f_{\bf k}\rangle_2}$. It produces a variance type term $e^{\frac{1}{4\pi}\| \dd f_{\bf k}\|_2^2}$ and it shifts the law of the GFF as $X_g\to X_g-(f_{\bf k}-m_g(f_{\bf k}))$ where $m_g(f_{\bf k}):=\frac{1}{{\rm v}_g(\Sigma)}\int_\Sigma f_{\bf k} \,\dd {\rm v}_g$. We can then shift the $c$-integral to absorb the constant $m_g(f_{\bf k})-f_{\bf k}(x_0)$.  Combining with the norm term in front of the expectation, we get the result.  

Finally if we consider another  basis of homology $\sigma'$.  Let $\hat{\omega}_{j} $, for $j=1,\dots, 2\mathfrak{g}$,  be another basis of cohomology, dual to $\sigma'$.  For  ${\bf k}\in \Z^{2\mathfrak{g}}$, we set  $\hat{\omega}_{\bf k} :=\sum_{j=1}^{2\mathfrak{g}}k_j\hat{\omega}_{j} $. Then there is $A\in {\rm GL}_{2\mathfrak{g}}(\Z)$ such that $\omega _{\bf k}=\hat{\omega}_{A\bf k} +\dd f_{A\bf k}$ for all ${\bf k}$ and for some exact form $\dd f_{A\bf k} $ (see Lemma \ref{basisH^1}). Then, the previous result tells us that we can replace, in the path integral associated to $\sigma$ and the $\omega_{\bf k}$'s, the closed 1-forms ${\omega}_{\bf k}$'s by the $\hat{\omega}_{\bf k}$'s. Next we want to change the homology basis. Only three terms depend now on $\sigma$: $F(\phi_g)$, the curvature term and the potential term. Note  that $e^{i\frac{1}{R}I^{\boldsymbol{\sigma}}_{x_0}(\hat{\omega}_{\bf k})}=  e^{i\frac{1}{R} I_{x_0}^{\sigma'}(\hat{\omega}_{\bf k}) }$. Also Lemma \ref{independence_basis} shows that $e^{ -\frac{i   Q}{4\pi}\int_{\Sigma_{\boldsymbol{\sigma}}}^{\rm reg} K_g \phi_g\,\dd v_g}=e^{ -\frac{i   Q}{4\pi}\int_{\Sigma_{\sigma'}}^{\rm reg} K_g \phi_g\,\dd v_g}$ because $Q\in \frac{1}{R}\Z$.  Hence we are done.
\end{proof}

\begin{remark}
The invariance under change of cohomology basis in  Proposition \ref{propdefpath} is quite intuitive from the path integral perspective. Indeed, note that formally, $\|\hat{\omega}_{\bf k} \|_2^2 +2\langle \dd X_g,\hat{\omega}_{\bf k} \rangle_2=\|\hat{\omega}_{\bf k} +\dd X_g\|^2-\|\dd X_g\|^2$. Next, the GFF expectation can be understood as a path integral $e^{-\frac{1}{4\pi}\|\dd X_g\|_2^2}DX_g$ so that, all in all, the path integral can be understood as $e^{-\frac{1}{4\pi}\|d\phi _g\|_2^2}D\phi_g$, namely a "Gaussian" measure over closed 1-forms.
\end{remark}
   
\subsection{Correlation functions: electric and magnetic operators}\label{sec:correls}
In this section, we introduce the correlation functions for all the operators we need in this theory. They are of two types, electric or magnetic, and we will construct each of them in the next two subsections. Finally we will construct mixed electric-magnetic operators by combining the two constructions.

\subsubsection{Magnetic operators}

Let $z_1,\dots, z_{n_\mathfrak{m}}$ be distinct points on a closed Riemann surface $\Sigma$. For each such a point $z_j$ we   assign a unit tangent vector $v_j\in T_{z_j}\Sigma$ and  a magnetic charge $m_j\in\Z$. We collect those datas in $\mathbf{z}=(z_1,\dots,z_{n_\mathfrak{m}}) \in\Sigma^{n_\mathfrak{m}}$, 
$\mathbf{v}=((z_1,v_1),\dots,(z_{n_\mathfrak{m}},v_{n_\mathfrak{m}}))\in (T\Sigma)^{n_\mathfrak{m}}$ and $\mathbf{m}\in\Z^{n_\mathfrak{m}}$. We assume that $\sum_{j=1}^{n_\mathfrak{m}} m_j=0$.
We wish to insert on $\Sigma$ magnetic defects at the $z_j$'s so that the  field $\phi_g(z)$ is multivalued and gains a factor $ 2\pi m_j  R$  when $z$ turns once around a small circle around $z_j$ (and not the other $z_{j'}$'s). As before, we choose a set of datas given by Assumption \ref{ass}. We assume that the geometric symplectic basis $\sigma$ as well as the base point $x_0$ are distinct from $\mathbf{z}$, in particular $\mathbf{z}\subset \Sigma_{\boldsymbol{\sigma}}$ (recall \eqref{sigma^2}).

The structure of the magnetic operators  relies on the construction of the harmonic 1-forms   of Proposition \ref{harmpoles}. Consider the harmonic 1-form $\nu^{\rm h}_{\mathbf{z},\mathbf{m}}$ with windings $2\pi Rm_j$ around the point $z_j$ given by this Proposition\footnote{Being harmonic is not necessary, we could choose closed 1-forms instead, according to the same proposition.}.

We  define the Liouville field 
 \begin{equation}\label{eq:liouvmag}
 \phi_g:=c+X_g+I^{\boldsymbol{\sigma}}_{x_0}( \omega_{\mathbf{k}})+I^{ \boldsymbol{\xi}}_{x_0}( \nu^{\rm h}_{\mathbf{z},\mathbf{m}}).
\end{equation}
As explained in Sections \ref{sub:fund} and Section \ref{curvature_mp}, this field belongs to $H^{s}(\Sigma\setminus \{\boldsymbol{\sigma}\cup \boldsymbol{\xi}\})$ but can alternatively be viewed as an element in $H^s_\Gamma (\tilde{\Sigma}_{\bf z})$ 
as $I^{\boldsymbol{\sigma}}_{x_0}( \omega_{\mathbf{k}})+I^{ \boldsymbol{\xi}}_{x_0}( \nu^{\rm h}_{\mathbf{z},\mathbf{m}})$ has a lift to a fundamental domain of $\pi_1(\Sigma_{\bf z},x_0)$ in $\tilde{\Sigma}_{{\bf z}}$ given by $I_{x_0}( \omega_{\mathbf{k}})+I_{x_0}( \nu^{\rm h}_{\mathbf{z},\mathbf{m}})$. 
 Recall that each $u\in H^s_\Gamma (\tilde{\Sigma}_{\bf z})$ decomposes uniquely as 
 $u=\pi^*f+I_{x_0}(\omega_{\bf k})+I_{x_0}(\nu^{\rm h}_{{\bf z}, {\bf m}})$ for some $f\in H^s(\Sigma)$, $({\bf k},{\bf m})\in \Z^{2{\mathfrak{g}}+n_{\mathfrak{m}}}$. We also write $f_0={\rm v}_g(\Sigma)^{-1}\int_{\Sigma} f{\rm dv}_g$ as in \eqref{f_Fourier}.
 Let us then consider the space $\mc{E}^{\rm m}_R(\Sigma)$ of functionals $F$, defined on $H^s_{\Gamma}(\tilde\Sigma_{\bf z})$, of the form
\begin{equation}\label{polytrigo}
F(u)=\sum_{n=-N}^Ne^{\frac{i}{R} n f_0}P_n(f-f_0)  G( e^{\frac{i}{R} I_{x_0}(\omega_{\bf k})} ) G'( e^{\frac{i}{R} I_{x_0}( \nu^{\rm h}_{\mathbf{z},\mathbf{m}})} ) 
\end{equation}
 for arbitrary $N\in\N$, where $P_n$ are polynomials of the form $P_n\big (\cjg f-f_0,g_1\cjd,\dots,\cjg f-f_0,g_{m_n}\cjd\big)$ where $g_1,\dots,g_{m_n}$  belong to $H^{-s}(\Sigma)$, and $G,G'$ is   continuous and bounded on $C(\Sigma;\S^1)$. Notice in particular that these functionals are $2\pi R$-periodic in the zero mode $f_0$. Next we define   the space  for ${\bf m}$ fixed:
 \begin{itemize}
\item $\mc{L}^{\infty,p}_m(H^s(\Sigma))$ as the closure of $\mc{E}^m_R(\Sigma,g)$ with respect to the norm defined by
$$\|F\|_{\mc{L}^{\infty,p}_m}:=\sup_{\bf k}\Big(\int_{\R/2\pi R\Z}\E\Big[e^{-\frac{1}{2\pi}\langle \dd X_g,\omega_{\bf k} \rangle_2-\frac{1}{4\pi}\|\dd f_{\bf k}\|_2^2}|F(c+\pi^*X_g+I_{x_0}(\omega_{\bf k})+I_{x_0}( \nu^{\rm h}_{\mathbf{z},\mathbf{m}}))|^p\Big]\,\dd c\Big)^{1/p}$$
where $(1-\Pi_1)\omega_{\bf k}=\dd f_{\bf k}$ with $\Pi_1$ is the projection on harmonic forms (recall Lemma \ref{dX_gomega}).
\end{itemize}

\begin{lemma}\label{invnormm} The norm $\|F\|_{\mc{L}^{\infty,p}_m}$ does not depend on the choice of cohomology basis.
 \end{lemma}

\begin{proof}
The invariance under change of cohomology basis is similar to Lemma \ref{invariance_of_norm}. 
\end{proof}

\begin{definition}
The definition of the path integral with magnetic operators at locations $\mathbf{z}=(z_1,\dots,z_{n_{\mathfrak{m}}})$ 
with magnetic charges $\mathbf{m}=(m_1,\dots,m_{n_{\mathfrak{m}}})$ and tangent vectors ${\bf v}=((z_1,v_1),\dots,(z_{n_{\mathfrak{m}}},v_{n_{\mathfrak{m}}}))\in (T\Sigma)^{n_{\mathfrak{m}}}$ reads for $F\in \mc{E}^{\rm m}_R(\Sigma)$
\begin{align}\label{defPI:mag}
   \langle F V^g_{(0,\mathbf{m})}({\bf v}) \rangle_{\Sigma,g } 
   :=&
 \big(\frac{{\rm v}_{g}(\Sigma)}{{\det}'(\Delta_{g})}\big)^\hf\sum_{{\bf k}\in \Z^{2\mathfrak{g}}}e^{-\frac{1}{4\pi}\|\omega _{\bf k}\|_2^2-\frac{1}{4\pi}\|\nu^{\rm h}_{\mathbf{z},\mathbf{m}}\|^2_{g,0}-\frac{1}{2\pi}\langle \omega_{\bf k},\nu^{\rm h}_{\mathbf{z},\mathbf{m}}\rangle_2}
 \\ & \int_{\R/2\pi R\Z}\E\Big[e^{-\frac{1}{2\pi}\langle \dd X_g,\omega_{\bf k} \rangle_2}F(\phi_g)e^{-\frac{i   Q}{4\pi}\int_{\Sigma }^{\rm reg} K_g\phi_g\,\dd {\rm v}_g -\mu  M^g_\beta(\phi_g,\Sigma)}\Big]\,\dd c\nonumber
\end{align}
where $V^g_{(0,\mathbf{m})}({\bf v}) $ is a formal notation to indicate no electric charge (the $0$ index) but 
the presence of magnetic charges ${\bf m}=(m_1,\dots,m_{n_{\mathfrak{m}}})$, and where the regularized curvature term has now a further magnetic  term 
\begin{equation}\label{curv:mag}
\int_\Sigma^{\rm reg} K_g\phi_g\,\dd {\rm v}_g :=\int_\Sigma   (c+X_g)K_g\,\dd {\rm v}_g +\int_{\Sigma_{\boldsymbol{\sigma}}}^{\rm reg} I^{\boldsymbol{\sigma}}_{x_0}(\omega_{\bf k})K_g \,\dd {\rm v}_g +\int_{\Sigma}^{\rm reg}  I^{ \boldsymbol{\xi}}_{x_0}( \nu^{\rm h}_{\mathbf{z},\mathbf{m}}) K_g \,\dd {\rm v}_g
\end{equation}
with $\int_{\Sigma}^{\rm reg}  I^{ \boldsymbol{\xi}}_{x_0}( \nu^{\rm h}_{\mathbf{z},\mathbf{m}}) K_g \,\dd v_g$ defined by 
\eqref{curv:mag2} and $\int_{\Sigma}^{\rm reg} I^{\boldsymbol{\sigma}}_{x_0}(\omega_{\bf k})K_g \,\dd {\rm v}_g$ by \eqref{def_reg_integtral}.
\end{definition}
This definition is similar to \eqref{defPI}; notice in particular that we have not put the 
term $e^{-\frac{1}{2\pi}\langle \dd X_g,\nu^{\rm h}_{{\bf z},{\bf m}} \rangle_2}$ since by Proposition \ref{harmpoles}, we know that
$\dd^*\nu^{\rm h}_{{\bf z},{\bf m}}\in C^\infty(\Sigma)$ ($d^*\nu^{\rm h}_{{\bf z},{\bf m}}$ is understood in the distributional sense), and in turn equal to $0$, 
thus
\[ \lim_{\eps\to 0} \langle \dd X_{g,\eps},\nu^{\rm h}_{{\bf z},{\bf m}}\cjd_2=\lim_{\eps\to 0} \langle X_{g,\eps},\dd^*\nu^{\rm h}_{{\bf z},{\bf m}}\cjd_2=0.\]
This term would however appear if we were using closed 1-form $\nu_{{\bf z},{\bf m}}$ (with prescribed windings) instead of the harmonic 1-form $\nu^{\rm h}_{{\bf z},{\bf m}}$\footnote{
Adding an exact form to  $ \nu^{\rm h}_{\mathbf{z},\mathbf{m}}$ in the path integral expression amounts to adding this exact form to $\omega_{\bf k}$ so that our statement is already completely equivalent to considering closed 1-forms instead of $ \nu^{\rm h}_{\mathbf{z},\mathbf{m}}$.
}.

Also  this path integral possesses the same properties as  \eqref{defPI}: this can be shown in the same way up to some caveats that we explain in the proof of the proposition below. One further important property is that the path integral does not depend on the choice of the defect graph.

\begin{proposition}\label{propdefpathmag} 
The path integral \eqref{defPI:mag} satisfies the following basic properties:
\begin{enumerate}
\item the quantity $ \langle F V^g_{(0,\mathbf{m})}({\bf v}) \rangle_{\Sigma,g }$  is well defined and finite for $F\in  \mc{E}^m_R(\Sigma)$, and extends to $F\in   \mc{L}^{\infty,p}_m(H^s(\Sigma))$ for $p>1$. 
\item the quantity $ \langle FV^g_{(0,\mathbf{m})}({\bf v}) \rangle_{\Sigma,g }$     depends neither on the base point $x_0\in \Sigma$, nor on the choice of geometric symplectic basis $\sigma$ of $\mc{H}^1_{R}(\Sigma)$, nor on the choice of the cohomology basis dual to  $\sigma$,
\item  the quantity $ \langle FV^g_{(0,\mathbf{m})}({\bf v}) \rangle_{\Sigma,g }$ only depends on the location of the points ${\bf v}=(z_j,v_j)_{j=1,\dots,n_{\mathfrak{m}}}$ in $T\Sigma$ and the charges $\mathbf{m}$, but not  on the defect graph.
\end{enumerate}
\end{proposition}

\begin{proof}
The proof of items (1) and (2) is similar to Proposition \ref{propdefpath}, but there is only one point to be careful with. We have to check the summability over ${\bf k}$ as this expression features now a further term $e^{-\frac{1}{2\pi}\langle \omega_{\bf k},\nu^{\rm h}_{\mathbf{z},\mathbf{m}}\rangle_2}$. This term is bounded by $e^{C_{{\bf m}}|{\bf k}|}$ for some $C_{\bf m}>0$ and thus does not affect the summability over ${\bf k}$ in the proof of Proposition \ref{propdefpath}.
To prove (3), it suffices to use Lemma \ref{inv_par_graphe}. 
\end{proof}

Below, we denote by $S\Sigma:=\{ (x,v)\in T\Sigma\,|\, |v|_{g_x}=1\}$ the unit sphere bundle.
 \begin{corollary}\label{corospin}
For $r_{\theta}$ being the rotation of angle $\theta$ in the tangent bundle, set 
\[r_{\boldsymbol{\theta}}\mathbf{v}:=((z_1,r_{\theta_1}v_1),\dots,(z_{n_{\mathfrak{m}}},r_{\theta_{n_\mathfrak{m}}} v_{{n_\mathfrak{m}}}) )\in (T\Sigma)^{n_{\mathfrak{m}}}.\]
Then
\[ \langle   FV^g_{(0,\mathbf{m})}(r_{\boldsymbol{\theta}}\mathbf{v}) \rangle_{\Sigma,g } =e^{-i QR\langle\mathbf{m},\boldsymbol{\theta}\rangle }  \langle F  V^g_{(0,\mathbf{m})}({\bf v}) \rangle_{\Sigma,g } .\]
Denoting $ 2\pi R Q=-\ell\in -\N$,  the correlation functions, viewed as functions 
\[ \mathbf{v}\in (S\Sigma)^{n_{\mathfrak{m}}}\mapsto  \langle   FV^g_{(0,\mathbf{m})}(\mathbf{v}) \rangle_{\Sigma,g } \]
are sections of $\mc{K}^{\ell m_1}\otimes \dots \otimes 
\mc{K}^{\ell m_{n_{\mathfrak{m}}}}$ where $\mc{K}=(T^{1,0}\Sigma)^*$ is the canonical line bundle and 
$\mc{K}^{-1}=(T^{0,1}\Sigma)^*$ the anti-canonical bundle; 
by convention, if $k\geq 1$, we write $\mc{K}^k:=\otimes_{j=1}^k\mc{K}$ and $\mc{K}^{-k}:=\otimes_{j=1}^k\mc{K}^{-1}$.
 \end{corollary}

\begin{proof} It suffices to consider the case where only one vector is rotated, as we can apply recursively the result to each angle.  Consider the defect graph $\mc{D}_{\mathbf{v},\boldsymbol{\xi}}$. Up to relabelling the sites $z_j$, we may assume that the charges are in increasing order $m_1\leq \dots\leq m_{n_{\mathfrak{m}}}$. Since the correlation functions don't depend on the graph, we may choose the canonical defect graph $z_1\to z_2\to\dots \to z_{n_{\mathfrak{m}}}$. Let us first investigate the case when the 1st vector is rotated. We proceed as in the proof of the previous proposition using Gauss-Bonnet. Denote by $(\xi_p)_p$ the paths of the defect graph $\mc{D}_{\mathbf{v},\boldsymbol{\xi}}$. Let us consider another  path $\tilde\xi_1$ such that $\tilde\xi(0)=z_1$, $\tilde\xi(1)=z_2$ with $\tilde\xi'(0)=\tilde \lambda_1 r_{\theta_1}v_1$, $\tilde\xi'(1)=\tilde \lambda_2 v_2$ for $\tilde \lambda_1,\tilde \lambda_2>0$. We compute the change in the correlation functions when replacing $\xi_1$ by $\tilde\xi$. Let us call $\tilde{\mc{D}}_{r_{\boldsymbol{\theta}}\mathbf{v},\boldsymbol{\xi}}$ the defect graph after this replacement. As before, the curves $\xi_1$ and $\tilde{\xi}$ bound a domain $D$ homeomorphic to a disk, and the boundary of $D$ inherits an orientation from $\Sigma$. Without loss of generality, we may assume this is $\tilde{\xi}$ is positively oriented, and $\xi_1$ negatively oriented with respect to the orientation of $\pl D$. The two defect graphs give rise to two different primitives $ I^{\boldsymbol{\xi}}_{x_0}( \nu^{\rm h}_{\mathbf{z},\mathbf{m}})$ and $   I^{\tilde {\boldsymbol{\xi}}}_{x_0}( \nu^{\rm h}_{\mathbf{z},\mathbf{m}})$ and, on $D$, we have $ I^{\tilde {\boldsymbol{\xi}}}_{x_0}( \nu^{\rm h}_{\mathbf{z},\mathbf{m}}) = I^{\boldsymbol{\xi}}_{x_0}( \nu^{\rm h}_{\mathbf{z},\mathbf{m}})-2\pi Rm_1$, where we noted that  $\kappa(\xi_1)=\kappa(\tilde\xi)=m_1$. 
The difference of the two regularized integrals is then
\begin{align*}
\int_{\Sigma}^{\rm reg}&  I^{\tilde {\boldsymbol{\xi}}}_{x_0}( \nu^{\rm h}_{\mathbf{z},\mathbf{m}}) K_g \,\dd {\rm v}_g-\int_{\Sigma}^{\rm reg}  I^{\boldsymbol{\xi}}_{x_0}( \nu^{\rm h}_{\mathbf{z},\mathbf{m}}) K_g \,\dd {\rm v}_g
\\
&=\int_D(  I^{\tilde {\boldsymbol{\xi}}}_{x_0}( \nu^{\rm h}_{\mathbf{z},\mathbf{m}})-  I^{\boldsymbol{\xi}}_{x_0}( \nu^{\rm h}_{\mathbf{z},\mathbf{m}}))K_g \dd {\rm v}_g+4\pi m_1 R(\int_{\xi_1}k_g\dd \ell_g-\int_{\tilde\xi}k_g\dd \ell_g.
\end{align*}
 Now we apply again the Gauss-Bonnet theorem  on $D$ to get
\[
\int_D(  I^{\tilde {\boldsymbol{\xi}}}_{x_0}( \nu^{\rm h}_{\mathbf{z},\mathbf{m}})-  I^{\boldsymbol{\xi}}_{x_0}( \nu^{\rm h}_{\mathbf{z},\mathbf{m}}))K_g \dd {\rm v}_g
 =-2\pi Rm_1\int_D  K_g \dd {\rm v}_g=4\pi Rm_1(\int_{\tilde{\xi}}k_g\dd\ell_g-\int_{\xi_1}k_g\dd\ell_g)+4\pi Rm_1\theta_1
\]
and we deduce that the difference of the two regularized integrals is $4\pi m_1\theta_1R$. Hence $ \langle  F V^g_{(0,\mathbf{m})}(r_{\boldsymbol{\theta}}\mathbf{v}) \rangle_{\Sigma,g } =e^{-i QRm_1\theta_1 } \langle  F V^g_{(0,\mathbf{m})}(\mathbf{v}) \rangle_{\Sigma,g } $. The same argument works when we rotate the last vector. 
 The proof has a little twist when rotating an intermediate point because turning an angle affects then two domains, each of which has to be applied the Gauss-Bonnet theorem. But there is no further subtlety   and this yields similarly $ \langle  F V^g_{(0,\mathbf{m})}(r_{\boldsymbol{\theta}}\mathbf{v}) \rangle_{\Sigma,g } =e^{-i QRm_{p}\theta_p } \langle   FV^g_{(0,\mathbf{m})}( \mathbf{v}) \rangle_{\Sigma,g } $ in case we rotate the $p$-th vector only. Hence our claim. 
 Any $C^0$ function $f$ on $S\Sigma$ can be decomposed in Fourier modes in the fibers (which are circles), the fact that for $k\in \Z$ one has
 $f(z,r_{\theta}v)=e^{ik\theta }f(z,v)$ for all $x$ means exactly that $f$ has only Fourier modes in the fibers of order $k$, which means that 
 $f$ is the restriction of an $C^0$ section of $\mc{K}^k$ to the unit sphere bundle (see \cite[Chapter 4.4.]{Guillarmou-Mazzuchelli} for instance).
\end{proof}

\subsubsection{Electric operators}
We construct now the electric operators in the presence of magnetic operators. Pure electric correlations can be obtained as a particular case of the following by taking the magnetic field to be $0$. Such fields need to be regularized. Recall that each $u\in H^s_\Gamma (\tilde{\Sigma}_{\bf z})$ decomposes uniquely as 
 $u=\pi^*f+I_{x_0}(\omega_{\bf k})+I_{x_0}(\nu^{\rm h}_{{\bf z}, {\bf k}})$ for some $f\in H^s(\Sigma)$, $({\bf k},{\bf m})\in \Z^{2{\mathfrak{g}}+n_{\mathfrak{m}}}$.  We introduce the regularized electric operators, for fixed electric charge $\alpha\in\R$ and $x\in \Sigma$,  
\begin{equation*}
V_{\alpha,g,\eps}(u,x)=\eps^{-\alpha^2/2}  e^{i\alpha u_{g,\epsilon}(x) } 
\end{equation*}
where $u_{g,\epsilon}$ is a $g$-regularization of the field $u$. When $u=\phi_g$ is the Liouville field (as below), we will shortcut this expression as  $V_{\alpha,g,\eps}(x)$.

Next, we choose distinct points $x_1,\dots,x_{n_{\mathfrak{e}}}$ on $\Sigma$ (and distinct from the locations $\mathbf{z}$ of the magnetic defects), which we collect in the vector $\mathbf{x}\in\Sigma^{n_{\mathfrak{e}}}$, with associated electric charges $\boldsymbol{\alpha}:=(\alpha_1,\dots,\alpha_{n_{\mathfrak{e}}})\in \R^{n_{\mathfrak{e}}}$. We denote $V_{(\boldsymbol{\alpha},0)}^{g,\epsilon}(u,\mathbf{x}):=\prod_{j=1}^{n_{\mathfrak{e}}}V_{\alpha_j,g,\eps}(u,x_j)$ (which we shortcut as $V_{(\boldsymbol{\alpha},0)}^{g,\epsilon}(\mathbf{x})$ if  $u$ is the Liouville field). Note that this functional belongs to $\mc{L}^{\infty,p}_m(H^s(\Sigma))$   iff the charges satisfy $\alpha_j\in \frac{1}{R}\Z$, which we will assume from now on.
Let us introduce the function $u_{\bf x}(x)=\sum_{j=1}^{n_{\mathfrak{e}}}i\alpha_jG_g(x,x_j)$ and note that $u_{\bf x}\in H^s ( \Sigma)$ for $s<1$.
 We consider the space $\mc{E}^{\rm m}_R(\Sigma)$ as before. Next we define   the space: 
 \begin{itemize}
\item $\mc{L}^{\infty,p}_{\rm e,m}(H^s(\Sigma))$ as the closure of $\mc{E}^{\rm m}_R(\Sigma)$ with respect to the  seminorm 
$$\|F\|_{\mc{L}^{\infty,p}_{\rm e,m}}:=\sup_{\bf k}\Big(\int_{\R/2\pi R\Z}\E\Big[e^{-\frac{1}{2\pi}\langle \dd X_g ,\omega_{\bf k} \rangle_2-\frac{1}{4\pi}\|\dd f_{\bf k}\|_2^2}|F(c+X_g+u_{\bf x}+I_{x_0}(\omega_{\bf k})+I_{x_0}( \nu^{\rm h}_{\mathbf{z},\mathbf{m}}))|^p\Big]\,\dd c\Big)^{1/p}$$
where $(1-\Pi_1)\omega_{\bf k}=\dd f_{\bf k}$ with $\Pi_1$ is the projection on harmonic forms (recall Lemma \ref{dX_gomega}).
\end{itemize}
Similarly to Lemma \ref{invnormm}, we claim
\begin{lemma}\label{invnormm_bis} The norm $\|F\|_{\mc{L}^{\infty,p}_{\rm e,m}}$ does not depend on the choice of cohomology basis 
 \end{lemma}
 \color{black}

The path integral with both electric and magnetic operators is defined by the limit
\begin{equation}\label{defcorrelg}
\cjg   F V_{(\boldsymbol{\alpha},0)}^{g}(\mathbf{x})V^g_{(0,\mathbf{m})}({\bf v})  \cjd_ {\Sigma, g}:=\lim_{\eps \to 0} \: \cjg F V_{(\boldsymbol{\alpha},0)}^{g,\epsilon}(\mathbf{x}) V^g_{(0,\mathbf{m})}({\bf v}) \cjd_ {\Sigma, g}
\end{equation}
for $F \in  \mc{E}^{\rm m}_R(\Sigma)$. The existence  of the limit is non trivial and only holds under some constraints that we summarize below:

\begin{theorem}\label{limitcorel} Assume that
\begin{align}\label{seiberg}
 \forall j,\quad \alpha_j> Q\quad \text{ and }\quad \alpha_j\in \frac{1}{R}\Z   ,\qquad  \sum_{j=1}^{n_{\mathfrak{m}}}m_j=0.
\end{align}
The mapping $F\in  \mc{E}^m_R(\Sigma)\mapsto \cjg  F V_{(\boldsymbol{\alpha},0)}^{g}(\mathbf{x})V^g_{(0,\mathbf{m})}({\bf v})  \cjd_ {\Sigma, g}$ satisfies the following properties:
\begin{enumerate}
\item {\bf Existence:} it is well-defined and extends to $F\in \mc{L}^{\infty,p}_{\rm e,m}(H^s(\Sigma))$ for $s<0$. For $F=1$, it defines  the correlation functions  $\cjg    V_{(\boldsymbol{\alpha},0)}^{g}(\mathbf{x})V^g_{(0,\mathbf{m})}({\bf v})  \cjd_ {\Sigma, g}$.
\item {\bf Conformal anomaly:} let $g'=e^{\rho}g$ be two conformal metrics on the closed Riemann surface $\Sigma$ for some $\rho\in C^\infty(\Sigma)$, and let $ \mathbf{x}=(x_1,\dots,x_{n_{\mathfrak{e}}})\in \Sigma^{n_{\mathfrak{e}}}$, ${\bf v}=((z_1,v_1),\dots,(z_{n_{\mathfrak{m}}},v_{n_{\mathfrak{m}}}))\in (T\Sigma)^{n_{\mathfrak{m}}}$ with $z_j$ and $x_i$ distincts for all $i,j$, and $\boldsymbol{\alpha}=(\alpha_1,\dots,\alpha_{n_{\mathfrak{e}}})\in\R^{n_{\mathfrak{e}}}$ obeying the constraint \eqref{seiberg}. Then for ${\bf m}=(m_1,\dots,m_{n_{\mathfrak{m}}})\in \Z^{n_{\mathfrak{m}}}$, we have
\begin{align}\label{confan} 
& \big\cjg  F  V_{(\boldsymbol{\alpha},0)}^{g'}(\mathbf{x})V^{g'}_{(0,\mathbf{m})}({\bf v})  \big\cjd_ {\Sigma, g'}
 \\
 &=\big\cjg   F(\cdot- \tfrac{i  Q}{2}\rho) V_{(\boldsymbol{\alpha},0)}^{g}(\mathbf{x})V^g_{(0,\mathbf{m})}({\bf v}) \big \cjd_ {\Sigma, g} 
e^{\frac{{\bf c}}{96\pi}\int_{\Sigma}(|d\rho|_g^2+2K_g\rho) {\rm dv}_g-\sum_{j=1}^{n_{\mathfrak{e}}}\Delta_{(\alpha_j,0)}\rho(x_j)-\sum_{j=1}^{n_{\mathfrak{m}}}\Delta_{(0,m_j)}\rho(z_j)}\nonumber
\end{align}
where the real numbers $\Delta_{\alpha,m}$, called {\it conformal weights}, are defined by the relation for $\alpha\in\R$
 \begin{align}\label{deltaalphadef}
\Delta_{(\alpha,m)}=\frac{\alpha}{2}(\frac{\alpha}{2}-  Q)+\frac{m^2R^2}{4}
\end{align} 
and the central charge is ${\bf c}:=1-6 Q^2 $.
\item {\bf Diffeomorphism invariance:} let $\psi:\Sigma'\to \Sigma$ be an orientation preserving diffeomorphism. Then  
\begin{equation}\label{diffinvariance}
\big\cjg  F (\phi_{\psi^*g}) V_{(\boldsymbol{\alpha},0)}^{\psi^*g}(\mathbf{x})V^{\psi^*g}_{(0,\mathbf{m})}({\bf v})  \big\cjd_ {\Sigma', \psi^*g}=\big\cjg  F(\phi_g\circ\psi)  V_{(\boldsymbol{\alpha},0)}^{g}(\psi(\mathbf{x}))V^{g}_{(0,\mathbf{m})}(\psi_*{\bf v})  \big\cjd_ {\Sigma, g}
\end{equation}
where we used the collective notations 
\[\psi_{*}\mathbf{v}:=((\psi(z_1),\dd \psi_{z_1}.v_1),\dots,(\psi(z_{n_\mathfrak{m}}),\dd \psi_{z_{n_{\mathfrak{m}}}}.v_{n_{\mathfrak{m}}})), \quad \psi({\bf x})=(\psi(x_1),\dots,\psi(x_{n_{\mathfrak{e}}})).\]
\item {\bf Spins:} with $r_{\boldsymbol{\theta}}\mathbf{v}:=((z_1,r_{\theta_1}v_1),\dots,(z_{n_\mathfrak{m}},r_{\theta_{n_\mathfrak{m}}} v_{\theta_{n_\mathfrak{m}}}))$, then
\begin{equation}\label{spinrelation} 
 \langle   FV_{(\boldsymbol{\alpha},0)}^{g}(\mathbf{x})V^g_{(0,\mathbf{m})}(r_{\boldsymbol{\theta}}\mathbf{v}) \rangle_{\Sigma,g } =e^{-i QR\langle\mathbf{m},\boldsymbol{\theta}\rangle  }  \langle F V_{(\boldsymbol{\alpha},0)}^{g}(\mathbf{x}) V^g_{(0,\mathbf{m})}({\bf v}) \rangle_{\Sigma,g } .
 \end{equation}
\end{enumerate}

 \end{theorem}
\begin{proof} We split the proof in 4 parts: the existence and convergence of the path integral, the conformal anomaly and the diffeomorphism invariance.

(1) \textbf{Existence.} The condition $ \alpha_j\in \frac{1}{R} \Z$ makes sure that the product $ \prod_j V_{\alpha_j,g,\eps}(x_j)$ is in $ \mc{L}^{\infty,p}_{\rm e,m}(H^s(\Sigma))$. We will use the Cameron-Martin theorem to transform the electric insertions into singularities in the potential. There is some caveat here: this theorem applies only for real valued Gaussians whereas we face here imaginary Gaussians. We thus need to use an analytic continuation argument. The fact that $F \in \mc{E}_R(\Sigma)$ is crucial for this: indeed $F$ is polynomial, hence analytic, in linear observables of the GFF. This argument only need to be applied to the GFF expectation and that is why we only average over $\E$ below (and irrelevant factors are removed from computations). So, consider the map
$$\mathbf{w}:=(w_1,\dots,w_{n_{\mathfrak{e}}})\in \C^{n_\mathfrak{e}}\mapsto A(\mathbf{w}) $$
with $A$ defined by (the variables $c$ and ${\bf k}$ are fixed)
$$A(\mathbf{w}):=\E\Big[e^{-\frac{1}{2\pi}\langle \dd X_g,\omega_{\bf k} \rangle_2}F(\phi_g)\prod_{j=1}^{n_{\mathfrak{e}}} \epsilon^{-w_j^2/2}e^{iw_j(c+X_{g,\epsilon}(x_j))}e^{-\frac{i   Q}{4\pi}\int_\Sigma^{\rm reg} K_g\phi_g\,\dd v_g -\mu M^g_\beta(\phi_g,\Sigma)}\Big]$$
where $\phi_g=c+X_g+I^{\boldsymbol{\sigma}}_{x_0}(\omega_{\bf k})+ I^{  {\boldsymbol{\xi}}}_{x_0}( \nu^{\rm h}_{\mathbf{z},\mathbf{m}})$ is the Liouville field.
For fixed $\epsilon>0$ this quantity is obviously holomorphic on $\C^{n_\mathfrak{e}}$. 
For $w_1,\dots,w_m\in i\R$, we can use the Cameron-Martin theorem to get that
\begin{align}
&A(\mathbf{w})=   e^{-\frac{1}{2}\E[(\sum_{j=1}^{n_{\mathfrak{e}}}w_jX_{g,\epsilon}(x_j))^2]}\prod_{j=1}^{n_{\mathfrak{e}}}\epsilon^{-\frac{w_j^2}{2}}\nonumber
\\
\label{mainquant}& e^{i\sum_jw_jc}\E\Big[e^{-\frac{1}{2\pi}\langle \dd X_g+\dd u_{0,\epsilon},\omega_{\bf k} \rangle_2}F(\phi_g+u_{0,\epsilon} )e^{-\frac{i   Q}{4\pi}\int_\Sigma^{\rm reg}  K_g(\phi_g+u_{0,\epsilon})\,\dd v_g -\mu  M^g_\beta(\phi_g +u_{0,\epsilon},\Sigma)}\Big] 
\end{align}
where we have set $G_{\epsilon,\epsilon'}(x,x'):=\E[X_{g,\epsilon}(x)X_{g,\epsilon}(x')]$ (with the convention that $X_{g,0}=X_g$) and $u_{\epsilon,\epsilon'}( x):=\sum_{j=1}^{n_{\mathfrak{e}}} i w_jG_{\epsilon,\epsilon'}(x,x_j)$, which is a continuous function of $x\in\Sigma$, holomorphic in $\mathbf{w}$. Now we would like to argue that the right hand side is a holomorphic function of $\mathbf{w}$. As already explained, the fact that $F$ is polynomial in the GFF is crucial but there is a further subtlety here in the potential: we stress that $M_\beta^g$ is not a.s. a measure, but a distribution of order 2. Therefore, to apply the theorem of complex differentiation for parametrized integrals, we need to control the quantities $ \partial_x^2u_{0,\epsilon}$ uniformly over $x$ and the compact subsets in $\mathbf{w}$. The point is that, because of our regularization along geodesic circles, the partial derivatives $\partial_x^2u_{0,\epsilon}$ do not exist as functions, hence are not bounded. Therefore, it is not clear that a.s. the mapping $\mathbf{w}\mapsto M^g_\beta(\phi_g+u_{0,\epsilon},\Sigma)$ is holomorphic (recall that the dependence on $\mathbf{w}$ is encoded in the function $u_{0,\epsilon}$). Furthermore, the term $I^{\boldsymbol{\xi}}_{x_0}(\nu^{\rm h}_{\mathbf{z,m}})$ appearing in the potential is not of class $C^2$. Yet, this statement is true at the level of expectation values and this is all what we need. To prove  this, we will approximate $u_{0,\epsilon}$ by a family of $\mathbf{w}$-holomorphic and two times $x$-differentiable functions. Let us thus consider a family $(u_{0,\epsilon,\delta})_\delta$ obtained by convolution in the $x$-variable of the function  $u_{0,\epsilon}$ with a mollifying family indexed by $\delta$, which stands for the regularization scale, and such that $\sup_{\Sigma}|u_{0,\epsilon}-u_{0,\epsilon,\delta}|\to 0$ as $\delta\to 0$. Such a family is holomorphic in $\mathbf{w}$ and two times differentiable in $x$ for each fixed $\delta>0$. The fact that $I^{\boldsymbol{\xi}}_{x_0}(\nu^{\rm h}_{\mathbf{z,m}})$ is not $C^2$ is not really problematic: indeed, since it is a deterministic smooth function outside of a set of zero Lebesgue measure (and extending as a piecewise smooth function at the singularities), the singularities are not seen by the imaginary GMC. To see this, observe that $ \int_\Sigma f(x)M^g_\beta(X_g+I^{\boldsymbol{\sigma}}_{x_0}(\omega_{\bf k})+ I^{  {\boldsymbol{\xi}}}_{x_0}( \nu^{\rm h}_{\mathbf{z},\mathbf{m}}),\dd x) $ can be obtained as $L^2$ limit of regularized approximations for each smooth $f$. Now we claim:
\begin{lemma}\label{expmomentunif}
The random variable  $ M^g_\beta(X_g+I^{\boldsymbol{\sigma}}_{x_0}(\omega_{\bf k})+ I^{  {\boldsymbol{\xi}}}_{x_0}( \nu^{\rm h}_{\mathbf{z},\mathbf{m}}),\dd x) $  is a random distribution (in the sense of Schwartz) of order 2 on $\Sigma$ and there exists some $L^2$ random variable $D_\Sigma$ such that
$$\forall f\in C^\infty(\Sigma),\qquad \big|\int_\Sigma f(x)  M^g_\beta(X_g+I^{\boldsymbol{\sigma}}_{x_0}(\omega_{\bf k})+ I^{  {\boldsymbol{\xi}}}_{x_0}( \nu^{\rm h}_{\mathbf{z},\mathbf{m}}),\dd x) \big|\leq D_\Sigma({\bf k}) (\|f\|_\infty+\|\Delta_gf\|_\infty).$$
\end{lemma}
\begin{proof}
Notice that all $f\in C^\infty(\Sigma)$ can be written as $f(x)=m_g(f)+\int G_g(x,y)\Delta_gf(y)\dd {\rm v}_g( y)$, where $m_g(f)=\frac{1}{{\rm v}_g(\Sigma)}\int f(y)\dd {\rm v}_g(y)$. Then for all fixed $f$ we have

\begin{align*}
  \Big|\int_\Sigma  f(x)  M^g_\beta(X_g+I^{\boldsymbol{\sigma}}_{x_0}(\omega_{\bf k})+& I^{  {\boldsymbol{\xi}}}_{x_0}( \nu^{\rm h}_{\mathbf{z},\mathbf{m}}),\dd x)\Big| \\
=&\Big|\int_\Sigma \big(m_g(f)+\int_\Sigma G_g(x,y)\Delta_gf(y)\dd {\rm v}_g(y)\big)  M^g_\beta(X_g+I^{\boldsymbol{\sigma}}_{x_0}(\omega_{\bf k})+ I^{  {\boldsymbol{\xi}}}_{x_0}( \nu^{\rm h}_{\mathbf{z},\mathbf{m}}),\dd x)\Big|
\\
  \leq &
 |m_g(f)||M^g_\beta(X_g+I^{\boldsymbol{\sigma}}_{x_0}(\omega_{\bf k})+ I^{  {\boldsymbol{\xi}}}_{x_0}( \nu^{\rm h}_{\mathbf{z},\mathbf{m}}),\Sigma)|
 \\
 &+ \int \Big| \Delta_gf(y) \Big(\int G_g(x,y)  M^g_\beta(X_g+I^{\boldsymbol{\sigma}}_{x_0}(\omega_{\bf k})+ I^{  {\boldsymbol{\xi}}}_{x_0}( \nu^{\rm h}_{\mathbf{z},\mathbf{m}}),\dd x)\Big)\Big|\dd {\rm v}_g(y)
 \\
    \leq & ( \|f\|_\infty+\|\Delta_gf\|_\infty)\Big(|M^g_\beta(X_g+I^{\boldsymbol{\sigma}}_{x_0}(\omega_{\bf k})+ I^{  {\boldsymbol{\xi}}}_{x_0}( \nu^{\rm h}_{\mathbf{z},\mathbf{m}}),\Sigma)|
    \\
    &+\int \Big| \int G_g(x,y)  M^g_\beta(X_g+I^{\boldsymbol{\sigma}}_{x_0}(\omega_{\bf k})+ I^{  {\boldsymbol{\xi}}}_{x_0}( \nu^{\rm h}_{\mathbf{z},\mathbf{m}}),\dd x)\Big|\dd {\rm v}_g(y) \Big).
\end{align*}
Since this bound is valid almost surely for a countable dense family of $C^\infty(\Sigma)$ equipped with the norm $|f|_\infty+|\Delta_gf|_\infty$, we deduce that $ M^g_\beta(X_g+I^{\boldsymbol{\sigma}}_{x_0}(\omega_{\bf k})+ I^{  {\boldsymbol{\xi}}}_{x_0}( \nu^{\rm h}_{\mathbf{z},\mathbf{m}}),\dd x) $  is a distribution of order 2 almost surely. The random variable in the right-hand side above is our $D_\Sigma({\bf k})$.  One can easily check that it is a $L^2$ random variable: this amounts to computing the following integral (in local coordinates, and using   that  $e^{i\beta(I^{\boldsymbol{\sigma}}_{x_0}(\omega_{\bf k})+ I^{  {\boldsymbol{\xi}}}_{x_0}( \nu^{\rm h}_{\mathbf{z},\mathbf{m}}))}$ are bounded)   
\[u^2(y):=\iint_{D^2}\ln\frac{1}{|x-y|}\ln\frac{1}{|x'-y|} |x-x'|^{-\beta^2}\dd x \dd x'<+\infty.\]
This is not only obviously true since $\beta^2<2$, but also this quantity is bounded uniformly in $y$ over compact subsets.  
\end{proof}

 Next we claim that, for $\delta>0$ fixed, the expectation
\begin{equation}\label{int}
 e^{i\sum_jw_jc}\E\Big[e^{-\frac{1}{2\pi}\langle \dd X_g+\dd u_{0,\epsilon},\omega_{\bf k} \rangle_2}F(\phi_g+u_{0,\epsilon} )e^{-\frac{i   Q}{4\pi}\int^{\rm reg}_\Sigma K_g(\phi_g+u_{0,\epsilon})\,\dd v_g -\mu  M^g_\beta(\phi_g+u_{0,\epsilon,\delta} ,\Sigma)}\Big] 
\end{equation}
is holomorphic in $\mathbf{w}\in \C^{n_{\mathfrak{e}}}$. We have to check that, for any compact subset $K\subset  \C^{n_{\mathfrak{e}}}$,
$$\sup_{\mathbf{w}\in K}\E\Big[\Big|e^{-\frac{1}{2\pi}\langle \dd X_g+\dd u_{0,\epsilon},\omega_{\bf k} \rangle_2}F(\phi_g+u_{0,\epsilon} )e^{-\frac{i   Q}{4\pi}\int^{\rm reg}_\Sigma K_g(\phi_g+u_{0,\epsilon})\,\dd v_g -\mu  M^g_\beta(\phi_g+u_{0,\epsilon,\delta} ,\Sigma)}\Big|\Big] <\infty.$$
Up to using Holder inequality, this amounts to proving that
$$\sup_{\mathbf{w}\in K}\E\Big[\Big| e^{  -\mu  M^g_\beta(\phi_g+u_{0,\epsilon,\delta} ,\Sigma)}\Big|\Big] <\infty.$$
This follows from Proposition \ref{expmoment}.

 Hence our claim for the holomorphicity of \eqref{int}. Now we claim that the integral \eqref{int} converges locally uniformly with respect to $\mathbf{w}$ towards the same expression with $\delta=0$. To see this, it is enough to observe that, locally uniformly in $\mathbf{w}$,
\begin{equation}\label{int2}
\E\Big[ \Big|e^{\mu  (M^g_\beta(\phi_g+u_{0,\epsilon,\delta} ,\Sigma)-M^g_\beta(\phi_g+u_{0,\epsilon} ,\Sigma))}-1\Big|^2\Big]\to 0,\quad \text{ as }\delta\to 0. 
\end{equation}
Indeed, from Proposition \ref{expmoment}, we have for each $\alpha\in\R$
\begin{align*}
\E\Big[  &e^{\alpha  |M^g_\beta(\phi_g+u_{0,\epsilon,\delta} ,\Sigma)-M^g_\beta(\phi_g+u_{0,\epsilon} ,\Sigma)|} \Big]\\
\leq &
e^{C\alpha  \|e^{i\beta u_{0,\epsilon,\delta}}-e^{i\beta u_{0,\epsilon}}\|_\infty}(1+C\alpha \|e^{i\beta u_{0,\epsilon,\delta}}-e^{i\beta u_{0,\epsilon}}\|_\infty e^{C\alpha^2  \|e^{i\beta u_{0,\epsilon,\delta}}-e^{i\beta u_{0,\epsilon}}\|_\infty^2})
\end{align*}
and the latter estimate goes to $1$ as $\delta\to 0$ locally uniformly in $\mathbf{w}$. The claim \eqref{int2} follows. 
 In conclusion the right hand side of \eqref{mainquant} defines a holomorphic quantity of $\mathbf{w}\in\C^{n_{\mathfrak{e}}}$. So does the left hand side, and both sides coincide on $(i\R)^{n_{\mathfrak{e}}}$, therefore on $\R^{n_{\mathfrak{e}}}$.
 
Next, we want to take $\mathbf{w}=\boldsymbol{\alpha}$  in  \eqref{mainquant}, integrate  over $c$ and ${\bf k}$, and then pass to the limit $\epsilon\to 0$ in the right hand side to give sense to the limit of the left hand side. The limit candidate is the same expression with $\epsilon=0$. The main issue is to make sense of the limit of the potential $ M^g_\beta(\phi_g+u_{0,\epsilon} ,\Sigma)$. In the limit, the contribution from $u_{0,\epsilon}$ will create a singularity in the surface $\Sigma$ and we have to show that we can integrate $M^g_\beta$ against those singularities. Actually, it is not clear that we can make sense of $\int_\Sigma e^{-\beta\sum_j\alpha_jG_g(x,x_j)} M^g_\beta(\phi_g,\dd x)$ almost surely for all possible values of the $x_j$'s, as we have an understanding of $M^g_\beta$ only as a distribution of order $2$. Yet, since we fix $x_1,\dots,x_{n_{\mathfrak{e}}}$, we can still make sense of this quantity on average. Indeed, under the condition \eqref{seiberg}, it is plain to see that the family $(M^g_\beta(\phi_g+u_{0,\epsilon} ,\Sigma))_\epsilon$ is Cauchy in $L^2$ and converges towards a random variable denoted by $M^g_\beta(\phi_g+u_{0,0} ,\Sigma)$ satisfying
\begin{multline*}
\E[|M^g_\beta(\phi_g+u_{0,0} ,\Sigma)|^2]\\
=\iint_{\Sigma^2}e^{i\beta (u_{0,0}(x)-  u_{0,0}(y))+\beta^2G_g(x,y)-\frac{\beta^2}{2}(W_g(x)+W_g(y))}e^{i\beta (I^{\boldsymbol{\sigma}}_{x_0}(\omega_{\bf k})+  I^{  {\boldsymbol{\xi}}}_{x_0}( \nu^{\rm h}_{\mathbf{z},\mathbf{m}}))(x)-i\beta (I^{\boldsymbol{\sigma}}_{x_0}(\omega_{\bf k})+ I^{  {\boldsymbol{\xi}}}_{x_0}( \nu^{\rm h}_{\mathbf{z},\mathbf{m}}))(y)}\dd {\rm v}_g(x)\dd {\rm v}_g(y).
\end{multline*}
The control of this integral uses the elementary computation  that the integral
\begin{equation}\label{elem}
\int_{|x|,|y| \leq 1}  \frac{|x|^{\beta\alpha}|y|^{\beta\alpha }\dd x\dd y}{|x-y|^{\beta^2}},
 \end{equation}
  is finite provided that $\beta\alpha>-2 + \frac{\beta^2}{2}$, i.e. $\alpha>  Q$. Furthermore, using Fatou's lemma in Proposition \ref{expmoment} gives the estimate for $\alpha\in\R$
  \begin{equation}
\E\Big[\exp\Big(\alpha |M^g_\beta(\phi_g+u_{0,0} ,\Sigma)|\Big)\Big]\leq   e^{C\alpha v}\big(1+C\alpha u\exp(C\alpha^2u^2)\big)
\end{equation}
with
$$u^2:=\iint_{\Sigma^2}e^{i\beta (u_{0,0}(x)+  u_{0,0}(y))+\beta^2G_g(x,y) }{\rm v}_g(\dd x){\rm v}_g(\dd y),\quad \text{ and }\quad v:=\int_{\Sigma}e^{i\beta u_{0,0}(x) }{\rm v}_g(\dd x).$$
Let us write $\T_R:=\R/2\pi R\Z$. Using Holder, we can then bound the difference between regularized amplitudes and  their candidate for the limit, call $\Delta_\epsilon$ this difference,  
\begin{align*}
|\Delta_\epsilon|\leq  &C\sum_{{\bf k}\in \Z^{2\mathfrak{g}}} e^{-\frac{1}{4\pi}\|\omega _{\bf k}\|_2^2-\frac{1}{4\pi}\|\nu^{\rm h}_{\mathbf{z},\mathbf{m}}\|^2_{g,0}-\frac{1}{2\pi}\langle \omega_{\bf k},\nu^{\rm h}_{\mathbf{z},\mathbf{m}}\rangle_2} \big(R_{\epsilon}^1+R_{\epsilon}^2+R_{\epsilon}^3\big)
\end{align*}
with
\begin{align*}
R_{\epsilon}^1:=&e^{-\frac{1}{2\pi}\langle  \dd u_{0,\epsilon},\omega_{\bf k}\rangle_2} \Big(\int_{\T_R}\E\Big[e^{-\frac{1}{2\pi}\langle \dd X_g ,\omega_{\bf k} \rangle_2 }|F(\phi_g+u_{0,\epsilon})-F(\phi_g+u_{0,0})|^p\Big]\,\dd c\Big)^{1/p}
\\
& \Big(\int_{\T_R}\E\Big[e^{-\frac{1}{2\pi}\langle \dd X_g ,\omega_{\bf k} \rangle_2 }|e^{-\mu  M^g_\beta(\phi_g+u_{0,\epsilon},\Sigma)}|^q\Big]\,\dd c\Big)^{1/q}\\
R_{\epsilon}^2:=&e^{-\frac{1}{2\pi}\langle  \dd u_{0,0},\omega_{\bf k}\rangle_2} \Big(\int_{\T_R}\E\Big[e^{-\frac{1}{2\pi}\langle \dd X_g ,\omega_{\bf k} \rangle_2 }|F(\phi_g+u_{0,
0}) |^p\Big]\,\dd c\Big)^{1/p}
\\
& \Big(\int_{\T_R}\E\Big[e^{-\frac{1}{2\pi}\langle \dd X_g ,\omega_{\bf k} \rangle_2 }|e^{-\mu  M^g_\beta(\phi_g+u_{0,\epsilon},\Sigma)}-e^{-\mu  M^g_\beta(\phi_g+u_{0,0},\Sigma)}|^q\Big]\,\dd c\Big)^{1/q}\Big)
\\
R_{\epsilon}^3:=&
\big|e^{-\frac{1}{2\pi}\langle  \dd u_{0,\epsilon},\omega_{\bf k}\rangle_2}  - e^{-\frac{1}{2\pi}\langle  \dd u_{0,0},\omega_{\bf k}\rangle_2}\big|
\Big(\int_{\T_R}\E\Big[e^{-\frac{1}{2\pi}\langle \dd X_g ,\omega_{\bf k} \rangle_2 }|F(\phi_g+u_{0,\epsilon})-F(\phi_g+u_{0,0})|^p\Big]\,\dd c\Big)^{1/p}
\\
& \Big(\int_{\T_R}\E\Big[e^{-\frac{1}{2\pi}\langle \dd X_g ,\omega_{\bf k} \rangle_2 }|e^{-\mu  M^g_\beta(\phi_g+u_{0,\epsilon},\Sigma)}|^q\Big]\,\dd c\Big)^{1/q}\Big)
\end{align*}
In the term $R_{\epsilon}^1$ and $R_{\epsilon}^2$ above, there are two trivial terms $e^{-\frac{1}{2\pi}\langle  \dd u_{0,\epsilon},\omega_{\bf k}\rangle_2}  $ and $e^{-\frac{1}{2\pi}\langle  \dd u_{0,0},\omega_{\bf k}\rangle_2} $, which we bound by $Ce^{C|{\bf k}|}$ for some constant $C>0$ uniformly in $\epsilon$. In $R^3_{\epsilon}$, the difference $\big|e^{-\frac{1}{2\pi}\langle  \dd u_{0,\epsilon},\omega_{\bf k}\rangle_2}  - e^{-\frac{1}{2\pi}\langle  \dd u_{0,0},\omega_{\bf k}\rangle_2}\big|$ is bounded by $Ce^{C|{\bf k}|}(e^{C|{{\bf k}}|o(1)}-1)$ (using Landau notation as $\epsilon\to 0$).

Next, after using the Girsanov transform in $R_{\epsilon}^1$, the first  integral term  is bounded by $e^{\frac{1}{4\pi p}\|\dd f_{\bf k}\|_2^2}\|F(\cdot+u_{0,\epsilon}-u_{0,0}) -F(\cdot)\|_{\mc{L}^{\infty,p}_{\rm e,m}}$ (recall that $\dd f_{\bf k}=(1-\Pi_1)\omega_{\bf k}$) by definition of $\|F(\cdot)\|_{\mc{L}^{\infty,p}_{\rm e,m}}$. It is straightforward to check   that $\|F(\cdot+u_{0,\epsilon}-u_{0,0}) -F(\cdot)\|_{\mc{L}^{\infty,p}_{\rm e,m}}\to 0$  as $\epsilon\to 0$ for $F\in \mc{E}^{\rm e,m}_R(\Sigma)$: indeed,  this follows from the fact that $u_{0,\epsilon}\to u_{0,0}$ in $H^{s}(\Sigma,g)$ for $s<1$ and from the fact that $F(f)$ depends on $f$ in terms of a polynomial in the variables $( f,g_1)_s,\dots, ( f,g_n)_s$ for some functions $g_1,\dots, g_n\in H^{-s}(\Sigma)$. The second integral in $R_{\epsilon}^1$ is bounded by $Ce^{\frac{1}{4\pi q}\|\dd f_{\bf k}\|_2^2}$ for some universal constant $C$  as a result of the Girsanov transform and Proposition \ref{expmoment}.  

Concerning $R_{\epsilon}^2$, the first integral is bounded by $e^{\frac{1}{4\pi p}\|\dd f_{\bf k}\|_2^2}\|F(\cdot)\|_{\mc{L}^{\infty,p}_{\rm e,m}}$, similarly to the first integral in $R_{\epsilon}^1$.
The main problem lies in evaluating the last integral. First, using the Girsanov transform in the first line and then Holder inequality for conjugate exponents $p_1,p_2$, we bound 
\begin{align*}
\sup_{c} \E\Big[&e^{-\frac{1}{2\pi}\langle \dd X_g ,\omega_{\bf k} \rangle_2 }|e^{-\mu  M^g_\beta(\phi_g+u_{0,\epsilon},\Sigma)}-e^{-\mu  M^g_\beta(\phi_g+u_{0,0},\Sigma)}|^q\Big]
 \\
 \leq & \sup_{c}e^{ \frac{1}{4\pi}\|\dd f_{\bf k}\|^2_2 }\E\Big[  |e^{-\mu  M^g_\beta(\phi_g+u_{0,\epsilon}+f_{\bf k},\Sigma)}-e^{-\mu  M^g_\beta(\phi_g+u_{0,0}+f_{\bf k},\Sigma)}|^q\Big]
  \\
 \leq & \sup_{c,{}}e^{ \frac{1}{4\pi}\|\dd f_{\bf k}\|^2_2 }\E\Big[  |e^{-\mu  M^g_\beta(\phi_g +u_{0,0}+f_{\bf k},\Sigma)}|^{qp_1}\Big]^{1/p_1}\E\Big[  |e^{-\mu  \big(M^g_\beta(\phi_g+u_{0,\epsilon}+f_{\bf k},\Sigma)-  M^g_\beta(\phi_g+u_{0,0}+f_{\bf k},\Sigma)\big)}-1|^{qp_2}\Big]^{1/p_2}.
\end{align*}
The first expectation above is bounded by constant (independent of $c,{\bf k}$) by Proposition \ref{expmoment} and as shown above. Now we focus on the second. Recall first the trivial inequality $|e^z-1|\leq e^{|z|}-1$ for $z\in\C$, and then $(e^u-1)^q\leq C(e^{qu}-1)$ for $u\in\R_+$ and $q>1$. Therefore the second expectation is bounded by 
$$C\E\Big[  e^{|\mu| qp_2 \big|M^g_\beta(\phi_g+u_{0,\epsilon}+f_{\bf k},\Sigma)-  M^g_\beta(\phi_g+u_{0,0}+f_{\bf k},\Sigma)\big|}-1 \Big]^{1/p_2},$$
which is bounded by (using Prop \ref{expmoment})
$$\Big( e^{C\alpha v}\big(1+C\alpha u\exp(C\alpha^2u^2)\big)-1\Big)^{1/p_2}
$$
with $\alpha=|\mu|qp_2$, $v:=\int_{\Sigma}|e^{i\beta u_{0,\epsilon}(x) }-e^{i\beta u_{0,0}(x) }|{\rm v}_g(\dd x)$ and
$$u^2:=\iint_{\Sigma^2}|e^{i\beta u_{0,\epsilon}(x) }-e^{i\beta u_{0,0}(x) }||e^{i\beta u_{0,\epsilon}(y) }-e^{i\beta u_{0,0}(y) }|e^{\beta^2G_g(x,y) }{\rm v}_g(\dd x){\rm v}_g(\dd y).$$
This quantity goes to $0$ as $\epsilon\to 0$ as a simple consequence of Lebesgue dominated convergence (recall $\beta^2<2$). 

For $R_{\epsilon}^3$, the two integral terms are bounded by $C  e^{\frac{1}{4\pi }\|\dd f_{\bf k}\|_2^2}\|F(\cdot)\|_{\mc{L}^{\infty,p}_{\rm e,m}}$, as above. Overall, we have the bound $R_{\epsilon}^3\leq C e^{\frac{1}{4\pi }\|\dd f_{\bf k}\|_2^2}\|F(\cdot)\|_{\mc{L}^{\infty,p}_{\rm e,m}} e^{C|{\bf k}|}(e^{C|{{\bf k}}|o(1)}-1)$.

Gathering these estimates, we deduce  (using the estimate $e^{-\frac{1}{2\pi}\langle \omega_{\bf k},\nu^{\rm h}_{\mathbf{z},\mathbf{m}}\rangle_2}\leq e^{C |{\bf k}|}$ for some $C>0$) 
\begin{align*}
|\Delta_\epsilon|\leq  &C\sum_{{\bf k}\in \Z^{2\mathfrak{g}}} e^{-\frac{1}{4\pi}\|\Pi_1\omega _{\bf k}\|_2^2 +C|{\bf k}|}\Big(\|F(\cdot+u_{0,\epsilon}-u_{0,0}) -F(\cdot)\|_{\mc{L}^{\infty,p}_{\rm e,m}}+(C_\epsilon+e^{C_\epsilon|{\bf k}|}-1)\|F(\cdot)\|_{\mc{L}^{\infty,p}_{\rm e,m}} \Big)
\end{align*}
for some constant $C_\epsilon$ such that $\lim_{\epsilon\to 0}C_\epsilon=0$. Therefore, up to the multiplicative factor $\big(\frac{{\rm v}_{g}(\Sigma)}{{\det}'(\Delta_{g})}\big)^\hf$ that is harmless, the regularized correlation functions in the right hand side of \eqref{defcorrelg} converge   as $\epsilon\to 0$ towards 
\begin{align}
&e^{-\frac{1}{2} \sum_j\alpha_j^2W_{g}(x_j)-\sum_{j<j'}\alpha_j\alpha_{j'}G_g(x_j,x_{j'})}  \sum_{{\bf k}\in \Z^{2\mathfrak{g}}}e^{-\frac{1}{4\pi}(\|\omega _{\bf k}\|_2^2+\|\nu^{\rm h}_{\mathbf{z},\mathbf{m}}\|^2_{g,0}+2\langle \omega_{\bf k},\nu^{\rm h}_{\mathbf{z},\mathbf{m}}\rangle_2)}\prod_{j}e^{i\alpha_j I_{x_0}(\omega_{\bf k}+\nu^{\rm h}_{\mathbf{z,m}})(x_j)}
\label{defelec}
\\
& \times \int_{\T_R}e^{i\sum_j\alpha_jc} \E\Big[e^{-\frac{1}{2\pi}\langle \dd X_g+\dd u_{0,0},\omega_{\bf k} \rangle_2}F(\phi_g+u_{0,0})e^{(-\frac{i   Q}{4\pi}\int_\Sigma^{\rm reg} K_g(\phi_g+u_{0,0})\,\dd v_g -\mu  M^g_\beta(\phi_g+u_{0,0},\Sigma))}\Big]\,\dd c,\nonumber
\end{align}
expression that we take as a definition of  $\cjg  F V_{(\boldsymbol{\alpha},0)}^{g}(\mathbf{x})V^g_{(0,\mathbf{m})}({\bf v})  \cjd_ {\Sigma, g}$. This expression extends  to functionals $F\in \mc{L}^{\infty,p}_{\rm e,m}(H^s(\Sigma))$ for $s<-1$.\\

(2) \textbf{Conformal anomaly.} Next, we prove the conformal anomaly. The argument is similar to \cite[Prop. 4.4]{GRV}, so we sketch the proof up to the crucial argument, following \cite[Prop. 4.4]{GRV}.
But first of all, and in order to simplify the proof, let us recall that the path integral is invariant under change of cohomology basis. It will then be convenient to choose a basis of harmonic 1-forms, hence the $\omega_{\bf k}$'s are harmonic in the following.
Let $g'=e^{\rho}g$ conformal to $g$. It suffices to consider the case $F\in  \mc{E}^{\rm m}_R(\Sigma)$. We recall the equality in law $X_{g'}=X_g-m_{g'}(X_g)$ with $m_{g'}(f):=\frac{1}{{\rm v}_{g'}(\Sigma)}\int_\Sigma f \dd {\rm v}_g$ (see \cite[Lemma 3.1]{GRV}). Using invariance under translations of the Lebesgue measure on the circle, we deduce
\begin{align*}
\cjg  F  V_{(\boldsymbol{\alpha},0)}^{g'}(\mathbf{x})V^{g'}_{(0,\mathbf{m})}({\bf v})  \cjd_ {\Sigma, g'}=& \big(\frac{{\rm v}_{g'}(\Sigma)}{{\det}'(\Delta_{g'})}\big)^\hf\sum_{{\bf k}\in \Z^{2\mathfrak{g}}}e^{-\frac{1}{4\pi}\|\omega _{\bf k}\|_2^2-\frac{1}{4\pi}\|\nu^{\rm h}_{\mathbf{z},\mathbf{m}}\|^2_{g',0}-\frac{1}{2\pi}\langle \omega_{\bf k},\nu^{\rm h}_{\mathbf{z},\mathbf{m}}\rangle_2}\\ 
 &\times 
 \int_{\T_R}\E\Big[V_{(\boldsymbol{\alpha},0)}^{g' }(\phi_g,\mathbf{x}) F(\phi_g)e^{-\frac{i   Q}{4\pi}\int_\Sigma^{\rm reg} K_{g'}\phi_g\,\dd v_{g'} -\mu  M^{g'}_\beta(\phi_g,\Sigma)}\Big]\,\dd c.
\end{align*}
The point, in the expression above, is that we are integrating the Liouville field $\phi_g=c+X_g+I_{x_0}(\omega_{\bf k}+\nu^{\rm h}_{\mathbf{z,m}})$ in the path integral regularized in the metric $g'$. So we have to remove every $g'$-dependency. We treat first the curvature term. For this we need to use Lemma \ref{change_conf_magnetic}.
Using this Lemma, the relations \eqref{relationentrenorm} and $K_{g'}=e^{-\rho}(K_g+\Delta_{g}\rho)$ and Lemma \ref{conformal_change_reg_int} (and note that $\langle \dd \rho,\omega_{\bf k}\rangle_2=0$ because $\omega_{\bf k}$ is harmonic), we deduce
\begin{align*}
\cjg  F V_{(\boldsymbol{\alpha},0)}^{g'}(\mathbf{x})V^{g'}_{(0,\mathbf{m})}({\bf v})  \cjd_ {\Sigma, g'}=   
&\big(\frac{{\rm v}_{g'}(\Sigma)}{{\det}'(\Delta_{g'})}\big)^\hf\sum_{{\bf k}\in \Z^{2\mathfrak{g}}}e^{-\frac{1}{4\pi}\|\omega _{\bf k}\|_2^2-\frac{1}{4\pi}\|\nu^{\rm h}_{\mathbf{z},\mathbf{m}}\|^2_{g',0}-\frac{1}{2\pi}\langle \omega_{\bf k},\nu^{\rm h}_{\mathbf{z},\mathbf{m}}\rangle_2}
 \\ 
 & \int_{\T_R}\E\Big[e^{-\frac{i Q}{4\pi}\int_\Sigma \Delta_g\rho X_g\,\dd {\rm v}_g}V_{(\boldsymbol{\alpha},0)}^{g' }(\phi_g,\mathbf{x}) F(\phi_g)e^{-\frac{i   Q}{4\pi}\int_\Sigma^{\rm reg} K_{g}\phi_g  \,\dd v_{g} -\mu  M^{g}_\beta(\phi_g+i\tfrac{  Q}{2}\rho,\Sigma)}\Big]\,\dd c.
\end{align*}
The same argument of analytic continuation as before allows us to use the (imaginary) Cameron-Martin theorem with the term $e^{-\frac{i Q}{4\pi}\int_\Sigma \Delta_g\rho X_g\,\dd {\rm v}_g}$: the field $X_g$ in the above expression is then replaced by $X_g-\tfrac{i  Q}{2}(\rho-m_g(\rho))$ and the variance of this transform is 
$$\frac{  Q^2}{16\pi^2}\iint_{\Sigma^2} \Delta_g\rho (x) G_g(x,x') \Delta_g\rho(x') \dd {\rm v}_g(x) \dd {\rm v}_g(x')=\frac{ Q^2}{8\pi}\int_\Sigma|d\rho|_g^2{\rm dv}_g.$$
Therefore, using \eqref{detpolyakov} and Lemma \ref{renorm_L^2} to transform both the det term and the regularized norm, we deduce
\begin{align*}
&\cjg  F  V_{(\boldsymbol{\alpha},0)}^{g'}(\mathbf{x})V^{g'}_{(0,\mathbf{m})}({\bf v})  \cjd_ {\Sigma, g'}=\\    
&e^{\frac{1-6  Q^2}{96\pi}\int_{\Sigma}(|d\rho|_g^2+2K_g\rho) {\rm dv}_g+\sum_j\tfrac{Q\alpha_j}{2}\rho(x_j)-\sum_j  \frac{R^2 m_j^2}{4}\rho(z_j)} \big(\frac{{\rm v}_{g}(\Sigma)}{{\det}'(\Delta_{g})}\big)^\hf\sum_{{\bf k}\in \Z^{2\mathfrak{g}}}e^{-\frac{1}{4\pi}\|\omega _{\bf k}\|_2^2-\frac{1}{4\pi}\|\nu^{\rm h}_{\mathbf{z},\mathbf{m}}\|^2_{g,0}-\frac{1}{2\pi}\langle \omega_{\bf k},\nu^{\rm h}_{\mathbf{z},\mathbf{m}}\rangle_2}\\ 
& \int_{\T_R}\E\Big[ V_{(\boldsymbol{\alpha},0)}^{g' }(\phi_g+i\tfrac{Q}{2}m_g(\rho),\mathbf{x}) F(\phi_g-\tfrac{i  Q}{2}(\rho-m_g(\rho)))e^{-\frac{i   Q}{4\pi}\int_\Sigma^{\rm reg} K_{g}(\phi_g+i\frac{Q}{2}m_g(\rho))\,\dd v_{g} -\mu  M^{g}_\beta(\phi_g+i\frac{Q}{2}m_g(\rho),\Sigma)}\Big]\,\dd c.
\end{align*}
Note now that  the vertex operator $V_{\alpha_j,g'}(\phi_g,x_j) $ is not regularized in the metric $g$, but $g'$ instead. Repeating the argument for the construction of the correlation functions before,  we see that this only affects the variance in the Cameron-Martin theorem. Otherwise stated, a straightforward consequence of  \eqref{varYg} is the relation
\begin{equation}\label{scalingvertex}
V_{\alpha,g'}(\phi_g,x)=  e^{-\frac{\alpha^2}{4} \rho(x)} V_{\alpha,g}(\phi_g,x)
\end{equation}
 when plugging this relation into the expectation. In conclusion, and using Gauss-Bonnet for the constant in the curvature term, we get
 \begin{align}\label{exptodev}
& \cjg  F  V_{(\boldsymbol{\alpha},0)}^{g'}(\mathbf{x})V^{g'}_{(0,\mathbf{m})}({\bf v})  \cjd_ {\Sigma, g'}
=\\
&  e^{\frac{1-6  Q^2}{96\pi}\int_{\Sigma}(|d\rho|_g^2+2K_g\rho) {\rm dv}_g-\sum_j\Delta_{\alpha_j}\rho(x_j)-\sum_j  \frac{R^2 m_j^2}{4}\rho(z_j)}  
\big(\frac{{\rm v}_{g}(\Sigma)}{{\det}'(\Delta_{g})}\big)^\hf\sum_{{\bf k}\in \Z^{2\mathfrak{g}}}e^{-\frac{1}{4\pi}\|\omega _{\bf k}\|_2^2-\frac{1}{4\pi}\|\nu^{\rm h}_{\mathbf{z},\mathbf{m}}\|^2_{g,0}-\frac{1}{2\pi}\langle \omega_{\bf k},\nu^{\rm h}_{\mathbf{z},\mathbf{m}}\rangle_2}
 \nonumber\\ 
 & \int_{\T_R}e^{-\frac{  Q}{2}(\sum_j\alpha_j-\chi(\Sigma)  Q)m_g(\rho)}\E\Big[ V_{(\boldsymbol{\alpha},0)}^{g }(\phi_g,\mathbf{x}) F(\phi_g-\tfrac{i  Q}{2}(\rho-m_g(\rho)))e^{-\frac{i   Q}{4\pi}\int_\Sigma^{\rm reg} K_{g}\phi_g \,\dd v_{g} -\mu e^{-\frac{Q\beta}{2}m_g(\rho)} M^{g}_\beta(\phi_g,\Sigma)}\Big]\,\dd c.\nonumber
\end{align}
In the case of standard Liouville theory, the further terms involving $m_g(\rho)$ are absorbed thanks to invariance of the Lebesgue measure under translations. This argument fails to work here: indeed it would  require  the Lebesgue measure (on the circle) to be invariant under complex shifts $c\to c+i a$ for $a$ real, which of course does not hold. We explain now how it works and, basically, translation invariance of the Lebesgue measure is replaced by a Fourier type argument. The function $F$ is a linear combination of the form \eqref{polytrigo}. So it is enough to consider $F$ of the form $F(c,\phi)=e^{  i n c/R}P_k(\phi)G(e^{i\frac{1}{R}I_{x_0}  }) $ (writing $I_{x_0}$ as a shortcut for $I_{x_0}(\omega_{\bf k}+\nu^{\rm h}_{\mathbf{z,m}})$). Expand now  the term  
$$e^{-\mu e^{-\frac{\beta  Q}{2} m_g(\rho) }M^{g}_\beta(\phi_g,\Sigma)}=\sum_{p=0}^\infty (-1)^p\frac{\mu^p}{p!}e^{ip\beta c}e^{-p\frac{\beta  Q}{2} m_g(\rho) }M^{g}_\beta(X_g+I_{x_0},\Sigma)^p$$
and plug this relation into \eqref{exptodev}. Performing the $c$-integral preserves only at most one term in the summation over $p$, i.e. the term corresponding to  $ \frac{1}{R}n+ \sum_j\alpha_j-\chi(\Sigma)  Q +p\beta=0$, if it exists. As a side remark, notice that this argument also shows that for $F(c,\phi)=e^{  i n c/R}P_k(\phi)G(e^{\frac{i }{R}I_{x_0}})$ we have 
\begin{equation}\label{0integral}
\forall p\in \N_0,\,\,  \frac{1}{R}n+ \sum_j\alpha_j-\chi(\Sigma)  Q +p\beta\not=0 \Longrightarrow \cjg  F  V_{(\boldsymbol{\alpha},0)}^{g'}(\mathbf{x})V^{g}_{(0,\mathbf{m})}({\bf v})  \cjd_ {\Sigma, g}=0.
\end{equation}
For this $p$, the contribution of all the terms involving $m_g(\rho)$ is a multiplicative factor given by
$$\exp\Big(-\frac{  Q}{2}m_g(\rho)\big(\sum_j\alpha_j-\chi(\Sigma)  Q+\frac{1}{R}n +\beta p\big)\Big).$$
But the condition above on $p$ implies that this term equals $1$. Therefore we end up with the final expression
\begin{align*}
 & \cjg  F  V_{(\boldsymbol{\alpha},0)}^{g'}(\mathbf{x})V^{g'}_{(0,\mathbf{m})}({\bf v})  \cjd_ {\Sigma, g'}=\\
&  e^{\frac{1-6  Q^2}{96\pi}\int_{\Sigma}(|d\rho|_g^2+2K_g\rho) {\rm dv}_g-\sum_j\Delta_{\alpha_j}\rho(x_j)-\sum_j  \frac{R^2 m_j^2}{4}\rho(z_j)} \times\big(\frac{{\rm v}_{g}(\Sigma)}{{\det}'(\Delta_{g})}\big)^\hf\sum_{{\bf k}\in \Z^{2\mathfrak{g}}}e^{-\frac{1}{4\pi}\|\omega _{\bf k}\|_2^2-\frac{1}{4\pi}\|\nu^{\rm h}_{\mathbf{z},\mathbf{m}}\|^2_{g,0}-\frac{1}{2\pi}\langle \omega_{\bf k},\nu^{\rm h}_{\mathbf{z},\mathbf{m}}\rangle_2}
 \nonumber\\ 
 & \times \int_{\T_R} \E\Big[ V_{(\boldsymbol{\alpha},0)}^{g }(\phi_g,\mathbf{x}) F(\phi_g-\tfrac{i  Q}{2}\rho )e^{-\frac{i   Q}{4\pi}\int_\Sigma^{\rm reg} K_{g}\phi_g \,\dd v_{g} -\mu M^{g}_\beta(\phi_g,\Sigma)}\Big]\,\dd c,\nonumber
\end{align*}
as claimed.  

(3) \textbf{Diffeomorphism invariance.} We turn now to the diffeomorphism invariance. We consider \eqref{defelec} in the metric $\psi^*g$ and we want to reformulate it in the metric $g$. For this, several observations are needed. First, as orientation preserving diffeomorphism preserve canonical basis, the natural choice of homology basis for \eqref{defelec} in the   metric $\psi^*g$ is $\psi^*\sigma$, with dual basis $\psi^*\omega_1,\dots,\psi^*\omega_{2\mathfrak{g}}$. Then $I_{x_0}^{\psi^*\sigma}(\psi^*\omega_{\bf k})=I_{\psi(x_0)}^{\sigma}(\omega_{\bf k})\circ\psi$. Similarly for the magnetic operators, the defect graph $\mc{D}_{\mathbf{v},\boldsymbol{\xi}}$ is mapped by $\psi$ to $\mc{D}_{\psi_*\mathbf{v},\psi\circ\boldsymbol{\xi}}$. Thus we deduce that 
$I_{x_0}^{\boldsymbol{\xi}}(\psi^*\nu_{\bf z,m}^{\rm h})=I_{\psi(x_0)}^{\psi\circ\boldsymbol{\xi}}(\nu_{\bf z,m}^{\rm h})\circ\psi$. 
Then we note  the standard relations
$$G_{\psi^*g}(x,y)=G_g(\psi(x),\psi(y)),\quad K_{\psi^*g}(x)=K_g(\psi(x)),\quad X_{\psi^*g} \stackrel{law}{=}X_g\circ\psi.$$
In particular $W_{\psi^*g}(x)=W_g(\psi(x))$ and, combining with the relations just above for the primitives  $I_{x_0}^{\psi^*\sigma}(\psi^*\omega_{\bf k})$ and $I_{x_0}^{\boldsymbol{\xi}}(\psi^*\nu_{\bf z,m}^{\rm h})$, we also obtain $M^{\psi^*g}_\beta(\phi_{\psi^*g}+u^{\psi^*g,\mathbf{x}}_{0,0},\Sigma)=M^{\psi^*g}_\beta(\phi_g+u^{g,\psi(\mathbf{x})}_{0,0},\Sigma)$, where we have made explicit the dependence on $g,\mathbf{x}$ of $u_{0,0}$ in the notations. Combining again with Lemma \ref{lemcurvdiff} for the curvature term, we get the result.

(4) \textbf{Spins.}  The spin property results from Corollary \ref{corospin} since the regularized electric operators are in  $\mc{E}^{\rm m}_R(\Sigma)$.
\end{proof}

\subsubsection{Electro-magnetic operators}
We complete this section with the operators that will be of utmost importance to describe the spectrum of this path integral: the electro-magnetic operators.  Basically they are obtained by merging the positions $\mathbf{x}$ and $\mathbf{z}$  in Proposition \ref{limitcorel}. So, the setup is the same as previously with the further condition that the numbers of electric or magnetic charges are the same, i.e. $n_{\mathfrak{e}}=n_{\mathfrak{m}}$.

The path integral with electro-magnetic operators is defined by the limit
\begin{equation}\label{defcorrelmixed}
\cjg   F  V^g_{(\alpha,\mathbf{m})}({\bf v})  \cjd_ {\Sigma, g}:=\lim_{t\to 1} \: \cjg F V_{(\boldsymbol{\alpha},0)}^{g}(\mathbf{x}(t)) V^g_{(0,\mathbf{m})}({\bf v}) \cjd_ {\Sigma, g}
\end{equation}
for $F \in   \mc{E}^{\rm m}_R(\Sigma)$ (with ${\bf x}={\bf z}$), where ${\bf x}(t)=(x_1(t),\dots,x_{n_{\mathfrak{m}}}(t))$ with
 $t\in [0,1]\mapsto x_j(t)$ being any $C^1$ curve such that $x_j(1)=z_j$ and $\dot{x}_j(1)=v_j$. 
 Indeed, the quantity in the right hand side only has a limit when $x_j\to z_j$ along a fixed direction, because of the winding around the points $\mathbf{z}$. This is why we need to fix a direction $v_j$ when $x_j$ approaches $z_j$.

 \begin{theorem}\label{limitcorelmixed} 
 Under the conditions \eqref{seiberg}, the limit \eqref{defcorrelmixed} exists. Moreover the mapping $F\in  \mc{E}^{\rm m}_R(\Sigma)\mapsto \cjg  F  V^g_{(\boldsymbol{\alpha},\mathbf{m})}({\bf v})  \cjd_ {\Sigma, g}$ satisfies the following properties:
\begin{enumerate}
\item {\bf Existence:} It is well-defined and extends to $F\in \mc{L}^{\infty,p}(H^s(\Sigma))$ for $s<0$. For $F=1$, it  gives the correlation functions  $\cjg    V^g_{(\alpha,\mathbf{m})}({\bf v})  \cjd_ {\Sigma, g}$.
\item {\bf Conformal anomaly:} let $g,g'$ be two conformal metrics on the closed Riemann surface $\Sigma$ with $g'=e^{\rho}g$ for some $\rho\in C^\infty(\Sigma)$. Then we have
\begin{align}\label{confanmixed} 
& \frac{\big\cjg  F   V^{g'}_{(\boldsymbol{\alpha},\mathbf{m})}({\bf v})  \big\cjd_ {\Sigma, g'}}
 {\big\cjg   F(\cdot- \tfrac{i  Q}{2}\rho)  V^g_{(\boldsymbol{\alpha},\mathbf{m})}({\bf v}) \big \cjd_ {\Sigma, g}}
 =  
\exp\Big(\frac{{\bf c}}{96\pi}\int_{\Sigma}(|d\rho|_g^2+2K_g\rho) {\rm dv}_g-\sum_{j=1}^{n_{\mathfrak{e}}}\Delta_{(\alpha_j,m_j)}\rho(z_j)\Big)
\end{align}
where the conformal weights $\Delta_{\alpha,m}$ are given by \eqref{deltaalphadef} and the central charge is ${\bf c}:=1-6 Q^2 $.
\item {\bf Diffeomorphism invariance:} let $\psi:\Sigma'\to \Sigma$ be an orientation preserving diffeomorphism. Then  
$$
\big\cjg  F (\phi_{\psi^*g}) V_{(\boldsymbol{\alpha},\mathbf{m})}^{\psi^*g} ({\bf v})  \big\cjd_ {\Sigma', \psi^*g}=\big\cjg  F(\phi_g\circ\psi)  V_{(\boldsymbol{\alpha},\mathbf{m})}^{g}(\psi_*{\bf v})  \big\cjd_ {\Sigma, g}.
$$
\item {\bf Spins:} with $r_{\boldsymbol{\theta}}\mathbf{v}:=(r_{\theta_1}v_1,\dots,r_{\theta_{n_\mathfrak{m}}} v_{\theta_{n_\mathfrak{m}}} )$, then
$$ 
 \langle   F V^g_{(\boldsymbol{\alpha},\mathbf{m})}(r_{\boldsymbol{\theta}}\mathbf{v}) \rangle_{\Sigma,g } =e^{-i QR\langle\mathbf{m},\boldsymbol{\theta}\rangle }  \langle F  V^g_{(\boldsymbol{\alpha},\mathbf{m})}({\bf v}) \rangle_{\Sigma,g } .
 $$
\end{enumerate}
 \end{theorem}

\begin{proof} The proof consists in taking the limit in the expression \eqref{defelec} as $(x_j(t),\dot{x}_j(t))\to (z_j,v_j)$ when $t\to 1$. 
The properties of the path integral then results from taking the limit in the related properties of Proposition \ref{limitcorel}. The crucial argument in the proof is the following: since the 1-form $\nu^{\rm h}_{\mathbf{z},\mathbf{m}}$ is of the form $m_j2\pi R\dd \theta$ in local radial coordinates $z-z_j=re^{i\theta}$ near $z_j$ (see Proposition  \ref{harmpoles}), then the function $e^{\frac{i}{R}I^{\boldsymbol{\xi}}_{x_0}(\nu^{\rm h}_{\mathbf{z},\mathbf{m}})(x)}$ has a limit when $(x_j(t),\dot{x}_j(t))\to (z_j,v_j)$ as $t\to 1$. An immediate consequence is the convergence of all terms of the form $e^{i\alpha_j I^{\boldsymbol{\xi}}_{x_0}(\nu^{\rm h}_{\mathbf{z},\mathbf{m}})(x_j(t))}$ as $(x_j(t),\dot{x}_j(t))\to (z_j,v_j)$. This makes the convergence obvious for all the prefactors in the expression \eqref{defelec}. It then remains to focus on the integral. To get the argument simpler, we can choose the cohomology basis to consist of harmonic 1-forms $\omega_{\bf k}$; in particular the term $e^{-\frac{1}{2\pi}\langle\dd X_g,\omega_{\bf k}\rangle_2}=1$ in  \eqref{defelec}. In \eqref{defelec}, the terms involving $F$ and the curvature depend on $\mathbf{x}$ (recall that this dependence is hidden in $u_{0,0}$) and  converge towards their value at $\mathbf{z}$ (in the direction $\mathbf{v}$) in $L^p$ for the measure $e^{-\frac{1}{4\pi}\|\omega _{\bf k}\|_2^2}\delta_{\bf k}\otimes\P\otimes\dd c$.
 Using H\"older inequality, it remains to investigate the interaction term. Let us write  $ M(\mathbf{x})$ as a shortcut for the random variable $M^g_\beta(\phi_g+u_{0,0},\Sigma)$ at $\mathbf{x}$, and $ M(\mathbf{z})$ for this random variable evaluated at $\mathbf{z}$. Therefore we have to show that 
$$\sup_{{\bf k}\in \Z^{2{\mathfrak{g}}}}\sup_{c\in \R/2\pi R\Z}\E\Big[\Big|e^{-\mu M(\mathbf{x}(t))}-e^{-\mu M(\mathbf{z})}\Big|^q\Big]\to 0$$ as $t\to 1$, for  $q>1$.
 Using H\"older inequality (taking $q$ slightly larger) and Proposition \ref{expmoment}, this amounts to showing that, as $t\to 1$,  
\[\sup_{{\bf k}\in \Z^{2{\mathfrak{g}}}}\sup_{c\in \R/2\pi R\Z}\E\Big[|e^{|\mu| |M(\mathbf{x}(t))- M(\mathbf{z})|}-1|^q \Big]\to 0.\]
Using super-additivity of the mapping $x\mapsto x^q$ (for $q>1$), this amounts to showing that
\[\sup_{{\bf k}\in \Z^{2{\mathfrak{g}}}}\sup_{c\in \R/2\pi R\Z}\E\Big[e^{A |M(\mathbf{x}(t))- M(\mathbf{z})|}-1 \Big]\to 0\]
for any $A>0$. Using Proposition \ref{expmoment}, we see that the above statement follows from the following Lemma  the proof of which is defered to the Appendix:
\begin{lemma}\label{estimeedechien}
We set $u_{\mathbf{x}(t)}(y):=\sum_j\alpha_jG_g(y,x_j(t))+I_{x_0}(\omega_{\bf k}+\nu^{\rm h}_{\mathbf{z},\mathbf{m}})(y) $. Then we have
\[\int_\Sigma\int_\Sigma \big|e^{-\beta u_{\mathbf{x}(t)}(y)}-e^{-\beta u_{\mathbf{z}}(y)}\big| .\big|e^{-\beta u_{\mathbf{x}(t)}(y')}-e^{-\beta u_{\mathbf{z}}(y')} \big| e^{\beta^2G_g(y,y')}\dd {\rm v}_g(y)\dd {\rm v}_g(y')\to 0,\quad \text{when }t\to 1,\]
uniformly in ${\bf k}$.
\end{lemma} 
This ends the proof of Theorem \ref{limitcorelmixed}.
\end{proof}

\section{Three point correlation functions on the Riemann sphere}\label{sec:3point}
 The condition   \eqref{seiberg} guarantees the existence of correlation functions. The condition $\alpha_j\in \frac{1}{R}\Z$ is however slightly misleading, for \eqref{0integral} implies that these correlation functions vanish if $\sum_j\alpha_j-\chi(\Sigma)  Q\in \frac{1}{R}\N$. Therefore the region of interest where the correlations are non-trivial is
 \begin{equation}
\sum_{j=1}^{n_{\mathfrak{m}}}\alpha_j-\chi(\Sigma)  Q\in -\beta \N,
\end{equation}
which, together with the condition $\alpha_j>  Q$ for all $j$, shows that the  non trivial correlation functions on the Riemann sphere have $n_{\mathfrak{m}}\geq 3$ (recall that $\chi(\Sigma)=2$ for the Riemann sphere). We thus recover the standard fact for CFTs that 3 point correlation functions on the Riemann sphere are of special importance. Below, we will relate them to   Dotsenko-Fateev type  integrals.

We identify the Riemann sphere with the extended complex plane $\hat\C$ by stereographic projection. On the sphere, every metric is (up to diffeomorphism) conformal to the special metric $g_0=|z|_+^{-4}|dz|^2$ with $|z|_+=\max(|z|,1)$.  Applying the transformation rules of Proposition \ref{limitcorelmixed} (namely items 3 and 4), one gets, with some straightforward computations,  that the  correlation functions 
are {\it conformally covariant}. More precisely, if $g=e^{\omega(z)}g_0=g(z)|dz|^2$ is a conformal metric and if $z_1, \cdots, z_{n_{\mathfrak{m}}}$ are $n_{\mathfrak{m}}$ distinct points in $\hat \C$, with associated unit tangent vectors $v_1,\dots, v_{n_\mathfrak{m}}$,  then for a M\"obius map $\psi(z)= \frac{az+b}{cz+d}$ (with $a,b,c,d \in \C$ and $ad-bc=1$) 
 \begin{equation}\label{KPZformula}
\langle V^g_{(\boldsymbol{\alpha},\mathbf{m})}(\psi_*{\bf v})     \rangle_{\hat\C,g}=  \prod_{j=1}^{n_{\mathfrak{m}}} \Big(\frac{|\psi'(z_j)|^2g(\psi(z_j))}{g(z_j)}\Big)^{-  \Delta_{(\alpha_j,m_j)}}     \langle V^g_{(\boldsymbol{\alpha},\mathbf{m})}({\bf v})     \rangle_{\hat\C,g}.
\end{equation}

Now we specify to the case $n_{\mathfrak{m}}=3$.
Without loss of generality, we may assume that the magnetic charges satisfy $m_1\leq m_2\leq m_3$. The M\"obius covariance implies in particular that the three point functions ($n_{\mathfrak{m}}=3$) are determined up to a constant denoted $C_{\beta,\mu} (\boldsymbol{\alpha},\boldsymbol{m})$, called the {\it structure constant}:
\begin{equation} \label{3point}
  \langle \prod_{j=1}^3 V^g_{(\alpha_j,m_j)}(z_j,v_j) \rangle_{\hat\C,g}
 =e^{\frac{{\bf c}}{96\pi}\int_{\Sigma}(|\dd\omega|_{g_0}^2+2K_{g_0}\omega) {\rm dv}_{g_0}}
e^{-i QR \sum_{j=1}^{3}m_j{\rm arg}(v_j)}  P_{\boldsymbol{\alpha},\boldsymbol{m}}({\bf z})\prod_{j=1}^3g(z_j)^{-\Delta_{(\alpha_j,m_j)}} C_{\beta,\mu} (\boldsymbol{\alpha},\boldsymbol{m})
 \end{equation}
with 
 \[\begin{split} 
  P_{\boldsymbol{\alpha},\boldsymbol{m}}({\bf z}):= & |z_1-z_3|^{2(\Delta_{(\alpha_2,m_2)}-\Delta_{(\alpha_1,m_1)}-\Delta_{(\alpha_3,m_3)})}|z_2-z_3|^{2(\Delta_{(\alpha_1,m_1)}-\Delta_{(\alpha_2,m_2)}-\Delta_{(\alpha_3,m_3)})}\\
 & \times |z_1-z_2|^{2(\Delta_{(\alpha_3,m_3)}-\Delta_{(\alpha_1,m_1)}-\Delta_{(\alpha_2,m_2)})}.
\end{split}\] 
We will compute the 3 point function in the metric $g_0$ to deduce the structure constants. For this we assume $z_1,z_2,z_3\in\R$ with $z_1<z_2<z_3$ and $v_1=v_2=v_3=1$ (identifying canonically $T_{z_j}\C$ with $\C$ using the global coordinate $z$ on $\C$). In that case, we choose the branch cut to be $[z_1,z_3]$ with defect graph $z_1\to z_2\to z_3$. The primitive is then  $I_{0}( \nu^{\rm h}_{\mathbf{z},\mathbf{m}})(z)=- Rm_1 {\rm arg}\big(\frac{z_2-z}{z_1-z}\big)+ Rm_3 {\rm arg}\big(\frac{z_3-z}{z_2-z}\big)$ (which vanishes on the half-line $]z_3,+\infty[$). By an abuse of notations, we will denote by $I_{0}( \nu^{\rm h}_{\mathbf{z},\mathbf{m}})(z_j):=\lim_{t\to 1}I_{0}( \nu^{\rm h}_{\mathbf{z},\mathbf{m}})(x_j(t))$ where $x_j(t)=z_j-1+t$ is converging to $z_j$ in the direction $v_j$ as $t\to 1$.

Using \eqref{defcorrelg} (recall also \eqref{defPI:mag}) and \eqref{defcorrelmixed} and the expression of the Green function in the metric $g_0$:
\begin{equation}\label{crepe}
G_{g_0}(z,z')=\ln\frac{1}{|z-z'|}+\ln|z|_++\ln|z'|_+, \quad (\text{ thus } W_{g_0}(z)=0)
\end{equation}
we deduce that (here we take the limit $t\to 1$ in the expression \eqref{defcorrelg} for correlation functions as explained in the proof of Proposition \ref{limitcorelmixed})  
\begin{align}\label{corr}
 \langle    \prod_{j=1}^3 V_{(\alpha_j,m_j)}(z_j,v_j)  \rangle_{\hat \C,g_0} = &\big(\frac{{\rm v}_{g_0}(\Sigma)}{{\det}'(\Delta_{g_0})}\big)^\hf 
 \prod_{j < j'}e^{-\alpha_j \alpha_{j'}G_{g_0}(z_j,z_{j'})} e^{-\frac{1}{4\pi}\| \nu^{\rm h}_{\mathbf{z},\mathbf{m}}\|^2_{g_0,0}+i\sum_{j=1}^3\alpha_jI_{0}( \nu^{\rm h}_{\mathbf{z},\mathbf{m}})(z_j)}
 \\
 &\times \int_0^{2\pi R} e^{i(\sum_j\alpha_j-2  Q)c}\E \left [  e^{-\mu e^{i\beta c}  \int_{\C}  F(x,{\bf z})\big(\prod_j|z_j|_+^{-\beta\alpha_j}\big)  M^{g_0}_\beta(\dd x) }  \right ]\dd c\nonumber
\end{align} 
with $\mathbf{z}=(z_1,z_2,z_3)$ and
\begin{equation}\label{coulomb}
F(x,{\bf z})=\prod_{j=1}^3 \left ( \frac{|x-z_j|}{ |x|_+} \right )^{\beta \alpha_j} e^{i\beta I_{0}( \nu^{\rm h}_{\mathbf{z},\mathbf{m}})(x)}.
\end{equation}
Set $s:=\frac{ 2 Q-\sum_j\alpha_j}{\beta}\in \N$. By performing the series expansion of the exponential in the expectation and computing the Fourier coefficients, we get
\begin{align*}
\langle    \prod_{j=1}^3& V_{(\alpha_j,m_j)}(z_j,v_j)   \rangle_{\hat \C,g_0}
\\ 
=&   2\pi R \frac{(-\mu)^s}{s!}\E \left [  \left( \int_{\C}  F(x,{\bf z})  M^{g_0}_\beta(\dd x)  \right )^s  \right ] \prod_{j < j'} |z_j-z_{j'}|^{\alpha_j \alpha_{j'}}\prod_j|z_j|_+^{4\Delta_{\alpha_j}}e^{ -i\pi \alpha_2Rm_1 +i \pi\alpha_3Rm_3 }.
\end{align*}
Note that (if $\beta= \frac{p}{R}$)
$$e^{i\beta I_{0}( \nu^{\rm h}_{\mathbf{z},\mathbf{m}})(x)}= \big(\frac{(z_2-x)(\bar{z}_1-\bar{x})}{(\bar{z}_2-\bar{x})(z_1-x)}\big)^{\frac{-pm_1}{2}}\big(\frac{(z_3-x)(\bar{z}_2-\bar{x})}{(z_2-x)(\bar{z}_3-\bar{x})}\big)^{\frac{pm_3}{2}}.$$

It will be convenient to fix the values of $(x_1,x_2,x_3)$ to be $(0,1,\infty)$ where the related correlation functions are defined by
$$\langle   V_{(\alpha_1,m_1)}(0)V_{(\alpha_2,m_2)}(1)V_{(\alpha_3,m_3)}(\infty)    \rangle_{\hat \C,g_0}=\lim_{|z_3|\to\infty} \langle   V_{(\alpha_1,m_1)}(0)V_{(\alpha_2,m_2)}(1)V_{(\alpha_3,m_3)}(z_3)    \rangle_{\hat \C,g_0},$$
in which case the above relation becomes for $s\in\N$ 
\begin{align*}
\langle  & V_{(\alpha_1,m_1)}(0)V_{(\alpha_2,m_2)}(1)V_{(\alpha_3,m_3)}(\infty)    \rangle_{\hat \C,g_0}\\
= &2\pi R  \frac{(-\mu)^s}{s!}\E \left [  \left( \int_{\C}  \frac{|x|^{\beta\alpha_1}|1-x|^{\beta\alpha_2}}{ |x|_+^{\beta\bar\alpha}}   \big(\frac{ x}{\bar{x}}\big)^{\frac{pm_1}{2}} \big(\frac{1-x}{1-\bar{x}}\big)^{\frac{p(-m_1-m_3)}{2}}    M^g_\beta(\dd x)  \right )^s  \right ]\\
=&2\pi R  \frac{(-\mu)^s}{s!}\E \left [  \left( \int_{\C}  \frac{x^{\Delta_1}\bar{x}^{\bar\Delta_1}(1-x)^{\bar\Delta_2}
(1-\bar x)^{\bar\Delta_2}
}{ |x|_+^{\beta\bar\alpha}}     M^g_\beta(\dd x)  \right )^s  \right ]
\end{align*}
where we have set $\bar\alpha=\sum_j\alpha_j$, $\Delta_1=\beta\alpha_1+\frac{pm_1}{2}$, $\bar\Delta_1=\beta\alpha_1-\frac{pm_1}{2}$, $\Delta_2=\beta\alpha_2+\frac{pm_2}{2}$, and $\bar\Delta_2=\beta\alpha_2-\frac{pm_2}{2}$.
This expression can be expanded as a multiple integral
\begin{align}
&\langle   V_{(\alpha_1,m_1)}(0)V_{(\alpha_2,m_2)}(1)V_{(\alpha_3,m_3)}(\infty)    \rangle_{\hat \C,g_0}\nonumber 
\\
&= 2\pi R  \frac{(-\mu)^s}{s!}\E \left [  \int_{\C^s}\prod_{j=1}^s \frac{x_j^{\Delta_1}\bar{x}_j^{\bar\Delta_1}(1-x_j)^{ \Delta_2}
(1-\bar x_j)^{\bar\Delta_2}}{ |x_j|_+^{\beta\bar\alpha}}    M^g_\beta(\dd x_1)\dots M_\beta^g(\dd x_s)     \right ]\nonumber\\
&=2\pi R  \frac{(-\mu)^s}{s!}\int_{\C^s}\prod_{j=1}^sx_j^{\Delta_1}\bar{x}_j^{\bar\Delta_1}(1-x_j)^{ \Delta_2}
(1-\bar x_j)^{\bar\Delta_2}\prod_{j<j'}|x_j-x_{j'}|^{\beta^2}\dd x_1\dots\dd x_s\nonumber\\
&=: 2\pi R \, \mathcal{I}_\mu(\beta,\alpha_1,\alpha_2,m_1,m_2,s).\label{npointfunction}
\end{align}
In the case when there are no magnetic charges, i.e. $m_1=m_2=0$,  this expression corresponds to the famous Dotsenko-Fateev integral \cite{dotsenko}, whose value is known (in case $m_1=m_2=0$)
\begin{equation}\label{ImDOZZ}
\mathcal{I}_\mu(\beta,\alpha_1,\alpha_2,0,0,s)=   (\frac{-\pi\mu }{l(\frac{\beta^2}{4})})^s\frac{\prod_{j=1}^{n} l(j \frac{\beta^2}{4}) }{ \prod_{j=0}^{s-1} l(- \frac{\beta \alpha_1}{2}  -j\frac{\beta^2}{4} )l(- \frac{\beta \alpha_2}{2}-   j\frac{\beta^2}{4} )l(- \frac{\beta \alpha_3}{2}    -j\frac{\beta^2}{4} )}
\end{equation}
with the convention that $\ell(x)=\Gamma(x)/\Gamma(1-x)$. This expression coincides with the imaginary DOZZ formula (see \cite{Za}), which is actually an analytic extension of this expression to all possible values of $\alpha_1,\alpha_2,\alpha_3$. In the presence of magnetic charges, it is an open question to find an explicit expression for the integral \eqref{npointfunction}.

 \section{Segal's axioms}\label{sec:segal}
In this section, we prove that our path integral satisfies Segal's axioms for CFT. These axioms are based on the notion of amplitude  that are building pieces of the path integral. They can be paired together to reconstruct the partition function or correlation functions, where pairing of amplitudes of surfaces with boundary correspond to amplitudes of glued surfaces along their boundaries. The pairing involves an integration over some functional space, and is associated to a Hilbert space (Section \ref{sub:hilbert}). 

In order to define the amplitudes properly, we shall need the gluing formalism for Riemann surfaces explained in Section \ref{sub:gluing}, and a few facts on Dirichlet-to-Neumann maps recalled below in Section \ref{subDTN}. Next, we give the setup for amplitudes and then study the main properties of the defect graph associated to the electro-magnetic operators for amplitudes in Section \ref{sub:setupamp}. This is all we need to give the definition of amplitudes in Section \ref{sub:defamp}.  Finally we shall prove the Segal axioms in Section \ref{sub:segal}.

\subsection{Hilbert space}\label{sub:hilbert} 
Consider the following real-valued random Fourier series defined on  the unit circle $\T=\{z\in \C\,|\, |z|=1\}\simeq \R/2\pi \Z$ 
\begin{equation}\label{GFFcircle0}
\forall \theta\in\R,\quad \varphi(\theta)=\sum_{n\not=0}\varphi_ne^{in\theta}, \quad \textrm{ with } \varphi_{n}=\frac{1}{2\sqrt{n}}(x_n+iy_n), \,\, \forall n>0
\end{equation}
where $x_n,y_n$ are i.i.d. standard real Gaussians with $0$ mean and variance $1$. The convergence holds in the  Sobolev space  $H^{s}(\T)$ with $s<0$, where $H^s(\T)$ can be identified with the set of sequences $(\varphi_n)_n$ in $\C^\Z$ such that
\begin{equation}\label{outline:ws}
\|\varphi\|_{H^s(\T)}^2:=\sum_{n\in\Z}|\varphi_n|^2(|n|+1)^{2s} <\infty.
\end{equation}
Such a random series arises naturally when considering the restriction of the whole plane GFF to the unit circle. Also, note that the series $ \varphi$ has no constant mode or equivariant part. Both of them will play an important role in what follows and this is why we want to single them out.  In a way similar to Section \ref{UnivCover}, we can define equivariant functions and distributions on $\T$. Consider the space $H^s_{\Z}(\R)$ to be the space of real-valued distributions $u$ on $\R$ such that their restriction on any finite size open interval belongs to $H^s(\R)$ and such that 
\[ \forall n\in \Z, \quad u(\theta+2\pi n)-u(\theta)\in 2\pi R\Z.\]
If one restricts to smooth functions in this space, this amounts precisely to the space of smooth functions $u$ on $\R$ such 
that $e^{\frac{i}{R}u}$ descends to a smooth function on $\T$.
As in Section \ref{UnivCover} we get an identification (with $\pi: \R\to \R/2\pi R\Z$ being the projection)
\[ \Z \times H^s(\T)\mapsto H^s_{\Z}(\R) , \quad (k,\tilde{\varphi})\mapsto \pi^*\tilde{\varphi}(\theta)+ kR\theta \]
and $k$ corresponds to the degree of the map $e^{\frac{i}{R}u}:\T\to \T$. Below, we will write $\tilde{\varphi}$ instead of $\pi^*\tilde{\varphi}$ and implicitely identify $\tilde{\varphi}$ with its periodic lift to $\R$.
Consider the probability space 
\begin{align}\label{omegat}
  \Omega_\T=(\R^{2})^{\N^*}
\end{align}
 equipped with the cylinder sigma-algebra $ \Sigma_\T=\mathcal{B}^{\otimes \N^*}$ ($\mathcal{B}$  stands for the Borel sigma-algebra on $\R^2$) and the   product measure 
 \begin{align}\label{Pdefin}
 \P_\T:=\bigotimes_{n\geq 1}\frac{1}{2\pi}e^{-\frac{1}{2}(x_n^2+y_n^2)}\dd x_n\dd y_n.
\end{align}
The push-forward of $\P_\T$ by the random variable $(x_n,y_n)_{n\in \N}\to \varphi$ defined in \eqref{GFFcircle0}
induces a measure, still denoted $ \P_\T$ by a slight abuse of notation, that is supported on $H^s(\T)$ for any $s<0$ in the sense that $ \P_\T(\varphi\in H^s(\T))=1$.
We next equip the space $\Z$ is with the discrete sigma-algebra and the counting measure $\mu_{\Z}=\sum_{k\in\Z}\delta_k$. 
Then we consider the Hilbert space 
\[\mc{H}:=L^2((\R/2\pi \Z )\times\Z\times  \Omega_{\T},\mu_0)\] 
where the underlying  measure is given by
\[\mu_0:=\dd c\otimes \mu_{\Z} \otimes \P_{\T}\] 
and Hermitian product denoted by $\langle\cdot,\cdot\rangle_{\mc{H}}$.
The random  variable (with $\varphi$ defined in \eqref{GFFcircle0})
\[
(c,k,(x_n,y_n)_{n\in \N})\mapsto  \tilde{\varphi}^{k}(\theta):=c+kR\theta+\varphi(\theta) \in H_{\Z}^{s}(\R)
\] 
induces, by push-forward of $\mu_0$, a measure on $H_{\Z}^{s}(\R)$, that we still denote $\mu_0$. This means that we can also rewrite our Hilbert space as 
\[\mc{H}=L^2(H_{\Z}^{s}(\R),\mu_0).\]

\subsection{Dirichlet-to-Neumann map}\label{subDTN}
Let  $\Sigma$ be a compact Riemann surface   with real analytic boundary $\partial\Sigma=\sqcup_{j=1}^{\mathfrak{b}}\pl_j\Sigma$  consisting of 
 ${\mathfrak{b}}$ closed simple curves, which do not intersect each other (here  ${\mathfrak{b}}$ could possibly be equal to $0$ in case $\pl\Sigma=\emptyset$) and parametrized by $\zeta_j:\T\to \pl_j\Sigma$ as in Section \ref{subsub:boundary}. We consider a metric $g$ on $\Sigma$ so that each boundary component has length $2 \pi$; except when mentioned, $g$ is not assumed to be admissible. We denote by $d\ell_g$ the Riemannian measure on $\pl\Sigma$ induced by $g$.
  
A generic (real valued) field $\tilde\varphi\in H^s(\T)$ (for $s<0$) will be decomposed into its constant mode $c$ and orthogonal part
$$\tilde\varphi=c+\varphi,\quad \varphi(\theta)=\sum_{n\not=0}\varphi_ne^{in\theta}$$
with $(\varphi_n)_{n\not=0}$  its other Fourier coefficients, which will be themselves parametrized by $\varphi_n=\frac{x_n+iy_n}{2\sqrt{n}}$  and $\varphi_{-n}=\frac{x_n-iy_n}{2\sqrt{n}}$  for $n>0$. Here, 
we make the slight abuse of notation of identifying a function on $\T$ with a function on $\R/2\pi \Z$.
In what follows, we will consider a family of such fields $\boldsymbol{\tilde \varphi}=(\tilde\varphi_1,\dots,\tilde\varphi_{b})\in (H^{s}(\T))^{{\mathfrak{b}}}$, in which case the previous notations referring to the $j$-th field will be augmented with an index $j$, namely, $c_j$, $\varphi_{j,n}$, $x_{j,n}$ or $y_{j,n}$.  We still denote by $\cjg \cdot,\cdot\cjd$ the distributional pairing between $(H^{s}(\T))^{{\mathfrak{b}}}$ and $(H^{-s}(\T))^{{\mathfrak{b}}}$ normalized so that $\cjg 1,u\cjd=\frac{1}{2\pi}\int u\dd \theta$ if $u\in C^\infty(\T)$.
 
For such a field  $\boldsymbol{\tilde \varphi}=(\tilde\varphi_1,\dots,\tilde\varphi_{{\mathfrak{b}}})\in (H^{s}(\T))^{{\mathfrak{b}}}$ with $s\in\R$, we will write   $P\tilde{\boldsymbol{\varphi}}$ for  the harmonic extension  of $\tilde{\boldsymbol{\varphi}}$, that is $\Delta_g P\tilde{\boldsymbol{\varphi}}=0$ on $\Sigma\setminus \bigcup_j\pl_j\Sigma$ with boundary values  $P\tilde{\boldsymbol{\varphi}}_{|\pl_j\Sigma}=\tilde\varphi_j\circ \zeta_j^{-1} $  for $j=1,\dots,{\mathfrak{b}}$.  The boundary value has to be understood in the following weak sense: for all $u\in C^\infty(\T)$, if $\zeta_j$ is the (analytic extension to a small annulus around $\T$ of the) parametrization of $\pl_j\Sigma$
\[\lim_{r\to 1^-}\int_{0}^{2\pi}P\tilde{\boldsymbol{\varphi}}(\zeta_j(re^{i\theta}))u(e^{i\theta})\dd \theta =2\pi \cjg \tilde\varphi_j,u\cjd.\]

The definition of our amplitudes will involve   the Dirichlet-to-Neumann operator (DN map for short). Recall that the DN map $\mathbf{D}_\Sigma:C^\infty(\T)^{{\mathfrak{b}}}\to C^\infty(\T)^{{\mathfrak{b}}}$ is defined as follows: for $\tilde{\boldsymbol{\varphi}}\in C^\infty(\T)^{{\mathfrak{b}}}$ 
\[\mathbf{D}_\Sigma\tilde{\boldsymbol{\varphi}}=(-\partial_{\nu } P\tilde{\boldsymbol{\varphi}}_{|\pl_j\Sigma}\circ\zeta_j)_{j=1,\dots,{\mathfrak{b}}} \]
where $\nu$ is the inward unit normal vector field to $\mc{C}_j$. Note that $\mathbf{D}_\Sigma$ is a non-negative symmetric operator with kernel $\ker {\bf D}_{\Sigma}=\R \tilde{1}$ where $\tilde{1}= (1, \dots, 1)$
Indeed, by Green's formula, for $\tilde{\boldsymbol{\varphi}}\in C^\infty(\T)^{{\mathfrak{b}}}$
\begin{equation}\label{Greenformula}
\int_{\Sigma} |dP\tilde{\boldsymbol{\varphi}}|_g^2{\rm dv}_g = 2 \pi \cjg \tilde{\boldsymbol{\varphi}},\mathbf{D}_\Sigma\tilde{\boldsymbol{\varphi}}\cjd.
\end{equation}
Consider ${\mathfrak{b}}'$ parametrized analytic closed simple non overlapping curves $\zeta_j':\T \to \mc{C}_j'$ in the interior of $\Sigma$ and denote 
$\mc{C}':=\sqcup_{j=1}^{{\mathfrak{b}}'}\mathcal{C}'_j$. 
For   a field  $\boldsymbol{\tilde \varphi}=(\tilde\varphi_1,\dots,\tilde\varphi_{{\mathfrak{b}}'})\in (H^{s}(\T))^{{\mathfrak{b}}'}$ with $s\in\R$, we will write   $P_{\mc{C}'}\tilde{\boldsymbol{\varphi}}$ for  the harmonic extension   $\Delta_g P_{\mc{C}'}\tilde{\boldsymbol{\varphi}}=0$ on $\Sigma\setminus \mc{C}'$ with boundary value $0$ on $\partial\Sigma$ and equal to $\tilde\varphi_j\circ \zeta'_j $ on $\mc{C}'_j$ for $j=1,\cdots,{\mathfrak{b}}'$. 
The DN map $\mathbf{D}_{\Sigma,\mc{C}'}:C^\infty(\T)^{{\mathfrak{b}}'}\to C^\infty(\T)^{{\mathfrak{b}}'}$ associated to $\mathcal{C}'$ is defined as the jump at $\mc{C}'$ of the harmonic extension:  for $\tilde{\boldsymbol{\varphi}}\in C^\infty(\T)^{{\mathfrak{b}}'}$ 
\begin{equation}\label{defDSigmaC}
\mathbf{D}_{\Sigma,{\mc{C}'}}\tilde{\boldsymbol{\varphi}}:=-((\partial_{\nu_-} P_{\mc{C}'}\tilde{\boldsymbol{\varphi}})|_{\mc{C}'_j}+(\partial_{\nu_+} P_{\mc{C}'}\tilde{\boldsymbol{\varphi}})|_{\mc{C}'_j})_{j=1,\dots,{\mathfrak{b}}'}.
\end{equation}
Here $\partial_{\nu_\pm}$ denote the two inward normal derivatives along $\mc{C}'_j$ from the right and from the left. Since $P_{\mc{C}'}(\tilde{\boldsymbol{\varphi}})$ is not $C^1$ at $\mc{C}'$ but it is piecewise smooth, the two normal derivatives are well-defined. 
The operator $\mathbf{D}_{\Sigma,{\mc{C}'}}$ is invertible and the Schwartz kernel of $\mathbf{D}_{\Sigma,{\mc{C}'}}^{-1}$ is (see e.g. the proof of \cite[Theorem 2.1]{Carron})
\begin{equation}\label{DNmapandGreen}
\mathbf{D}_{\Sigma,{\mc{C}'}}^{-1}(y,y')=\frac{1}{2\pi}G_{g,D}(y,y'), \quad y\not=y' \in \mc{C}'.
\end{equation}  
For respectively $k={\mathfrak{b}}$ or $k={\mathfrak{b}'}$, we let 
$\mathbf{D}:H^1(\T)^{k}\to L^2(\T)^{k}$ 
defined by
 \begin{equation}\label{defmathbfD}
\forall \tilde{\boldsymbol{\varphi}} \in C^\infty(\T;\R)^{k},\quad (\mathbf{D}\tilde{\boldsymbol{\varphi}},\tilde{\boldsymbol{\varphi}}):= 2\sum_{j=1}^{k}\sum_{n>0} n|\varphi_{j,n}|^2=\frac{1}{2}\sum_{j=1}^{k}\sum_{n>0}((x_{j,n})^2+(y_{j,n})^2).
\end{equation}
and  we define the operators on $C^\infty(\T)^{{\mathfrak{b}}}$ and $C^\infty(\T)^{{\mathfrak{b}}'}$
\begin{align}
\label{tildeD}& \widetilde{\mathbf{D}}_{\Sigma}  :=\mathbf{D}_{\Sigma}-\mathbf{D},  
& \widetilde{\mathbf{D}}_{\Sigma,{\mc{C}'}} :=\mathbf{D}_{\Sigma,{\mc{C}'}}-2\mathbf{D}.
\end{align}
We recall from \cite[Lemma 4.1]{GKRV} that the operators $\widetilde{\mathbf{D}}_{\Sigma}$ and $\widetilde{\mathbf{D}}_{\Sigma,{\mc{C}'}}$ are smoothing operators in the sense that they are operators with smooth Schwartz kernel that are bounded for all $s,s'\in \R$ as maps 
\[(H^{s}(\T))^{b} \to (H^{s'}(\T))^{b}, \textrm{ respectively } (H^{s}(\T))^{b'} \to (H^{s'}(\T))^{b'}.\]

\subsection{Curvature terms in the case with boundary}\label{sub:setupamp}
In what follows, we consider  a Riemann surface  $\Sigma$  with an analytic parametrization $
\boldsymbol{\zeta}=(\zeta_1,\dots, \zeta_{\mathfrak{b}})$ of the boundary with $\zeta_j:\T\to \pl_j\Sigma$ where $\pl_1\Sigma,\dots,\pl_{\mathfrak{b}}\Sigma$ are the   ${\mathfrak{b}}>0$  boundary components of $\Sigma$. 
We consider the following data that we call \emph{geometric data} of $\Sigma$:\\

\noindent \textbf{Geometric data of $\Sigma$:}

(i) Let $g$ be an admissible  metric on $\Sigma$, let $x_0$ be a base point $x_0$ chosen to be on the boundary $\pl\Sigma$ and distinct from $p_j:=\zeta_j(1)$ for $j=1,\dots,{\mathfrak{b}}$. Let  ${\bf z}:=(z_1,\dots,z_{n_{\mathfrak{m}}})$ be some marked points in its interior $\Sigma^\circ$ and let 
$\mathbf{v}=((z_1,v_1),\dots,(z_n,v_{n_{\mathfrak{m}}}))\in (T\Sigma)^{n_\mathfrak{m}}$ be unit tangent vectors at these points, to which we attach magnetic charges $\mathbf{m}=(m_1,\dots,m_{n_{\mathfrak{m}}})\in\Z^{n_\mathfrak{m}}$. 

(ii) Let us fix a canonical geometric basis $\boldsymbol{\sigma}:=(\sigma_1,\dots,\sigma_{2{\mathfrak{g}}+{\mathfrak{b}}-1})$ of the relative homology $\mc{H}_1(\Sigma,\pl \Sigma)$ (following Lemma \ref{compactsupp}), consisting of $2\mathfrak{g}$ interior cycles $(\sigma_1,\dots,\sigma_{2{\mathfrak{g}}})=(a_1,b_1,\dots,a_{\mathfrak{g}},b_{\mathfrak{g}})$ satisfying the intersection pairings 
\[\iota(a_j,b_i)=\delta_{ij}, \quad  \iota(a_i,a_j)=0,\quad \iota(b_i,b_j)=0,\]
and ${\mathfrak{b}}-1$ arcs $(\sigma_{2{\mathfrak{g}}+1},\dots, \sigma_{2{\mathfrak{g}}+{\mathfrak{b}}-1})=(d_1,\dots,d_{\mathfrak{b}-1})$ with endpoints on the boundary and no intersection with $\cup_j (a_j \cup b_j)$. 
We consider a basis $\omega^{c}_1,\dots, \omega^{c}_{2{\mathfrak{g}}+{\mathfrak{b}}-1}$ of $\mc{H}^1(\Sigma,\pl\Sigma)$ dual to $\boldsymbol{\sigma}$, made of closed forms that are compactly supported inside $\Sigma^\circ$. We ask the arc $d_j$ to have endpoints at $p_j=\zeta_j(1)\in\partial_j\Sigma$ and $p_{j+1}=\zeta_{j+1}(1)\in\partial_j\Sigma$ while making an (non-oriented) angle $\pi/2$ with $\partial_j\Sigma$ at the endpoints. The orientation of $d_j$ can be taken either way, this will not play any role later. 
Then for each ${\bf k}^c=(k_1^c,\dots,k^c_{2{\mathfrak{g}}+{\mathfrak{b}}-1})\in\Z^{2\mathfrak{g}+{\mathfrak{b}}-1}$, the form $\omega^c_{\bf k}:=\sum_{j=1}^{2{\mathfrak{g}}+{\mathfrak{b}}-1}k^c_j \omega^c_j$ satisfies $\int_{\sigma_i}\omega^c_{{\bf k}^c}=2\pi  k^c_i R$.\\
 \begin{figure}[h] 
\centering
\includegraphics[width=0.45\textwidth]{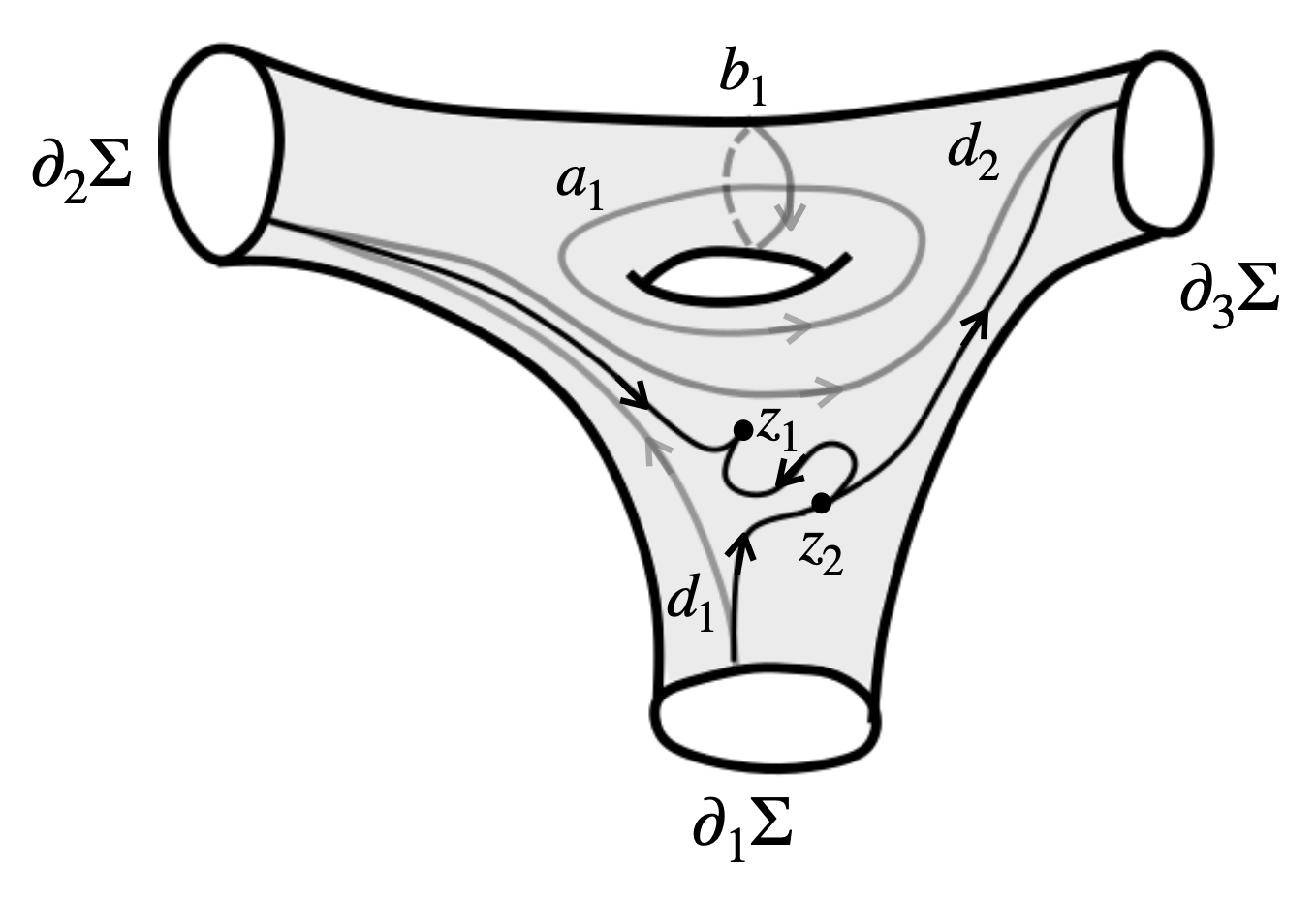} 
\caption{Homology on open surfaces. In gray, relative homology. In black, defect graph.}
\label{fig:amplitude}
\end{figure}


(iii) We encode the structure of the absolute cohomology and the magnetic operators in the closed 1-forms  $\nu_{\mathbf{z},\mathbf{m},\mathbf{k}}$ of Proposition \ref{harmpoles} using the basis of $\mc{H}_1(\Sigma)
$
\[(a_1,b_1,\dots,a_{\mathfrak{g}},b_{\mathfrak{g}}, \pl_1\Sigma,\dots,\pl_{\mathfrak{b}-1}\Sigma)\]
where  $\nu_{\mathbf{z},\mathbf{m},\mathbf{k}}$ are labeled by   ${\bf k}=(k_1,\dots,k_{\mathfrak{b}})\in\Z^{\mathfrak{b}}$ satisfying 
\begin{equation}\label{condition_charges}
\qquad \frac{1}{2\pi R}\int_{\partial_j\Sigma}\nu_{\mathbf{z},\mathbf{m},\mathbf{k}}=\varsigma_j k_j,\qquad \sum_{j=1}^{n_\mathfrak{m}} m_j+\sum_{j=1}^{{\mathfrak{b}}}\varsigma_jk_j=0, 
\end{equation}
where we recall that $\varsigma_j=-1$ if the boundary is outgoing and $\varsigma_j=1$ if it is incoming. The choice of such a form is not unique: indeed the difference of two such forms is an exact 1-form $\dd f$ with $\pl_\nu f|_{\pl \Sigma}=0$ (see Lemma \ref{existence_primitive}). By possibly adding such an exact form, we can require that the form $\nu_{\mathbf{z},\mathbf{m},\mathbf{k}}$ takes values in $2\pi R\Z$ along the boundary-to-boundary arcs: 
\begin{equation}\label{btbaRz}
\forall j=2\mathfrak{g}+1,\dots,2\mathfrak{g}+{\mathfrak{b}}-1 ,\quad \int_{\sigma_{j}}  \nu_{\mathbf{z},\mathbf{m},\mathbf{k}}\in 2\pi R\Z.
\end{equation}
 
(iv) Next we construct a defect graph associated to the 1-forms $\nu_{\mathbf{z},\mathbf{m},\mathbf{k}}$ as follows. The construction, detailed below, is similar to the case of closed Riemann surfaces except that we will  see the point $\zeta_j(1)\in \pl_j\Sigma$ as extra marked points with a magnetic charge $ \varsigma_jk_j$   assigned. So, for notational simplifications,  we set $z_{n_{\mathfrak{m}}+j}=\zeta_j(1)$ and $m_{n_{\mathfrak{m}}+j}= \varsigma_{j}k_{j}$ for $j=1,\dots,{\mathfrak{b}}$. We then associate the total magnetic charges (which are now ${\bf k}$ dependent) defined by
\begin{equation}\label{defm(k)} 
{\bf m}({\bf k})=(m_1({\bf k}),\dots, m_{n_{\mathfrak{m}}+\mathfrak{b}}({\bf k})):= 
(m_1,\dots, m_{n_{\mathfrak{m}}}, \varsigma_1 k_1, \dots, \varsigma_{\mathfrak{b}} k_{\mathfrak{b}}).
\end{equation}
Also, for $j=n_{\mathfrak{m}}+1,\dots,n_{\mathfrak{m}}+{\mathfrak{b}}$, we set $v_j\in T_{z_j}\Sigma$ to be the inward unit vector at $z_j$ normal to $\pl \Sigma$.

\begin{definition}\label{def:defect}{\bf Defect graph:} 
We consider a family of $n_\mathfrak{m}+{\mathfrak{b}}-1$ arcs as follows:
\begin{itemize}
\item these arcs are indexed by $p\in\{1,\dots, n_\mathfrak{m}+{\mathfrak{b}}-1\}$, are simple and do not intersect except eventually at their endpoints,
\item each arc is a smooth oriented curve $\xi_p:[0,1]\to  \Sigma $ parametrized by arclength with endpoints $\xi_p(0)=z_j$ and $\xi_p(1)=z_{j'}$ for $j\not =j'$, with orientation in the direction of increasing charges, meaning $m_j\leq m_{j'}$.
\item  these arcs satisfy $\dot{\xi}_{p}(0)=\lambda_{p,j} v_{j}$ and $\dot{\xi}_{p}(1)=\lambda_{p,j'} v_{j'}$ for some  $\lambda_{p,j}>0$ and $\lambda_{p,j'} >0$ if $\xi_{p}(1)\notin \pl \Sigma$, while $\lambda_{p,j'}<0$ if $\xi_{p}(1)\in \pl \Sigma$.
\item consider the  oriented  graph with vertices $\mathbf{z}$ and edges $(z_j,z_{j'})$, if there is an oriented arc with basepoint $z_j$ and endpoint $z_{j'}$. This graph must be connected and without cycle, i.e. there is no sequence of edges $(z_{j_1},z_{j_2}),\dots,(z_{j_k},z_{j_{k+1}})$ with $j_1=j_{k+1}$. 
\end{itemize}
In what follows, the union $\mc{D}_{\mathbf{v},\boldsymbol{\xi}}:=\bigcup_{p\in\{1,\dots, n_\mathfrak{m}+{\mathfrak{b}}-1\}}\xi_p([0,1])$ will be called the defect graph associated to $\mathbf{v}$ and the collection of arcs $\boldsymbol{\xi}:=(\xi_{1},\dots, \xi_{n_\mathfrak{m}+{\mathfrak{b}}-1})$.
\end{definition}

 \begin{figure}[h] 
\centering
\includegraphics[width=.5\textwidth]{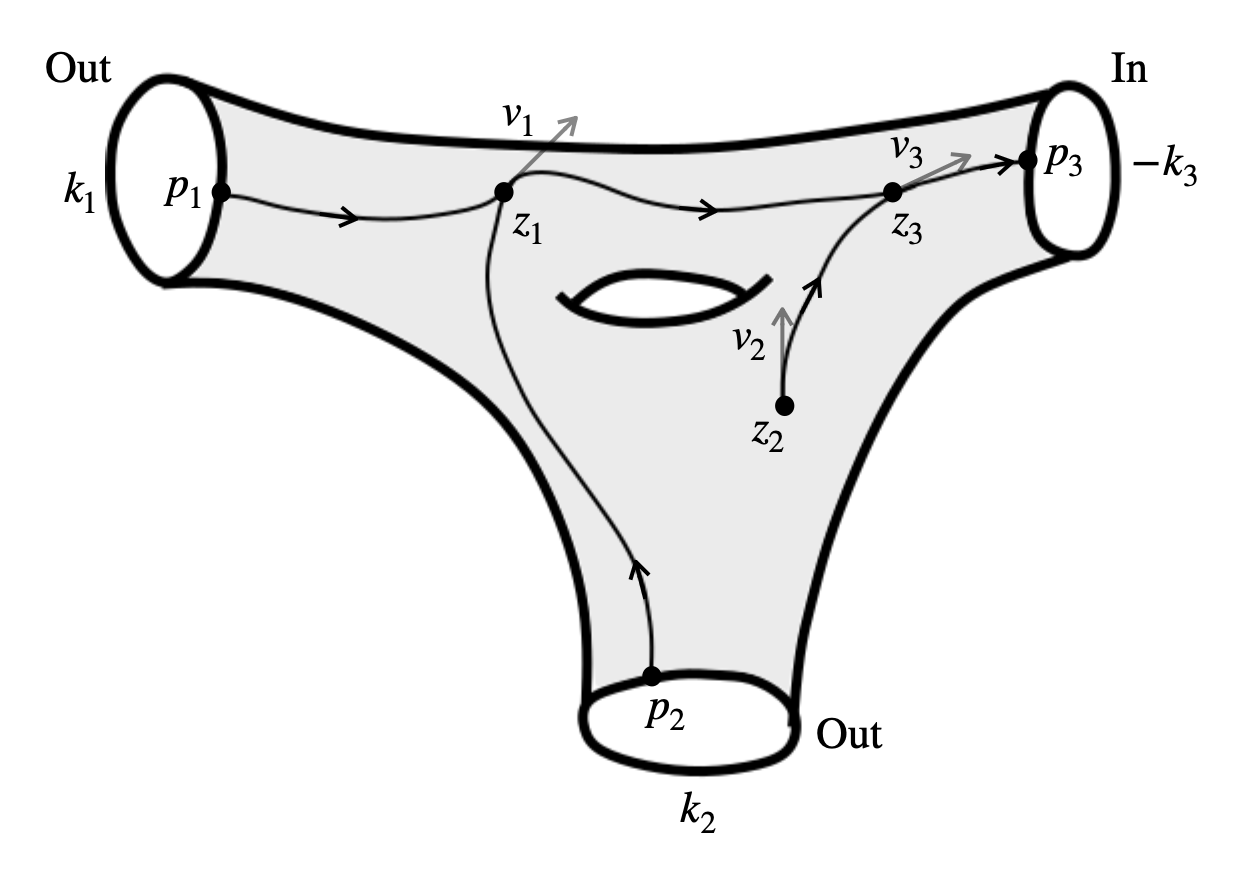} 
\caption{Exemple of defect graph (black) in case $ k_2\leq k_1\leq m_1\leq m_2\leq m_3 \leq -k_3$.}
\label{fig:case1}
\end{figure}

Notice that the graph $\mc{D}_{\mathbf{v},\boldsymbol{\xi}}$ is contractible to a point. The form $\nu_{\mathbf{z},\mathbf{m},\mathbf{k}}$ is exact on $\Sigma\setminus \mc{D}_{\mathbf{v},\boldsymbol{\xi}}$ so that we can consider the primitive $I^{\boldsymbol{\xi}}_{x_0}(\nu_{\mathbf{z},\mathbf{m},\mathbf{k}})$ on $\Sigma\setminus \mc{D}_{\mathbf{v},\boldsymbol{\xi}}$. As in the closed case (recall \eqref{defkappaxip}), we assign to each arc $\xi_p$ a value $\kappa(\xi_p)\in 2\pi R\Z$, corresponding to the difference of the values of $I^{ \boldsymbol{\xi}}_{x_0}(\nu_{\mathbf{z},\mathbf{m},\mathbf{k}})$ on 
both sides of the arc. The value $\kappa(\xi_p)$ is defined by first gluing disks $\mc{D}_j$ to $\pl _j \Sigma$ for each $j$ in order to be in the setting of a closed surface and using the same definition as in the closed case, see \eqref{defkappaxip}.

The regularized curvature terms are then defined by the same formula as \eqref{def_reg_integtral} and  \eqref{curv:mag2}:
\begin{equation}\label{def_reg_integtral_open}
\int_{\Sigma_{\boldsymbol{\sigma}}}^{\rm reg} K_g I^{\boldsymbol{\sigma}}_{x_0}(\omega^c_{{\bf k}^c})\dd {\rm v}_g:=  \int_{\Sigma_{\boldsymbol{\sigma}}}  I^{\boldsymbol{\sigma}}_{x_0}(\omega^c_{{\bf k}^c}) K_g\dd {\rm v}_g+2\sum_{j=1}^{{\mathfrak{g}}}
\Big(\int_{a_j}\omega^c_{{\bf k}^c} \int_{b_j}k_{g}\dd \ell_{g}-\int_{b_j}\omega^c_{{\bf k}^c}\int_{a_j}k_{g}\dd \ell_{g}\Big),
\end{equation}
\begin{equation}\label{curv:mag2open}
\int_{\Sigma }^{\rm reg}  I^{ \boldsymbol{\xi}}_{x_0}( \nu_{\mathbf{z},\mathbf{m},\mathbf{k}}) K_g \,\dd v_g:=\int_{\Sigma\setminus \mc{D}_{\mathbf{v},\boldsymbol{\xi}}}   I^{ \boldsymbol{\xi}}_{x_0}( \nu_{\mathbf{z},\mathbf{m},\mathbf{k}}) K_g \,\dd v_g-2\sum_{p=1}^{n_\mathfrak{m}+{\mathfrak{b}}-1}\kappa(\xi_p)\int_{\xi_p}k_g\dd \ell_g,
\end{equation}
where $k_g$ is the geodesic curvature as defined before. Remark that in the expression \eqref{def_reg_integtral_open}, there is no boundary 
term involving the arcs $d_j$ or the boundary cycles $c_j$: the reason is that the curves $c_j$ will be chosen to be geodesics (since our metrics are admissible) and that $\iota_{\pl\Sigma}^*\omega^c_{{\bf k}^c}=0$  (thus $\int_{\pl_j\Sigma}\omega^c_{{\bf k}^c}=0$) so that the natural boundary terms that one could add actually vanish: 
\[\int_{\pl_j\Sigma}\omega^c_{{\bf k}^c} \int_{d_j}k_{g}\dd \ell_{g}-\int_{d_j}\omega^c_{{\bf k}^c}\int_{c_j}k_{g}\dd \ell_{g}=0.\]
With the same proof as Lemma \ref{conformal_change_reg_int}, we have 
\begin{lemma}\label{conformal_change_reg_int_open}
If $\hat{g}=e^{\rho}g$ with $\pl_\nu \rho|_{\pl \Sigma}=0$ and $\omega\in \mc{H}^1_R(\Sigma,\pl \Sigma)$ 
with compact support in $\Sigma^\circ$, the following identity holds true
\[ \int^{\rm reg}_{\Sigma_{\boldsymbol{\sigma}}} I^{\boldsymbol{\sigma}}_{x_0}(\omega)K_{\hat{g}}\dd {\rm v}_{\hat{g}}
=\int^{\rm reg}_{\Sigma_{\boldsymbol{\sigma}}} I^{\boldsymbol{\sigma}}_{x_0}(\omega)K_{g}\dd {\rm v}_{g}+\cjg  \dd\rho,\omega\cjd_2.\]
 \end{lemma} 
Similarly to Lemma \ref{independence_basis}, we also have 
\begin{lemma}\label{independence_basis_open}
 Let $\boldsymbol{\sigma}=(\sigma_j)_{j=1,\dots,{\mathfrak{g}}+{\mathfrak{b}}-1}$ and $\boldsymbol{\sigma}'=(\sigma_j')_{j=1,\dots,{\mathfrak{g}}+{\mathfrak{b}}-1}$ 
 be canonical geometric bases  as described above and $\omega\in \mc{H}^1_R(\Sigma,\pl\Sigma)$ with compact support in $\Sigma^\circ$. Then the following identity holds true 
\[ \int_{\Sigma_{\boldsymbol{\sigma}}}^{\rm reg} K_g I^{\boldsymbol{\sigma}}_{x_0}(\omega)\dd {\rm v}_g -\int_{\Sigma_{\boldsymbol{\sigma}'}}^{\rm reg} K_{g}I^{\boldsymbol{\sigma}'}_{x_0}(\omega)\dd {\rm v}_{g}\in 8\pi^2 R\Z.\]
 \end{lemma} 
\begin{proof} We can double the surface $\Sigma^{\#2}=\Sigma\# \Sigma$, view $(\sigma_j)_{j\leq {\mathfrak{g}}}$ as cycles of $\Sigma^{\#2}$ contained in the right copy of $\sigma$, and let $\tau$ denote the natural involution on $\Sigma^{\#2}$. The cycles
\begin{equation}\label{doublecycles}
(a_1,b_1,\dots,a_{\mathfrak{g}},b_{\mathfrak{g}}, -\tau(a_1),\tau(b_1),\dots, -\tau(a_{\mathfrak{g}}),\tau(b_{\mathfrak{g}})) \in \mc{H}_1(\Sigma^{\sharp 2})
\end{equation}
form a part of a geometric symplectic basis of $\mc{H}_1(\Sigma^{\sharp 2})$, which can be completed in a full geometric symplectic basis denoted $\sigma^\sharp$ by adding ${\mathfrak{b}}-1$ cycles $c_1,\dots,c_{{\mathfrak{b}}-1}$ coming from boundary cycles of $\Sigma$ and ${\mathfrak{b}}-1$ non intersecting cycles $d_1,\dots,d_{{\mathfrak{b}}-1}$ with intersection pairing $\iota(c_i,d_j)=\delta_{ij}$ and $d_j$ not intersecting $a_i,b_i$. Let us extend $\omega$ by $0$ from the left copy of $\Sigma$ to $\Sigma^{\# 2}$. Since the integral $\int_{c_i}\omega=0$ for all $i$ and since $c_i$ are also geodesic curves, the geodesic curvature terms in the regularized integral on $\Sigma^{\#2}$ coming from the cycles are $0$ except for those in \eqref{doublecycles}. The regularized integral of $\omega$  on $\Sigma^{\sharp 2}$ satisfies 
\[ \int_{\Sigma^{\sharp 2}_{\boldsymbol{\sigma}^\sharp}}^{\rm reg} K_gI^{\boldsymbol{\sigma}^\#}_{x_0}(\omega)\dd {\rm v}_g= \int_{\Sigma_{\boldsymbol{\sigma}}}^{\rm reg} K_g I^{\boldsymbol{\sigma}}_{x_0}(\omega)\dd {\rm v}_g.\]
We can then apply Lemma \ref{independence_basis} on $\Sigma^{\#2}$ to a change of geometric symplectic basis which preserve all the cycles contained in the left copy of $\Sigma^{\#2}$ and map $(a_i,b_i)_{i\leq {\mathfrak{g}}}$ to other elements $(a'_i,b_i')_{i\leq {\mathfrak{g}}}$ contained in the interior of the right copy of $\Sigma$. This proves the desired the result.
\end{proof}
The invariance by diffeomorphism is proved in the same exact way as Lemma \ref{lemcurvdiff} and reads 
\begin{lemma}\label{lemcurvdiffopen}
 For $\psi: \Sigma\to \Sigma$ an orientation preserving diffeomorphism and $\boldsymbol{\sigma}$ a canonical geometric basis  of $\mc{H}_1(\Sigma,\pl \Sigma)$, let $\psi(\boldsymbol{\sigma})$ be  the image canonical geometric basis of $\mc{H}_1(\Sigma,\pl \Sigma)$.  Let $x_0\in \Sigma$ and $\omega\in \mc{H}^1_R(\Sigma,\pl \Sigma)$ with compact support in $\Sigma^\circ$. The following identity then holds true 
\[ \int_{\Sigma_{\boldsymbol{\sigma}}}^{\rm reg} K_g I^{\boldsymbol{\sigma}}_{x_0}(\omega)\dd {\rm v}_g =\int_{\Sigma_{\psi(\boldsymbol{\sigma})}}^{\rm reg} K_{\psi_*g} I^{\psi(\boldsymbol{\sigma})}_{\psi(x_0)}(\psi_*\omega)\dd {\rm v}_{\psi_*g}.\]
 \end{lemma} 

Now we state all the main properties of defect graphs needed in the sequel. We claim
\begin{lemma}\label{defectopen}
The magnetic regularized curvature term only depends on  the points $\mathbf{z}\in \Sigma^{n_{\mathfrak{m}}+{\mathfrak{b}}}$, the charges $\mathbf{m}\in \Z^{n_{\mathfrak{m}}+{\mathfrak{b}}}$ and the unit tangent vectors $\mathbf{v}$ but not on the defect graph (i.e. the choice of arcs constructed from this data).
\end{lemma}
\begin{proof} Same proof as Lemma \ref{inv_par_graphe}.
\end{proof}

\begin{lemma} \label{magcurvconfopen}
Consider two conformal metrics $g'=e^{\rho}g$ with $\pl_\nu \rho=0$ on $\pl \Sigma$. The regularized magnetic curvature term defined by \eqref{curv:mag2open}  satisfies
$$\int_{\Sigma }^{\rm reg}  I^{\boldsymbol{\xi}}_{x_0}(\nu_{\mathbf{z},\mathbf{m},\mathbf{k}}) K_{g'} \, {\rm dv}_{g'}=\int_{\Sigma }^{\rm reg}  I^{\boldsymbol{\xi}} _{x_0}( \nu_{\mathbf{z},\mathbf{m},\mathbf{k}}) K_g \, {\rm dv}_g+\cjg  \dd\rho,  \nu_{\mathbf{z},\mathbf{m},{\bf k}}\cjd_2.$$
\end{lemma} 

\begin{proof} Same proof as Lemma \ref{conformal_change_reg_int}.
\end{proof}
 
 \begin{lemma} \label{magcurvdiffopen}
Let $\psi:\Sigma'\to\Sigma$ be an  orientation preserving diffeomorphism. The regularized magnetic curvature term defined by \eqref{curv:mag2open}  satisfies the relation
\[\int_{\Sigma' }^{\rm reg}  I^{\boldsymbol{\xi }}_{x_0}(\nu_{\mathbf{z},\mathbf{m},\mathbf{k}}) K_{g} \, {\rm dv}_{g}=\int_{\Sigma }^{\rm reg}  I^{\boldsymbol{\psi\circ\xi }} _{\psi(x_0)}( \psi_*\nu_{\mathbf{z},\mathbf{m},\mathbf{k}}) K_{\psi_*g} \,{\rm  dv}_{\psi_*g}.\]
\end{lemma} 

\begin{proof}
The proof is the same as that of Lemma \ref{lemcurvdiff}.
\end{proof} 

Now we treat the additivity of regularized curvature integrals. First we treat the case when we glue  two Riemann surfaces:
\begin{lemma} \label{magcurvtwo}
Consider two Riemannian surfaces $(\Sigma_i,g_i)$  for $i=1,2$ with admissible metrics, with ${\mathfrak{b}}_i>0$ boundary components and genus ${\mathfrak{g}}_i$ for $i=1,2$ and with analytic parametrizations $
\boldsymbol{\zeta}_i$ of their respective boundary. Assume $\pl_{{\mathfrak{b}}_1}\Sigma_1$ is outgoing and $\pl_{{\mathfrak{b}}_2}\Sigma_2$ is incoming and let $\Sigma=\Sigma_1\#\Sigma_2$ the glued surface obtained by identifying $\pl_{{\mathfrak{b}}_1}\Sigma_1$ with $\pl_{{\mathfrak{b}}_2}\Sigma_2$ and $g$ the metric induced by $g_1,g_2$.  Consider the geometric data (i), (ii), (iii), (iv) described above for $\Sigma_1$ and for $\Sigma_2$, and denote them with either a subscript $i$ or a superscript $i$ when they are associated to $\Sigma_i$, for example 
$\boldsymbol{\sigma}_i,{\bf k}^c_i,\boldsymbol{\xi}_i,x_{0}^{i}, \nu^{i}_{{\bf z}_i,{\bf m}_i,{\bf k}_i},\omega^{i,c}_{{\bf k}_i^c}$ etc.
We choose these defect graphs by imposing that there is only one arc $\xi_{1j_0}$ having $\zeta_{1{\mathfrak{b}}_1}(1)$ as endpoint on $\Sigma_1$ and only one arc $\xi_{2j'_0}$ having $\zeta_{2{\mathfrak{b}}_2}(1)$ as endpoint on $\Sigma_2$. Assume $k_{1\mathfrak{b}_1}=k_{2\mathfrak{b}_2}$ and set ${\bf z}:=({\bf z}_1,{\bf z}_2)$, ${\bf m}:=({\bf m}_1,{\bf m}_2)$ and ${\bf v}:=({\bf v}_1,{\bf v}_2)$. Then we obtain a defect graph $\mc{D}_{{\bf v},\boldsymbol{\xi}}$   on $\Sigma$ by gathering all the arcs in $\boldsymbol{\xi}_1$ but $\xi_{1j_0}$, all the arcs in $\boldsymbol{\xi}_2$ but $\xi_{2j_0'}$, and the arc $\xi$ obtained by gluing the arcs $\xi_{1j_0}$ and $\xi_{2j_0'}$ with orientation in the sense of increasing charges. Furthermore, if ${\bf k}:=({\bf k}^-_{1},{\bf k}_2^-)$ where ${\bf k}_i^-\in\Z^{{\mathfrak{b}}_i-1}$ stands for the vector ${\bf k}_i$ with its ${\mathfrak{b}}_i$-th component removed and if one glues together the $1$-forms  $\nu^{i}_{{\bf z}_i,{\bf m}_i,{\bf k}_i}$, $i=1,2$, to get the $1$-form $\nu_{{\bf z},{\bf m},{\bf k}}$ on $\Sigma$, then
\begin{equation}\label{addmag}
\int_\Sigma^{\rm reg}I_{x_0}^{\boldsymbol{\xi}}(\nu_{{\bf z},{\bf m},{\bf k}})K_g \dd {\rm v}_g=\int_{\Sigma_1}^{\rm reg}I_{x_0^{1}}^{\boldsymbol{\xi}_1}(\nu^{1}_{{\bf z}_1,{\bf m}_1,{\bf k}_1})K_{g_1} \dd {\rm v}_{g_1}+\int_{\Sigma_2}^{\rm reg}I_{x_0^{2}}^{\boldsymbol{\xi}_2}(\nu^{2}_{{\bf z}_2,{\bf m}_2,{\bf k}_2})K_{g_2} \dd {\rm v}_{g_2}
\end{equation}
where $x_0^{1}$ and $x_0^{2}$ are two base points respectively on $\pl_{\mathfrak{b}_1}\Sigma_1\subset\Sigma_1$ and $\pl_{\mathfrak{b}_2}\Sigma_2\subset\Sigma_2$, which we assume are identified under gluing to $x_0=x_0^{1}=x_0^{2}\in\Sigma$.
Similarly, let $\boldsymbol{\sigma}=\boldsymbol{\sigma}_1\# \boldsymbol{\sigma}_2$ be the gluing of the canonical geometric bases $\boldsymbol{\sigma}_1$ and $\boldsymbol{\sigma}_2$ from Lemma \ref{baseglue}, let ${\bf k}^c=({\bf k}_1^c,{\bf k}_2^c)$ and $\omega_{{\bf k}^c}^c:=\omega_{{\bf k}_1^c}^{1,c}+\omega_{{\bf k}_2^c}^{2,c}$.
Then the following holds true
\begin{equation}\label{add_curvature}
\int_\Sigma^{\rm reg}I_{x_0}^{\boldsymbol{\sigma}}(\omega^c_{{\bf k}^c})K_g \dd {\rm v}_g=\int_{\Sigma_1}^{\rm reg}I_{x^{1}_0}^{\boldsymbol{\sigma}_1}(\omega^{1,c}_{{\bf k}_1^c})K_{g_1} \dd {\rm v}_{g_1}+\int_{\Sigma_2}^{\rm reg}I_{x^{2}_0}^{\boldsymbol{\sigma}_2}(\omega^{2,c}_{{\bf k}_2^c})K_{g_2} \dd {\rm v}_{g_2}.
\end{equation}
\end{lemma} 

\begin{proof}
The fact that $\mc{D}_{{\bf v},\boldsymbol{\xi}}$ is a defect graph is clear. Since $\xi$ is obtained by gluing of the arcs $\xi_{1j_0}$ and $\xi_{2,j_0'}$, then $I^{\boldsymbol{\xi}}_{x_0}(\nu_{{\bf z},{\bf m},{\bf k}})|_{\Sigma_i}=I_{x^{i}_0}^{\boldsymbol{\xi}_i}(\nu^{i}_{{\bf z}_i,{\bf m}_i,{\bf k}_i})$ (indeed the branch cuts remain unchanged). What is less straightforward is the additivity of regularized curvature integrals. It results from a simple observation concerning the computation of the coefficients $\kappa(\xi_{ip})$ on $\Sigma_i$.  Let us begin with $\Sigma_1$ ($\Sigma_2$ is similar). To compute the coefficient $\kappa(\xi_{1p})$ of the arc $\xi_{1p}$ (for $p\not=j_0$) in the defect graph associated to $\Sigma_1$, recall  (see after Definition \ref{def:defect}) that we consider a positively oriented closed contractible curve $\alpha_x:[0,1]$ in  the closed surface $\hat{\Sigma}_1$ (obtained by gluing disks at the boundary components of $\Sigma_1$), with $x=\alpha_x(0)=\xi_{1p}(t)=\alpha_x\cap \boldsymbol{\xi}$ for some $t$ and $\dot{\alpha}_x(0)=J\dot{\xi}_{1p}(t)$; note that $\alpha_x$ bounds a disk $D_\alpha$ with positively oriented boundary. 
If $D_\alpha$ does not contain $\zeta_{1\mathfrak{b}_1}(1)$, this means the same curve $\alpha_x$ can be used to compute $\kappa(\xi_{1p})$, both when the arc $\xi_{1p}$ is seen as an element of the defect graph of $\Sigma_1$ or $\Sigma$, hence it takes the same value in both cases.
If the point  $\zeta_{1\mathfrak{b}_1}(1)$ belongs to $D_\alpha$, then its contribution to $\kappa(\xi_{1p})$ is   $-2\pi Rk_{1\mathfrak{b}_1}$  (recall that $\pl_{\mathfrak{b}_1}\Sigma_1$ is outgoing). Next, if  we compute the   coefficient $\kappa(\xi_{1p})$ for the   arc  $\xi_{1p}$ seen as an arc in the defect graph $\mc{D}_{{\bf v},\boldsymbol{\xi}}$   on $\Sigma$, the curve $\alpha_x$ bounds a disk containing all the points of the defect graphs located  in $\Sigma_2$ as well, producing  a total contribution $ 2\pi R(\sum_{j=1}^{n^{(2)}_\mathfrak{m}} m_{2j}+\sum_{j=1}^{\mathfrak{b}_2-1} \varsigma_{2j}k_{2j})$, which by \eqref{condition_charges} is equal to $ -2\pi Rk_{2\mathfrak{b}_2} =  -2\pi Rk_{1\mathfrak{b}_1}$. Hence $\kappa(\xi_{1p})$ has the same value when viewed in the graph $\mc{D}_{{\bf v},\boldsymbol{\xi}}$  or in $\mc{D}_{{\bf v}_1, \boldsymbol{\xi}_1}$. The situation is slightly different if we compute $\kappa(\xi_{1j_0})$ because of the orientation of $\xi|_{\Sigma_1}$. Of course, if this  orientation is the same as $\xi_{1j_0}$ then  the argument is the same as above. If the orientation is reversed, then   the sum defining $\kappa(\xi)$ in $\mc{D}_{{\bf v},\boldsymbol{\xi}}$ involves the complement of the charges used to compute  $\kappa(\xi_{1j_0})$ in $\Sigma_1$, and since the total sum of all charges is $0$, this means that $\kappa (\xi|_{\Sigma_1})=-\kappa(\xi_{1j_0})$. But changing the orientation also changes the sign of the geodesic curvature so that 
$$\kappa(\xi_{1j_0})\int_{\xi_{1j_0}}k_g\dd \ell_g =\kappa (\xi|_{\Sigma_1})\int_{\xi|_{\Sigma_1}}k_g\dd \ell_g.$$
The argument is the same on $\Sigma_2$, and all in all, this proves the relation \eqref{addmag}.  

The identity \eqref{add_curvature} is clear, using that  $\omega^c_{i{\bf k}_i^c}$ vanishes near the boundaries $\pl \Sigma_i$. 
\end{proof} 

It remains to consider the case of self-gluing. We claim: 

\begin{lemma} \label{magcurvtwoself}
Consider  a Riemann surface  $\Sigma $ with $\mathfrak{b}\geq 2$ boundary components  and with analytic parametrizations $
\boldsymbol{\zeta}$ of their respective boundary. We consider marked points ${\bf z}_1:=(z_{1},\dots,z_{n_{\mathfrak{m}}})$   on  $\Sigma $ and associated respective magnetic charges ${\bf m}:=(m_{1},\dots,m_{n_{\mathfrak{m}}})$  and unit tangent vectors $\mathbf{v}=((z_1,v_{1}),\dots,(z_n,v_{n_{\mathfrak{m}}}))\in (T\Sigma)^{n_\mathfrak{m}}$ at the points $z_j$.  We assume $\pl_{\mathfrak{b}-1}\Sigma$ is outgoing and $\pl_{\mathfrak{b}}\Sigma$ is incoming. 
Finally we consider ${\bf k}\in \Z^{\mathfrak{b}}$   and the 1-forms $\nu_{{\bf z},{\bf m},{\bf k}}$ on $\Sigma$ given by   Lemma \ref{harmpoles} such that (recall that $\varsigma_{j}=-1$ is the boundary $\partial_j\Sigma$ is outgoing and $\varsigma_{j}=1$ if it is incoming)
\begin{align*}
 &\int_{\partial_j\Sigma}\nu_{\mathbf{z},\mathbf{m},\mathbf{k}}= 2\pi R\varsigma_{j}k_{j}\text{ for }j=1,\dots,\mathfrak{b},\qquad   \int_{d_g(z_{j},\cdot)=\eps}\nu_{\mathbf{z},\mathbf{m},\mathbf{k}}=2\pi R m_{j}\quad \text{ for } j=1,\dots,n_{\mathfrak{m}}\\
& \sum_{j=1}^{n_\mathfrak{m}} m_{j}+\sum_{j=1}^{\mathfrak{b}} \varsigma_{j}k_{j}=0
\end{align*}
and a  defect graph  $\mc{D}_{{\bf v}, \boldsymbol{\xi}}$  on $\Sigma$. We choose this defect graphs by imposing that there is only one arc $\xi_{j_0}$ having $\zeta_{\mathfrak{b}}(1)$ as endpoint, and that this arc has $\zeta_{\mathfrak{b}-1}(1)$ as other endpoint.
Let $\Sigma^{\#}$ be the Riemann surface obtained by self-gluing $\Sigma$ by identifying $\partial_{\mathfrak{b}-1}\Sigma$ and $\partial_{\mathfrak{b}}\Sigma$ using the corresponding parametrizations, and we assume $k_{\mathfrak{b}-1}=k_{\mathfrak{b}}$.   Then we can obtain a defect graph $\mc{D}_{{\bf v},\boldsymbol{\xi}^\#}$   on $\Sigma^\#$ by taking all the arcs in $\boldsymbol{\xi}$ and removing the two arcs having $\zeta_{\mathfrak{b}-1}(1)$ or $\zeta_{\mathfrak{b}}(1)$ as endpoints. The curves $\xi_{j_0}$ and $\zeta_{\mathfrak{b}-1}$ have intersection number 1. Let $\boldsymbol{\sigma}$ be a basis of the relative homology on $\Sigma^{\#}$ having both $\xi_{j_0}$ and $\zeta_{\mathfrak{b}-1}$ has interior cycles.

We write ${\bf k}\in\Z^\mathfrak{b}$ as ${\bf k}=({\bf k}_-,k_{\mathfrak{b}-1} ,k_\mathfrak{b})\in \Z^{\mathfrak{b}-2}\times \Z\times \Z$. Then $\nu_{{\bf z},{\bf m},{\bf k}}$ can be split as $\nu_{{\bf z},{\bf m},{\bf k}}=\nu_{{\bf z},{\bf m},({\bf k}_-,0,0)}+\nu_{{\bf z},{\bf m},({\bf 0},k_\mathfrak{b},k_\mathfrak{b})}$, and each 1-form involved in this expression makes sense as 1-form on $\Sigma^{\#}$ by Lemma \ref{baseselfglue}. Then
\begin{equation}\label{addmag}
\int_\Sigma^{\rm reg}I_{x_0}^{\boldsymbol{\xi}}(\nu_{{\bf z},{\bf m},{\bf k}})K_g \dd {\rm v}_g=\int_{\Sigma^{\#}}^{\rm reg}I_{x_0}^{\boldsymbol{\xi}^\#}(\nu_{{\bf z},{\bf m},({\bf k}_-,0,0)})K_g \dd {\rm v}_g+\int_{\Sigma^{\#}}^{\rm reg}I_{x_0}^{\boldsymbol{\sigma}}(\nu_{{\bf z},{\bf 0},({\bf 0},k_\mathfrak{b},k_\mathfrak{b})})K_g \dd {\rm v}_g.
\end{equation}
\end{lemma} 

\begin{proof}
It is obvious that $\mc{D}_{{\bf v},\boldsymbol{\xi}^\#}$ is a defect graph associated to the charges ${\bf m}$ for the points in ${\bf z}$ and ${\bf k}_-$ for the boundary components of $\Sigma^{\#}$. Recall that $\xi_{j_0}$ is the arc with basepoint $\zeta_{\mathfrak{b}}(1)$ and endpoint $\zeta_{\mathfrak{b}-1}(1)$. Let us call 
$\xi_{j_0'}$, $j_0'\not=j_0$, the other arc with basepoint/endpoint $\zeta_{\mathfrak{b}-1}(1)$. Let us assume for convenience that all the arcs in the defect graph $\mc{D}_{{\bf v},\boldsymbol{\xi}^\#}$ keep the same label as in $\mc{D}_{{\bf v},\boldsymbol{\xi}}$. Therefore the arcs in the defect graph $\mc{D}_{{\bf v},\boldsymbol{\xi}^\#}$ is labelled with $p=1,\dots, n_{\mathfrak{m}}+\mathfrak{b}-1$ with $p\not=j_0,j_0'$.

Let us compute the coefficients $\kappa(\xi^\#_p)$.   Then $\kappa(\xi_{j_0})=-k_{\mathfrak{b}}$ (since $\zeta_\mathfrak{b}(1)$ is an end of the graph tree) and $\kappa(\xi_{j_0'})=0$ (because $k_{\mathfrak{b}-1}=k_\mathfrak{b}$). For $p\not=j_0,j_0'$, we have $\kappa(\xi^\#_p)=\kappa(\xi_p)$ because the structure for the defect graph $\mc{D}_{{\bf v},\boldsymbol{\xi}}$ we chose imposes that each time we meet $\zeta_{\mathfrak{b}-1}(1)$ when following counterclockwise the contour of the graph, we also meet $\zeta_{\mathfrak{b}}(1)$, for a total contribution of both points given by $k_{\mathfrak{b}-1}-k_{\mathfrak{b}}=0$. Therefore
\begin{align*}
\int_\Sigma^{\rm reg}I_{x_0}^{\boldsymbol{\xi}}(\nu_{{\bf z},{\bf m},{\bf k}})K_g \dd {\rm v}_g=&\int_{\Sigma\setminus \mc{D}_{\mathbf{v},\boldsymbol{\xi}}}   I^{ \boldsymbol{\xi}}_{x_0}( \nu_{\mathbf{z},\mathbf{m},\mathbf{k}}) K_g \ {\rm dv}_g-2\sum_{p=1}^{n_\mathfrak{m}+\mathfrak{b}-1}\kappa(\xi_p)\int_{\xi_p}k_g\dd \ell_g
\\
=&\int_{\Sigma\setminus \mc{D}_{\mathbf{v},\boldsymbol{\xi}}}   I^{ \boldsymbol{\xi}}_{x_0}(\nu_{{\bf z},{\bf m},({\bf k}_-,0,0)}) K_g \,{\rm dv}_g+\int_{\Sigma\setminus \mc{D}_{\mathbf{v},\boldsymbol{\xi}}}   I^{ \boldsymbol{\xi}}_{x_0}( \nu_{{\bf z},{\bf m},({\bf 0},k_{\mathfrak{b}},k_{\mathfrak{b}})}) K_g \,{\rm dv}_g
\\
&-2\sum_{p=1,p\not =j_0,j_0'}^{n_\mathfrak{m}+b-1}\kappa(\xi_p)\int_{\xi_p}k_g\dd \ell_g-2\sum_{p=j_0,j_0'} \kappa(\xi_p)\int_{\xi_p}k_g\dd \ell_g
\\
=&\int_{\Sigma\setminus \mc{D}_{{\bf v},\boldsymbol{\xi}^\#}}   I^{ \boldsymbol{\xi}^\#}_{x_0}(\nu_{{\bf z},{\bf m},({\bf k}_-,0,0)}) K_g \,\dd v_g+\int_{\Sigma_{\boldsymbol{\sigma}}}   I^{ \sigma}_{x_0}( \nu_{{\bf z},{\bf m},({\bf 0},k_{\mathfrak{b}},k_{\mathfrak{b}})}) K_g \,{\rm dv}_g
\\
&-2\sum_{p=1,p\not =j_0,j_0'}^{n_\mathfrak{m}+\mathfrak{b}-1}\kappa(\xi^\#_p)\int_{\xi_p}k_g\dd \ell_g-2 \kappa(\xi_{j_0})\int_{\xi_{j_0}}k_g\dd \ell_g
\\
=&\int_{\Sigma}^{\rm reg}   I^{ \boldsymbol{\xi}^\#}_{x_0}(\nu_{{\bf z},{\bf m},({\bf k}_-,0,0)}) K_g \,{\rm dv}_g+\int_{\Sigma}^{\rm reg}   I^{\boldsymbol{ \sigma}}_{x_0}( \nu_{{\bf z},{\bf m},({\bf 0},k_{\mathfrak{b}},k_{\mathfrak{b}})}) K_g \,
{\rm dv}_g,
\end{align*}
where we have used in the last line the fact that the 1-form $ \nu_{{\bf z},{\bf m},({\bf 0},k_b,k_b)}$ on $\Sigma^{\#}$ possesses no trivial winding only around the cycle $\pl_{\mathfrak{b}-1}\Sigma$. Therefore the regularizing term in expression \eqref{def_reg_integtral} reduces to 
$$-2\chi_{\nu_{{\bf z},{\bf m},({\bf 0},k_{\mathfrak{b}},k_{\mathfrak{b}})}}( \zeta_{\mathfrak{b}-1})\int_{\xi_{j_0} }k_g\,\dd\ell_g=-2k_{\mathfrak{b}}\int_{\xi_{j_0} }k_g\,\dd\ell_g=2\kappa(\xi_{j_0})\int_{\xi_{j_0} }k_g\,\dd\ell_g.$$
\end{proof}

\subsection{Definition of the amplitudes}\label{sub:defamp}
We are now in position to define  the  amplitudes. In the case of closed Riemann surfaces, amplitudes basically correspond to the path integral  \eqref{defcorrelmixed}. In the case of Riemann surfaces with $\mathfrak{b}$ boundary components, the amplitudes will be a  functional of the boundary fields  
\[(\tilde{\varphi}_1^{k_1},\dots, \tilde{\varphi}_{\mathfrak{b}}^{k_{\mathfrak{b}}}):=(k_1R\theta+ \tilde{\varphi}_1, \dots,k_{\mathfrak{b}}R\theta+  \tilde{\varphi}_{\mathfrak{b}}) \in (H^s_\Z (\T))^{\mathfrak{b}}.\]
Below we shall identify a pair $(k,\tilde{\varphi})\in \Z\times H^{s}(\T)$ with the field 
$kR\theta+ \pi^*\tilde{\varphi}(\theta)\in H^s_\Z (\T)$ and we use the notation  
\[ ({\bf k},\tilde{\boldsymbol{\varphi}}):=(k_1, \dots, k_{\mathfrak{b}},\tilde{\varphi}_1, \dots, \tilde{\varphi}_{\mathfrak{b}})\in \Z^{\mathfrak{b}} \times (H^s (\T))^{\mathfrak{b}}, \quad \tilde{\boldsymbol{\varphi}}^{{\bf k}}:=(\tilde{\varphi}_1^{k_1},\dots,\tilde{\varphi}_{\mathfrak{b}}^{k_{\mathfrak{b}}})\in (H_\Z^s(\T))^{\mathfrak{b}}.\]
We notice that we recover the fields $\tilde{\boldsymbol{\varphi}}$ from $\tilde{\boldsymbol{\varphi}}^{{\bf k}}$ by setting ${\bf k}={\bf 0}=(0,\dots,0)$, i.e. 
\[\tilde{\boldsymbol{\varphi}}=\tilde{\boldsymbol{\varphi}}^{{\bf 0}}\in H^s(\T)^{\mathfrak{b}}.\]
Recall from Section \ref{subDTN} that $P\tilde{\boldsymbol{\varphi}}$ is the harmonic extension of $\tilde{\boldsymbol{\varphi}}\in H^{s}(\T)^{\mathfrak{b}}$. The definition of amplitudes will also involve the $1$-form $\nu_{\mathbf{z},\mathbf{m},{\bf k}} $ of Proposition \ref{harmpoles}. Recall that this 1-form is not in $L^2$, which is why we introduced its regularized norm  (see Lemma \ref{renorm_L^2}). By extension and for notational simplicity, we define the regularized norm of $\nu_{\mathbf{z},\mathbf{m},{\bf k}} +\omega$, where $\omega  $ is a smooth 1-form on $\Sigma$, by  
\begin{equation}
\|\nu_{\mathbf{z},\mathbf{m},{\bf k}} +\omega \|_{g,0}^2:=\|\nu_{\mathbf{z},\mathbf{m},{\bf k}} \|_{g,0}^2+2\langle \nu_{\mathbf{z},\mathbf{m},{\bf k}} ,\omega \rangle^2+\|\omega \|_{2}^2.
\end{equation}

 Finally we need to introduce  a space of reasonable test functions for our path integral in the case when $\Sigma$ has a non-empty boundary.   
  Equivariant maps $u\in H^s_{\Gamma}(\tilde\Sigma_{\bf z})$ can then be uniquely  decomposed as  
 \[u = \pi^*f +I_{x_0}(\omega^c_{\bf k}) + I_{x_0}(\nu_{\bf z,m,k})\]
 for some ${\bf k}^c\in\Z^{2\mathfrak{g}+\mathfrak{b}-1}$, ${\bf k}\in\Z^{\mathfrak{b}}$ and ${\bf m}\in \Z^{n_{\mathfrak{m}}}$, where $\pi: \tilde{\Sigma}_{\bf z}\to \Sigma$ is the projection from the universal cover of $\Sigma_{\bf z}=\Sigma\setminus \{{\bf z}\}$ to the base. We consider the space $\mc{E}^{\rm m}_R(\Sigma)$ of functionals $F$, defined on $H^s_{\Gamma}(\tilde\Sigma_{\bf z})$ for $s<-1$, of the form
\begin{equation}\label{polytrigo}
F(u)=  P(f) G( e^{\frac{i}{R} I_{x_0}(\omega^c_{\bf k})} ) G'( e^{\frac{i}{R} I_{x_0}(\nu_{\bf z,m,k})} ) 
\end{equation}
  where $P$ is a polynomial  of the form $P\big(\cjg f,g_1\cjd,\dots,\cjg f,g_{m_n}\cjd\big)$ where $g_1,\dots,g_{m_n}$  belong to $H^{-s}(\Sigma)$ (hence $s<0$), and $G,G'$ are  continuous and bounded on $C^0(\Sigma;\S^1)$.  Let $u^0_{\bf z}(x):=\sum_{j=1}^{n_{\mathfrak{m}}}i\alpha_jG_{g,D}(x,z_j)$. Next we define   the seminorm on $\mc{E}^{\rm m}_R(\Sigma)$
\[\|F\|_{\mc{L}^{\infty,p}_{\rm e,m}}:=\sup_{\bf k}\Big( \E\Big[e^{-\frac{1}{2\pi}\langle \dd X_{g,D},\nu_{\mathbf{z},\mathbf{m},{\bf k}} +\omega^{c}_{{\bf k}^c} \rangle_2-\frac{1}{4\pi}\|\dd f_{\bf k}\|_2^2}|F( X_{g,D}+P\tilde{\boldsymbol{\varphi}}+u^0_{\bf z}+I_{x_0}(\omega^c_{\bf k^c})+I_{x_0}( \nu_{\mathbf{z},\mathbf{m},{\bf k}}))|^p\Big] \Big)^{1/p}\]
where $(1-\Pi^c_1)\omega^{c}_{{\bf k}^c}=\dd f_{\bf k^c}$ with $\Pi^c_1$ is the projection on the space of harmonic $1$-forms with relative boundary condition.

With the same argument as in Lemma  \ref{invariance_of_norm}, we have:
\begin{lemma}\label{invnormm_bis_bis} The seminorm $\|F\|_{\mc{L}^{\infty,p}_{\rm e,m}}$ does not depend on the choice of cohomology basis 
 \end{lemma}

 \begin{definition}[\textbf{Amplitudes}]\label{def:amp}

\noindent {\bf (A)}
Let $\partial\Sigma=\emptyset$. We suppose that the condition \eqref{seiberg} holds.  For $F$ continuous  bounded function on $\mc{E}^{\rm m}_{R}(\Sigma)$ for some $s<0$, we define 
\begin{equation}\label{defampzerobound} 
\caA_{\Sigma,g,{\bf v},\boldsymbol{\alpha},{\bf m}}(F):= \langle F(\phi_g)  V^g_{(\boldsymbol{\alpha},{\bf m}) }({\bf v}) \rangle_{\Sigma,g}  
\end{equation}
using  \eqref{defcorrelmixed}.
 If $F=1$ then the amplitude is just the  correlation function  and will simply be denoted by $\caA_{\Sigma,g,{\bf v},\boldsymbol{\alpha},\boldsymbol{m}}$.
 \vskip 3mm
 
\noindent {\bf (B)}  If $\partial\Sigma$ has $\mathfrak{b}>0$ boundary components, consider a set of geometric data of $\Sigma$, as described above. Then the amplitude 
$\caA_{\Sigma,g,{\bf v},\boldsymbol{\alpha},\bf{m},\boldsymbol{\zeta}}$ is
  a function 
 \[(F, \tilde{\boldsymbol{\varphi}}^{ {\bf k}})\in \mc{E}^{\rm m}_R(\Sigma)\times (H^s_\Z (\R))^{\mathfrak{b}}  \mapsto\caA_{\Sigma,g,{\bf v},\boldsymbol{\alpha},\bf{m},\boldsymbol{\zeta}}(F,\tilde{\boldsymbol{\varphi}}^{ {\bf k}}),\] 
  that depends on a marked point $x_0$ in $\pl\Sigma$. It is  defined as follows 
\begin{align}\label{amplitude}
 \caA_{\Sigma,g,{\bf v},\boldsymbol{\alpha},\bf{m},\boldsymbol{\zeta}}&(F,\tilde{\boldsymbol{\varphi}}^{ {\bf k}}) \\
: =&
 \delta_0(\sum_{j=1}^{n_{\mathfrak{m}}+\mathfrak{b}} m_j({\bf k})) \lim_{t\to 1}\lim_{\eps\to 0}\sum_{{\bf k}^c\in \Z^{2\mathfrak{g}+\mathfrak{b}-1}} e^{-\frac{1}{4\pi}\|\nu_{\mathbf{z},\mathbf{m},{\bf k}} +\omega^{c}_{{\bf k}^c}\|_{g,0}^2 }Z_{\Sigma,g}\caA^0_{\Sigma,g}(\tilde{\boldsymbol{\varphi}}  ) 
 \nonumber\\
 &
 \E \Big[e^{-\frac{1}{2\pi}\langle \dd X_{g,D}+\dd P\tilde{\boldsymbol{\varphi}},\nu_{\mathbf{z},\mathbf{m},{\bf k}} +\omega^{c}_{{\bf k}^c}\rangle_2}F(\phi_g)\prod_{j=1}^{n_{\mathfrak{m}}} V_{\alpha_j,g,\epsilon}(x_j(t))e^{-\frac{i  Q}{4\pi}\int^{\rm reg}_\Sigma K_g\phi_g\dd {\rm v}_g
 -\mu M_\beta^g (\phi_g,\Sigma)}\Big]\nonumber
\end{align}
where $\delta_0$ is the Dirac mass at $0$, $t\in [0,1] \mapsto x_j(t)$ is a $C^1$-curve so that $(x_j(t),\dot{x}_j(t))\to (z_j,v_j)$ as $t\to 1$, the Liouville field is $\phi_g= X_{g,D}+P\tilde{\boldsymbol{\varphi}} +I^{\boldsymbol{\xi}}_{x_0}(\nu_{\mathbf{z},\mathbf{m},\mathbf{k}})+I^{\boldsymbol{\sigma}}_{x_0}(\omega^{c}_{{\bf k}^c})$, the expectation $\E$ is over the Dirichlet GFF $X_{g,D}$, $M_\beta^g (\phi_g,\Sigma)$ is defined as a limit in \eqref{GMCg}, $m_j({\bf k})$ is defined in \eqref{defm(k)},   
and $Z_{\Sigma,g}$ is the normalization constant 
 \begin{align}\label{znormal}
 Z_{\Sigma,g}:=\det (\Delta_{g,D})^{-\hf}.
 \end{align}
The regularized curvature term is
 \begin{equation}\label{defcurvamp}
 \int^{\rm reg}_\Sigma K_g\phi_g\dd {\rm v}_g:=\int_\Sigma K_g(X_{g,D}+P\tilde{\boldsymbol{\varphi}} )\dd {\rm v}_g+\int^{\rm reg}_\Sigma K_gI^{\boldsymbol{\sigma}}_{x_0}(\omega^{c}_{{\bf k}^c})\dd {\rm v}_g+\int^{\rm reg}_\Sigma K_gI^{\boldsymbol{\xi}}_{x_0}(\nu_{\mathbf{z},\mathbf{m},\mathbf{k}})\dd {\rm v}_g
 \end{equation}
and $\caA^0_{\Sigma,g}(\tilde{\boldsymbol{\varphi}})$ is the free field amplitude defined as  
\begin{align}\label{amplifree}
\caA^0_{\Sigma,g}(\tilde{\boldsymbol{\varphi}})=e^{-\frac{1}{2}\cjg\tilde{\boldsymbol{\varphi}}, (\mathbf{D}_\Sigma-\mathbf{D})  \tilde{\boldsymbol{\varphi}}\cjd}.
\end{align}
When $F=1$, we will simply write $\caA_{\Sigma,g,{\bf v},\boldsymbol{\alpha},\bf{m},\boldsymbol{\zeta}}( \tilde{\boldsymbol{\varphi}}^{\bf k})$.
 \end{definition}
 Let us make a couple of remarks on the definition of the amplitudes for what follows.
\begin{remark}\label{remarkamplitude}
(1) We first remark that if $\pl \Sigma\not=\emptyset$, the amplitude depends on the marked point $x_0$. When this dependance needs to be precised, we shall add the uperscript $x_0$ to the amplitude, i.e. we will write 
\[\caA^{x_0}_{\Sigma,g,{\bf v},\boldsymbol{\alpha},\bf{m},\boldsymbol{\zeta}}(F, \tilde{\boldsymbol{\varphi}}^{\bf k}).\]
(2) The amplitude does not depend on the choice of orientation of the arcs $d_j$, since such a change just amounts to make a resummation in the sum $\sum_{{\bf k}^c\in \Z^{2\mathfrak{g}+\mathfrak{b}-1}}$ appearing in the definition of the amplitude.
 \end{remark} 

We now state the main properties of the amplitudes as well as their gluing properties in Section \ref{sub:segal}. The proofs are postponed to Section \ref{proof_of_propositions}.
 
  \begin{proposition}\label{prop:symmetry}
  If $\pl\Sigma$  has $\mathfrak{b}>0$ connected components, then the amplitudes  satisfy:  
  \vskip1mm

\noindent $\bullet$ 
  The limit \eqref{amplitude} is well defined for $F\in\mc{E}^{\rm m}_R(\Sigma)$ in $\mu_0^{\otimes \mathfrak{b}}$-probability and belongs to $L^2(\mc{H}^{\otimes \mathfrak{b}})$.

 \vskip1mm
 \noindent $\bullet$ 
The amplitudes do not depend on the choice of the relative homology basis $\boldsymbol{\sigma}$, nor on the choice of the relative cohomology basis $(\omega_j^c)_j$ dual to $\boldsymbol{\sigma}$.

\vskip1mm

\noindent $\bullet$ The amplitudes do not depend on the choice of the 1-form $\nu_{\mathbf{z},\mathbf{m},{\bf k}} $ in the absolute cohomology, satisfying the conditions of Proposition \ref{harmpoles} and \eqref{btbaRz}.
\vskip1mm

\noindent $\bullet$  {\bf Conformal anomaly:} let $g,g'$ be two conformal admissible metrics on   $\Sigma$ with $g'=e^{\rho}g$ for some $\rho\in C^\infty(\Sigma)$, vanishing on $\pl\Sigma$. Then we have
\begin{align}\label{confanmixed} 
& {\caA_{\Sigma,g',{\bf v},\boldsymbol{\alpha},\bf{m},\boldsymbol{\zeta}}(F, \tilde{\boldsymbol{\varphi}}^{ {\bf k}})}
=  e^{\frac{{\bf c}}{96\pi}\int_{\Sigma}(|d\rho|_g^2+2K_g\rho) {\rm dv}_g-\sum_{j=1}^{n_{\mathfrak{e}}}\Delta_{(\alpha_j,m_j)}\rho(z_j)}
 {\caA_{\Sigma,g,{\bf v},\boldsymbol{\alpha},\bf{m},\boldsymbol{\zeta}}(F(\cdot- \tfrac{i  Q}{2}\rho),{\bf k},\tilde{\boldsymbol{\varphi}})  }
\end{align}
where the conformal weights $\Delta_{\alpha,m}$ are given by \eqref{deltaalphadef} and the central charge is ${\bf c}:=1-6 Q^2 $.
\vskip1mm

\noindent $\bullet$  {\bf Diffeomorphism invariance:} let $\psi:\Sigma'\to \Sigma$ be an orientation preserving diffeomorphism. Then  
$$
\caA^{x_0}_{\Sigma',\psi^*g,{\bf v},\boldsymbol{\alpha},\bf{m},\boldsymbol{\zeta}}(F, \tilde{\boldsymbol{\varphi}}^{ {\bf k}})=\caA^{\psi(x_0)}_{\Sigma,g,\psi_*{\bf v},\boldsymbol{\alpha},\bf{m},\psi\circ\boldsymbol{\zeta}}(F_\psi, \tilde{\boldsymbol{\varphi}}^{ {\bf k}}).
$$
where $F_\psi(\phi):=F(\phi\circ\psi)$.
\vskip1mm

\noindent $\bullet$  {\bf Spins:} with $r_{\boldsymbol{\theta}}\mathbf{v}:=(r_{\theta_1}v_1,\dots,r_{\theta_{n_\mathfrak{m}}} v_{\theta_{n_\mathfrak{m}}} )$, then  
\[
 \caA_{\Sigma, g,r_{\boldsymbol{\theta}}\mathbf{v} ,\boldsymbol{\alpha},\bf{m},\boldsymbol{\zeta}}(F, \tilde{\boldsymbol{\varphi}}^{ {\bf k}})
  =e^{-i QR\langle\mathbf{m},\boldsymbol{\theta}\rangle }  \caA_{\Sigma, g,r_{\boldsymbol{\theta}}\mathbf{v} ,\boldsymbol{\alpha},\bf{m},\boldsymbol{\zeta}}(F, \tilde{\boldsymbol{\varphi}}^{ {\bf k}}).
\]

\end{proposition}

 \subsection{Gluing of surfaces and amplitudes: statements of the results}\label{sub:segal}

Let us consider $(\Sigma_1,g_1, {\bf z}_1,\boldsymbol{\zeta}_1)$ and $(\Sigma_2,g_2, {\bf z}_2,\boldsymbol{\zeta}_2)$ two admissible surfaces with $\partial\Sigma_i\neq\emptyset$ for $i=1,2$. We enumerate  the boundary components of $\partial\Sigma_i$ as $\partial_j\Sigma_i$ for $j=1,\dots,\mathfrak{b}_i$.
Assume that  $\caC_1:=\partial_{\mathfrak{b}_1}\Sigma_1$ is  outgoing  and $\caC_{2}:=\partial_{\mathfrak{b}_2}\Sigma_2$ is incoming, and that $x_0^{i}\in \partial_{\mathfrak{b}_i}\Sigma_i$ for $i=1,2$. Then we can glue the two surfaces $(\Sigma_i,g_i)$, $i=1,2$ by identifying  $\mc{C}_1\sim \mc{C}_2$ using the parametrizations.
This forms an admissible surface denoted $(\Sigma,g, {\bf z},\boldsymbol{\zeta})$, 
 with $\mathfrak{b}=\mathfrak{b}_1+\mathfrak{b}_2-2$ boundary components.  
 We assume that $x_0^{1}=x_0^{2}$ on the glued surface.
 
 At the marked points $\mathbf{z}_i=(z_{i1},\dots,z_{in^i_{\mathfrak{m}}})$ on $\Sigma_i$ we choose unit vectors ${\bf v}_i=(v_{i1},\dots,v_{in^i_{\mathfrak{m}}})$ and we attach some  weights $\boldsymbol{\alpha}_i=(\alpha_{11},\dots,\alpha_{in^i_{\mathfrak{m}}})$ and  magnetic charges ${\bf m}_i=(m_{i1},\dots,m_{in^i_{\mathfrak{m}}})$. 
We use the notation 
\[\mathbf{z}:=(\mathbf{z}_1,\mathbf{z}_2), \quad \boldsymbol{\alpha}:=(\boldsymbol{\alpha}_1,\boldsymbol{\alpha}_2),\quad  {\bf m}:=({\bf m}_1,{\bf m}_2), \quad {\bf v}:=({\bf v}_1,{\bf v}_2)\] 
and denote by $\boldsymbol{\zeta}$ the collection of parametrizations of the boundaries $\partial_j\Sigma_1$ with $j=1,\dots,\mathfrak{b}_1-1$ and $\partial_j\Sigma_2$ with $j=1,\dots,\mathfrak{b}_2-1$.
The surface $\Sigma$ thus has an analytic boundary consisting of the curves $\partial_j\Sigma_1$ with $j=1,\dots,\mathfrak{b}_1-1$ and $\partial_j\Sigma_2$ with $j=1,\dots,\mathfrak{b}_2-1$. We denote by $\mc{C}=\mc{C}_{1}=\mc{C}_{2}$ the glued curve on $\Sigma$.
 Boundary conditions on $\Sigma$ will thus be written as couples $(\tilde{\boldsymbol{\varphi}}_1^{{\bf k}_1},\tilde{\boldsymbol{\varphi}}_2^{{\bf k}_2})\in   \mc{H}^{\mathfrak{b}_1-1}\times \mc{H}^{\mathfrak{b}_2-1}$. Similarly, for $i=1,2$, the surface $\Sigma_i$ has analytic boundary made up of the curves  $\partial_j\Sigma_i $ for $j=1,\dots,\mathfrak{b}_i-1$ and  $\caC_{i} $   and boundary conditions on $\Sigma_i$ will thus be written as couples $(\tilde{\boldsymbol{\varphi}}_i^{{\bf k}_i},\tilde{\varphi}^{k})\in   \mc{H}^{\mathfrak{b}_i-1}\times \mc{H}$.

 The corresponding amplitudes compose as follows: 

\begin{proposition}\label{glueampli}
Let $F_1,F_2$   respectively  in $\mc{E}^{\rm m}_R(\Sigma_1)$ and  $\mc{E}^{\rm m}_R(\Sigma_2)$  and let us denote by $F_1\otimes F_2$ the functional on $\mc{E}^{\rm m}_R(\Sigma,g)$ defined by  \[F_1\otimes F_2(\phi_g):=
F_1(\phi_{g}|_{\Sigma_1})F_2(\phi_{g}|_{\Sigma_2}).\] 
The following holds true
 \begin{align*}
&\caA_{\Sigma,g, {\bf v},\boldsymbol{\alpha},{\bf m},\boldsymbol{\zeta}}
(F_1\otimes F_2,\tilde{\boldsymbol{\varphi}}_1^{{\bf k}_1},\tilde{\boldsymbol{\varphi}}_2^{{\bf k}_2})\\
&\quad =C\int_{H^s_\Z(\R)} \caA_{\Sigma_1,g_1,{\bf v}_1,\boldsymbol{\alpha}_1,{\bf m}_1,\boldsymbol{\zeta}_1}(F_1, \tilde{\boldsymbol{\varphi}}_1^{{\bf k}_1},\tilde{ \varphi}^{k})\caA_{\Sigma_2,g_2,{\bf v}_2,\boldsymbol{\alpha}_2,{\bf m}_2,\boldsymbol{\zeta}_2}(F_2,\tilde{\boldsymbol{\varphi}}_2^{{\bf k}_2},\tilde{ \varphi}^k) \dd\mu_0  (\tilde{ \varphi}^k).
\end{align*}
 where $C= \frac{1}{(\sqrt{2} \pi) }$ if $\partial\Sigma \not =\emptyset$ and $C= \sqrt{2}  $ if $\partial\Sigma =\emptyset$.
  \end{proposition}

The proof of this Proposition will be done in the next sections.
We remark that, in the case of a curve disconnecting the surface, the summation over $k$ in the Hilbert space will always reduce to a Dirac mass at $k=0$.\\

The situation will be different for self-gluing (or non disconnecting curve), which we focus on now.
Let $(\Sigma,g, {\bf z},\boldsymbol{\zeta})$ be an admissible surface with $\mathfrak{b}$ boundary components such that the boundary contains an  outgoing boundary component  $\pl_{\mathfrak{b}-1} \Sigma\subset\partial\Sigma$ and an incoming boundary component $\pl_{\mathfrak{b}}\Sigma\subset\partial\Sigma$. We glue these two boundary components  to produce the surface denoted $(\Sigma^{\#} ,g, {\bf z},\boldsymbol{\zeta}_{\#})$. In this context, we will write the various field living on the boundary components of $\Sigma$ 
as  $(\tilde{\boldsymbol{\varphi}}_\#^{{\bf k}},\tilde\varphi_{\mathfrak{b}-1}^{k_{\mathfrak{b}-1}},\tilde\varphi_{\mathfrak{b}}^{k_{\mathfrak{b}}} ) \in      \mc{H}^{\mathfrak{b}-2} \times\mc{H} \times \mc{H}$ where $\tilde\varphi_{\mathfrak{b}-1}^{k_{\mathfrak{b}-1}}$ corresponds to  $\pl_{\mathfrak{b}-1} \Sigma$, $\tilde\varphi_{\mathfrak{b}}^{k_{\mathfrak{b}}}$ corresponds to  $\pl_{\mathfrak{b}} \Sigma$, and $\tilde{\boldsymbol{\varphi}}_\#^{{\bf k}}$ corresponds to the boundary components of $\Sigma^{\#}$. The location of the base point $x_0$ could be on any boundary component and we still denote by $g$ the glued metric on $\Sigma^{\#}$.

\begin{proposition}\label{selfglueampli}
For  $F\in \mc{E}^{m}_R(\Sigma)$, the following holds true
 \[
\caA_{\Sigma^{\#},g, {\bf v},\boldsymbol{\alpha},{\bf m},\boldsymbol{\zeta}_\#}
(F,\tilde{\boldsymbol{\varphi}}_\#^{{\bf k}})=C \int \caA_{\Sigma,g,{\bf v},\boldsymbol{\alpha},{\bf m},\boldsymbol{\zeta}}(F,  \tilde{\boldsymbol{\varphi}}_\#^{{\bf k}},\tilde\varphi^k , \tilde\varphi^k) \dd\mu_0(\tilde{\varphi}^k).
\]
where $C= \frac{1}{\sqrt{2} \pi}$ if $\partial\Sigma \not =\emptyset$ and $C= \sqrt{2}  $ if $\partial\Sigma =\emptyset$.
 \end{proposition}
The proof of this Proposition will be done in the next sections.\\

The gluing of amplitudes is a property that is valid for the compactified boson, i.e. the $\mathbb{T}$-valued free field. So we will first state a general lemma for the gluing of the compactified boson as a starting point. This will allow us to establish that Liouville amplitudes are in $L^2$ and we will later deduce the gluing for Liouville amplitudes. This will not be a straightforward consequence because of subtleties related to the curvature term.


 \subsection{Gluing for the compactified boson}\label{s:gluingboson}
  

We consider the setup drawn for amplitudes with the difference that we will consider amplitudes where $\boldsymbol{\alpha}=0$, $\mu=0$ and $Q=0$. This means that we consider
\begin{align}\label{amplitudegff}
 \caA^{0}_{\Sigma,g,{\bf z},\bf{m},\boldsymbol{\zeta}}&(F,\tilde{\boldsymbol{\varphi}}^{{\bf k}}) \\
: =&
 \delta_0(\sum_{j=1}^{n_{\mathfrak{m}}+\mathfrak{b}} m_j({\bf k}))  \sum_{{\bf k}^c\in \Z^{2\mathfrak{g}+\mathfrak{b}-1}} e^{-\frac{1}{4\pi}\|\nu_{\mathbf{z},\mathbf{m},{\bf k}} +\omega^{c}_{{\bf k}^c}\|_{g,0}^2 }Z_{\Sigma,g}\caA^0_{\Sigma,g}(\tilde{\boldsymbol{\varphi}}  ) 
 \E \Big[e^{-\frac{1}{2\pi}\langle \dd X_{g,D}+\dd P\tilde{\boldsymbol{\varphi}},\nu_{\mathbf{z},\mathbf{m},{\bf k}} +\omega^{c}_{{\bf k}^c}\rangle_2}F(\phi_g)  \Big]\nonumber
\end{align}
where expectation is taken with respect to the Dirichlet GFF and $F$ can be any  measurable positive functional, or any  measurable functional such that $ \caA^{0}_{\Sigma,g,{\bf v},\bf{m},\boldsymbol{\zeta}} (|F|,\tilde{\boldsymbol{\varphi}}^{{\bf k}})<\infty$, $\tilde{\boldsymbol{\varphi}}^{{\bf k}}$ almost surely (condition that we will shortcut as {\it integrable}). Such expression thus perfectly makes sense and we note that there is no dependence of this amplitude with respect to ${\bf v}$. We will further require $F$ to be $2\pi R$-periodic, meaning that $F(\phi_g+2\pi nR)=F(\phi_g)$ for all $n\in\Z$.

We first claim that these amplitudes glue as prescribed by Segal. Under the conditions stated just before Proposition \ref{glueampli}, we claim

\begin{proposition}\label{glueampli0}
Let $F_1,F_2$    be  periodic measurable positive or periodic integrable respectively  on $\Sigma_1$  and  $ \Sigma_2$  and let us denote by $F_1\otimes F_2$ the functional on the glued surface  $\Sigma$ defined by  \[F_1\otimes F_2(\phi_{g}):=F_1(\phi_{g_1}|_{\Sigma_1})F_2(\phi_{g}|_{\Sigma_2}).\] 
Then the following holds true
 \begin{align*}
&\caA^{0}_{\Sigma,g,{\bf z},{\bf m},\boldsymbol{\zeta}}
(F_1\otimes F_2,\tilde{\boldsymbol{\varphi}}_1^{{\bf k}_1},\tilde{\boldsymbol{\varphi}}_2^{{\bf k}_2})\\
&\quad =C\int_{H^s_\Z(\R)} \caA^{0}_{\Sigma_1,g_1,{\bf z}_1,{\bf m}_1,\boldsymbol{\zeta}_1}(F_1, \tilde{\boldsymbol{\varphi}}_1^{{\bf k}_1},\tilde{\varphi}^{k})\caA^{0}_{\Sigma_2,g_2,{\bf z}_2, {\bf m}_2,\boldsymbol{\zeta}_2}(F_2,\tilde{\boldsymbol{\varphi}}_2^{{\bf k}_2},\tilde{ \varphi}^k) \dd\mu_0  (\tilde{ \varphi}^{k}).
\end{align*}
 where $C= \frac{1}{(\sqrt{2} \pi) }$ if $\partial\Sigma \not =\emptyset$ and $C= \sqrt{2}  $ if $\partial\Sigma =\emptyset$.
  \end{proposition}

 \begin{proof}
We split the proof in two cases depending on whether the glued surface $\Sigma$ has a  trivial boundary (case 1) or not (case 2). 
We write $n_{\mathfrak{m}}=n_{\mathfrak{m}}^{1}+n_{\mathfrak{m}}^{2}$ and  $\mathfrak{g}$ for   $\mathfrak{g}_1 +\mathfrak{g}_2$. 
 
\medskip
{\bf 1) Assume first $\partial\Sigma  = \emptyset$}.   Let us call $\boldsymbol{\sigma}_{i}$ a canonical geometric basis of the relative homology on $\Sigma_i$ and note that this basis contains no boundary-to-boundary arcs, only interior cycles. 
Since the glued curve $\mc{C}$ is homologically trivial in $\Sigma_1$, $\Sigma_2$ or $\Sigma$,  the family  $\boldsymbol{\sigma}:=\boldsymbol{\sigma}_1\cup\boldsymbol{\sigma}_2$ forms a homology  basis on $\Sigma$ (see Figure \ref{fig:case1}).  Let  $\omega^{i,c}_1,\dots, \omega^{i,c}_{2{\mathfrak{g}}_i}$ be a cohomology basis of $\mc{H}^1(\Sigma_i,\pl \Sigma_i)$ dual to $\boldsymbol{\sigma}_i$ made of closed forms that are compactly supported inside $\Sigma^\circ_i$. Since they are compactly supported, all these forms can be obviously extended to $\Sigma$ by prescribing their value to be $0$ on $\Sigma\setminus\Sigma_i$. Then $\omega^{1,c}_1,\dots, \omega^{1,c}_{2{\mathfrak{g}}_1},\omega^{2,c}_1,\dots, \omega^{2,c}_{2{\mathfrak{g}}_2}$ is a basis of $\mc{H}^1(\Sigma)$ made of closed forms, dual to $\boldsymbol{\sigma}$. For ${\bf k}^c:=({\bf k}^{c}_1,{\bf k}^{c}_2)\in\Z^{2\mathfrak{g}_1}\times \Z^{2\mathfrak{g}_2}$, we set $\omega_{\bf k^c}:=\omega^{1,c}_{{\bf k}^{c}_1} +\omega^{2,c}_{{\bf k}^c_2} $.

 \begin{figure}[h] 
\centering
\includegraphics[width=.5\textwidth]{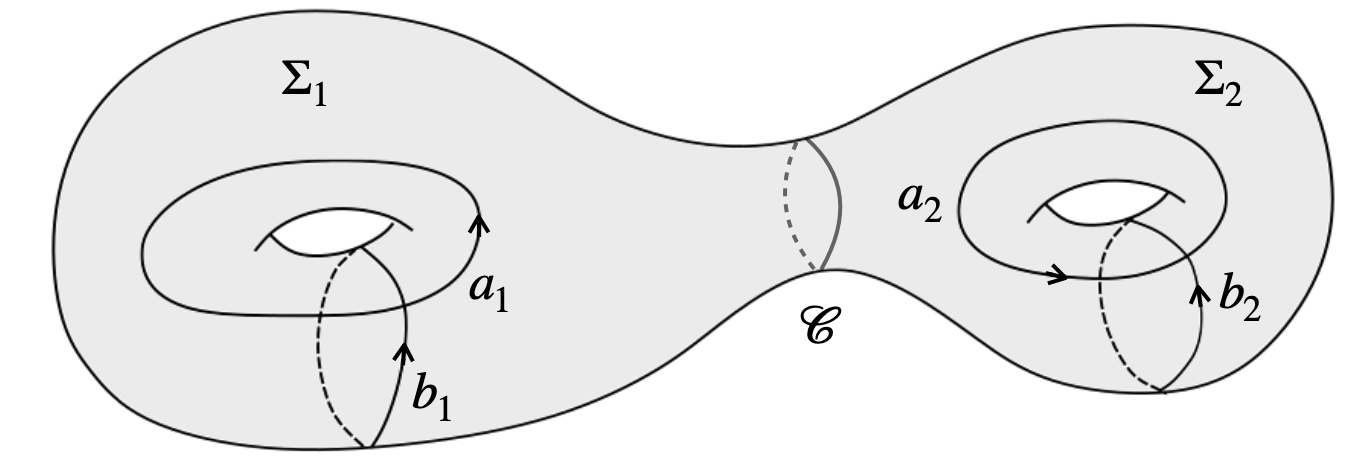} 
\caption{Case $\pl\Sigma=\emptyset$ with $\Sigma=\Sigma_1\#\Sigma_2$ cut along $\mc{C}$.}
\label{fig:case1}
\end{figure}

%
%
%

Next we consider the 1-forms $\nu_{{\bf z}_1,{\bf m}_1,k_1 }^1$, $\nu_{{\bf z}_2,{\bf m}_2,k_2 }^2$ respectively on $\Sigma_1,\Sigma_2$ given by Proposition \ref{harmpoles}. Notice that its proof (based on Lemma \ref{boundary_forms_absolute}) shows that we can choose the forms $\nu_{{\bf z}_i,{\bf m}_i,k_i}^i$ to be equal to 
$-\varsigma_i k_i R \dd \theta$ near $\pl \Sigma_i$, in the chart $\omega_j:U_j\to \mathbb{A}_\delta$ associated to $\pl \Sigma_i$. 
Note that, in order for the amplitudes on $\Sigma_1,\Sigma_2$ to be non zero, we must have  $k_1=\sum_{j=1}^{n^1_{\mathfrak{k}}}m_{1j}$ and $k_2=-\sum_{j=1}^{n^2_{\mathfrak{k}}}m_{2j}$. Furthermore, for the amplitude on $\Sigma$ to exist, we must have $\sum_{j=1}^{n^1_{\mathfrak{k}}}m_{1j}+\sum_{j=1}^{n^2_{\mathfrak{k}}}m_{2j}=0$. All in all, $k_1=k_2=\sum_{j=1}^{n^1_{\mathfrak{k}}}m_{1j}$. In particular $\nu_{{\bf z}_2,{\bf m}_2,k_2 }^2$ is a smooth extension of $\nu_{{\bf z}_1,{\bf m}_1,k_1 }^1$ (viewed as a form on $\Sigma\setminus \Sigma_2$)  to $\Sigma$, which means that we get a smooth closed 1-form $\nu_{{\bf z},{\bf m} }$ on $\Sigma$ with winding $m_{ij}$ around the point $z_{ij}$, for $i=1,2$ and $j=1,\dots n^i_{\mathfrak{m}}$. The defect graph $\mc{D}_{{\bf v}_1,\boldsymbol{\xi}_1}$ on $\Sigma_1$ is chosen so that only one arc has the boundary point $\zeta(1)$ as endpoint, and similarly for $\Sigma_2$. We get a defect graph $\mc{D}_{{\bf v},\boldsymbol{\xi}}$  by gluing the two defect graphs, i.e. we keep all the arcs in $\Sigma_1,\Sigma_2$ that do not have endpoint on the boundary and we form one arc out of the two arcs (one on $\Sigma_1$ and one on $\Sigma_2$) with an endpoint on the boundary, that we orient in the direction of increasing charges.

On $\Sigma$, the path integral  can be expressed, for all positive or integrable $F$, as 
\begin{align}\label{defPI1}
 \caA^{0}_{\Sigma,g, {\bf z},{\bf m},\boldsymbol{\zeta}}
(F)
   :=&
 \big(\frac{{\rm v}_{g}(\Sigma)}{{\det}'(\Delta_{g})}\big)^\hf\sum_{{\bf k}\in \Z^{2\mathfrak{g}}}e^{-\frac{1}{4\pi}\|\omega _{\bf k}+\nu_{\mathbf{z},\mathbf{m}}\|^2_{g,0} }
  \int_{\R/2\pi R\Z}\E\Big[e^{-\frac{1}{2\pi}\langle \dd X_g,\omega_{\bf k} +\nu_{{\bf z},{\bf m} }\rangle_2}F(\phi_g) \Big]\,\dd c 
\end{align}
where $\phi_g =c+X_g  +I^{\boldsymbol{\sigma}}_{x_0}(\omega_{\bf k})+ I^{\boldsymbol{\xi}}_{x_0}(\nu_{\bf z,m})$.
Let now $X_1$ and $X_2$ be two  independent Dirichlet GFF respectively on $\Sigma_1$ and $\Sigma_2$. We assume that they are both defined on $\Sigma$ by setting $X_i=0$ outside of $\Sigma_i$. Then we have the following decomposition in law (see Proposition \ref{decompGFF})
\begin{align*}
X_g\stackrel{\rm law}=X_1+X_2+P{\bf X}-c_g
\end{align*}
where ${\bf X}$ is the restriction of the GFF $X_g$ to the  glued boundary component  $\mathcal{C}$ expressed in parametrized coordinates, i.e. ${\bf X} = X_{g|_{{\caC}}}\circ \zeta_{\mc{C}}  $ with    $\zeta_{\mc{C}} $   the parametrization of $\mc{C}$,   $P{\bf X}$ is its harmonic extension to $ \Sigma $ and $c_g:=\frac{1}{{\rm v}_g(\Sigma)}\int_\Sigma (X_1+X_2+P{\bf X})\,\dd {\rm v}_g$. We can then plug this decomposition into the path integral \eqref{defPI1} and then shift the $c$-integral by $c_g$ (i.e. $c\mapsto c+c_g$)
to get
\begin{multline} \label{productE}
 \caA^{0}_{\Sigma,g, {\bf z},{\bf m},\boldsymbol{\zeta}}
(F_1\otimes F_2)
   :=
 \big(\frac{{\rm v}_{g}(\Sigma)}{{\det}'(\Delta_{g})}\big)^\hf\sum_{{\bf k}\in \Z^{2\mathfrak{g}}}e^{-\frac{1}{4\pi}\|\omega _{\bf k}+\nu_{\mathbf{z},\mathbf{m}}\|^2_{g,0} }\\ \int_{\R/2\pi R\Z} \E[\mc{B}_1(c,{\bf X}, \omega^{1,c}_{{\bf k}_1^c}+\nu^1_{{\bf z}_1,{\bf m}_1,k_1})\mc{B}_2(c,{\bf X}, \omega^{2,c}_{{\bf k}_2^{c}}+\nu^2_{{\bf z}_2,{\bf m}_2,k_2})]\,\dd c
\end{multline}
where we have set  
$$\mc{B}_i(c,\varphi, \omega)=\E\Big[e^{-\frac{1}{2\pi}\langle \dd X_i+\dd P\varphi,\omega \rangle_2}F_i(\phi_i ) \Big]$$
with $\phi_i:=c+X_i+P\varphi+I^i_{x_0}(\omega)$ and expectation is over the GFF $X_i$.
Here we have used a shortcut notation $I^i_{x_0}(\omega)$: it means that for $\omega=\omega^{i,c}_{{\bf k}_i^{c}}+\nu^i_{{\bf z}_i,{\bf m}_i,k_i}$ we have $I^i_{x_0}(\omega)=I_{x_0}^{\boldsymbol{\sigma}_i}(\omega^{i,c}_{{\bf k}_i^{c}})+I^{\boldsymbol{\xi}_i}_{x_0}(\nu^i_{{\bf z}_i,{\bf m}_i,k_i})$.  

Now we make a further shift in the $c$-variable in the expression above to subtract the mean $m_{\mc{C}}({\bf X}):=\frac{1}{2\pi}\int_0^{2\pi}{\bf X}(e^{i\theta})\,\dd \theta$ to the field ${\bf X}$. As a consequence we can replace the law $ \P_{{\bf X}} $ of ${\bf X}$ in \eqref{productE} (expectation is there w.r.t $ \P_{{\bf X}} $) by the law $  \P_{{\bf X}-m_{\mc{C}}({\bf X})}$ of the recentered field ${\bf X}-m_{\mc{C}}({\bf X})$ so that we end up with (using the description of the law of $X-m_{\mc{C}}({\bf X})$ proved in \cite[eq (5.14)]{GKRV2} together with the  computation of determinants \cite[eq (5.15)]{GKRV2}) 
\[\begin{split} 
\caA^{0}_{\Sigma,g, {\bf z},{\bf m},\boldsymbol{\zeta}}(F)
  = & Z_{\Sigma_1,g_1}Z_{\Sigma_2,g_2}
 \sum_{{\bf k}\in \Z^{2\mathfrak{g}}}e^{-\frac{1}{4\pi}\|\omega _{\bf k}+\nu_{\mathbf{z},\mathbf{m}}\|^2_{g,0} }\\
&  \int_{\R/2\pi R\Z}\int 
\mc{B}_1(c,\varphi, \omega^{1,c}_{{\bf k}_1^{c}})\mc{B}_2(c,\varphi, \omega^{2,c}_{{\bf k}_2^{c}})
e^{-\frac{1}{2}\cjg \varphi,\widetilde{\mathbf{D}}_{\Sigma,{\mc{C}}}   \varphi\cjd}   \dd \P_\T   (  \varphi) \,\dd c.
\end{split}\]
Here we recall that $Z_{\Sigma_i,g_i}$ and $\widetilde{\mathbf{D}}_{\Sigma,{\mc{C}}}$ were defined in \eqref{znormal} and \eqref{tildeD}. 
Next we observe that
$$\|\omega_{\bf k}+\nu_{\mathbf{z},\mathbf{m}}\|_{g,0}^2=\|\omega^{1,c}_{{\bf k}_1^{c}}+\nu^1_{\mathbf{z}_1,\mathbf{m}_1,k_1}\|_{g_1,0}^2+\|\omega^{2,c}_{{\bf k}_2^{c}}+\nu^2_{\mathbf{z}_2,\mathbf{m}_2,k_2}\|_{g_2,0}^2$$
and   that $\exp(-\frac{1}{2}\cjg\varphi,\widetilde{\mathbf{D}}_{\Sigma,{\mc{C}}}   \varphi\cjd)=\caA^0_{\Sigma_1,g_1}(\tilde{\boldsymbol{   \varphi}})\caA^0_{\Sigma_2,g_2}(\tilde{\boldsymbol{  \varphi}} ) $ (whatever the value of $c$ is, since $\mathbf{D}_{\Sigma,\mc{C}}1={\bf D}1=0$). Also, note that summation over $k$ in the measure $ \mu_0$ reduces to $k=k_1=k_2=\sum_{j=1}^{n^1_{\mathfrak{m}}}m_{1j}$. This completes the proof of the first case.

\medskip
{\bf 2) Assume now $\partial\Sigma  \not= \emptyset$.} In that case, $\Sigma_1$ or $\Sigma_2$ has at least $2$ boundary connected components and we will assume (even if it means re-labelling the surfaces) that this will be the case for $\Sigma_1$. Let us take 
\[\boldsymbol{\sigma}_i=(a_{i1},b_{i1}, \dots,a_{i{\mathfrak{g}}_i},b_{i{\mathfrak{g}}_i}, d_{i1},\dots,d_{i(\mathfrak{b}_i-1)})\] 
a canonical geometric basis of  $\mc{H}_1(\Sigma_i,\pl \Sigma_i)$ (see Figure \ref{fig:partition}). Let  $\omega^{i,c}_1,\dots, \omega^{i,c}_{2{\mathfrak{g}}_i+\mathfrak{b}_i-1}$ be a basis of $\mc{H}^1_R(\Sigma_i,\pl \Sigma_i)$ dual to $\boldsymbol{\sigma}_i$ made of closed forms that are compactly supported inside $\Sigma^\circ_i$. Since they are compactly supported, all these forms can be obviously extended to $\Sigma$ by prescribing their value to be $0$ on $\Sigma\setminus\Sigma_i$. 

$\bullet$ We first consider the case where $\Sigma_2$ has  $\mathfrak{b}_2\geq 2$ boundary components. By Lemma \ref{baseglue}, we get a basis of the relative homology $\mc{H}_1(\Sigma,\pl\Sigma)$
\[\boldsymbol{\sigma}=\boldsymbol{\sigma}_1\# \boldsymbol{\sigma}_2\]  
by gathering the curves for $i=1,2$: $a_{ij},b_{ij}$ for $j=1,\dots,\mathfrak{g}_i$, $d_{ij}$  for $j=1,\dots,\mathfrak{b}_i-2$,  and finally  the curve $d_{1(\mathfrak{b}_1-1)}-d_{2(\mathfrak{b}_2-1)}$ (see Figure \ref{fig:partition}). Then $\omega^{1,c}_1,\dots, \omega^{1,c}_{2{\mathfrak{g}}_1+\mathfrak{b}_1-1},\omega^{2,c}_1,\dots, \omega^{2,c}_{2{\mathfrak{g}}_2+\mathfrak{b}_2-2} $ is a basis of $\mc{H}^1_R(\Sigma,\partial\Sigma)$ made of closed forms, dual to $\boldsymbol{\sigma}$ and compactly supported. For ${\bf k}^c:=({\bf k}_1^{c},{\bf k}_2^{c})\in\Z^{2\mathfrak{g}_1+\mathfrak{b}_1-1}\times \Z^{2\mathfrak{g}_2+\mathfrak{b}_2-2} $, we set $\omega^c_{\bf k^c}:=\omega^{1,c}_{{\bf k}_1^c} +\omega^{2,c}_{({\bf k}_2^{c},0)} $.

\begin{figure}[h]
     \begin{subfigure}[b]{0.5\textwidth}
         \includegraphics[width=1\textwidth]{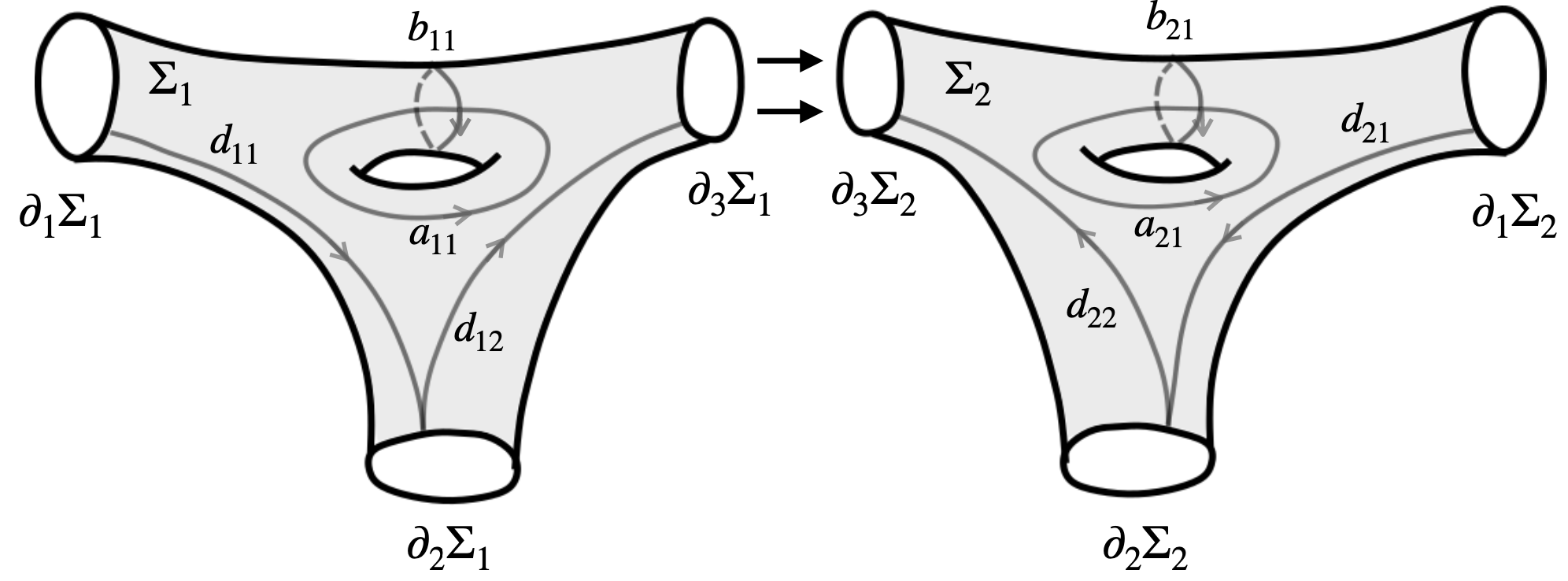}
         \caption{Disconnected surface. In gray, the relative homology on both surfaces.}
     \end{subfigure}
     \hfill
     \begin{subfigure}[b]{0.46\textwidth}
         \centering
         \includegraphics[width=1\textwidth]{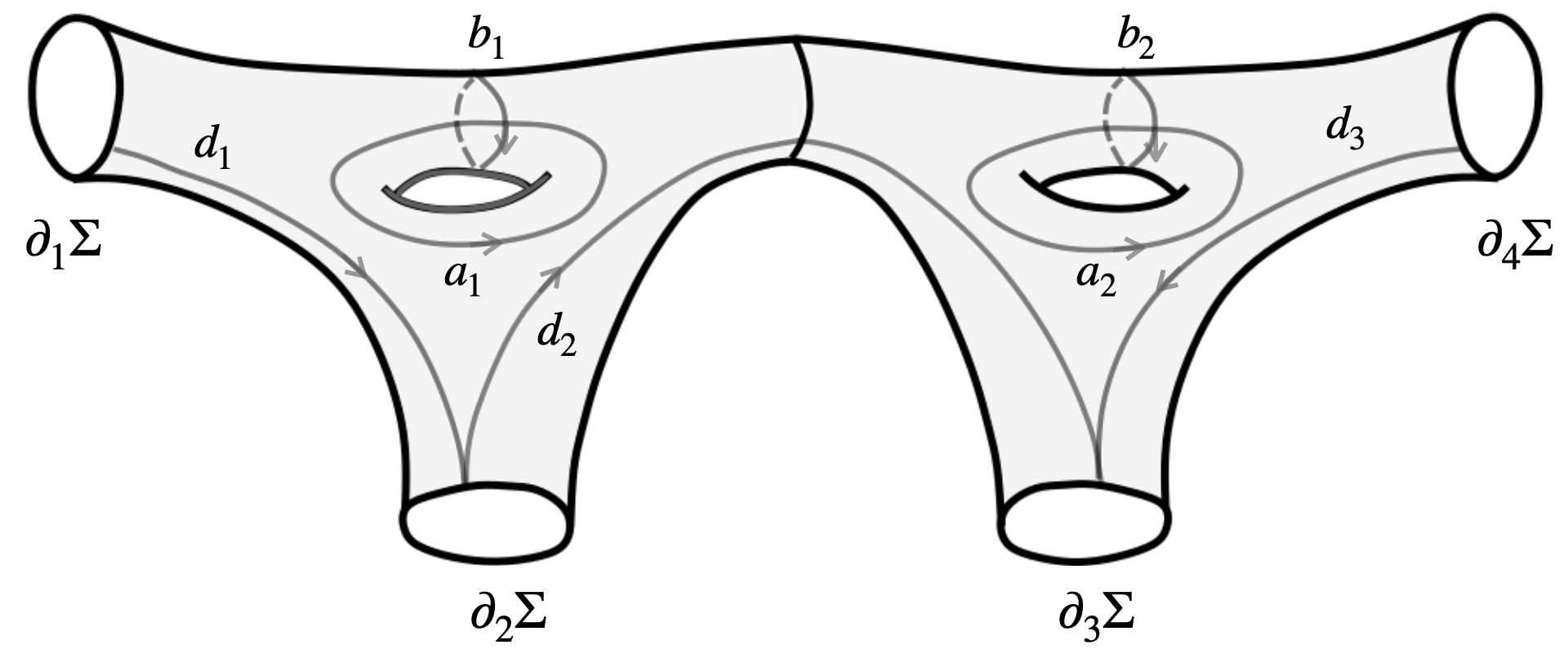}
         \caption{Glued surface. In gray, the relative homology built from those on $\Sigma_1$ and $\Sigma_2$.}
     \end{subfigure} 
        \begin{subfigure}[b]{0.6\textwidth}
\includegraphics[width=1\textwidth]{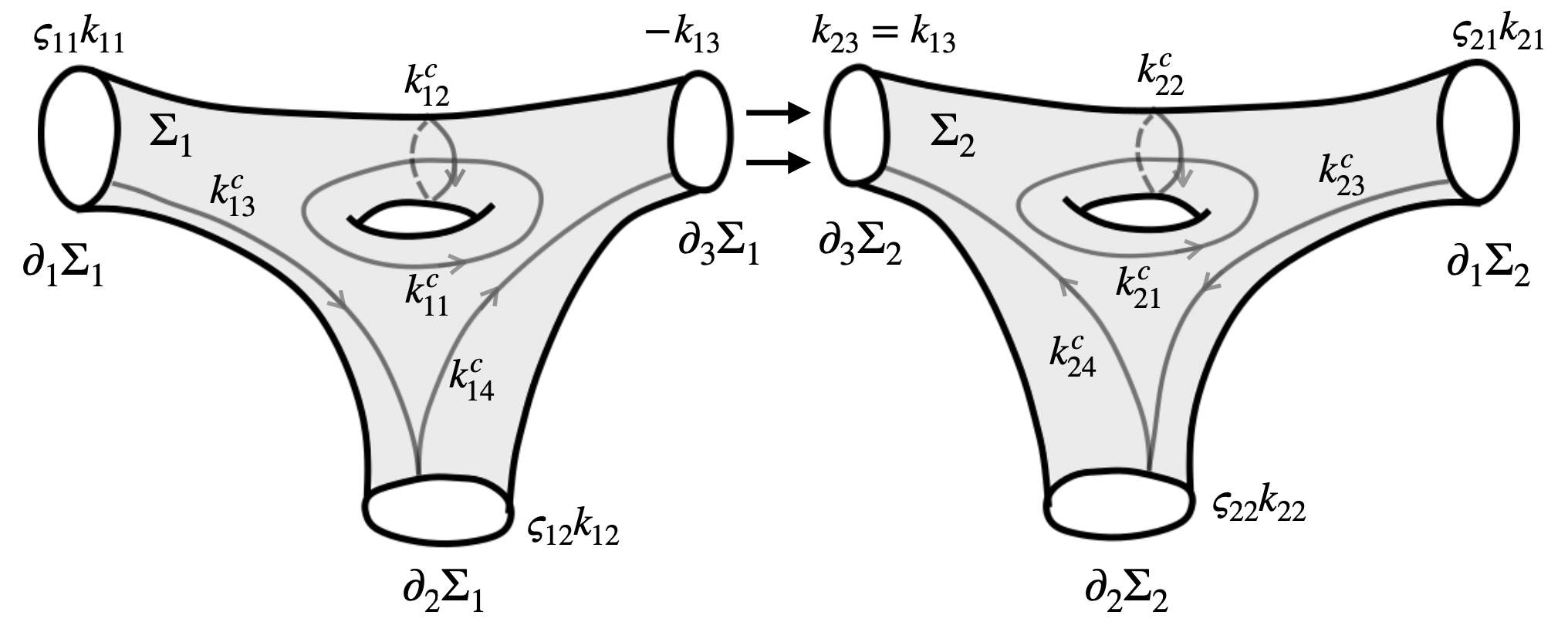}
\caption{Assignment of values of $k$-parameters to each homology curve.}
\end{subfigure}
  \caption{Case $\pl\Sigma\not=\emptyset$.}
 \label{fig:partition}\end{figure}




Now we focus on the absolute cohomology and magnetic $1$-forms. Notice first that in case $\sum_{j=1}^{\mathfrak{b}_1-1} \varsigma_{1j}k_{1j}+\sum_{j=1}^{\mathfrak{b}_2-1} \varsigma_{2j}k_{2j}+\sum_{j=1}^{n^1_\mathfrak{m}}m_{1j}+\sum_{j=1}^{n^2_\mathfrak{m}}m_{2j}\not=0$, both sides in the gluing statement vanish (because of $\delta$-masses in the definition of amplitudes) so that equality obviously  holds. So the case of interest is 
\begin{equation}\label{neutral}
\sum_{j=1}^{\mathfrak{b}_1-1} \varsigma_{1j}k_{1j}+\sum_{j=1}^{\mathfrak{b}_2-1} \varsigma_{2j}k_{2j}+\sum_{j=1}^{n^1_\mathfrak{m}}m_{1j}+\sum_{j=1}^{n^2_\mathfrak{m}}m_{2j}=0,
\end{equation} 
in which case the only contributing term in the summation over $k$ in the Hilbert space comes from the case when $\sum_{j=1}^{n^1_\mathfrak{m}}m_{1j}+\sum_{j=1}^{\mathfrak{b}_1} \varsigma_{1j}k_{1j}=0$, $\sum_{j=1}^{n^2_\mathfrak{m}}m_{2j}+\sum_{j=1}^{\mathfrak{b}_2} \varsigma_{2j}k_{2j}=0$ and \eqref{neutral}, which implies in particular $k_{1\mathfrak{b}_1}=k_{2\mathfrak{b}_2}$. So we fix $k_{1\mathfrak{b}_1}=k_{2\mathfrak{b}_2}$ such that these conditions are fulfilled. Under these conditions, the forms $\nu^{i}_{{\bf z}_i,{\bf m}_i,({\bf k}_i,k_{i\mathfrak{b}_i})}$ (for $i=1,2$) satisfy 
\begin{align}\label{3mag}
&\int_{\partial_j\Sigma_i}\nu^{i}_{{\bf z}_i,{\bf m}_i,({\bf k}_i,k_{i\mathfrak{b}_i})}=  2\pi R\varsigma_{ij}k_{ij},\quad\text{ for }j=1,\dots, \mathfrak{b}_i,
\\
&\int_{\partial_{b_1}\Sigma_1}\nu^{1}_{{\bf z}_1,{\bf m}_1,({\bf k}_1,k_{1\mathfrak{b}_1})} =-\int_{\partial_{\mathfrak{b}_2}\Sigma_2}\nu^{2}_{{\bf z}_2,{\bf m}_2,({\bf k}_2,k_{2b_2})} \nonumber\\
&\int_{\pl \mc{D}_{ij}}\nu^{i}_{{\bf z}_i,{\bf m}_i,({\bf k}_i,k_{i\mathfrak{b}_i})}= 2\pi R m_{ij}\nonumber
\end{align}
where $\mc{D}_{ij}$ is a small containing $z_{ij}$ (and of course all integrals along interior cycles vanish).
The second condition (and the fact that both forms involved are of the form $ \varsigma_{i\mathfrak{b}_i}k_{i\mathfrak{b}_i}\,R\dd\theta$ over a neighborhood of the boundary $\partial_{\mathfrak{b}_i}\Sigma_i$, see Proposition \ref{harmpoles}) allows us to glue together these forms to get a closed 1-form on $\Sigma$ by the relation 
$\nu_{{\bf z},{\bf m},{\bf k}}=\nu^{1}_{{\bf z}_1,{\bf m}_1,({\bf k}_1,k_{1\mathfrak{b}_1})} \mathbf{1}_{\Sigma_1}+\nu^{2}_{{\bf z}_2,{\bf m}_2,({\bf k}_2,k_{2\mathfrak{b}_2})} \mathbf{1}_{\Sigma_2}$, 
for ${\bf k}:=({\bf k}_1,{\bf k}_2)\in\Z^{\mathfrak{b}_1-1}\times \Z^{\mathfrak{b}_2-1}$, satisfying our basic conditions. Finally we consider two defect graphs on $\Sigma_1$ and $\Sigma_2$ as in Lemma \ref{magcurvtwo},   and glue them to get a defect graph on $\Sigma$. The amplitude on $\Sigma$ then reads
\begin{align*} 
 & \caA^{0}_{\Sigma,g,{\bf z},{\bf m},\boldsymbol{\zeta}}(F_1\otimes F_2,\tilde{\boldsymbol{\varphi}}_1^{{\bf k}_1},\tilde{\boldsymbol{\varphi}}_2^{{\bf k}_2}) \\
 &=
 \sum_{{\bf k}^c\in \Z^{2\mathfrak{g}+\mathfrak{b}_1+\mathfrak{b}_2-3}} e^{-\frac{1}{4\pi}\|\nu_{{\bf z},{\bf m},{\bf k}}+\omega^{c}_{{\bf k}^c}\|_{g,0}^2} Z_{\Sigma,g}\caA^0_{\Sigma,g}(\tilde{\boldsymbol{\varphi}}_1 ,\tilde{\boldsymbol{\varphi}}_2 )\mathcal{B}_{\Sigma,g}(F_1\otimes F_2,\tilde{\boldsymbol{\varphi}}_1,\tilde{\boldsymbol{\varphi}}_2,{\bf k},{\bf k}^c) 
\end{align*} 
 with
 \begin{align} 
 \mathcal{B}_{\Sigma,g}(F_1\otimes F_2,\tilde{\boldsymbol{\varphi}}_1,\tilde{\boldsymbol{\varphi}}_2,{\bf k},{\bf k}^c):=&
 \E \Big[e^{-\frac{1}{2\pi}\langle \dd X_{g,D}+\dd P\tilde{\boldsymbol{\varphi}},\nu_{{\bf z},{\bf m},{\bf k}}+\omega^{c}_{{\bf k}^c}\rangle}F_1\otimes F_2(\phi_g)  \Big]\nonumber
\end{align}
where the Liouville field is $\phi_g= X_{g,D}+P(\tilde{\boldsymbol{\varphi}}_1,\tilde{\boldsymbol{\varphi}}_2)+ I_{x_0}^{\boldsymbol{\sigma}}(\nu_{{\bf z},{\bf m},{\bf k}})+I_{x_0}^{\boldsymbol{\xi}}(\omega^c_{{\bf k}^c})$, the expectation $\E$ is over the Dirichlet GFF $X_{g,D}$ on $\Sigma$, $P(\tilde{\boldsymbol{\varphi}}_1,\tilde{\boldsymbol{\varphi}}_2)$ stands for the harmonic extension to $\Sigma$ of the boundary fields $\tilde{\boldsymbol{\varphi}}_1,\tilde{\boldsymbol{\varphi}}_2$, which stand respectively for the  boundary conditions on the remaining (i.e. unglued) components of $\partial \Sigma_1$ and  $\partial \Sigma_2$, namely
\[\Delta_g P(\tilde{\boldsymbol{\varphi}}_1,\tilde{\boldsymbol{\varphi}}_2)=0\quad  \text{on }\Sigma ,\qquad   P(\tilde{\boldsymbol{\varphi}}_1,\tilde{\boldsymbol{\varphi}}_2)_{|\partial_j\Sigma_i}= \tilde{\varphi}_{ij}\circ \zeta_{ij}^{-1} \text{   for }j<\mathfrak{b}_i 
.\]
Let now $X_1:=X_{g_1,D}$ and $X_2:=X_{g_2,D}$ be two  independent Dirichlet GFF respectively on $\Sigma_1$ and $\Sigma_2$. Then we have the following decomposition in law (see Proposition \ref{decompGFF})
\begin{align*}
X_{g,D}\stackrel{{\rm law}}=X_1+X_2+P{\bf Y}
\end{align*}
where ${\bf Y}$ is the restriction of $X_{g,D}$ to the  glued boundary component  $\mathcal{C}$ expressed in parametrized coordinates, i.e. ${\bf Y} =  X_{g,D|_{{\caC}}}\circ \zeta  $, and  $P{\bf Y}$ is its harmonic extension to $ \Sigma $ vanishing on $\partial \Sigma $, which is non empty. We stress that, since  $X_{g,D}$ is only a distribution,  making sense of ${\bf Y}$ is not completely obvious but, using the parametrization, this can be done in the same way as making sense of the restriction of the GFF to a circle: since this is a standard argument, we do not elaborate more on this point. Finally we denote by $h_\mathcal{C}$ the restriction of  the harmonic function $P(\tilde{\boldsymbol{\varphi}}_1,\tilde{\boldsymbol{\varphi}}_2)$ to $\mathcal{C}$ in parametrized coordinates
\[ h_{\mc{C}}: =P(\tilde{\boldsymbol{\varphi}}_1,\tilde{\boldsymbol{\varphi}}_2)_{|\mc{C} }\circ\zeta  .\]
Observe now the trivial fact that, on $\Sigma_i$ ($i=1,2$), the function $P(\tilde{\boldsymbol{\varphi}}_1,\tilde{\boldsymbol{\varphi}}_2)+P{\bf Y}$ is harmonic with boundary values (expressed in parametrized coordinates  on $\Sigma_i$) $ \tilde\varphi_{ij}$ on $\partial_j\Sigma_i$ for $j<\mathfrak{b}_i$ and $({\bf Y}+h_{\mathcal{C}})$ on $\mc{C}$.  Thus we get, using Lemma  \ref{magcurvtwo}
\begin{align*}
& \mathcal{B}_{\Sigma,g}(F_1\otimes F_2,\tilde{\boldsymbol{\varphi}}_1,\tilde{\boldsymbol{\varphi}}_2,{\bf k},{\bf k}^c)\\
 &=   \int \mathcal{B}_ {\Sigma_1,g_1}(F_1,\tilde{\boldsymbol{\varphi}}_1,\tilde{ \varphi} +h_{\mc{C}},{\bf k}_1 ,k_{1\mathfrak{b}_1},{\bf k}^c_1 )\mathcal{B}_ {\Sigma_2,g_2}(F_2, \tilde{\boldsymbol{\varphi}}_2, \tilde{\varphi}   +h_{\mc{C}},{\bf k}_2 ,k_{2\mathfrak{b}_2},{\bf k}^c_2 ) \P_{{\bf Y}} (\dd \tilde{ \varphi})
\end{align*}
with $\P_{{\bf Y}}$ the law of $ {\bf Y} $ and, for $i=1,2$, 
\begin{equation}
\mathcal{B}_{\Sigma_i,g_i}(F_i,\tilde{\boldsymbol{\varphi}}_i,\tilde{ \varphi},{\bf k}_i ,k_{i\mathfrak{b}_i} ,{\bf k}_i^{c}):=\E \big[F_i(\phi_i)e^{-\frac{1}{2\pi}\langle \dd X_i+\dd P( \tilde{\boldsymbol{\varphi}}_i,\tilde{ \varphi} ),\nu^{i}_{{\bf z}_i,{\bf m}_i,({\bf k}_i,k{i\mathfrak{b}_i})} +\omega^{i,c}_{{\bf k}^c_i}\rangle}  \big]
\end{equation}
where $\E$ is taken with respect to the Dirichlet GFF  $X_i$  on $\Sigma_i$, 
\[\phi_i=X_i+P( \tilde{\boldsymbol{\varphi}}_i,\tilde{ \varphi} )+I^{\boldsymbol{\sigma}_i}_{x_0} (\omega^{i,c}_{{\bf k}^c_i})
+I^{\boldsymbol{\xi}_i}_{x_0}(\nu^{i}_{{\bf z}_i,{\bf m}_i,({\bf k}_i,k_{i\mathfrak{b}_i})} )\]
(here with an abuse of notations we identify ${\bf k}^c_2$ with $({\bf k}^c_2,0)$) and $P( \tilde{\boldsymbol{\varphi}}_i,\tilde{ \varphi} )$ stands for the harmonic extension on $\Sigma_i$ of the boundary fields $  \tilde{\boldsymbol{\varphi}}_i,\tilde{ \varphi} $ respectively on   $\partial\Sigma_i\setminus \mathcal{C}$ and $\mathcal{C}$. 

Now we use Lemma \cite[Lemmas 5.3 \&  5.4]{GKRV2} to get
\[
\P_{{\bf Y}+h_{\mc{C}}}(\dd \tilde{ \varphi})
:=
\frac{Z_{\Sigma_1,g_1}Z_{\Sigma_2,g_2}}{\sqrt{2}\pi Z_{\Sigma,g}}   \frac{\caA^0_{\Sigma_1,g_1}(\tilde{\boldsymbol{\varphi}}_1 ,c+ \varphi  )\caA^0_{\Sigma_2,g_2}(\tilde{\boldsymbol{\varphi}}_2,c+ \varphi ) }{\caA^0_{\Sigma,g}(\tilde{\boldsymbol{\varphi}}_1 ,\tilde{\boldsymbol{\varphi}}_2 )}  \dd c\otimes \P_\T(\dd  {\varphi}).
\]
 Thus we get (the expectation $\E$ is taken with respect to $\varphi$)
 \begin{align*}
 & \caA^{0}_{\Sigma,g,{\bf z},{\bf m},\boldsymbol{\zeta}}(F_1\otimes F_2,\tilde{\boldsymbol{\varphi}}_1^{{\bf k}_1},\tilde{\boldsymbol{\varphi}}_2^{{\bf k}_2})\\
 =&
 \sum_{{\bf k}^c\in \Z^{2\mathfrak{g}+\mathfrak{b}-1}} e^{-\frac{1}{4\pi}\|\nu_{{\bf z},{\bf m},{\bf k}}+\omega^{c}_{{\bf k}^c}\|_{g,0}^2}Z_{\Sigma_1,g_1}Z_{\Sigma_2,g_2} \int_\R \E\Big[\caA^0_{\Sigma_1,g_1}(\tilde{\boldsymbol{\varphi}}_1 ,c+ \varphi  )\caA^0_{\Sigma_2,g_2}(\tilde{\boldsymbol{\varphi}}_2,c+ \varphi )
 \\
 &\times  \mathcal{B}_ {\Sigma_1,g_1}(F_1,\tilde{\boldsymbol{\varphi}}_1,c+ \varphi   ,{\bf k}_1 ,k_{\mathfrak{b}_1},{\bf k}^c_1 )\mathcal{B}_ {\Sigma_2,g_2}(F_2, \tilde{\boldsymbol{\varphi}}_2, c+ \varphi     ,{\bf k}_2 ,k_{\mathfrak{b}_2} ,{\bf k}^c_2)\Big]\,\dd c 
 \\
 =&
 \sum_{{\bf k}^c\in \Z^{2\mathfrak{g}+\mathfrak{b}-1}}\sum_{n\in\Z}  e^{-\frac{1}{4\pi}\|\nu_{{\bf z},{\bf m},{\bf k}}+\omega^{c}_{{\bf k}^c}\|_{g,0}^2}Z_{\Sigma_1,g_1}Z_{\Sigma_2,g_2} \int_0^{2\pi R}\E\Big[\caA^0_{\Sigma_1,g_1}(\tilde{\boldsymbol{\varphi}}_1,c+ \varphi +n2\pi R   )\caA^0_{\Sigma_2,g_2}(\tilde{\boldsymbol{\varphi}}_2,c+ \varphi +n2\pi R   )
 \\
 &\times  \mathcal{B}_ {\Sigma_1,g_1}(F_1,\tilde{\boldsymbol{\varphi}}_1,c+ \varphi +n2\pi R  ,{\bf k}_1 ,k_{1\mathfrak{b}_1},{\bf k}_1^c )\mathcal{B}_ {\Sigma_2,g_2}(F_2, \tilde{\boldsymbol{\varphi}}_2, c+ \varphi +n2\pi R   ,{\bf k}_2 ,k_{2\mathfrak{b}_2} ,{\bf k}^c_2)\Big]\,\dd c .
 \end{align*}
The last relation was obtained using the Chasles relation on the $c$-integral. Next, we introduce the harmonic functions $P^i_{n}$, for $i=1,2$ and $n\in\Z$, that are harmonic on $\Sigma_i$ with boundary values $n2\pi R$ on $\mc{C}$ and $0$ on the other boundary components of $\Sigma_i$. Then we notice that, writing $\tilde{\varphi}$ for $c+\varphi$,
\begin{equation}\label{A^0phi+nR}
\begin{split}
\caA^0_{\Sigma_1,g_1}(\tilde{\boldsymbol{\varphi}}_1,\tilde{\varphi}+n2\pi R   ) =&
e^{-\frac{1}{2}\cjg (\mathbf{D}_{\Sigma_1}-\mathbf{D})(\tilde{\boldsymbol{\varphi}}_1,\tilde{\varphi} +n2\pi R   ),(   \tilde{\boldsymbol{\varphi}}_1,\tilde{\varphi}+n2\pi R)\cjd}\\
=& \caA^0_{\Sigma_1,g_1}(\tilde{\boldsymbol{\varphi}}_1,\tilde{\varphi}    ) e^{-\frac{1}{2\pi}\langle \dd P( \tilde{\boldsymbol{\varphi}}_1,\tilde{ \varphi} ),\dd P^1_n\rangle} e^{-\frac{1}{4\pi}\|\dd P^1_n\|^2_2}
\end{split}
\end{equation}
and similarly for $\caA^0_{\Sigma_2,g_2} (\tilde{\boldsymbol{\varphi}}_2 ,\tilde{\varphi}+n2\pi R  )$. 
Let us write $\omega^{1,c}_{{\bf k}^{1,c}}=\omega^{1,h}_{{\bf k}^{1,c}}+\dd u$ for some $u\in C^\infty(\Sigma_1)$ with $u|_{\pl \Sigma_1}=0$ and $\dd^*\omega^{1,h}_{{\bf k}^{1,c}}=0$ satisfies the relative condition $\iota_{\pl \Sigma_1}^*\omega^{1,h}_{{\bf k}^{1,c}}=0$, then by Stokes formula 
\[ \langle  \omega^{1,c}_{{\bf k}^{1,c}},\dd P_n^1\rangle_{\Sigma_1}=\langle  \dd^*\omega^{1,h}_{{\bf k}^{1,c}}, P_n^1\rangle_{\Sigma_1}+\cjg   u, \dd^*\dd P_n^1\rangle_{\Sigma_1}=0.\]
Similarly we use \eqref{nuh-nu} to write 
$ \nu^{1}_{{\bf z}_1,{\bf m}_1,({\bf k}^1,k^1_{\mathfrak{b}_1})}= \nu^{1,h}_{{\bf z}_1,{\bf m}_1,({\bf k}^1,k^1_{\mathfrak{b}_1})}+\dd u'$ for some $u'\in C^\infty(\Sigma_1)$ with $u'|_{\pl \Sigma}=0$ and $\dd^*\nu^{1,h}_{{\bf z}_1,{\bf m}_1,({\bf k}^1,k^1_{\mathfrak{b}_1})}=0$: then by Stokes and the fact that $\nu^{1}_{{\bf z}_1,{\bf m}_1,({\bf k}^1,k^1_{\mathfrak{b}_1})}$ satisfies the absolute boundary condition
 \[\langle \nu^{1}_{{\bf z}_1,{\bf m}_1,({\bf k}^1,k^1_{\mathfrak{b}_1})} ,d P_n^1\rangle_{\Sigma_1}=\int_{\partial \Sigma_1}P_n^1 *  \nu^{1,h}_{{\bf z}_1,{\bf m}_1,({\bf k}^1,k^1_{\mathfrak{b}_1})}=\int_{\partial \Sigma_1}P_n^1  \nu^{1}_{{\bf z}_1,{\bf m}_1,({\bf k}^1,k^1_{\mathfrak{b}_1})}(\nu)d\ell=0.
\]
These two identities above imply on $\Sigma_i$ with $g_i:=g|_{\Sigma_i}$
\begin{equation}\label{doublesquare}
 \|\nu^{i}_{{\bf z}_i,{\bf m}_i,({\bf k}^i,k^i_{b_i})}+\omega^{i,c}_{{\bf k}^{i,c}}\|_{g_i,0}^2+ \|\dd P^i_n\|^2_{2} = \|\nu^{i}_{{\bf z}_i,{\bf m}_i,({\bf k}^1,k^i_{\mathfrak{b}_i})}+\omega^{i,c}_{{\bf k}^{i,c}}+\dd P_n^i\|_{g_i,0}^2 .
 \end{equation} 
We deduce, by combining with \eqref{A^0phi+nR}, that
 \begin{align*}
 & \caA^{0}_{\Sigma,g, {\bf z},\boldsymbol{\alpha}{\bf m},\boldsymbol{\zeta}}(F_1\otimes F_2,\tilde{\boldsymbol{\varphi}}_1^{{\bf k}_1},\tilde{\boldsymbol{\varphi}}_2^{{\bf k}_2}) \\
  =&
 \sum_{{\bf k}^c\in \Z^{2\mathfrak{g}+\mathfrak{b-}1}}\sum_{n\in\Z}  e^{-\frac{1}{4\pi}\|\nu^{1}_{{\bf z}_1,{\bf m}_1,({\bf k}_1,k_{1\mathfrak{b}_1})}+\omega^{1,c}_{{\bf k}^c_1}+\dd P_n^1\|_{g_1,0}^2} 
e^{-\frac{1}{4\pi}\|\nu^{2}_{{\bf z}_2,{\bf m}_2,({\bf k}_2,k_{2\mathfrak{b}_2})}+\omega^{2,c}_{({\bf k}^c_2,0)}+\dd P_n^2\|_{g_2,0}^2} 
 \\
& \int_0^{2\pi R}\E\Big[\caA^0_{\Sigma_1,g_1}( \tilde{\boldsymbol{\varphi}}_1,  \tilde{\varphi}    )\caA^0_{\Sigma_2,g_2}( \tilde{\boldsymbol{\varphi}}_2,\tilde{\varphi} )
   \hat{\mc{B}}_ {\Sigma_1,g_1}(F_1,\tilde{\boldsymbol{\varphi}}_1, \tilde{ \varphi}   ,{\bf k}_1 ,k_{1\mathfrak{b}_1},{\bf k}^c_1,n )\hat{\mc{B}}_ {\Sigma_2,g_2}(F_2, \tilde{\boldsymbol{\varphi}}_2, \tilde{ \varphi}     ,{\bf k}_2 ,k_{2\mathfrak{b}_2} ,({\bf k}^c_2,0),n)\Big]\,\dd c 
 \end{align*}
  with
\[
\hat{\mathcal{B}}_{\Sigma_i,g_i}(F_i,\tilde{\boldsymbol{\varphi}}_i,\tilde{ \varphi},{\bf k}_i ,k_{i\mathfrak{b}_i} ,{\bf k}^c_i,n):= 
\E \Big[F_i(\phi_i+P^i_n)e^{-\frac{1}{2\pi}\langle dX_i+dP( \tilde{\boldsymbol{\varphi}}_i,\tilde{ \varphi} ),\nu^{i}_{{\bf z}_i,{\bf m}_i,({\bf k}_i,k_{i\mathfrak{b}_i})}+\omega^{i,c}_{{\bf k}^c_i}+\dd P^i_n\rangle}   \Big].
\]
 
Now the point is to see  that the term $\dd P^i_n$ encodes part of the relative cohomology on $\Sigma_i$. For this, let us first introduce the notation ${\bf n}^i$ for the vector $(0,\dots,0,n)\in \Z^{2\mathfrak{g}_i+\mathfrak{b}_i-1}$.  Since the 1-form $\dd P^i_n$ is exact and since it takes the value $0$ on $\partial_j\Sigma_i$ and $n$ on  $\partial_j\Sigma_{\mathfrak{b}_i}$,  we have 
$$\int_{a^i_j}\dd P^i_n=0,\qquad \int_{b^i_j}\dd P^i_n=0,\qquad \int_{d^i_j}\dd P^i_n=0 \text{ for }j=1,\dots,\mathfrak{b}_i-1,\qquad \int_{d^i_{\mathfrak{b}_i}}\dd P^i_n=n2\pi R.$$
The 1-form $\dd P^i_n$ has therefore the same cycles/arcs as the 1-form $  \omega^{i,c}_{ {\bf n}^i} $.
Thus,  we have $ \dd P^i_n=\omega^{i,c}_{ {\bf n}^i}+\dd f^i_{\bf n}$, for some smooth function on $\Sigma_i$ vanishing on the boundary $\partial\Sigma_i$. We can absorb the term $\dd f_{\bf n}^i$ by means of the Girsanov transform. More precisely, on $\Sigma_i$ ($i=1,2$), we apply the Girsanov transform to the term $e^{-\frac{1}{2\pi}\langle dX_i,\dd f^i_{\bf n}\rangle-\frac{1}{4\pi}\|\dd f^i_{\bf n}\|_{2}^2}  $, which has the effect of shifting the law of the GFF as $X_{g,D}\to X_{g,D}- f^i_{\bf n}$, to get that (note that $\langle dP( \tilde{\boldsymbol{\varphi}}_i,\tilde{ \varphi} ),\dd f^i_{\bf n}\rangle=0$ on $\Sigma_i$)
\begin{multline*}
e^{-\frac{1}{4\pi}\|\nu^{i}_{{\bf z}_i,{\bf m}_i,({\bf k}_i,k_{i\mathfrak{b}_i})}+\omega^{i,c}_{{\bf k}_i^{c}}+\dd P_n^i\|_{g,0,\Sigma_i}^2}\hat{\mathcal{B}}_{\Sigma_i,g_i}(F_i,\tilde{\boldsymbol{\varphi}}_i,\tilde{ \varphi},{\bf k}_i ,k_{i\mathfrak{b}_i} ,{\bf k}^{c}_i,n )
\\
=e^{-\frac{1}{4\pi}\|\nu^{i}_{{\bf z}_i,{\bf m}_i,({\bf k}_i,k_{i\mathfrak{b}_i})}+\omega^{i,c}_{{\bf k}_i^{c}}+\omega^{i,c}_{ {\bf n}^i}\|_{g_i,0}^2}\hat{\mathcal{B}}'_{\Sigma_i,g_i}(F_i,\tilde{\boldsymbol{\varphi}}_i,\tilde{ \varphi},{\bf k}_i ,k_{i\mathfrak{b}_i} ,{\bf k}^{c}_i ,n)
\end{multline*}
where
\begin{multline*}
\hat{\mathcal{B}}'_{\Sigma_i,g_i}(F_i,\tilde{\boldsymbol{\varphi}}_i,\tilde{ \varphi},{\bf k}^i ,k_{\mathfrak{b}_i}^i ,{\bf k}^{i,c},n):= 
\E \Big[F_i(\phi_i+I^i_{x_0}(\omega^{i,c}_{ {\bf n}^i}))e^{-\frac{1}{2\pi}\langle dX_i+dP( \tilde{\boldsymbol{\varphi}}_i,\tilde{ \varphi} ),\nu^{i}_{{\bf z}_i,{\bf m}_i,({\bf k}_i,k_{i\mathfrak{b}_i})}+\omega^{i,c}_{{\bf k}_i^{c}}+\omega^{i,c}_{ {\bf n}^i}\rangle}   \Big].
\end{multline*}
In particular, note  the relation $\omega^{i,c}_{{\bf k}^{i,c}}+\omega^{i,c}_{ {\bf n}^i} =\omega^{i,c}_{{\bf k}^{i,c}+ {\bf n}^i}$. Making next the change of variables $k^{i,c}_{2\mathfrak{g}_i+\mathfrak{b}_i-1}+n\to k^{i,c}_{2\mathfrak{g}_i+\mathfrak{b}_i-1}$ in the summation over ${\bf k}^c$, we end up with the gluing statement we claimed.

\medskip
$\bullet$ Now we have to consider the case when   $\Sigma_2$ has only $\mathfrak{b}_2=1$ boundary components, in which case we get a basis $\boldsymbol{\sigma}$ of the relative homology by gathering all the curves $a_{ij},b_{ij}$ (for $j=1,\dots,\mathfrak{g}_i$), $d_{1j}$ (for $j=1,\dots,\mathfrak{b}_1-2$). Then $\omega^{1,c}_1,\dots, \omega^{1,c}_{2{\mathfrak{g}}_1+\mathfrak{b}_1-2},\omega^{2,c}_1,\dots, \omega^{2,c}_{2{\mathfrak{g}}_2} $ is a basis of $H^1(\Sigma,\partial\Sigma)$ made up of closed forms, dual to $\sigma$ and compactly supported. For ${\bf k}^c:=({\bf k}_1^c,{\bf k}_2^c)\in\Z^{2\mathfrak{g}_1+\mathfrak{b}_1-2}\times \Z^{2\mathfrak{g}_2} $, we set $\omega^c_{\bf k^c}:=\omega^{1,c}_{({\bf k}_1^{c},0)} +\omega^{2,c}_{{\bf k}_2^c} $.

Regarding  the absolute cohomology and magnetic 1-forms, the situation is somewhat simpler. The Dirac masses in the definition of amplitudes make it clear that the non trivial case is (in other cases, both hands of the gluing statement are $0$)
\begin{equation}\label{neutral2}
\sum_{j=1}^{\mathfrak{b}_1-1} \varsigma_{1j}k_{1j} +\sum_{j=1}^{n^1_\mathfrak{m}}m_{1j}+\sum_{j=1}^{n^2_\mathfrak{m}}m_{2j}=0,
\end{equation}
together with  $\sum_{j=1}^{n^1_\mathfrak{m}}m_{1j}+\sum_{j=1}^{\mathfrak{b}_1} \varsigma_{1j}k_{1j}=0$, $\sum_{j=1}^{n^2_\mathfrak{m}}m_{2j}+\varsigma_{21}k_{21}=0$, which implies again $k_{1\mathfrak{b}_1}=k_{21}$. So we fix $k_{1\mathfrak{b}_1}=k_{21}$ such that these conditions are fulfilled. Under these conditions, the forms $\nu^{1}_{{\bf z}_1,{\bf m}_1,({\bf k}^1,k_{1\mathfrak{b}_1})}$ and $\nu^{2}_{{\bf z}_2,{\bf m}_2, k_{21}}$  satisfy still \eqref{3mag} (the only difference is that, now, $\mathfrak{b}_2=1$). The magnetic forms on $\Sigma_1$ and $\Sigma_2$ glue to form a magnetic 1-form on $\Sigma$  by the relation $\nu_{{\bf z},{\bf m},{\bf k}}=\nu^{1}_{{\bf z}_1,{\bf m}_1,({\bf k}_1,k_{1\mathfrak{b}_1})} \mathbf{1}_{\Sigma_1}+\nu^{2}_{{\bf z}_2,{\bf m}_2, k_{21}} \mathbf{1}_{\Sigma_2}$, for ${\bf k}:={\bf k}_1\in\Z^{\mathfrak{b}_1-1}$. Finally we consider two defect graphs on $\Sigma_1$ and $\Sigma_2$ as in Lemma \ref{magcurvtwo},   and glue them to get a defect graph on $\Sigma$.

We can then follow the proof of the case $\mathfrak{b}_2\geq 2$ with the difference that, on $\Sigma_2$, the harmonic extension $P(\tilde{\varphi}+n2\pi R)$ is now trivial in the sense that it is equal to   $P\tilde{\varphi}+n2\pi R$. Furthermore, the free field amplitude on $\Sigma_2$ is now
 \begin{align*}
\caA^0_{\Sigma_2,g_2}(\tilde{\varphi}+n2\pi R   ) =&
e^{-\frac{1}{2}\cjg (\mathbf{D}_{\Sigma_1}-\mathbf{D})(\tilde{\varphi} +n2\pi R   ),\tilde{\varphi}+n2\pi R \cjd}=
\caA^0_{\Sigma_2,g_2}(\tilde{\varphi}   )  
\end{align*}
so that no change in the relative cohomology on $\Sigma_2$ is involved (note that adding $n2\pi R$ to $\phi_2$ in $\hat{\mc{B}}_2$ does not change $\hat{\mc{B}}_2$ by periodicity). We still have a change of relative cohomology   on $\Sigma_1$ (i.e. again with an extra $\dd P_n$) and it produces the summation over the 1-form that was absent on $\Sigma$, i.e. $\omega^{1,c}_{\mathfrak{b}_1-1}$. We obtain this way the gluing statement.
\end{proof}

Now we focus on the case of self-gluing for our reduced form amplitudes. The setup is the same as that drawn just before Proposition \ref{selfglueampli}. We claim
\begin{proposition}\label{selfglueampli0}
Let  $F$ be periodic positive or periodic integrable.
 \[
\caA^{0}_{\Sigma^{\#},g, {\bf z},{\bf m},\boldsymbol{\zeta}_\#}
(F,\tilde{\boldsymbol{\varphi}}_\#^{{\bf k}})=C \int \caA^{0}_{\Sigma,g,{\bf z}, {\bf m},\boldsymbol{\zeta}}(F,  ,\tilde{\boldsymbol{\varphi}}_\#^{{\bf k}},\tilde\varphi^k , \tilde\varphi^k) \dd\mu_0(\tilde{\varphi}^k).
\]
where $C= \frac{1}{\sqrt{2} \pi}$ if $\partial\Sigma \not =\emptyset$ and $C= \sqrt{2}  $ if $\partial\Sigma =\emptyset$.
 \end{proposition}

\begin{proof}
We denote by $\mc{C}$ the glued curve on $\Sigma$. Again, we split the proof in two parts depending whether $\Sigma^{\#}$ has a non empty boundary or not.

\medskip
{\bf 1) Assume first $\partial\Sigma^{\#}  \not= \emptyset$}. Let us call $\boldsymbol{\sigma}$ a basis of the relative homology on $\Sigma$. We choose this basis   by taking 
$a_1,\dots,a_{\mathfrak{g}}$, $b_1,\dots,b_{\mathfrak{g}}, d_1,\dots,d_{\mathfrak{b}-1}$ where $a_j,b_j$ are chosen  inside $\Sigma^\circ$ such that the intersection numbers are given as in \eqref{int_numbers}, and the $d_j$ are non-intersecting simple curves (not closed) with the base point on 
$\pl_{j}\Sigma$ and the  endpoint on $\pl_{j+1}\Sigma$, and each $d_j$ is not intersecting any other curve of the basis (see Figure \ref{fig:selfopen}). Let  $\omega^{c}_1,\dots, \omega^{c}_{2{\mathfrak{g}}+\mathfrak{b}-1}$ be a basis of $\mc{H}^1_R(\Sigma,\pl\Sigma)$ dual to $\boldsymbol{\sigma}$ made of closed forms that are compactly supported inside $\Sigma^\circ$ (hence can be viewed as compactly supported closed 1-forms on $ \Sigma^{\#}  $ too).  
 We stress that the last boundary-to-boundary arc $d_{\mathfrak{b}-1}$ joining $ \partial_{\mathfrak{b}-1}\Sigma$ to $\partial_{\mathfrak{b}}\Sigma$ will form a cycle in the glued surface, and therefore will play a special role in what follows. 
For ${\bf k}^c  \in\Z^{2\mathfrak{g} +\mathfrak{b}-1} $, we set $\omega^c_{\bf k^c}:=\sum_{j=1}^{2\mathfrak{g}+\mathfrak{b}-1}k^c_j\omega^{c}_{j} $ as usual.

\begin{figure}[h]
     \begin{subfigure}[b]{0.45\textwidth}
         \includegraphics[width=1\textwidth]{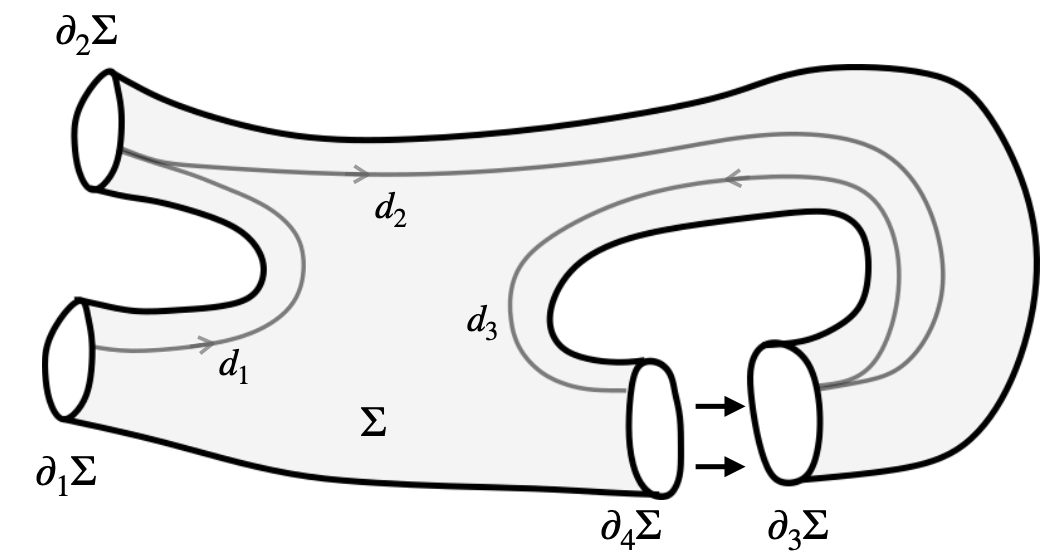}
         \caption{Unglued surface $\Sigma$. In gray, the relative homology on $\Sigma$. \vspace{1.75cm}}
     \end{subfigure}
     \hfill
     \begin{subfigure}[b]{0.45\textwidth}
         \includegraphics[width=1\textwidth]{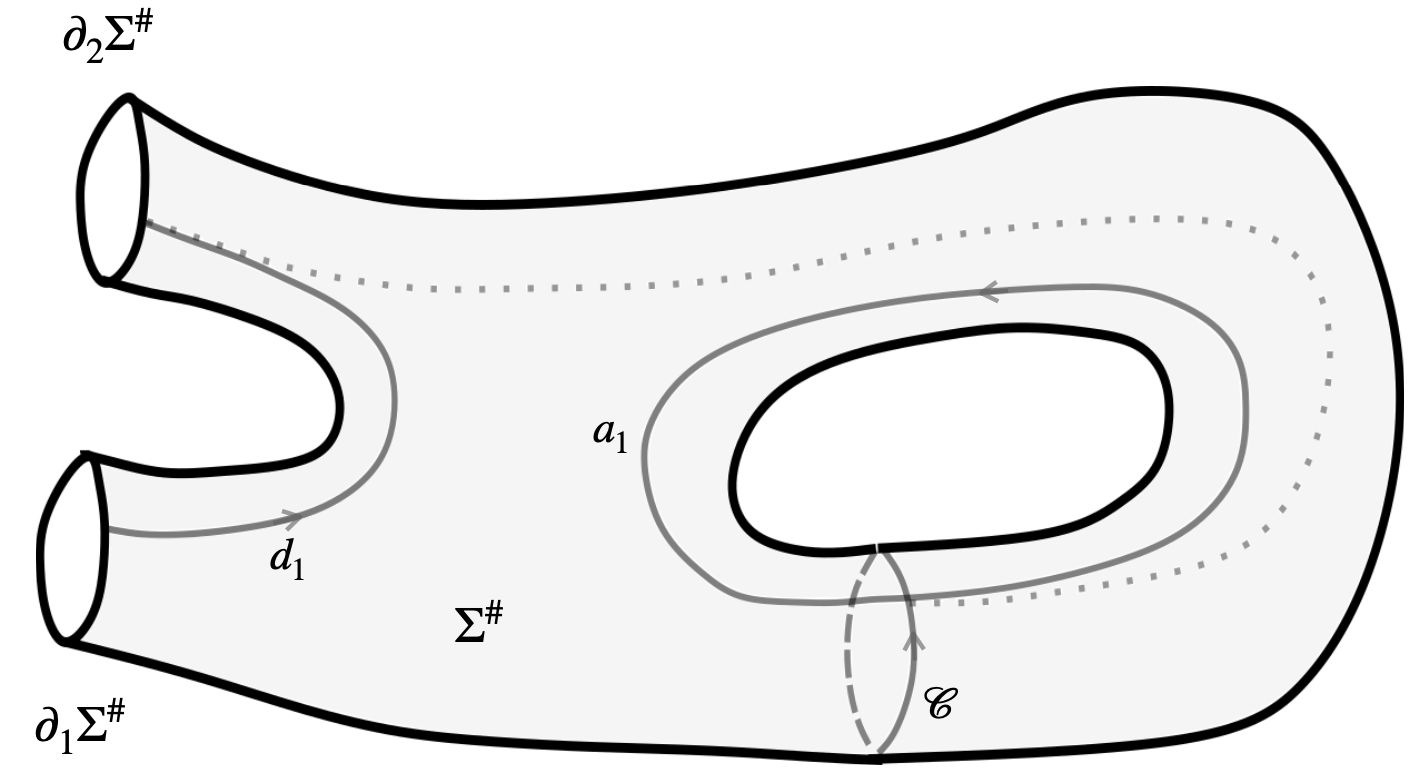}
         \caption{Glued surface. The curve $d_1$ remains in $\mc{H}_1(\Sigma^{\#},\pl \Sigma^{\#})$ after gluing. 
         The curve $d_3$ becomes a closed curve denoted $a_1$ in $\mc{H}_1(\Sigma^{\#},\pl \Sigma^{\#})$. The glued circle $\mc{C}$ is a cycle on $\Sigma^\#$, with intersection number $1$ with $a_1$. The dotted arc $d_2\in \mc{H}_1(\Sigma\pl\Sigma)$ has been  removed to construct $\mc{H}_1(\Sigma^{\#},\pl \Sigma^{\#})$.}
     \end{subfigure}
        \caption{Case $\pl\Sigma\not=\emptyset$. }
        \label{fig:selfopen}
\end{figure}

%

Now we focus on the   absolute cohomology and magnetic 1-forms. We discard first trivial cases: indeed, if ${\bf k}=(k_1,\dots,k_{\mathfrak{b}-2})$, notice that in case $\sum_{j=1}^{n_{\mathfrak{m}}}m_j+\sum_{j=1}^{\mathfrak{b}-2} \varsigma_jk_j \not=0$, both sides in the gluing statement vanish (because of Dirac masses) so that equality obviously  holds. So the case of interest is 
\begin{equation}\label{neutralbis}
\sum_{j=1}^{n_{\mathfrak{m}}}m_j+\sum_{j=1}^{\mathfrak{b}-2} \varsigma_jk_j =0,
\end{equation} in which case the winding numbers around $\partial_{\mathfrak{b}-1}\Sigma$ and $\partial_{\mathfrak{b}}\Sigma$ must satisfy $k_{\mathfrak{b}-1}=k_{\mathfrak{b}}$. We will simply write $k$ for $k_{\mathfrak{b}-1}=k_{\mathfrak{b}}$. Under these conditions, we consider the 1-form $\nu_{{\bf z},{\bf m},({\bf k},k,k)}$ on $\Sigma$ given by Lemma \ref{harmpoles} satisfying 
\[
\int_{\partial_j\Sigma}\nu_{{\bf z},{\bf m},({\bf k},k,k)}= 2\pi R\varsigma_jk_j,\quad\text{ for }j=1,\dots, \mathfrak{b}-2,\quad \text{and }\int_{\partial_{\mathfrak{b}}\Sigma}\nu_{{\bf z},{\bf m},({\bf k},k,k)}=-\int_{\partial_{\mathfrak{b}-1}\Sigma}\nu_{{\bf z},{\bf m},({\bf k},k,k)}= 2\pi R k.
\]
The last condition (and the fact that both forms involved are of the form $\pm k\,R\dd\theta$ over a neighborhood of the boundaries $\partial_{\mathfrak{b}-1}\Sigma$ and $\partial_{b}\Sigma$) allows us to self-glue  (see Lemma  \ref{baseselfglue}) these forms 
by identifying $\partial_{\mathfrak{b}-1}\Sigma$ and $\partial_{\mathfrak{b}}\Sigma$
to get a closed 1-form on $\Sigma^{\#}$  denoted by $\nu_{{\bf z},{\bf m},({\bf k},k)}$, satisfying our basic conditions for the definition of amplitude on the glued surface. Note that the curve $\mc{C}$ will become a cycle on the glued surface, and therefore the form $\nu_{{\bf z},{\bf m},({\bf k},k)}$, which has a winding $k$ along this curve, will produce the missing part in the relative cohomology on the glued surface.  For this we will split it as   $\nu_{{\bf z},{\bf m},({\bf k},k)}=\nu_{{\bf z},{\bf m},({\bf k},0)}+ \nu_{{\bf z},{\bf 0},({\bf 0},k)}$, the last term producing 1-forms with winding $k$ along $\mc{C}$ and vanishing along any other cycle.

As outlined above, we get now a basis $\boldsymbol{\sigma}_\#$ of the relative homology on $ \Sigma^{\#}$, which has $\mathfrak{b}-2$ boundary components,    by taking the cycles $a_1,\dots,a_{\mathfrak{g}}$, $b_1,\dots,b_{\mathfrak{g}}$ together with the cycles $ d_{\mathfrak{b}-1},\mc{C}$, and the boundary-to-boundary arcs $ d_1,\dots,d_{\mathfrak{b}-3}$. Note that the arc $d_{\mathfrak{b}-2}$ has been discarded. A basis of the relative cohomology, dual to $\boldsymbol{\sigma}_\#$, is now $\omega^{c}_1,\dots,  \omega^{c}_{2{\mathfrak{g}}+\mathfrak{b}-3},\omega^{c}_{2{\mathfrak{g}}+\mathfrak{b}-1},\nu_{{\bf z},{\bf 0},({\bf 0},k)}$ ($\omega^c_{2{\mathfrak{g}}+\mathfrak{b}-2}$ has been discarded). Absolute cohomology and magnetic 1-forms  are  then encoded in the forms $\nu_{{\bf z},{\bf m},({\bf k},0)}$ for ${\bf k}\in\Z^{\mathfrak{b}-2}$ with defect graph given by   Lemma \ref{magcurvtwoself}.
Since we have removed the $1$-form $\omega^c_{2{\mathfrak{g}}+\mathfrak{b}-2}$, we need to consider 1-forms $\omega^c_{{\bf k}^c}$ where the $(2\mathfrak{g}+\mathfrak{b}-2)$-th component of ${\bf k}^c$ has been set to $0$. We will thus need to consider  the vector $(k^c_1,\dots,k^c_{2\mathfrak{g}+\mathfrak{b}-3},0,k^c_{2\mathfrak{g}+\mathfrak{b}-1}) \in\Z^{2\mathfrak{g} +\mathfrak{b}-1} $, which will be identified with ${\bf k}^c_- \in\Z^{2\mathfrak{g} +\mathfrak{b}-2}$.

The definition of the amplitude on $\partial\Sigma^{\#}$ then yields
\begin{align*} 
&\caA^{0}_{\Sigma^{\#},g,{\bf z}, {\bf m},\boldsymbol{\zeta}_\# }(F,\tilde{\boldsymbol{\varphi}}^{{\bf k}}):=\\
    &
 \sum_{({\bf k}^c_-,k)\in \Z^{2\mathfrak{g}+\mathfrak{b}-2}\times\Z} e^{-\frac{1}{4\pi}\|\nu_{{\bf z},{\bf m},({\bf k'},0)}+\omega^{c}_{{\bf k}^c_-}+\nu_{{\bf z},{\bf 0},({\bf 0},k)}\|_{g,0}^2} Z_{\Sigma^{\#},g}\caA^0_{\Sigma^{\#},g}(\tilde{\boldsymbol{\varphi}})
 \mathcal{B}_{\Sigma^{\#},g}(F ,\tilde{\boldsymbol{\varphi}},{\bf k},{\bf k}^c_-,k) 
\end{align*} 
 with
 \begin{align} 
 \mathcal{B}_{\Sigma^{\#},g}(F ,\tilde{\boldsymbol{\varphi}},{\bf k},{\bf k}^c_-,k) :=&
 \E \Big[e^{-\frac{1}{2\pi}\langle \dd X_{g,D}+\dd P\tilde{\boldsymbol{\varphi}},\nu_{{\bf z},{\bf m},({\bf k},0)}+\omega^{c}_{{\bf k}^c_-}+\nu_{{\bf z},{\bf 0},({\bf 0},k)}\rangle}F(\phi_g)  \Big]\nonumber
\end{align}
where the Liouville field is $\phi_g= X_{g,D}+P\tilde{\boldsymbol{\varphi}}+I_{x_0}^{\boldsymbol{\xi}_\#}(\nu_{{\bf z},{\bf m},({\bf k},0)})+I_{x_0}^{\boldsymbol{\sigma}_\#}(\omega^{c}_{{\bf k}^c_-})+I_{x_0}^{\boldsymbol{\sigma}_\#}(\nu_{{\bf z},{\bf 0},({\bf 0},k)})$, the expectation $\E$ is over the Dirichlet GFF $X_{g,D}$ on $\Sigma^{\#} $, $P\tilde{\boldsymbol{\varphi}}$ stands for the harmonic extension to $\Sigma^{\#}$ of the boundary fields $\tilde{\boldsymbol{\varphi}}$, which stand respectively for the  boundary conditions on the remaining (i.e. unglued) components of $\partial\Sigma^{\#}$, namely
\[\Delta_g P\tilde{\boldsymbol{\varphi}}=0\quad  \text{on }\partial\Sigma^{\#}  ,\qquad   P\tilde{\boldsymbol{\varphi}}|_{\partial_j\Sigma}= \tilde{\varphi}_{j}\circ \zeta_{j}^{-1} \text{   for }j\leq \mathfrak{b}-2 
.\]

Let now $X$ be an independent Dirichlet GFF  on $\Sigma$. Then we have the following decomposition in law (see Proposition \ref{decompGFF})
\begin{align*}
X_{g,D}\stackrel{{\rm law}}=X+P{\bf Y}
\end{align*}
where ${\bf Y}$ is the restriction of $X_{g,D}$ to the  glued boundary component  $\mathcal{C}$ expressed in parametrized coordinates, i.e. ${\bf Y} =  X_{g,D}|_{{\caC}}\circ \zeta  $, and  $P{\bf Y}$ is its harmonic extension to $\Sigma  $ vanishing on $\partial\Sigma^{\#}$, which is non empty.  Again we denote by $h_\mathcal{C}$ the restriction of  the harmonic function $P\tilde{\boldsymbol{\varphi}}$ to $\mathcal{C}$ in parametrized coordinates
\[ h_{\mc{C}}: =P\tilde{\boldsymbol{\varphi}}|_{\mc{C} }\circ\zeta  \]
and, on $\Sigma$, the function $P\tilde{\boldsymbol{\varphi}}+P{\bf Y}$ is harmonic with boundary values (expressed in parametrized coordinates  on $\Sigma$) $ \tilde\varphi_{j}$ on $\partial_j\Sigma$ for $j<\mathfrak{b}-2$ and $({\bf Y}+h_{\mathcal{C}})$ on $\partial_{\mathfrak{b}-1}\Sigma$ and $\partial_{\mathfrak{b}}\Sigma$.  Proceeding as in the proof of the previous proposition, using the law of the field ${\bf Y}$ proved in \cite[Lemma 5.3 and Lemma 5.4]{GKRV2} and applying the Chasles relation, we arrive at the expression (expectation $\mathbb{E}$ is taken wrt $\varphi$)
\begin{align*}
&\caA^{0}_{\Sigma^{\#},g,{\bf z}, {\bf m},\boldsymbol{\zeta}_\# }(F,\tilde{\boldsymbol{\varphi}}^{{\bf k}})\\
 &=   \sum_{({\bf k}^c_-,k,n)\in \Z^{2\mathfrak{g}+\mathfrak{b}-2}\times\Z\times\Z} e^{-\frac{1}{4\pi}\|\nu_{{\bf z},{\bf m},({\bf k},0)}+\omega^{c}_{{\bf k}^c_-}+\nu_{{\bf z},{\bf 0},({\bf 0},k)}\|_{g,0}^2} Z_{\Sigma,g}
 \\
 &\times \int_0^{2\pi R}\E\Big[\caA^0_{ \Sigma,g}(\tilde{\boldsymbol{\varphi}},c+ \varphi +n2\pi R,c+ \varphi +n2\pi R   )  \mathcal{B}_ {\Sigma,g}(F,\tilde{\boldsymbol{\varphi}},c+ \varphi +n2\pi R  ,c+ \varphi +n2\pi R,{\bf k}',{\bf k}^c_-,k ) \Big]\,\dd c ,
\end{align*}
where expectation is taken over $\varphi$, with  
\begin{equation*}
\mathcal{B}_{\Sigma,g}(F,\tilde{\boldsymbol{\varphi}},\tilde{ \varphi},\tilde{ \varphi},{\bf k},{\bf k}^c_-,k):=\E \big[F(\phi)e^{-\frac{1}{2\pi}\langle \dd X+\dd P( \tilde{\boldsymbol{\varphi}},\tilde{ \varphi},\tilde{ \varphi} ),\nu_{{\bf z},{\bf m},({\bf k},k)}+\omega^{c}_{{\bf k}^c_-} \rangle}   \big]
\end{equation*}
where $\E$ is taken with respect to the Dirichlet GFF  $X$  on $\Sigma$, $\phi=X+P( \tilde{\boldsymbol{\varphi}},\tilde{ \varphi},\tilde{ \varphi} )+I^{\boldsymbol{\sigma}}_{x_0}(\omega^{c}_{{\bf k}^c_-})+I_{x_0}^{\boldsymbol{\xi}}(\nu_{{\bf z},{\bf m},({\bf k},k)}) $ and $P( \tilde{\boldsymbol{\varphi}},\tilde{ \varphi},\tilde{ \varphi} )$ stands for the harmonic extension to $\Sigma$ of the boundary fields $  \tilde{\boldsymbol{\varphi}},\tilde{ \varphi},\tilde{ \varphi} $. Also, we have obtained the defect graph on $\Sigma$ from the one on  $\Sigma^{\#} $  by using Lemma \ref{magcurvtwoself}.

Next, we introduce the harmonic functions $P_{n}$, for  $n\in\Z$, that are harmonic on $\Sigma$ with boundary values $n2\pi R$ on both boundary components corresponding to  $\mc{C}$ in $\Sigma$, and $0$ on the other boundary components of $\Sigma$. Then we notice that, writing $\tilde{\varphi}$ for $c+\varphi$,
\[ \begin{split}
\caA^0_{\Sigma,g}( \tilde{\boldsymbol{\varphi}},\tilde{\varphi}+n2\pi R ,\tilde{\varphi}+n2\pi R   ) =&
e^{-\frac{1}{2}(\cjg \mathbf{D}_{\Sigma}-\mathbf{D})(\tilde{\boldsymbol{\varphi}},\tilde{\varphi}+n2\pi R ,\tilde{\varphi}+n2\pi R ),(\tilde{\boldsymbol{\varphi}},\tilde{\varphi}+n2\pi R ,\tilde{\varphi}+n2\pi R )\cjd}\\
=& 
\caA^0_{\Sigma,g}( \tilde{\boldsymbol{\varphi}},\tilde{\varphi}  ,\tilde{\varphi}  ) e^{-\frac{1}{2\pi}\langle \dd P( \tilde{\boldsymbol{\varphi}},\tilde{\varphi} ,\tilde{\varphi} ),\dd P_n\rangle} e^{-\frac{1}{4\pi}\|\dd P_n\|^2_{2}}.
\end{split}\]
Notice also, using the similar argument as for \eqref{doublesquare}, that 
\[ \|\nu_{{\bf z},{\bf m},({\bf k},k)}+\omega^{c}_{{\bf k}_-^{c}}\|_{g,0}^2+ \|\dd P_n\|^2_{2} = \|\nu_{{\bf z},{\bf m},({\bf k},k)}+
\omega^{c}_{{\bf k}_-^{c}}+\dd P_n\|_{g,0}^2.\]
Therefore, one obtains
\begin{align*}
 &\caA^{0}_{\Sigma^{\#},g,{\bf z}, {\bf m},\boldsymbol{\zeta}_\# }(F,\tilde{\boldsymbol{\varphi}}^{\bf k})\\
  =&
  \sum_{k}\sum_{({\bf k}^c_-,n)\in \Z^{2\mathfrak{g}+\mathfrak{b}-2}\times\Z}    e^{-\frac{1}{4\pi}\|\nu_{{\bf z},{\bf m},({\bf k},k)}+\omega^{c}_{{\bf k}^c_-} +\dd P_n\|_{g,0}^2} Z_{\Sigma,g} \int_0^{2\pi R}\E\Big[\caA^0_{ \Sigma,g}(\tilde{\boldsymbol{\varphi}}, \tilde{\varphi}  ,\tilde{\varphi}    ) \hat{\mathcal{B}}_{\Sigma,g}(F,\tilde{\boldsymbol{\varphi}},\tilde{ \varphi},\tilde{ \varphi},{\bf k},{\bf k}^c_-,k,n)
 \Big]\,\dd c 
 \end{align*}
  with
\begin{equation*}
 \hat{\mathcal{B}}_{\Sigma,g}(F,\tilde{\boldsymbol{\varphi}},\tilde{ \varphi},\tilde{ \varphi},{\bf k},{\bf k}^c_-,k,n):
 =
\E \Big[F(\phi+P_n)e^{-\frac{1}{2\pi}\langle \dd X+\dd P( \tilde{\boldsymbol{\varphi}},\tilde{ \varphi},\tilde{ \varphi} ),\nu_{{\bf z},{\bf m},({\bf k},k)}+\omega^{c}_{{\bf k}^c_-} +\dd P_n\rangle}   \Big].
\end{equation*}

Again, the point is now to relate $\dd P_n$ to  the relative cohomology: indeed, it encodes the 1-forms dual to the boundary-to-boundary arc $d_{\mathfrak{b}-2}$ that we have previously removed. To see this, let us first introduce the notation ${\bf n}$ for the vector $(0,\dots,n,0)\in \Z^{2\mathfrak{g}+\mathfrak{b}-1}$. Since the 1-form $\dd P_n$ is exact and since it takes the value $0$ on $\partial_j\Sigma$ for $j\leq b-2$ and $n$ on both $\partial_{b-1}\Sigma $ and $\pl_b\Sigma$,  we have 
$$\int_{a_j}\dd P_n=0,\qquad \int_{b_j}\dd P_n=0,\qquad \int_{d_j}\dd P_n=0 \text{ for }j\not=\mathfrak{b}-2,\qquad \int_{d_{\mathfrak{b}-2}}\dd P_n=n2\pi R.$$
The 1-form $\dd P_n$ has therefore the same cycles/arcs as the 1-form $  \omega^{c}_{ {\bf n}} $. Thus, since $dP_n\in \mc{H}_1(\Sigma,\pl\Sigma)$ satisfies the relative boundary condition,  we have $ \dd P_n=\omega^{c}_{ {\bf n}}+\dd f_{\bf n}$, for some smooth function on $\Sigma$ vanishing on the boundary $\partial\Sigma$.  As in the proof of Proposition \ref{glueampli0}, we can replace $\dd P_n $
by $\omega^{c}_{ {\bf n}}+\dd f_{\bf n}$ in the expression of $\hat{\mathcal{B}}_{\Sigma,g}$ and apply Girsanov to the term $e^{-\frac{1}{2\pi}\langle \dd X , \dd f_{\bf n}\rangle-\frac{1}{4\pi}\|\dd f_{\bf n}\|_2^2}    $ to get that 
\begin{align*}
 &\caA^{0}_{\Sigma^{\#},g,{\bf v}, {\bf m},\boldsymbol{\zeta}_\# }(F,\tilde{\boldsymbol{\varphi}}^{\bf k})\\
  =&
  \sum_{k}\sum_{({\bf k}^c_-,n)\in \Z^{2\mathfrak{g}+\mathfrak{b}-2}\times\Z}    e^{-\frac{1}{4\pi}\|\nu_{{\bf z},{\bf m},({\bf k'},k)}+\omega^{c}_{{\bf k}^c_-} +\omega^{c}_{ {\bf n}}\|_2^2} Z_{\Sigma,g} \int_0^{2\pi R}\E\Big[\caA^0_{ \Sigma,g}(\tilde{\boldsymbol{\varphi}}, \tilde{\varphi}  ,\tilde{\varphi}    ) \hat{\mathcal{B}}_{\Sigma,g}(F,\tilde{\boldsymbol{\varphi}},\tilde{ \varphi},\tilde{ \varphi},{\bf k},{\bf k}^c_-,k,n)
 \Big]\,\dd c 
 \end{align*}
  with
\begin{equation*}
 \hat{\mathcal{B}}_{\Sigma,g}(F,\tilde{\boldsymbol{\varphi}},\tilde{ \varphi},\tilde{ \varphi},{\bf k},{\bf k}^c_-,k,n):
 =
\E \Big[F(\phi+I_{x_0}(\omega^{c}_{ {\bf n}}))e^{-\frac{1}{2\pi}\langle \dd X+\dd P( \tilde{\boldsymbol{\varphi}},\tilde{ \varphi},\tilde{ \varphi} ),\nu_{{\bf z},{\bf m},({\bf k},k)}+\omega^{c}_{{\bf k}^c_-} +\omega^{c}_{ {\bf n}}\rangle_{2}}   \Big].
\end{equation*}

This means that, in the expression of $\hat{\mc{B}}_{\Sigma,g}$, we have the relation $\omega^{c}_{{\bf k}^{c}_-}+\omega^{c}_{ {\bf n}}=\omega^{c}_{{\bf k}^{c}_-+ {\bf n}}$.   Therefore, the relative cohomology term in $ \hat{\mathcal{B}}_{\Sigma,g}$ is $\omega^{c}_{{\bf k}^{c}_-+ {\bf n}}$ and the summation $\sum_{({\bf k}^c_-,n)\in \Z^{2\mathfrak{g}+\mathfrak{b}-2}\times\Z} $ thus corresponds to a sum over the whole relative cohomology basis on $\Sigma$. This proves the claim.

 \medskip
{\bf 2) Assume now $\partial\Sigma^{\#}   = \emptyset$}.  The surface $\Sigma$  has now two boundary components $\pl_1\Sigma,\pl_2\Sigma$, which we want to glue to get a surface  $\Sigma^{\#}$. We take a basis $\boldsymbol{\sigma}$ of $\mc{H}_1(\Sigma,\pl \Sigma)$ made of $a_1,\dots,a_{\mathfrak{g}}$, $b_1,\dots,b_{\mathfrak{g}}, d_1 $ where $a_j,b_j$ are cycles chosen  inside $\Sigma^\circ$ such that the intersection numbers are given as in \eqref{int_numbers}, and  $d_1$ is a non-intersecting simple curve  (not closed) with the base point on 
$\pl_{1}\Sigma$ and the  endpoint on $\pl_{2}\Sigma$, and   $d_1$ is not intersecting any other curve of the basis. Let  $\omega^{c}_1,\dots, \omega^{c}_{2{\mathfrak{g}}+ 1}$ be a basis of $\mc{H}^1_R(\Sigma,\pl\Sigma)$ dual to $\sigma$ made of closed forms that are compactly supported inside $\Sigma^\circ$ (hence can be viewed as   closed 1-forms on $ \Sigma^{\#}  $ too).   For ${\bf k}^c  \in\Z^{2\mathfrak{g} +1} $, we set $\omega^c_{\bf k^c}:=\sum_{j=1}^{2\mathfrak{g}+1}k^c_j\omega^{c}_{j} $ as usual.

We focus now on  the absolute cohomology and magnetic forms. Again, we identify the only non trivial case to be treated. It corresponds to
\begin{equation}
\sum_{j=1}^{n_{\mathfrak{m}}}m_j=0,\qquad k_1=k_2.
\end{equation}
As such, we will simply write $k$ for $k_1=k_2$. Next, we consider the closed  $1$-form $\nu_{{\bf z},{\bf m},(k,k)}$  with  winding numbers given by 
$$
 \int_{\partial_{1}\Sigma}\nu_{{\bf z},{\bf m},(k,k)}= 2\pi R k,\qquad \int_{\partial_{2}\Sigma}\nu_{{\bf z},{\bf m},(k,k)}=- 2\pi R k,
$$
winding $m_j$ around $z_j$, and $0$ along any other interior cycle.
The   condition above (and the fact that both forms involved are of the form $\pm  R k\,\dd\theta$ over a neighborhood of the boundaries $\partial_{1}\Sigma$ and $\partial_{2}\Sigma$) allows us to self-glue   this form
by identifying $\partial_{1}\Sigma$ and $\partial_{2}\Sigma$
to get a closed 1-form on $\Sigma^{\#}$  denoted by $\nu_{{\bf z},{\bf m},k}$. Note that the curve $\mc{C}$ will become a cycle on the glued surface, and therefore the form $\nu_{{\bf z},{\bf m},k}$, which has a winding $k$ along this curve, will be part of the cohomology on the glued surface.  For this, we split $\nu_{{\bf z},{\bf m},k}$ as $\nu_{{\bf z},{\bf m},0}+\nu_{{\bf z},{\bf 0},k}$.

Now we get a basis $\boldsymbol{\sigma}_\#$ of   homology on $ \Sigma^{\#}$, which is a closed surface,    by taking the cycles $a_1,\dots,a_{\mathfrak{g}}$, $b_1,\dots,b_{\mathfrak{g}}$ together with the cycles $ d_{1},\mc{C}$. A basis of   cohomology, dual to $\boldsymbol{\sigma}_\#$, is now $\omega^{c}_1,\dots, \omega^{c}_{2{\mathfrak{g}}+1},\nu_{{\bf z},{\bf 0},k}$.   The magnetic 1-form on $\Sigma^{\#}$ is now $\nu_{{\bf z},{\bf m},0}$.

Based on this, the rest of the proof    is a combination of arguments  already used so we just outline the proof.    We condition the amplitude on $\Sigma^{\#}$ on the values of the GFF along $\mc{C}$ and use the description of this law given by \cite[eq (5.14)]{GKRV2} (similar to proof of Prop. \ref{glueampli}  case $\partial\Sigma=\emptyset$). We apply the Chasles relation on the $c$-integral, then we use
\begin{align*}
\caA^0_{\Sigma,g}(\tilde{\varphi}+n2\pi R,\tilde{\varphi}+n2\pi R  ) =&
\caA^0_{\Sigma,g}(\tilde{\varphi}   ,\tilde{\varphi} )  
\end{align*}
because harmonic functions on $\Sigma$, worth $n2\pi R$ on both $\partial_1\Sigma$ and $\partial_2\Sigma$, must be constant and equal to $n2\pi R$. Therefore there is no contribution coming from the shift by $n2\pi R$ and we get this way the expression of the glued amplitude on $\Sigma$. Details are left to the reader.
\end{proof} 

\subsection{Amplitudes are $L^2$}\label{s:L^2amplitude}  
In this section, we shall prove that the Liouville amplitudes of surfaces with boundary are in the Hilbert space $\mc{H}^{\otimes \mathfrak{b}}$. This is done by doubling the surface and using that the correlation functions on closed surfaces exist.

First we state a lemma about the effect of reverting orientation on a given Riemann surface with or without boundary. So we consider the setup for the definition of amplitudes. We denote by $\Sigma'$ the surface $\Sigma$ but with orientation reversed.
The following lemma directly follows from definitions:
\begin{lemma}\label{reverting}
We have the relation
\[ \caA^0_{\Sigma,g,{\bf z},\bf{m},\boldsymbol{\zeta}} (F,\tilde{\boldsymbol{\varphi}}^{\bf k}) = \caA^{0}_{\Sigma',g,{\bf z},-\bf{m},\boldsymbol{\zeta}} (F,\tilde{\boldsymbol{\varphi}}^{\bf k}) .\]
\end{lemma}
\begin{proof}
In view of the properties of the form $\nu^{\Sigma}_{{\bf z},{\bf k},{\bf m}}$ on $\Sigma$ and $\nu^{\Sigma}_{{\bf z},{\bf k},{\bf m}}$ on $\Sigma'$ given in Proposition \ref{harmpoles},
we observe that 
\[ \nu^{\Sigma'}_{{\bf z},{\bf m},{\bf k}}=\nu^{\Sigma}_{{\bf z},-{\bf m},{\bf k}}.\]
This directly implies the result by using the definition of \eqref{amplitudegff} and the fact that all the other quantities involved in the free field amplitude are independent of the orientation of $\Sigma$. 
\end{proof}

\subsubsection{Regularized amplitudes}
We consider a surface $\Sigma$ with non empty boundary and all the datas from the setup for amplitudes. For $\epsilon>0$, we denote   by $ \caA^\eps_{\Sigma,g,{\bf x},{\bf v},\boldsymbol{\alpha},\bf{m},\boldsymbol{\zeta}} (F,\tilde{\boldsymbol{\varphi}}^{\bf k})$ the regularized amplitudes:
\begin{align}\label{amplitudereg}
 \caA^\eps_{\Sigma,g,{\bf x},{\bf v},\boldsymbol{\alpha},\bf{m},\boldsymbol{\zeta}}&(F,\tilde{\boldsymbol{\varphi}}^{\bf k}) \\
: =&
\delta_0(\sum_{j=1}^{n_{\mathfrak{m}}+\mathfrak{b}} m_j) \sum_{{\bf k}^c\in \Z^{2\mathfrak{g}+\mathfrak{b}-1}} e^{-\frac{1}{4\pi}\|\nu_{\mathbf{z},\mathbf{m},{\bf k}} +\omega^{c}_{{\bf k}^c}\|_{g,0}^2 }Z_{\Sigma,g}\caA^0_{\Sigma,g}(\tilde{\boldsymbol{\varphi}}  ) 
 \nonumber\\
 &
 \E \Big[e^{-\frac{1}{2\pi}\langle \dd X_{g,D}+\dd P\tilde{\boldsymbol{\varphi}},\nu_{\mathbf{z},\mathbf{m},{\bf k}} +\omega^{c}_{{\bf k}^c}\rangle}F(\phi_g)\prod_{j=1}^{n_{\mathfrak{m}}} V_{\alpha_j,g,\epsilon}(x_j)e^{-\frac{i  Q}{4\pi}\int^{\rm reg}_\Sigma K_g\phi_g\dd {\rm v}_g
 -\mu M_\beta^g (\phi_g,\Sigma)}\Big]\nonumber.
\end{align}
Recall that amplitudes are defined as the limit $\lim_{{\bf x}\to {\bf z}}\lim_{\epsilon\to 0}$ of this quantity, where the limit when $x_i\to z_i$ is understoof as in \eqref{defcorrelmixed} along a curve with a tangent vector given by $v_i$. We will essentially follow the argument in Propositions \ref{limitcorel} and  \ref{limitcorelmixed} to take these limits, with adaptations due to the presence of a boundary. 

The first step is to adapt the imaginary Girsanov transform. Let us define 
\[u^{\epsilon,\epsilon'}_{\bf x}(y):=\sum_{j=1}^{n_{\mathfrak{m}}}i\alpha_j\E[X_{g,D,\epsilon}(x_j)X_{g,D,\epsilon}(y)].\] 
We write simply $u^{\epsilon}_{\bf x}(y)$ for the function $u^{\epsilon,\epsilon'}_{\bf x}(y)$ for $\epsilon'=0$. We claim:
\begin{lemma}\label{fuckgirsanov}
The following identity holds
\begin{align*} 
 \caA^{\epsilon}_{\Sigma,g,{\bf x},{\bf v},\boldsymbol{\alpha},\bf{m},\boldsymbol{\zeta}}&(F,\tilde{\boldsymbol{\varphi}}^{\bf k})=   \delta_0(\sum_{j=1}^{n_{\mathfrak{m}}+\mathfrak{b}} m_j({\bf k})) 
 C_\epsilon \sum_{{\bf k}^c\in \Z^{2\mathfrak{g}+b-1}} e^{-\frac{1}{4\pi}\|\nu_{\mathbf{z},\mathbf{m},{\bf k}} +\omega^{c}_{{\bf k}^c}\|_{g,0}^2 }Z_{\Sigma,g}\caA^0_{\Sigma,g}(\tilde{\boldsymbol{\varphi}}  )   \prod_{j=1}^{n_{\mathfrak{m}}}e^{i\alpha_j P\tilde{\boldsymbol{\varphi}}_\epsilon(z_j)}
 \\
 &
 \E \Big[e^{-\frac{1}{2\pi}\langle \dd X_{g,D}+\dd P\tilde{\boldsymbol{\varphi}}+\dd u^{\epsilon}_{\bf x},\nu_{\mathbf{z},\mathbf{m},{\bf k}} +\omega^{c}_{{\bf k}^c}\rangle_2}F(\phi_g+u^{\epsilon}_{\bf x}) e^{-\frac{i  Q}{4\pi}\int^{\rm reg}_\Sigma K_g(\phi_g+u^{\epsilon}_{\bf x})\dd {\rm v}_g
 -\mu M_\beta^g (\phi_g+u^{\epsilon}_{\bf x},\Sigma)}\Big]\nonumber.
\end{align*}
where the constant $C_\epsilon$ is given by
$$C_\epsilon:=e^{-\sum_{j<j'}\alpha_j\alpha_{j'}\E[X_{g,D,\epsilon}(x_j)X_{g,D,\epsilon}(x_{j'})]} \prod_{j=1}^{n_{\mathfrak{m}}}\epsilon^{-\frac{\alpha_j^2}{2}}e^{-\frac{\alpha_j^2}{2}\E[X_{g,D,\epsilon}(x_j)]}.$$
\end{lemma}

\begin{proof}
This lemma relies on the Girsanov transform applied to the product $\prod_{j=1}^{n_{\mathfrak{m}}}e^{i\alpha_jX_{g,D,\epsilon}(z_j)}$. Therefore it follows the argument explained in the beginning of the proof of Proposition \ref{limitcorel}. To reproduce the argument, we need the following tools:
\begin{enumerate}
\item $\tilde{\boldsymbol{\varphi}}$ a.s., the mapping $f\mapsto \int_\Sigma f(x) M_\beta^g(X_{g,D}+P\tilde{\boldsymbol{\varphi}}+I_{x_0}^{\boldsymbol{\sigma}}(\omega^{c}_{{\bf k}^c})+I^{\boldsymbol{\zeta}}_{x_0}(\nu_{\mathbf{z},\mathbf{m},{\bf k}} ),\dd x)$ is a distribution (in the sense of Schwartz)  of order $2$ and there exists some $L^2$ random variable $D_\Sigma$ such that
$$\forall f\in C_c^\infty(\Sigma),\qquad \big|\int_\Sigma f(x)  M^g_\beta(X_{g,D}+P\tilde{\boldsymbol{\varphi}} +I_{x_0}^{\boldsymbol{\sigma}}(\omega^{c}_{{\bf k}^c})+I^{\boldsymbol{\zeta}}_{x_0}(\nu_{\mathbf{z},\mathbf{m},{\bf k}} ),\dd x) \big|\leq D_\Sigma   \|\Delta_gf\|_\infty. $$
\item  we have the estimate $$\E\Big[\exp\big(\big|\int_\Sigma f(x)  M^g_\beta(X_g+I_{x_0}^{\boldsymbol{\sigma}}(\omega^{c}_{{\bf k}^c})+I^{\boldsymbol{\zeta}}_{x_0}(\nu_{\mathbf{z},\mathbf{m},{\bf k}} ),\dd x) \big|\big)\,|\,\tilde{\boldsymbol{\varphi}} \Big]\leq C\|f\|_\infty$$
for some deterministic constant $C>0$.
\end{enumerate}
To establish the first property, we can write $f$ as $f(x)=\int_\Sigma \Delta f(y)G_{g,D}(x,y)\,{\rm dv}_g(y)$ and then follow the argument in Lemma \ref{expmomentunif} to get that 
  $D_\Sigma=\int_\Sigma \Big|\int_\Sigma G_{g,D}(x,y)M^g_\beta(X_g+P\tilde{\boldsymbol{\varphi}} + I_{x_0}^\sigma(\omega^{c}_{{\bf k}^c})+I^{\boldsymbol{\zeta}}_{x_0}(\nu_{\mathbf{z},\mathbf{m},{\bf k}} ),\dd x) \Big| {\rm dv}_g(y)$, which is $\tilde{\boldsymbol{\varphi}} $ a.s. in $L^2$ in the expectation with respect to the Dirichlet GFF.
  
For the exponential moment estimate, this follows from Proposition \ref{expmoment} again. Note that the contribution from the harmonic extension is trivial using that $|e^{i\beta  P\tilde{\boldsymbol{\varphi}}}|=1$, and this is why we get a deterministic bound.
\end{proof}

Now we claim
\begin{lemma}
The following convergence result holds 
 $(\dd c \otimes \P_\T)^{\otimes \mathfrak{b}}$ almost surely:
\[  \caA^{\epsilon}_{\Sigma,g,{\bf x},{\bf v},\boldsymbol{\alpha},\bf{m},\boldsymbol{\zeta}} (F,\tilde{\boldsymbol{\varphi}}^{\bf k}) \to  \caA_{\Sigma,g,{\bf v},\boldsymbol{\alpha},\bf{m},\boldsymbol{\zeta}} (F,\tilde{\boldsymbol{\varphi}}^{\bf k}) \]
as $\eps\to 0$ and  ${\bf x}\to {\bf z}$ in the direction ${\bf v}=((z_1,v_1),\dots,(z_n,v_n))\in (T\Sigma)^{n}$, where 
\[\begin{split}
 \caA_{\Sigma,g,{\bf v},\boldsymbol{\alpha},\bf{m},\boldsymbol{\zeta}} (F,\tilde{\boldsymbol{\varphi}}^{\bf k}) = &  \delta_0(\sum_{j=1}^{n_{\mathfrak{m}}+\mathfrak{b}} m_j({\bf k}))   Z_{\Sigma,g}\caA^0_{\Sigma,g}(\tilde{\boldsymbol{\varphi}}  )e^{-\sum_{j<j'}\alpha_j\alpha_{j'}G_{g,D}(z_j,z_{j'})} \prod_{j=1}^{n_{\mathfrak{m}}}e^{i\alpha_j P\tilde{\boldsymbol{\varphi}}(z_j)-\frac{1}{2} \sum_j\alpha_j^2W_{g}(z_j)} \nonumber
\\
& \sum_{{\bf k}^c\in \Z^{2\mathfrak{g}+b-1}} e^{-\frac{1}{4\pi}\|\nu_{\mathbf{z},\mathbf{m},{\bf k}} +\omega^{c}_{{\bf k}^c}\|_{g,0}^2 } \prod_{j}e^{i\alpha_j (I^{\boldsymbol{\sigma}}_{x_0}(\omega^c_{\bf k^c})+I_{x_0}^{\boldsymbol{\xi}}(\nu_{\mathbf{z,m,k}}))(z_j)}
\label{defelecamp}
\\
& \E\Big[e^{-\frac{1}{2\pi}\langle \dd X_{g,D}+\dd P\tilde{\boldsymbol{\varphi}}+\dd u_{\bf z},\nu_{\mathbf{z},\mathbf{m},{\bf k}} +\omega^{c}_{{\bf k}^c} \rangle}F(\phi_g+u_{\bf z})e^{-\frac{i   Q}{4\pi}\int_\Sigma^{\rm reg} K_g(\phi_g+u_{\bf z})\,\dd v_g -\mu  M^g_\beta(\phi_g+u_{\bf z},\Sigma)}\Big] ,\nonumber
\end{split}\]
 where  $u_{\bf z}(x)=\sum_{j=1}^{n_{\mathfrak{m}}}i\alpha_jG_{g,D}(x,z_j)$ and the Liouville field is $\phi_g=X_{g,D}+P\tilde{\boldsymbol{\varphi}}+ I^{\boldsymbol{\xi}}_{x_0}(\nu_{\mathbf{z},\mathbf{m},\mathbf{k}})+I^{\boldsymbol{\sigma}}_{x_0}(\omega^{c}_{{\bf k}^c})$. Expectation is over the Dirichlet GFF and the evaluation of $I_{x_0}^{\boldsymbol{\xi}}(\nu_{\mathbf{z,m,k}})$ at the points $z_j$ is done in the direction prescribed by the vectors ${\bf v}$. 
\end{lemma}

\begin{proof} Beware that the pairing $\cjg \dd u_{\bf z},\nu_{{\bf z},{\bf m},{\bf k}}\cjd$  makes sense due to the form of the singularity at $z_j$: it is radial for the Green $ \dd u_{\bf z}$ whereas it is angular for $\nu_{{\bf z},{\bf m},{\bf k}}$.
The limit in $\epsilon$ is the same as in the proof of Propositions \ref{limitcorel} and \ref{limitcorelmixed}.  Denoting by $\Delta_{\epsilon,{\bf x}}$ the difference between regularized amplitudes, the argument of Propositions  \ref{limitcorel} and \ref{limitcorelmixed} leads to 
\begin{align*}
|\Delta_{\epsilon,{\bf x}}|\leq  &\sum_{{\bf k}\in \Z^{2\mathfrak{g}}} \caA^0_{\Sigma,g}(\tilde{\boldsymbol{\varphi}}  ) e^{-\frac{1}{4\pi}\|\Pi^c_1\omega _{\bf k^c}^c\|_2^2 +C|{\bf k}|}e^{-\frac{1}{2\pi}\langle \dd P\tilde{\boldsymbol{\varphi}}, \Pi_1^c \omega^c_{\bf k^c}\rangle_2}C_{\epsilon,{\bf x},{\bf z}}(F, \tilde{\boldsymbol{\varphi}})
\end{align*}
for some constant $C_{\epsilon,{\bf x},{\bf z}}(F, \tilde{\boldsymbol{\varphi}})$ such that $\lim_{{\bf x}\to {\bf z}}\lim_{\epsilon\to 0}C_{\epsilon,{\bf x},{\bf z}}(F, \tilde{\boldsymbol{\varphi}})=0$, $\mu_0^{\otimes \mathfrak{b}}$ almost surely in $ \tilde{\boldsymbol{\varphi}}$. Here again, we complete the argument with the fact that $e^{-\frac{1}{2\pi}\langle \dd P\tilde{\boldsymbol{\varphi}}, \Pi_1^c \omega^c_{\bf k^c}\rangle}\leq e^{C|{\bf k}|}$.
\end{proof}

 We prove now that this expression makes sense as an element in  $\mc{H}^{\otimes \mathfrak{b}}$.
 \begin{lemma}\label{amplitudeL^2}
 Let  $(\Sigma,g)$ be an admissible surface with $\mathfrak{b}$ boundary components and parametrizations $\boldsymbol{\zeta}$. 
Let ${\bf m}\in \Z^n$ and $\boldsymbol{\alpha}\in (\frac{1}{R}\Z)^{n}$ satisfying \eqref{seiberg}, ${\bf v}\in (T\Sigma)^n$. If $F\in \mc{E}^{\bf m}_R$,  the amplitudes defined in Definition \ref{def:amp} satisfy
\[\caA_{\Sigma,g,{\bf v},\boldsymbol{\alpha},\bf{m},\boldsymbol{\zeta}} (F,\cdot)\in \mc{H}^{\otimes \mathfrak{b}}\] 
\end{lemma}
\begin{proof} Consider another copy of $\Sigma$, call it $\Sigma'$, with reverted orientation. We can glue $\Sigma$ to $\Sigma'$ along the $\mathfrak{b}$ boundary components (the $i$-th boundary component of $\Sigma$ is glued to the corresponding $i$-th boundary component of $\Sigma'$) to get the double surface $\Sigma^{\# 2}$ without boundary, which comes equipped with an involution $\tau: \Sigma^{\# 2}\to \Sigma^{\# 2}$ mapping a point $x$ in $\Sigma$ to its copy in $\Sigma'$. The metric $g$ also extends to $\Sigma^{\# 2}$ to a symmetric metric under $\tau$. We want to prove that the amplitude $(\tilde{\boldsymbol{\varphi}},{\bf k})\mapsto  \caA_{\Sigma,g,{\bf v},\boldsymbol{\alpha},\bf{m},\boldsymbol{\zeta}} (F,\tilde{\boldsymbol{\varphi}}^{\bf k}) $ is in $\mc{H}^{\otimes b}$.

Let us consider the amplitude
\[ \caA^{0}_{\Sigma,g,{\bf z},\bf{m},\boldsymbol{\zeta}} (F_{\bf z},\tilde{\boldsymbol{\varphi}}^{\bf k})\] 
where 
\[F_{\bf z}(\phi)=e^{-\frac{1}{2\pi}\langle \dd u_{\bf z},\nu_{\mathbf{z},\mathbf{m},{\bf k}} +\omega^c_{\bf k^c}\rangle}|F(\phi+u_{\bf z})| |e^{-\mu  M^g_\beta(\phi+u_{\bf z},\Sigma)}|\]
and the function  $u_{\bf z}$ is given by $u_{\bf z}(x)=\sum_{j=1}^{n_{\mathfrak{m}}}i\alpha_jG_{g,D}(x,z_j)$. Note that $ \caA^{0}_{\Sigma,g,{\bf z},\bf{m},\boldsymbol{\zeta}} (F_{\bf z},\tilde{\boldsymbol{\varphi}}^{\bf k})\geq 0$. Furthermore, given the formula for amplitudes \eqref{defelecamp}, we have 
\[ |\caA_{\Sigma,g,{\bf v},\boldsymbol{\alpha},\bf{m},\boldsymbol{\zeta}} (F,\tilde{\boldsymbol{\varphi}}^{\bf k}) |\leq C \caA^{0}_{\Sigma,g,{\bf v},\bf{m},\boldsymbol{\zeta}} (F_{\bf z},\tilde{\boldsymbol{\varphi}}^{\bf k})\]
for some $C$, which takes into account all trivial factors (it may depend on  ${\bf z}$, $\boldsymbol{\alpha}$).  

We can then glue the amplitudes $\caA^{0}_{\Sigma,g,{\bf z},\bf{m},\boldsymbol{\zeta}} (F_{\bf z},\tilde{\boldsymbol{\varphi}}^{\bf k})$ and $\caA^{0}_{\Sigma',g,{\bf z},-\bf{m},\boldsymbol{\zeta}} (F_{\bf z},\tilde{\boldsymbol{\varphi}}^{\bf k})$ using Propositions \ref{glueampli0} and \ref{selfglueampli0} to get
 \[C'\int \caA^{0}_{\Sigma,g,{\bf z},\bf{m},\boldsymbol{\zeta}} (F_{\bf z},\tilde{\boldsymbol{\varphi}}^{\bf k}) \caA^{0}_{\Sigma',\tau^*g,{\bf z},-\bf{m},\boldsymbol{\zeta}} (F_{\bf z},\tilde{\boldsymbol{\varphi}}^{\bf k})\,\dd \mu_0^{\otimes b}(\tilde{\boldsymbol{\varphi}}^{\bf k}) =\caA^{0}_{\Sigma^{\#2},g,\hat{\bf z},\hat{\bf m}} (\hat F_{\bf z}) \]
 for some explicit constant $C'$ appearing in Propositions \ref{glueampli0} and \ref{selfglueampli0}, where $\hat{\bf z}=({\bf z},\tau({\bf z}))$, and $\hat{\bf m}=({\bf m},-{\bf m})$. The functional $\hat{F}_{\bf z}(\phi)$ is given by 
\[\hat{F}_{\bf z}(\phi):=e^{-\frac{1}{2\pi}\langle \dd \hat u_{\bf z},\nu_{\mathbf{z},\mathbf{m} } + \omega_{\bf k}\rangle}|F(\phi+\hat u_{\bf z}|_{\Sigma})|.|F(\phi+\hat u_{\bf z}|_{\Sigma'})|.|e^{-\mu  M^g_\beta(\phi+\hat u_{\bf z},\Sigma^{\#2})}|\] 
and $\hat u_{\bf z}=u_{\bf z}\mathbf{1}_{\Sigma}+\tau^*u_{\bf z} \mathbf{1}_{\Sigma'}$. An argument similar to the proof of Proposition \ref{limitcorelmixed} shows that the amplitude $\caA^{0}_{\Sigma^{\#2},g,\hat{\bf z},\hat{\bf m}} (\hat F_{\bf z})$ is finite (there, we have the Green function $G_g$ on $\Sigma^{\#2}$ instead of the Dirichlet Green function  in the definition of $\hat u_{\bf z}$, but this is harmless to adapt). Therefore, using Lemma \ref{reverting}, we deduce
\[\begin{split}
\int |\caA_{\Sigma,g,{\bf v},\boldsymbol{\alpha},\bf{m},\boldsymbol{\zeta}} (F,\tilde{\boldsymbol{\varphi}}^{\bf k}) |^2\,\dd \mu_0^{\otimes \mathfrak{b}}(\tilde{\boldsymbol{\varphi}}^{\bf k})
&\leq 
C^2\int | \caA^{0}_{\Sigma,g,{\bf z},\bf{m},\boldsymbol{\zeta}} (F_{\bf z},\tilde{\boldsymbol{\varphi}}^{\bf k}) |^2\,\dd \mu_0^{\otimes \mathfrak{b}}(\tilde{\boldsymbol{\varphi}}^{\bf k})
\\
&= C^2\int   \caA^{0}_{\Sigma,g,{\bf z},\bf{m},\boldsymbol{\zeta}} (F_{\bf z},\tilde{\boldsymbol{\varphi}}^{\bf k}) \caA^{0}_{\Sigma',\tau^*g,{\bf z},-\bf{m},\boldsymbol{\zeta}} (F_{\bf z},\tilde{\boldsymbol{\varphi}}^{\bf k})  \,\dd \mu_0^{\otimes \mathfrak{b}}(\tilde{\boldsymbol{\varphi}}({\bf k}))
\\
&=C^2/C' \caA^{0}_{\Sigma^{\#2},g,\hat{\bf z},\hat{\bf m}} (\hat F_{\bf z}) <+\infty.
\end{split}\]
This shows that the amplitudes are in $\mc{H}^{\otimes \mathfrak{b}}$.    Note that the regularized amplitudes are in $\mc{H}^{\otimes b}$ for the same reason. \end{proof}

\subsection{Proof of Propositions \ref{glueampli}, \ref{selfglueampli} and \ref{prop:symmetry}}\label{proof_of_propositions}  
Now we focus on the proof of Propositions \ref{glueampli}  and \ref{selfglueampli}. We begin with the first one. Recall the setup is described just before the statement of Proposition \ref{glueampli}. First we claim that the gluing property holds for $\epsilon$-regularized amplitudes
 \begin{align}\label{regsegal}
&\caA^{\epsilon}_{\Sigma^{\#},g,{\bf x},{\bf v},\boldsymbol{\alpha},{\bf m},\boldsymbol{\zeta}}
(G_1\otimes G_2,\tilde{\boldsymbol{\varphi}}_1^{{\bf k}_1},\tilde{\boldsymbol{\varphi}}_2^{{\bf k}_2})\\
& =C\int \caA^{\epsilon}_{\Sigma_1,g_1,{\bf x}_1,{\bf v}_1,\boldsymbol{\alpha}_1,{\bf m}_1,\boldsymbol{\zeta}_1}(G_1, \tilde{\boldsymbol{\varphi}}_1^{{\bf k}_1},\tilde{ \varphi}^{k})\caA^{\epsilon}_{\Sigma_2,g_2,{\bf x}_2,{\bf v}_2,\boldsymbol{\alpha}_2,{\bf m}_2,\boldsymbol{\zeta}_2}(G_2,\tilde{\boldsymbol{\varphi}}_2^{{\bf k}_2},\tilde{ \varphi}^k) \dd\mu_0  (\tilde{ \varphi}^k).\nonumber
\end{align}
 where $C= \frac{1}{(\sqrt{2} \pi) }$ if $\partial\Sigma \not =\emptyset$ and $C= \sqrt{2}  $ if $\partial\Sigma =\emptyset$. For this we apply Proposition \ref{glueampli0} with 
$F_i(\phi)=G_i(\phi)\prod_{j=1}^{n^i_{\mathfrak{e}}} V_{\alpha_{ij},g,\epsilon}(x_j)e^{-\frac{i  Q}{4\pi}\int^{\rm reg}_{\Sigma_i} K_g\phi_g\dd {\rm v}_g -\mu M_\beta^g (\phi_g,\Sigma_i)}$ for $i=1,2$. Lemma \ref{magcurvtwo} makes sure that 
$$F_1\otimes F_2(\phi)=G_1\otimes G_2(\phi)\prod_{i=1,2}\Big(\prod_{j=1}^{n^i_{\mathfrak{e}}} V_{\alpha_{ij},g,\epsilon}(x_j)\Big)e^{-\frac{i  Q}{4\pi}\int^{\rm reg}_{\Sigma} K_g\phi_g\dd {\rm v}_g -\mu M_\beta^g (\phi_g,\Sigma)}.$$
In particular, we stress that it is not a priori straightforward, because of the regularization, to see that the regularized curvature on $\Sigma_1$ and $\Sigma_2$ sum up to produce the regularized curvature on $\Sigma$; this is the main outcome of Lemma \ref{magcurvtwo}. 
Then we claim

\begin{lemma}\label{cvepsamp}
We have the convergence
\[\lim_{\mathbf{x}\to\mathbf{z}}\lim_{\epsilon\to 0} \caA^{\epsilon}_{\Sigma,g,{\bf x},{\bf v},\boldsymbol{\alpha},\bf{m},\boldsymbol{\zeta}} (F,\tilde{\boldsymbol{\varphi}}^{\bf k}) = \caA_{\Sigma,g,{\bf v},\boldsymbol{\alpha},\bf{m},\boldsymbol{\zeta}} ( F,\tilde{\boldsymbol{\varphi}}^{\bf k})\text{ in }\mc{H}^{\otimes b}. \]
where ${\bf x}$ tends to ${\bf z}$ in the direction of ${\bf v}$.
\end{lemma}

\begin{proof} 
We already know that the convergence holds almost surely. Now we show  that regularized amplitudes form a Cauchy sequence in $\mc{H}^{\otimes b}$. For this we consider the doubled surface $\Sigma^{\#2}$ as before equipped with its involution $\tau$ and we still write $\hat{g}$ for the symmetrized metric $g+\tau_*g$ induced by $g$ on $\Sigma^{\#2}$. For $G\in \mc{E}^{\rm e,m}_R(\Sigma)$, we denote by $\tau(G)\in \mc{E}^{\rm e,m}_R(\tau(\Sigma))$ the functional $\tau(G)(\phi_{g}):=G(\phi_{g}\circ\tau)$. We consider the surfaces with boundary $\Sigma_1=\Sigma$ and $\Sigma_2=\tau(\Sigma)$. We consider the amplitudes 
$ \caA^{\epsilon}_{\Sigma,g,{\bf x},{\bf v},\boldsymbol{\alpha},{\bf m},\boldsymbol{\zeta}}(G, \tilde{\boldsymbol{\varphi}}^{\bf k})$ and $\overline{\caA^{\epsilon'}_{\tau(\Sigma),\tau_*g,\tau({\bf x}),\tau_*{\bf v},\boldsymbol{\alpha},-{\bf m},\tau\circ\boldsymbol{\zeta}}(\tau(G),\tilde{\boldsymbol{\varphi}}^{\bf k})}$, where $\bar{a}$ means complex conjugate of $a$. We make now two observations. First, by Lemma \ref{reverting}, we have the relation
\[\overline{\caA^{\epsilon'}_{\tau(\Sigma),\tau_*g,\tau({\bf x}),\tau_*{\bf v},\boldsymbol{\alpha},-{\bf m},\tau\circ\boldsymbol{\zeta}}(\tau(G),\tilde{\boldsymbol{\varphi}}^{\bf k})}
=\overline{\caA^{\epsilon'}_{\Sigma,g,{\bf x},{\bf v},\boldsymbol{\alpha},{\bf m},\boldsymbol{\zeta}}(G, \tilde{\boldsymbol{\varphi}}{\bf k})}.\]
Second, we have obviously
\begin{align*}
&\overline{\caA^{\epsilon'}_{\tau(\Sigma),\tau_*g,\tau({\bf x}),\tau_*{\bf v},\boldsymbol{\alpha},-{\bf m},\tau\circ\boldsymbol{\zeta}}(\tau(G),\tilde{\boldsymbol{\varphi}}^{\bf k})}
\\
&=\delta_0(\sum_{j=1}^{n_{\mathfrak{m}}} -m_j+\sum_{j=1}^{\mathfrak{b}} k_j)
\sum_{{\bf k}^c\in \Z^{2\mathfrak{g}+\mathfrak{b}-1}} e^{-\frac{1}{4\pi}\|\nu_{\tau(\mathbf{z}),-\mathbf{m},{\bf k}} +\omega^{c}_{{\bf k}^c}\|_{g,0}^2 }Z_{\tau(\Sigma),\tau_*g}\caA^0_{\tau(\Sigma),\tau_*g}(\tilde{\boldsymbol{\varphi}}  ) 
 \nonumber\\
 &
\quad  \E \Big[e^{-\frac{1}{2\pi}\langle \dd X_{\tau_*g,D}+\dd P\tilde{\boldsymbol{\varphi}},\nu_{\tau(\mathbf{z}),-\mathbf{m},{\bf k}} +\omega^{c}_{{\bf k}^c}\rangle}\tau(G)(\phi_{\tau_*g})\prod_{j=1}^{n_{\mathfrak{e}}} V_{-\alpha_j,\tau_*g,\epsilon'}(\tau(x_j))e^{\frac{i  Q}{4\pi}\int^{\rm reg}_{\Sigma'} K_{\tau_*g}\phi_{\tau_*g}\dd {\rm v}_{\tau_*g}
 -\bar{\mu} M_{-\beta}^{\tau_*g} (\phi_{\tau_*g},\tau(\Sigma))}\Big]\nonumber.
\end{align*}
 which we glue together to form an amplitude on $\hat\Sigma$ using Proposition \ref{glueampli0}:
 \begin{align*} 
&\mc{A}^{\epsilon,\epsilon'}_{\Sigma^{\#2},\hat{g},\hat{\bf x},\hat{\bf v},\hat{\boldsymbol{\alpha}},\hat{\bf m}}
(G \otimes \overline{\tau(G)})\\
&\quad =C\int \caA^{\epsilon}_{\Sigma,g,{\bf x},{\bf v},\boldsymbol{\alpha},{\bf m},\boldsymbol{\zeta}}(G, \tilde{\boldsymbol{\varphi}}({\bf k}))\overline{\caA^{\epsilon'}_{\tau(\Sigma),\tau_*g,\tau({\bf x}),\tau_*{\bf v},\boldsymbol{\alpha},-{\bf m},\boldsymbol{\zeta}}(\tau(G),\tilde{\boldsymbol{\varphi}}^{\bf k})}\dd\mu_0^{\otimes b}  (\tilde{\boldsymbol{\varphi}}^{\bf k}).
\end{align*}
where $\hat{\bf x}:=({\bf x},\tau({\bf x})),\hat{\bf v}:=({\bf v},\tau_*{\bf v}), \hat{\boldsymbol{\alpha}}=(\boldsymbol{\alpha},\boldsymbol{\alpha}), \hat{\bf m}=({\bf m},-{\bf m})$, and the amplitude type functional has the following expression
\begin{align*}
&\caA^{\epsilon,\epsilon'}_{\Sigma^{\#2},\hat{g},\hat{\bf x},\hat{\bf v},\hat{\boldsymbol{\alpha}},\hat{\bf m}}
(G \otimes  \tau(G))\\
=&
 \big(\frac{{\rm v}_{g}(\Sigma)}{{\det}'(\Delta_{g})}\big)^\hf\sum_{{\bf k}\in \Z^{2\mathfrak{g}}}e^{-\frac{1}{4\pi}\|\omega _{\bf k}\|_2^2-\frac{1}{4\pi}\|\nu_{\mathbf{z},\mathbf{m}}\|^2_{\hat{g},0}-\frac{1}{2\pi}\langle \omega_{\bf k},\nu_{\mathbf{z},\mathbf{m}}\rangle_2}
 \\ 
 & \int_{0}^{2\pi R}\E\Big[e^{-\frac{1}{2\pi}\langle \dd X_{g},\omega_{\bf k}+\nu_{\mathbf{z},\mathbf{m}} \rangle}G\otimes \overline{\tau(G)}(\phi_{\hat{g}})
 V_{(\boldsymbol{\alpha},0)}^{g,\epsilon}(\mathbf{x})\overline{V_{(\boldsymbol{\alpha},0)}^{g,\epsilon'}(\tau(\mathbf{x}))}e^{-\frac{i   Q}{4\pi}\int_{\Sigma_1 }^{\rm reg} K_{\hat{g}}\phi_{\hat{g}}\,\dd {\rm v}_{\hat{g}}+\frac{i   Q}{4\pi}\int_{\Sigma_2 }^{\rm reg} K_{\hat{g}}\phi_{\hat{g}}\,\dd {\rm v}_{\hat{g}} -  \hat{M}^{\hat{g}}_\beta(\phi_{\hat{g}},\hat\Sigma)}\Big]\,\dd c 
\end{align*}
where $\phi_{\hat{g}}:=c+X_{\hat{g}}+I_{x_0}^{\boldsymbol{\sigma}}(\omega _{\bf k})+I^{\hat{\boldsymbol{\xi}}}_{x_0}(\nu_{\hat{\mathbf{z}},\hat{\mathbf{m}}})$, the potential  is given by 
\begin{align*}
& \hat M^{\hat g,\eps}_{\beta}(h,\dd x):= \mu \eps^{-\frac{\beta^2}{2}}e^{i\beta h_{ \eps}(x)}\mathbf{1}_{\Sigma_1}(x){\rm dv}_{\hat{g}}(x)+\bar\mu  \eps^{-\frac{\beta^2}{2}}e^{-i\beta h_{ \eps}(x)}\mathbf{1}_{\Sigma_2}(x){\rm dv}_{\hat{g}}(x),\\
& \hat{M}^{\hat{g}}_\beta(\phi_{\hat{g}},\Sigma^{\#2})=\lim_{\epsilon\to 0}\hat{M}^{\hat{g},\epsilon}_\beta(\phi_{\hat{g}},\Sigma^{\#2}),\end{align*}
and $V_{(\boldsymbol{\alpha},0)}^{g,\epsilon}(u,\mathbf{x})$ are   electric type operators.
Note that the zero mode contributions from both curvature terms cancel out, hence the curvature term remains $c$-periodic. We can then follow the proof of Propositions \ref{limitcorel} and \ref{limitcorelmixed} to take the limit as $\epsilon,\epsilon'\to 0$, then $\mathbf{x}{\to}\mathbf{z}$ in the direction of ${\bf v}$ and obtain a limit that does not depend on choice of the families $\epsilon,\epsilon'$. We stress that the condition to follow these proofs remain $\alpha_j>Q$ since the electric insertions with point $x$ in $\Sigma_1$ will create a singularity  in $\Sigma_1$, and the electric insertions at point $\tau(x)$ create a singularity only on $\Sigma_2$ where the potential has a reversed sign $-\beta$.

For the last part of the argument, we shortcut $\caA^{\epsilon}_{\Sigma,g,{\bf x},{\bf v},\boldsymbol{\alpha},\bf{m},\boldsymbol{\zeta}} (F,\tilde{\boldsymbol{\varphi}}^{\bf k})$ as $\caA_{\bf x}(\epsilon)$ and we denote by $L$ the limit of $\langle \caA_{\bf x}(\epsilon),\caA_{\bf x}(\epsilon')\rangle_{\mc{H}^{\otimes \mathfrak{b}}}=\int \caA_{\bf x}(\epsilon)\overline{\caA_{\bf x}(\epsilon')}\dd\mu_0^{\otimes \mathfrak{b}}  (\tilde{\boldsymbol{\varphi}}^{\bf k})$. We have shown that 
\[\begin{split}
\|\caA_{\bf x}(\epsilon)-\caA_{\bf x}(\epsilon')\|^2_{\mc{H}^{\otimes \mathfrak{b}}}=& \langle \caA_{\bf x}(\epsilon),\caA_{\bf x}(\epsilon) \rangle_{\mc{H}^{\otimes \mathfrak{b}}}-\langle \caA_{\bf x}(\epsilon),\caA_{\bf x}(\epsilon')\rangle_{\mc{H}^{\otimes \mathfrak{b}}}-\langle  \caA_{\bf x}(\epsilon'),\caA_{\bf x}(\epsilon)\rangle_{\mc{H}^{\otimes \mathfrak{b}}}+\langle \caA_{\bf x}(\epsilon') ,\caA_{\bf x}(\epsilon')\rangle_{\mc{H}^{\otimes \mathfrak{b}}}
\\
 \to&  L-L-L+L=0
\end{split}\]
as $\epsilon,\epsilon'\to 0$, then $\mathbf{x}{\to}\mathbf{z}$ in the direction ${\bf v}$. 
Hence, the sequence $\caA_{\bf x}(\epsilon)$ is Cauchy in $\mc{H}^{\otimes b}$, and then converges in $\mc{H}^{\otimes \mathfrak{b}}$ towards the amplitude (since we know almost sure convergence already holds).
\end{proof}

Now we complete the proof of Proposition \ref{glueampli}. Back to \eqref{regsegal}, the above Lemma shows that the regularized amplitudes on respectively $\Sigma_1$ and $\Sigma_2$, in probability in  $ \tilde{\boldsymbol{\varphi}}_1^{{\bf k}_1}$ and $ \tilde{\boldsymbol{\varphi}}_2^{{\bf k}_2} $, converge as a function of $ \tilde{ \varphi}^k)$ in $\mc{H}$ towards the limiting amplitudes. We can then pass to the limit  $\epsilon\to 0$ and then $\lim_{\mathbf{x}\to \mathbf{z}}$ in \eqref{regsegal} to get our claim.

\bigskip
The case of self-gluing, namely Prop. Proposition   \ref{selfglueampli},  is slightly more subtle. Indeed self-gluing can be seen as a partial trace and it is not clear that this partial trace makes sense in generality. In \cite[Lemma 7.2]{GKRV2}, it is shown that the partial trace makes sense if we can show that amplitudes are composition of Hilbert-Schmidt operators. Observe that an annulus with Out/In boundary component can be seen as the integral kernel of some Hilbert-Schmidt operator $\mc{H}\to\mc{H}$, since we have proved that amplitudes are $L^2$. Therefore, any (regularized) amplitude can be seen as a composition of Hilbert-Schmidt operators because, for any surface $\Sigma$ with boundary, one can  see $\Sigma $ as the gluing  of the surface obtained by removing from $\Sigma$ small annuli around the boundary components and those annuli. The corresponding regularized amplitudes converge in $L^2$ towards the limiting amplitudes by Lemma \ref{cvepsamp}. It is then straightforward to pass to the limit in $\epsilon\to 0$  and then $\lim_{\mathbf{x}{\to}\mathbf{z}}$ in the analog of \eqref{regsegal} for the case of self-gluing to deduce Proposition \ref{selfglueampli}. Note that in the case of self-gluing, the behaviour of the regularized curvature is treated in Lemma \ref{magcurvtwoself}.

\subsubsection{Proof of Proposition \ref{prop:symmetry}}  

  Before proving the remaining statements, let us explain here a point that is crucial to understand the definition of amplitudes. The GFF expectation with respect to the Dirichlet GFF can absorb via the Girsanov transform any shift in the path integral by exact forms of the form $\dd f$ for some smooth function $f$ vanishing along the boundary (we call these forms exact forms of Dirichlet type). This means that the path integral features invariance along the relative cohomology classes. For instance if we want to change the relative (co-)homology basis, this can be performed as in the proof of Proposition \ref{propdefpath} because changing relative cohomology basis amounts to shifting the original basis by exact forms of Dirichlet type, up to the invariance of the curvature term with respect to homology basis.  This point that does not rise any difficulty. 
What is   more subtle is that the amplitudes are invariant under change of 1-forms in the absolute cohomology, i.e. the 1-form $\nu_{\mathbf{z},\mathbf{m},\mathbf{k}}$. Let us consider another 1-form $ \nu'_{\mathbf{z},\mathbf{m},\mathbf{k}}$ satisfying our assumptions (i.e.  Proposition \ref{harmpoles} and \eqref{btbaRz}). Then the difference can be expressed (see Lemma \ref{existence_primitive}) as $\nu_{\mathbf{z},\mathbf{m},\mathbf{k}}-\nu'_{\mathbf{z},\mathbf{m},\mathbf{k}}=\dd f$ for some smooth smooth function $f$ on $\Sigma$  such that $\pl_\nu f|_{\pl \Sigma}=0$ with $\nu$ the unit vector normal to $\pl\Sigma$ (such exact forms are said to be of Neumann type). This type of shifts cannot be absorbed by the Dirichlet GFF. This is where our complete set of assumptions will be crucial. These forms $\nu_{\mathbf{z},\mathbf{m},\mathbf{k}}$ are required to be of the form $k\dd \theta$ in local coordinates near the boundary components and near the marked points, this then forces $f$ to be locally constant near the boundary components.
From \eqref{btbaRz}, it is then plain to check that one can decompose $f$ as $f(x)=C+h(x)+\bar{f}(x)$, where $h$ is  smooth  on $\Sigma$ and constant near the boundary components/marked points with values in $R\Z$, $\bar{f}$ is smooth and vanishes on $\pl \Sigma$ and $C$ is a constant fixed to have $f(x_0)=0$. Next, we observe that there is ${\bf k}_0^c\in \Z^{2\mathfrak{g}+\mathfrak{b}-1}$ such that $\dd h=\omega^c_{{\bf k}_0^c}+\dd f_{{\bf k}_0^c}$ for some exact form $\dd f_{{\bf k}_0^c}$ of Dirichlet type by Lemma   \ref{existence_primitive} (in fact only the boundary-to-boundary arcs matter, meaning ${\bf k}_0^c=({\bf k}_0^{ic},{\bf k}_0^{bc})\in \Z^{2\mathfrak{g}}\times\Z^{\mathfrak{b}-1}$ with ${\bf k}_0^{ic}=0$). Next we plug the relation $\nu_{\mathbf{z},\mathbf{m},\mathbf{k}}=\nu'_{\mathbf{z},\mathbf{m},\mathbf{k}}+\dd f $ in the expression for amplitudes \eqref{amplitude}, we perform a change of variables ${\bf k}^c\to{\bf k}^c-{\bf k}_0^c $ in the summation $\sum_{{\bf k}^c\in \Z^{2\mathfrak{g}+\mathfrak{b}-1}} $ to absorb the $\omega^c_{{\bf k}_0^c}$-component of $f$. Finally the exponential term in the expectation in \eqref{amplitude} produces a term $e^{-\frac{1}{2\pi}\langle \dd X_{g,D},\dd \bar{f}+  \dd f_{{\bf k}_0^c}\rangle}$, to which we apply the Girsanov transform. This absorbs the $\bar{f}+ f_{{\bf k}_0^c}$ component of $f$ and in conclusion we get the expression \eqref{amplitude} where $\nu_{\mathbf{z},\mathbf{m},\mathbf{k}}$  has been replaced by $\nu'_{\mathbf{z},\mathbf{m},\mathbf{k}}$. Hence the invariance under changes of representatives of 1-form in the absolute cohomology.

\medskip Now we turn to the Weyl anomaly. We have to deal with a change of conformal metrics $g'=e^{\rho} g$ in the definition \eqref{amplitude}.    
Recall from \eqref{relationentrenorm} that $M_\beta^{g'}(X_{g',D},\dd x)=e^{-\frac{\beta Q}{2}\rho(x)}M_\beta^{g}(X_{g,D},\dd x)$ and  (recall \eqref{defvertex}) from \eqref{varYg} that $V_{\alpha_i,g'}(x_i)=e^{-\frac{\alpha_i^2}{4}\rho(x_i)}V_{\alpha_i,g}(x_i)$ (the last identity making sense when inserted inside expectation values). 
Also, from the relation for curvatures $K_{g'}=e^{-\rho}(K_g +\Delta_g\rho)$ we deduce that the term involving the curvature in  \eqref{amplitude}, given by \eqref{defcurvamp}, reads (recall that $X_{g',D}=X_{g,D}$)
\begin{align*}
 \int^{\rm reg}_\Sigma K_{g'} \phi_g\dd {\rm v}_{g'}=& \int_\Sigma  K_g (X_{g,D}+P\tilde{\boldsymbol{\varphi}})\dd {\rm v}_g+  \int_\Sigma  \Delta_g\rho (X_{g,D}+P\tilde{\boldsymbol{\varphi}})\dd {\rm v}_{g}
\\
&
+\int^{\rm reg}_\Sigma K_{g'}I^{\boldsymbol{\sigma}}_{x_0}(\omega^{c}_{{\bf k}^c})\dd {\rm v}_{g'}+\int^{\rm reg}_\Sigma K_{g'}I^{\boldsymbol{\xi}}_{x_0}(\nu_{\mathbf{z},\mathbf{m},\mathbf{k}})\dd {\rm v}_{g'}.
\end{align*}
The last two terms can be expressed in the metric $g$ thanks to Lemmas \ref{conformal_change_reg_int_open} and \ref{magcurvconfopen}.
In the second term in the right hand side, the contribution of the harmonic extension is treated  by the Green identity to produce (with $\nu$ the unit inward pointing normal at $\pl\Sigma$)
$$ \int_\Sigma  \Delta_g\rho  P\tilde{\boldsymbol{\varphi}} \dd {\rm v}_g= \int_{\partial \Sigma}\partial_\nu \rho P\tilde{\boldsymbol{\varphi}}\,\dd \ell_{g}=0$$
since the fact that both $g,g'$ are admissible entails that $\partial_\nu \rho $ must vanish.
Therefore   \eqref{amplitude} expressed in the metric $g'$ can be rewritten as
\begin{align}\label{amplitudeg'}
& \caA_{\Sigma,g',{\bf v},\boldsymbol{\alpha},\bf{m},\boldsymbol{\zeta}} (F,\tilde{\boldsymbol{\varphi}}^{\bf k}) \\
& : =
 \delta_0(\sum_{j=1}^{n_{\mathfrak{m}}+\mathfrak{b}} m_j({\bf k}))\lim_{\mathbf{x}\to\mathbf{z}}\lim_{\eps\to 0}\sum_{{\bf k}^c\in \Z^{2\mathfrak{g}+\mathfrak{b}-1}} e^{-\frac{1}{4\pi}\|\nu_{\mathbf{z},\mathbf{m},{\bf k}} +\omega^{c}_{{\bf k}^c}\|_{g,0}^2 }Z_{\Sigma,g'}\caA^0_{\Sigma,g'}(\tilde{\boldsymbol{\varphi}}  ) e^{-\sum_{j=1}^{n_{\mathfrak{m}}}\frac{\alpha_j^2}{4}\rho(x_j)-\sum_{j=1}^{n_{\mathfrak{m}}}\frac{m_j^2R^2}{16\pi^2}\rho(z_j)}
 \nonumber\\
 &
 \times \E \Big[e^{ -i\frac{Q}{4\pi}\int_\Sigma  \Delta_g\rho  X_{g,D} \dd {\rm v}_g}e^{-\frac{1}{2\pi}\langle \dd X_{g,D}+\dd P\tilde{\boldsymbol{\varphi}},\nu_{\mathbf{z},\mathbf{m},{\bf k}} +\omega^{c}_{{\bf k}^c}\rangle}F(\phi_g)\prod_{j=1}^{n_{\mathfrak{m}}} V_{\alpha_i,g,\epsilon}(x_i)e^{-\frac{i  Q}{4\pi}\int^{\rm reg}_\Sigma K_g\phi_g\dd {\rm v}_g
 -\mu M_\beta^g (\phi_g+iQ\rho/2,\Sigma)}\Big]\nonumber
\end{align}
where we have used Lemma \ref{renorm_L^2}   to switch the metric in the term involving the regularized norm. Notice that $\mc{A}^0_{\Sigma,g}(\tilde{\boldsymbol{\varphi}})=\mc{A}^0_{\Sigma,g'}(\tilde{\boldsymbol{\varphi}})$ since the quadratic form 
$({\bf D}_\Sigma \tilde{\boldsymbol{\varphi}},\tilde{\boldsymbol{\varphi}})_{2}=\int_{\Sigma}|\nabla^gP\tilde{\boldsymbol{\varphi}}|_g^2{\rm dv}_g$ 
depends only on the conformal class of $g$ for $\tilde{\boldsymbol{\varphi}}$ smooth (and using the density of $C^\infty(\pl \Sigma)\subset 
H^{s}(\pl\Sigma)$).

Now we apply the Girsanov transform to the first exponential term $e^{ -i\frac{Q}{4\pi}\int_\Sigma  \Delta_g\rho  X_{g,D} \dd {\rm v}_g}$. Again, we have to be cautious because of the imaginary factor $i$; we should perform an analytic continuation argument, which is possible here because $F\in\mc{E}_R(\Sigma)$, as in Proposition \ref{propdefpath}. We omit the details.  The variance of this transform is given by (it is negative because of the $i^2$)
$$- \frac{Q^2}{16\pi^2}\int_{\Sigma^2}  \Delta_g\omega (x) G_D(x,y)   \Delta_g\omega (y)   {\rm dv}_g(x) {\rm dv}_g(y)=- \frac{Q^2}{8\pi}\int_{\Sigma}|\dd\omega|_g^2{\rm dv}_g.$$
It has the effect of shifting the mean of the GFF $X_{g,D}$, i.e. $X_{g,D}$ becomes $X_{g,D}-i\frac{Q}{2}\rho$.  Then we see that we get almost the result, up to the $+1$ in front of the Liouville functional $\frac{  1-6Q^2}{96\pi}\int_{\Sigma}(|d\rho|_g^2+2K_g\rho) {\rm dv}_g$. This $+1$ comes from the variations of the regularized determinant of Laplacian, here $Z_{\Sigma,g'}$. The Polyakov formula for the regularized determinant of Laplacian \cite[section 1]{OPS}  implies that (for admissible $g,g'$)
\begin{align*}
 Z_{\Sigma,  g'}=& \det (\Delta_{g,D})^{-\hf}\exp\Big(\frac{1 }{96\pi}\int_{\Sigma}(|\dd\omega|_g^2+2K_g\omega) {\rm dv}_g   \Big).
\end{align*}
This completes the argument for the Weyl anomaly.  

The spin property follows the same argument as in Corollary \ref{corospin}.

\section{The irrational case} \label{irrational}
  Again, we work  throughout this section  under the constraint $\beta^2<2$.  We also assume that $\beta>0$  and we further impose the compactification radius $R>0$ to obey 
\begin{equation}
 \beta\Z \subset\frac{1}{R}\Z.
\end{equation}
Let us introduce  a further parameter    $\mu\in \C\setminus\{0\}$ and set
 \begin{equation}\label{valueQbis}
  Q= \frac{\beta}{2}- \frac{2}{\beta}.
 \end{equation} 
 Finally we introduce the central charge   ${\bf c}=1-6Q^2$. The crucial assumption is now
 \begin{equation}\label{valueQirra}
 Q\notin \frac{1}{R}\Z,
  \end{equation}
 in which case the central charge is irrational.

\subsection{Path integral and correlation functions}\label{subsecPIirra}

 Consider now a closed Riemann surface $\Sigma$  equipped with a metric $g$.  To construct the path integral, we need: the same material as in the rational case, namely we assume Assumption \ref{ass} is in force. The construction of the path integral is similar to the rational with the main difference that we need to consider electric operators with charges that do not belong to $\frac{1}{R}\Z$, which we will abusively call irrational charges as opposed to charges $\alpha_j\in  \frac{1}{R}\Z$ which we will call rational. The reason is that the contribution of the zero mode in the curvature term is no  more rational (since $Q\notin \frac{1}{R}\Z$) and therefore those irrational charges  will be used to compensate for the curvature contribution.
 
 More precisely, let $z_1,\dots, z_{n_\mathfrak{m}},z_1',\dots, z'_{n_\mathfrak{i}}$ be distinct points on $\Sigma$.  For each  point $z_j$ we   assign a unit tangent vector $v_j\in T_{z_j}\Sigma$,  a magnetic charge $m_j\in\Z$ and an (rational) electric charge   $\alpha_j\in\frac{1}{R}\Z$. For each point $z'_j$ we   assign   an (irrational) electric charge   $\alpha'_j\in\R$.
 
 We collect those datas in $\mathbf{z}=(z_1,\dots,z_{n_\mathfrak{m}}) \in\Sigma^{n_\mathfrak{m}}$,   
$\mathbf{v}=((z_1,v_1),\dots,(z_{n_\mathfrak{m}},v_{n_\mathfrak{m}}))\in (T\Sigma)^{n_\mathfrak{m}}$ and $\mathbf{m}\in\Z^{n_\mathfrak{m}}$, $\boldsymbol{\alpha}:=(\alpha_1,\dots,\alpha_{n_{\mathfrak{e}}})\in \R^{n_{\mathfrak{m}}}$,  $\mathbf{z}'=(z'_1,\dots,z'_{n_\mathfrak{i}}) \in\Sigma^{n_\mathfrak{i}}$ and $\boldsymbol{\alpha}':=(\alpha'_1,\dots,\alpha'_{n_{\mathfrak{i}}})\in \R^{n_{\mathfrak{i}}}$. We assume that $\sum_{j=1}^{n_\mathfrak{m}} m_j=0$ and $\sum_j\alpha_j'\in \chi(\Sigma)Q+\frac{1}{R}\Z$. Note that there are no magnetic charges attached to the irrational charges.

The path integral is then defined as in Section \ref{secPI} and correlation functions will be denoted by 
$$\cjg   F  V^g_{(\boldsymbol{\alpha},\mathbf{m})}({\bf v})V^g_{\boldsymbol{\alpha}' }({\bf z}')  \cjd_ {\Sigma, g}$$
 for $F\in\mc{L}^{\infty,p}_{\rm e,m}(H^s(\Sigma))$. The same proofs as for the rational case shows the following result:
\begin{theorem}\label{irrational_case}
Assume that $\sum_{j=1}^{n_\mathfrak{m}} m_j=0$ and $\sum_j\alpha_j'\in \chi(\Sigma)Q+\frac{1}{R}\Z$ and that for all $j$, $\alpha_j>Q$ and $\alpha_j'>Q$. Then:
\begin{enumerate}
\item the correlation functions are well defined for $F\in\mc{L}^{\infty,p}_{\rm e,m}(H^s(\Sigma))$. 
\item the quantity $\cjg   F  V^g_{(\boldsymbol{\alpha},\mathbf{m})}({\bf v})V^g_{\boldsymbol{\alpha}' }({\bf z}')  \cjd_ {\Sigma, g}$ does not depend on the base point $x_0$ and it does not depend on the choice of cohomology basis dual to $\sigma$.
\item the correlation function does not depend on the curves $\sigma$ provided 
we compare two families of curves that are images one to each other by a diffeomorphism isotopic to the identity and being the identity on the points with weight that do not beloong to $\frac{1}{R}\Z$ 
\item  the quantity $\cjg   F  V^g_{(\boldsymbol{\alpha},\mathbf{m})}({\bf v})V^g_{\boldsymbol{\alpha}' }({\bf z}')  \cjd_ {\Sigma, g}$  depends on the location of the points ${\bf v}=(z_j,v_j)_{j=1,\dots,n_{\mathfrak{m}}}$ in $T\Sigma$ and the charges $\mathbf{m}$, but not  on the defect graph.
\item 
{\bf Conformal anomaly:} let $g'=e^{\rho}g$ be two conformal metrics on the closed Riemann surface $\Sigma$ for some $\rho\in C^\infty(\Sigma)$. We have
\begin{align}\label{confanirra} 
& \cjg   F  V^{g'}_{(\boldsymbol{\alpha},\mathbf{m})}({\bf v})V^{g'}_{\boldsymbol{\alpha}' }({\bf z}')  \cjd_ {\Sigma, g'}  \\
 &=\big\cjg   F(\cdot- \tfrac{i  Q}{2}\rho) V^{g}_{(\boldsymbol{\alpha},\mathbf{m})}({\bf v})V^{g}_{\boldsymbol{\alpha}' }({\bf z}')  \big \cjd_ {\Sigma, g} 
e^{\frac{{\bf c}}{96\pi}\int_{\Sigma}(|d\rho|_g^2+2K_g\rho) {\rm dv}_g-\sum_{j=1}^{n_{\mathfrak{m}}}\Delta_{(\alpha_j,m_j)}\rho(z_j)-\sum_{j=1}^{n_{\mathfrak{i}}}\Delta_{(\alpha'_j,0)}\rho(z_j')}\nonumber
\end{align}
where the real numbers $\Delta_{\alpha,m}$ are defined by the relation \eqref{deltaalphadef} for $\alpha\in\R$
and the central charge is ${\bf c}:=1-6 Q^2 $.
\item {\bf Diffeomorphism invariance:} let $\psi:\Sigma\to \Sigma$ be an orientation preserving diffeomorphism whose mapping class is an element of the Torelli group. Then, recalling the notations of Theorem \ref{limitcorel},
\begin{equation}\label{diffinvariance}
\big\cjg  F (\phi_{\psi^*g}) V_{(\boldsymbol{\alpha},\mathbf{m})}^{\psi^*g}(\mathbf{v})V^{\psi^*g}_{\boldsymbol{\alpha}'}({\bf z}')  \big\cjd_ {\Sigma, \psi^*g}=\big\cjg  F(\phi_g\circ\psi)  V_{(\boldsymbol{\alpha},\mathbf{m})}^{g}(\psi_{*}\mathbf{v})V^{g}_{\boldsymbol{\alpha}'}(\psi({\bf z}'))  \big\cjd_ {\Sigma, g}.
\end{equation}
\item {\bf Spins:} with $r_{\boldsymbol{\theta}}\mathbf{v}:=((z_1,r_{\theta_1}v_1),\dots,(z_{n_\mathfrak{m}},r_{\theta_{n_\mathfrak{m}}} v_{\theta_{n_\mathfrak{m}}}))$, then
\begin{equation}\label{spinrelation} 
 \big\langle   F  V^g_{(\boldsymbol{\alpha},\mathbf{m})}(r_{\boldsymbol{\theta}}{\bf v})V^g_{\boldsymbol{\alpha}' }({\bf z}')  \big\rangle_{\Sigma,g } =e^{-i QR\langle\mathbf{m},\boldsymbol{\theta}\rangle  }  \langle F  V^g_{(\boldsymbol{\alpha},\mathbf{m})}( {\bf v})V^g_{\boldsymbol{\alpha}' }({\bf z}') \rangle_{\Sigma,g } .
 \end{equation}
\end{enumerate}
\end{theorem}

\appendix

 \section{Markov property of the GFF}
 We recall here the domain Markov property of the GFF on Riemann surfaces, whose proof can be found in the appendix of \cite{GKRV2}.
 \begin{proposition}\label{decompGFF}
Let $(\Sigma,g)$ be a Riemannian manifold with smooth boundary $\partial \Sigma$. Let $\mathcal{C}$ be a  union of  smooth non overlapping closed simple curves separating  $\Sigma$ into two connected components $\Sigma_1$ and $\Sigma_2$.\\
1) if $\partial \Sigma\not=\emptyset$ then the Dirichlet GFF $X_D$ on $\Sigma$ admits the following decomposition in law as a sum of independent processes
$$X_D\stackrel{law}{=}Y_1+Y_2+P$$
with $Y_q$ a Dirichlet GFF on $\Sigma_q$ for $q=1,2$ and $P$ the harmonic extension on $\Sigma\setminus\mathcal{C}$ of the restriction of $X_D$ to $\mathcal{C}$ with boundary value $0$ on $\partial \Sigma$.\\
2) if $\partial \Sigma=\emptyset$ then the GFF $X_g$ on $\Sigma$ admits the following decomposition in law  
$$X_g\stackrel{law}{=}Y_1+Y_2+P-c_g$$
where $Y_1,Y_2,P$ are independent, $Y_q$ is a Dirichlet GFF on $\Sigma_q$ for $q=1,2$, $P$ is the harmonic extension on $\Sigma\setminus\mathcal{C}$ of the restriction of $X_g$ to $\mathcal{C}$  and $c_g:=\frac{1}{{\rm v}_g(\Sigma)}\int_\Sigma(Y_1+Y_2+P)\,\dd {\rm v}_g $.
 \end{proposition}   
 
 \section{Proof of Lemma \ref{independence_basis}}\label{Lemma_indep}

Using Lemma \ref{conformal_change_reg_int}, it suffices to prove the result for one choice of conformal representative in the conformal 
class of $[g]$. By the result of Aubin \cite{Aubin}, one can prescribe the scalar curvature $K_{\hat{g}}$ of a conformal metric $\hat{g}:=e^{\rho}g$ (for some $\rho\in C^\infty(\Sigma)$ as long as $K_{\hat{g}}\leq 0$ with $\int_\Sigma K_{\hat{g}}\dd {\rm v}_{\hat{g}}<0$. We have thus reduced to studying the case where the curvature $K_g$ can be assumed to be non-positive and $K_g=0$ everywhere but on an arbitrarily small open set $\mc{K}$.

The link between a symplectic basis $([a_j],[b_j])_j$ to another one $([a_j'],[b_j'])_j$ is given by a matrix $A\in {\rm Sp}(2{\mathfrak{g}},\Z)$. 
The group ${\rm Sp}(2{\mathfrak{g}},\Z)$ is generated by four types of elements, called Burkhardt generators, see 
\cite[Proof of Th. 6.4]{Farb-Margalit}.  

The first is the factor rotation $r_{j_1}\in {\rm Sp}(2{\mathfrak{g}},\Z)$, which keeps the basis $\sigma$ 
fixed except for the two elements $[a_{j_1}],[b_{j_1}]$ where $r_{j_1}([a_{j_1}])=[b_{j_1}]$ and $r_{j_1}([b_{j_1}])=-[a_{j_1}]$. In terms of our regularized integral, this amounts to check that 
\[ \chi(\gamma_{a_{j_1}})\int_{b_{j_1}}k_{g}\dd \ell_{g}-\chi(\gamma_{b_{j_1}})\int_{a_{j_1}}k_{g}\dd \ell_{g}=\chi(\gamma_{b_{j_1}})\int_{-a_{j_1}}k_{g}\dd \ell_{g}-\chi(\gamma_{-a_{j_1}})\int_{b_{j_1}}k_{g}\dd \ell_{g},\]
which is straightforward (here we write $-a_{j_1}$ for the curve $a_{j_1}$ with reverse orientation). 

The second Burkhardt generator is the factor swap $s_{j_1,j_2}$ 
which is the identity on $[a_j],[b_j]$ for $j\notin \{k,\ell\}$ and satisfies $s_{j_1,j_2}([a_{j_1}])=[a_{j_2}]$ and $s_{j_1,j_2}([b_{j_1}])=[b_{j_2}]$, i.e. it swaps 
$\mc{T}_{j_1}$ and $\mc{T}_{j_2}$ in our geometric representation of $\Sigma$ in Figure \ref{fig:domaine}. 
For this elementary move, it is clear from the definition of the regularized integral that this is invariant. 

The third Burkhardt generator is the transvection $t_{j_1}$, which preserves all $[a_j],[b_j]$ except for 
$t_k([a_{j_1}])=[a_{j_1}]+[b_{j_1}]$. This is realized by a Dehn twist $T_{b_{j_1}}$ along $b_{j_1}$. Without loss of generality, we can 
assume that $j_1=1$ to simplify the notations. We consider the equivariant extension of 
$I^{\boldsymbol{\sigma}}_{x_0}(\omega)$ from $K_{1}$ to $\tilde{\mc{T}}_{1}$ (which is the plane with the translated disks $D(e,\eps)+k$ removed, for $k\in \Z^2$) and still denote it by $I^{\boldsymbol{\sigma}}_{x_0}(\omega)$. 
The new basis obtained by applying the transvection can be represented by simply changing the curve $a_{1}$ by a simple smooth 
curve $a'_{1}$ that we represent in $K_{1}$ by 
\[ \sigma_{a'_{1}}(t)=t+i\alpha(t) \, \quad \alpha(t)=0 \textrm{ for } t\in [0,\eps], \quad \alpha(t)=1 \textrm{ for }t\in [1-\eps,1]\] for some $\eps>0$
and $\alpha(t)\geq 0$ non decreasing in the interval $t\in [\eps,1-\eps]$. We obtain a new fundamental domain $K_{1}'$ 
of $\mc{T}_{1}$ by considering the domain in $\tilde{\mc{T}}_{1}$ bounded by $\sigma_{a'_1},\bar{\sigma}_{a'_1}:=
\gamma_{b_1}(\sigma_{a_1})$ and the vertical lines $i\R$ and $1+i\R$; see Figure \ref{fig:transvection}.
 \begin{figure}[h] 
\includegraphics[width=0.4\textwidth]{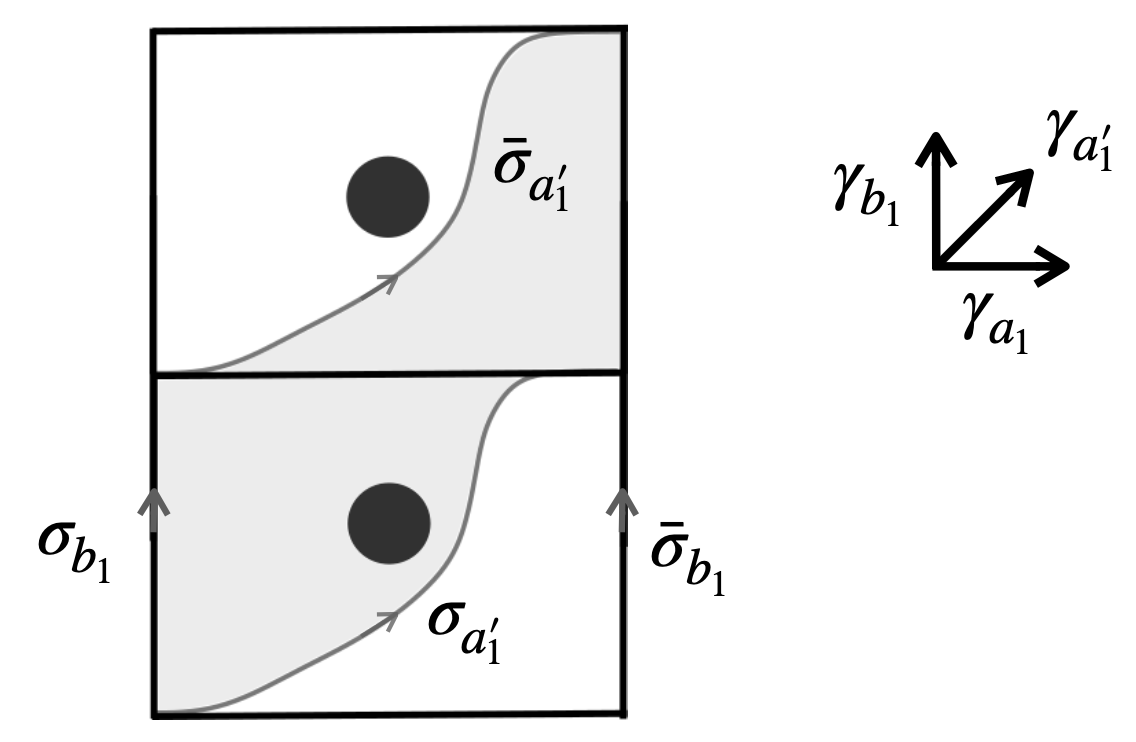} 
\caption{The domain $K_1'$ in gray and the new curves $a'_1$ lifted to $\tilde{\mc{T}}_1$.} 
\label{fig:transvection}
\end{figure} 
If $\sigma'$ denotes the new canonical basis of $\mc{H}_1(\Sigma)$, we notice that $I^{\boldsymbol{\sigma}'}_{x_0}(\omega)$ in $K_1'$ is equal to the equivariant extension of 
$I^{\boldsymbol{\sigma}}_{x_0}(\omega)$. Let us denote by $D= K_1'\setminus K_1$: we compute using 
$I^{\boldsymbol{\sigma}}_{x_0}(\omega)(z+i)=I^{\boldsymbol{\sigma}}_{x_0}(\omega)(z)+\chi(\gamma_{b_1})$ for $z\in K_1$
\[  \begin{split}
\int_{K_1'} K_gI^{\boldsymbol{\sigma}'}_{x_0}(\omega)\dd {\rm v}_g=& \int_{K_1'} K_g I^{\boldsymbol{\sigma}}_{x_0}(\omega)\dd {\rm v}_g=\int_{K_1\cap K_1'}K_g I^{\boldsymbol{\sigma}}_{x_0}(\omega)\dd {\rm v}_g+\int_{D}K_gI^{\boldsymbol{\sigma}}_{x_0}(\omega) \dd {\rm v}_g \\
=& \int_{K_1}K_g I^{\boldsymbol{\sigma}}_{x_0}(\omega)\dd {\rm v}_g+\chi(\gamma_{b_1})\int_{D}K_g \dd {\rm v}_g
\end{split}\]
and using Gauss-Bonnet, 
\[\int_{D}K_g \dd {\rm v}_g= 2(\int_{b_1}k_g\dd\ell_g+ \int_{a_1}k_g\dd\ell_g-\int_{a'_1}k_g\dd\ell_g).\]
On the other hand, the difference of the boundary terms of $\int_{\Sigma_{\boldsymbol{\sigma}'}}^{\rm reg} K_gI^{\boldsymbol{\sigma}'}_{x_0}(\omega)\dd {\rm v}_g$ with that of 
$\int_{\Sigma_{\boldsymbol{\sigma}}}^{\rm reg} K_g I^{\boldsymbol{\sigma}}_{x_0}(\omega)\dd \ell_g$ is given by 
\[ 2(\chi(\gamma_{a_1})+\chi(\gamma_{b_1}))\int_{b_1}k_g\dd\ell_g-2\chi(\gamma_{b_1})\int_{a'_1}k_g\dd\ell_g+2\chi(\gamma_{b_1})\int_{a_1}k_g\dd\ell_g -2\chi(\gamma_{a_1})\int_{b_1}k_g\dd\ell_g\]
thus implying that 
\[\int_{\Sigma_{\boldsymbol{\sigma}'}}^{\rm reg} K_gI^{\boldsymbol{\sigma}'}_{x_0}(\omega)\dd {\rm v}_g=\int_{\Sigma_{\boldsymbol{\sigma}}}^{\rm reg} K_g I^{\boldsymbol{\sigma}}_{x_0}(\omega)\dd {\rm v}_g.\]

It remains to deal with the case of the 4th Burkhardt generator, which is the factor mix $f_{j_1,j_2}$ which preserves all $[a_j],[b_j]$ except for $j=j_1,j_2$ where 
\[ f_{j_1,j_2} :  ([a_{j_1}],[b_{j_1}], [a_{j_2}],[b_{j_2}])\mapsto ([a_{j_1}]-[b_{j_2}],[b_{j_1}],[a_{j_2}]-[b_{j_1}],[b_{j_2}]).\] 
As before, witout loss of generality we can assume that $j_1=1$ and $j_2=2$ to simplify the notations. We choose a curve $a_1'$ which represent the class $[a_1]-[b_2]$ and a curve $a_2'$ which represent the class $[a_2]-[b_1]$: they must have intersection pairing equal to $0$ with all $a_j,b_j$ except for the following $j$:
\begin{align*} 
& \iota(a_1',a_1)=0, \,\, \iota(a_1',a_2)=1, \,\,  \iota(a_1',b_1)=1, \,\,  \iota(a_1',b_2)=0,\\
& \iota(a_2',a_1)=1, \,\, \iota(a_2',a_2)=0, \,\,  \iota(a_2',b_1)=0, \,\,  \iota(a_2',b_2)=1.
\end{align*}
We can assume that $S_{\boldsymbol{\sigma}}=\hat{\C}\setminus \cup_{j=1}^{\mathfrak{g}} \mc{D}_j$ and $\mc{D}_{1}=D(0,\eps)$, $\mc{D}_{2}=D(1,\eps)$. 
The curve $a_1'$ can then be chosen so that:\\ 
1) its intersection with $\mc{T}_{1}$ lifts to $\sigma_{a'_{1}}(t)=i/2+t$ for $t\in [0,1/2-\eps]$ and 
$\sigma_{a'_{1}}(t)=i/2+(t-1)$ for $t\in [3/2+\eps,2]$, and up to making a small deformation of the curve near 
$\sigma_{a_{1}'}\cap \pl \mc{D}_{1}'$, we can assume that this curve intersects $\pl \mc{D}'_{1}$ with an angle $\pi/2$;\\
2)  its intersection with $S_{\boldsymbol{\sigma}}$ decomposes into two pieces of curves 
$a'_{1}(t)$ for $t\in [1/2-\eps,1/2+\eps]$ and $t\in  [3/2-\eps,3/2+\eps]$, with $a_{1}'(1/2-\eps)=-i\eps/2$, $a_{1}'(1/2+\eps)=1-i\eps/2$ and
$a_{1}'(3/2-\eps)=1+i\eps/2$, $a_{1}'(3/2+\eps)=i\eps/2$, and the angle between $a_{1}'$ and $\pl \mc{D}_{1}$ and $\pl \mc{D}_{2}$ are $\pi/2$;\\
3) its intersection with $\mc{T}_{2}$ lifts to $\sigma_{a'_{1}}(t)=1/2+i(1-t)$ for $t\in [1/2+\eps,1]$ and 
$\sigma_{a'_{1}}(t)=1/2+i(2-t)$ for $t\in [1,3/2-\eps]$, and up to  making a small deformation of the curve near $\sigma_{a_{1}'}\cap \pl \mc{D}_{2}'$
we can assume that this curve intersects $\pl \mc{D}'_{2}$ with an angle $\pi/2$. We define similarly $a_{2}'$ by reverting the role of $a'_1$ and $a'_2$. 
See Figures \ref{fig:factormix1} and \ref{fig:factormix2}.

 \begin{figure}[h] 
\includegraphics[width=0.6\textwidth]{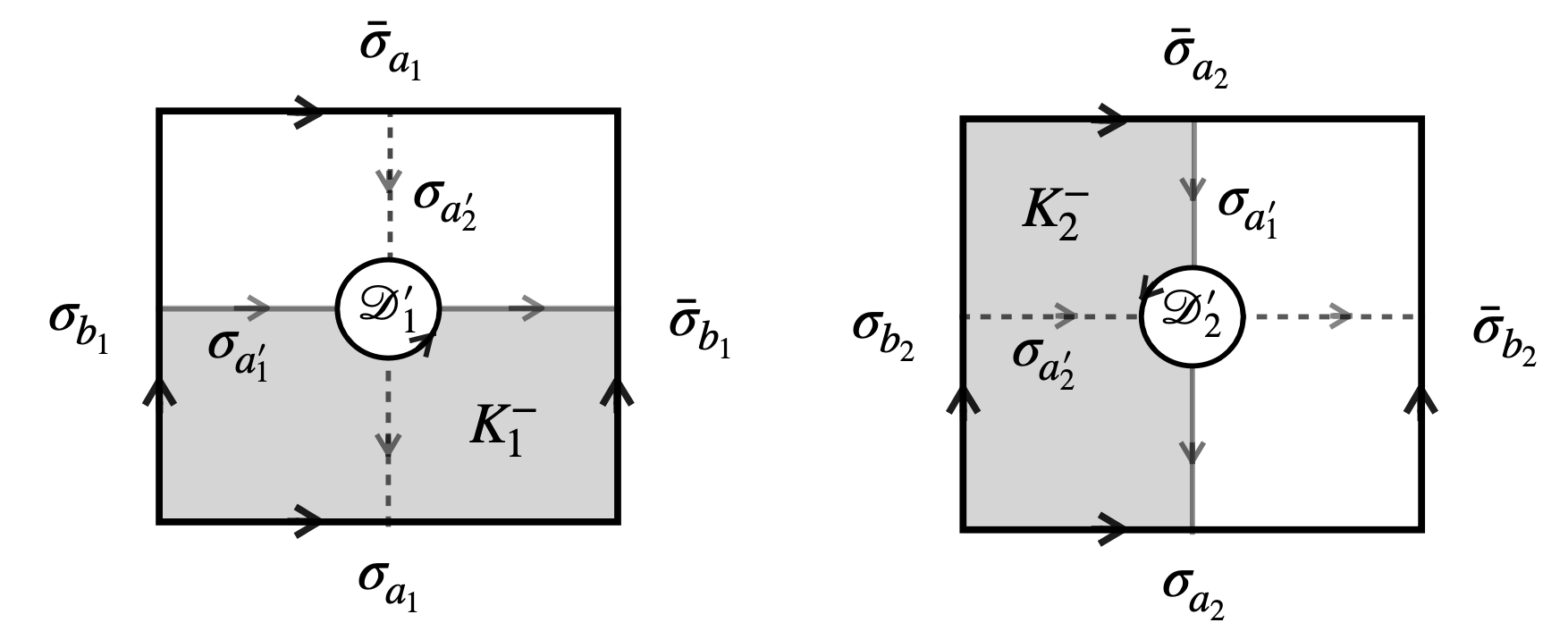} 
\caption{The domains $K_{1},K_{2}$ and the new curves $a'_{1},a'_{2}$ lifted to $\tilde{\mc{T}}_{1},\tilde{\mc{T}}_{2}$.} 
\label{fig:factormix1}
\end{figure} 
 \begin{figure}[h] 
\includegraphics[width=0.4\textwidth]{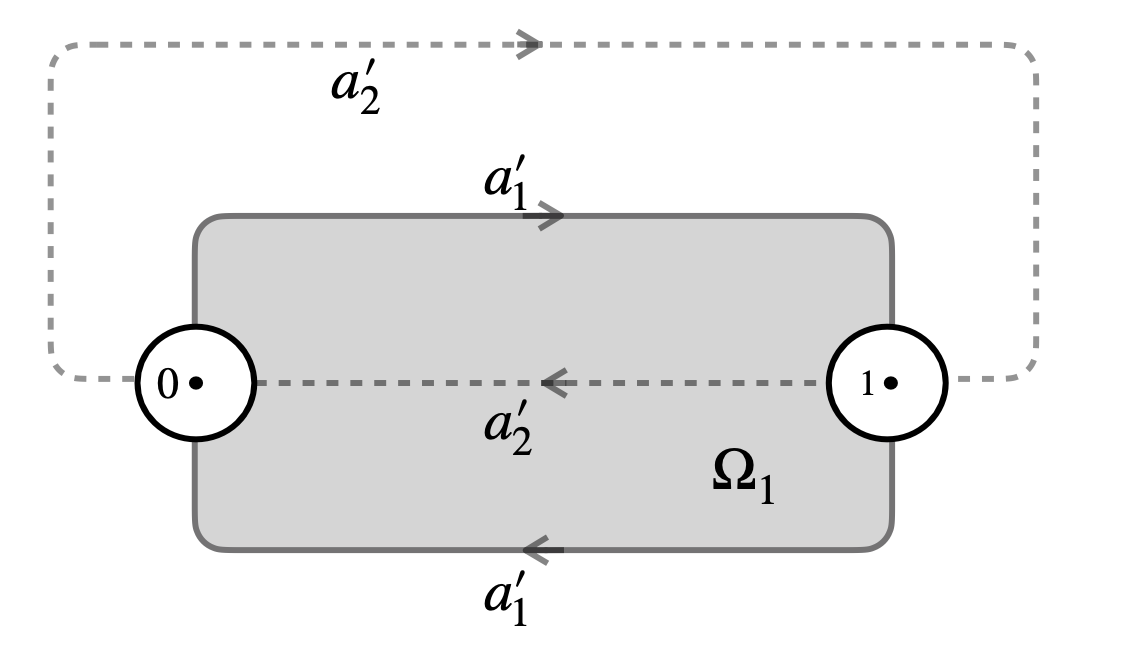} 
\caption{The domains $S_{\boldsymbol{\sigma}}$ near $\mc{D}_{1}$ and $\mc{D}_{2}$, and the new curves $a'_{1},a'_{2}$.} 
\label{fig:factormix2}
\end{figure}
For $j=1,2,$ let us denote $\Omega_{j}\subset S_\sigma\cap \C$ the connected set bounded by the two connected components of $a'_{j}\cap S_{\boldsymbol{\sigma}}$.    
We can freely  choose the curves $a'_{1},a'_{2}$ so that the set $\mc{K}$ where the curvature is non zero is contained 
in $S_{\boldsymbol{\sigma}}^\circ \setminus (\Omega_{1}\cup \Omega_{2})$.
We observe that $I^{\boldsymbol{\sigma}'}_{x_0}(\omega)(x)=I^{\boldsymbol{\sigma}}_{x_0}(\omega)(x)$ for $x\in S_{\boldsymbol{\sigma}}\setminus (\Omega_{1}\cup \Omega_{2})$. 
We compute 
\[ \int_{\Sigma_{\boldsymbol{\sigma}}} K_gI^{\boldsymbol{\sigma}'}_{x_0}(\omega)\dd {\rm v}_g=\int_{\mc{K}} K_gI^{\boldsymbol{\sigma}'}_{x_0}(\omega)\dd {\rm v}_g=
\int_{\mc{K}} K_gI^{\boldsymbol{\sigma}}_{x_0}(\omega)\dd {\rm v}_g=\int_{\Sigma_{\boldsymbol{\sigma}}} K_gI^{\boldsymbol{\sigma}}_{x_0}(\omega)\dd {\rm v}_g.\]
Let $K_{1}^-\subset K_{1}$ be the connected set bounded by $\sigma_{b_{1}},\bar{\sigma}_{b_{1}},\sigma_{a_{1}},\sigma_{a'_{1}}$ and $K_{2}^-\subset K_{2}$ the connected set 
bounded by $\sigma_{b_{2}},\sigma_{a_{2}},\bar{\sigma}_{a_{2}},\sigma_{a'_{1}}$
We compute using the Gauss-Bonnet formula in $K_{1}^-$ 
\begin{equation}\label{GBOmega_j}
\begin{split}
\int_{\sigma_{a'_{1}}\cap K_{1}}k_g\dd \ell_g=\pi+\int_{\sigma_{a_{1}}}k_g\dd \ell_g-\int_{\mc{C}'_{1}}k_g\dd \ell_g
\end{split}\end{equation}
where $\mc{C}'_{1}$  is the semi-circle $\pl\mc{D}_{1}'\cap K_{1}^-$  oriented counter-clockwise. 
We next apply Gauss-Bonnet in the region $\Omega_{1} \subset S_{\boldsymbol{\sigma}}$ bounded by the two pieces of curves representing 
$a'_{1}\cap S_{\boldsymbol{\sigma}}$: 
\begin{equation}\label{GBOmegaj1}
\int_{a'_{1}\cap S_{\boldsymbol{\sigma}}}k_g\dd \ell_g=\int_{\mc{C}_{1}}k_g\dd \ell_g+\int_{\mc{C}_{2}}k_g \dd \ell_g
\end{equation} 
where, for $j=1,2$, $\mc{C}_{j}$  is the semi-circle $\Omega_{1}\cap \pl \mc{D}_{j}$ oriented clockwise. Finally we apply Gauss-Bonnet in the domain 
$K_{2}^-$: 
\begin{equation}\label{GBK_j2}
\int_{\sigma_{a'_{1}}\cap K_{2}}k_g\dd \ell_g=-\int_{\sigma_{b_{2}}}k_g\dd \ell_g -\int_{\mc{C}'_2}k_g\dd \ell_g+\pi
\end{equation} 
where $\mc{C}'_{2}$ is the semi-circle $\mc{C}'_{2}=\pl \mc{D}'_{2}\cap K_{2}^-$ oriented counter clockwise. We can use that 
$\int_{\mc{C}_{j}}k_g\dd \ell_g=\int_{\mc{C}'_{j}}k_g\dd \ell_g$ for $j=1,2$, and summing \eqref{GBK_j2}, \eqref{GBOmegaj1} and \eqref{GBOmega_j}, we obtain 
\[ \int_{a'_{1}}k_g\dd \ell_g=2\pi+\int_{a_{1}}k_g\dd \ell_g-\int_{b_{2}}k_g\dd \ell_g.\]
By symmetry, the same argument yields 
\[ \int_{a'_{2}}k_g\dd \ell_g=2\pi+\int_{a_{2}}k_g\dd \ell_g-\int_{b_{1}}k_g\dd \ell_g.\]
We then obtain (with $\chi(a_1')=\chi(a_1)-\chi(b_2)$ and $\chi(a_2')=\chi(a_2)-\chi(b_1)$) 
\[ \Big(-\sum_{j=1}^2\chi(b_j)\int_{a_j'}k_g\dd \ell_g+ \chi(a'_j)\int_{b_j}k_g\dd \ell_g\Big)-\Big( -\sum_{j=1}^2\chi(b_j)\int_{a_j}k_g\dd \ell_g
+\chi(a_j)\int_{b_j}k_g\dd \ell_g
\Big) =-(\chi(b_1)+\chi(b_2))2\pi.\]
This shows that 
\[ \int_{\Sigma_{\boldsymbol{\sigma}}}^{\rm reg} K_gI^{\boldsymbol{\sigma}'}_{x_0}(\omega)\dd {\rm v}_g=\int_{\Sigma_{\boldsymbol{\sigma}'}}^{\rm reg} K_gI^{\boldsymbol{\sigma}}_{x_0}(\omega)\dd {\rm v}_g-4\pi(\chi(b_1)+\chi(b_2)).\]

\section{Proof of Lemma \ref{estimeedechien}}
The term $I_{x_0}(\omega_{\bf k}+\nu^{\rm h}_{\mathbf{z},\mathbf{m}})(y) $ being bounded on the complement of the defect graph, we can obviously get rid of this term. Then, observing that the Green function in isothermal local coordinates is of the form $-\ln|x-y|+f(x,y)$ for some smooth function $f$, it is plain to see that the control of our integral amounts to estimating the following quantities
\begin{align*}
\int_{B(z_j,\delta)}&\int_{B(z_{j'},\delta)}|y-x_j(t)|^{\beta\alpha_j}\big(|x_{j'}(t)-y'|^{\beta\alpha_{j'}}-|z_{j'}-y'|^{\beta\alpha_{j'}}\big)\frac{\dd y \dd y'}{|y-y'|^{\beta^2}}\\
 \int_{B(z_j,\delta)}&\int_{B(z_{j'},\delta)}|y-z_j|^{\beta\alpha_j}\big(|x_{j'}(t)-y'|^{\beta\alpha_{j'}}-|z_{j'}-y'|^{\beta\alpha_{j'}}\big)\frac{\dd y \dd y'}{|y-y'|^{\beta^2}}
 \end{align*}
when $t\to 1$, for $\delta>0$ small but fixed. We treat only  the first integral because the second one is similar. If $j\not=j'$, then choosing $\delta>0$ small enough so that the balls $ B(z_j,\delta)$ and $B(z_{j'},\delta) $ do not intersect, it is obvious to show that the integral converges to $0$ because the on-diagonal term $|y-y'|^{-\beta^2}$ is bounded and $\beta\alpha_j>-2$ for all $j$. We can thus focus on the case $j=j'$ and, by invariance under translation, we can assume $z_j=0$. We set
$$F_\delta(x):=\int_{B(0,\delta)}\int_{B(0,\delta)}|x-y|^{\beta\alpha_j}\big(|x-y'|^{\beta\alpha_{j}}-|y'|^{\beta\alpha_{j}}\big)\frac{\dd y \dd y'}{|y-y'|^{\beta^2}}.$$
We have to show that $F_\delta(x)\to 0$ as $x\to 0$. We have already seen that $F_\delta$ is finite using \eqref{elem}. Next, we want to  show that $F_\delta$ is bounded. Indeed
\begin{align*}
|F_\delta(x)|\leq &\int_{B(0,\delta)}\int_{B(0,\delta)}|x-y|^{\beta\alpha_j} |x-y'|^{\beta\alpha_{j}} \frac{\dd y \dd y'}{|y-y'|^{\beta^2}}+\int_{B(0,\delta)}\int_{B(0,\delta)}|x-y|^{\beta\alpha_j}|y'|^{\beta\alpha_{j}} \frac{\dd y \dd y'}{|y-y'|^{\beta^2}}.
 \end{align*}
 The first integral is bounded by an argument of translation invariance and we call $G_\delta(x)$ the second integral. For $|x|\geq \delta/4$ we claim that $G_\delta(x)$ is bounded by some $\delta$ dependent constant. Indeed we can split the $y$-integral in two parts depending on $|y|\leq \delta/8$ or  $|y|\geq \delta/8$. On the first part, we can remove the first singularity $|x-y|^{\beta\alpha_j}$ since it is bounded and the remainder is then obvious to bound. On the second part, one of the two singularities involving $y'$ is bounded (because $|y|\geq \delta/8$, $y'$ cannot be close to both $0$ and $y$) so that it is bounded in the same way.
  
 Now we bound  $G_\delta(x)$ for $|x|\leq \delta/2$.  Using the same argument as above, we show that $G_\delta(x)-G_{\delta/2}(x)$ is bounded by some $\delta$ dependent constant and $G_{\delta/2}(x)= (1/2)^{2\beta\alpha_j-\beta^2+4}G_\delta(2x)$ by a change of variables. Therefore $G_\delta(x)\leq  (1/2)^{2\beta\alpha_j-\beta^2+4}G_\delta(2x)+C_\delta$. Now we conclude that $G_\delta$ is bounded. Indeed, if $|x|\leq \delta/2$, we can find $n$ such that 
 $2^{-n-1}\delta<|x|\leq 2^{-n}\delta$. We can then iterate the previous relation to get $G_{\delta}(x)\leq (1/2)^{(2\beta\alpha_j-\beta^2+4)n}G_\delta(2^nx)+C_\delta\sum_{k=0}^{n-1}(1/2)^{(2\beta\alpha_j-\beta^2+4)k}$. Using that $G_\delta(x)$ is bounded for  $|x|\geq \delta/4$, we deduce that $G_\delta$ is bounded.

Now we are back to the study of $F_\delta$, which we know now to be bounded.  Next, for $|x|\leq \delta/4$ we have

\begin{align}\label{chien1}
|F_\delta(x)|\leq &\iint_{B(0,\delta/2)^2} |x-y|^{\beta\alpha_j}\big(|x-y'|^{\beta\alpha_{j'}}-|y'|^{\beta\alpha_{j'}}\big)\frac{\dd y \dd y'}{|y-y'|^{\beta^2}}\\
& +\iint_{B(0,\delta)^2\setminus B(0,\delta/2)^2} |x-y|^{\beta\alpha_j}\big(|x-y'|^{\beta\alpha_{j'}}-|y'|^{\beta\alpha_{j'}}\big)\frac{\dd y \dd y'}{|y-y'|^{\beta^2}}.\nonumber
 \end{align}
The second integral in the rhs can be split in two parts depending on $|y|\geq \delta/2$ or  $|y|\leq \delta/2$ and $|y'|\geq \delta/2$. On the first part, we can remove the first singularity $|x-y|^{\beta\alpha_j}$ since it is bounded and the remainder is then obvious to bound with the relation \eqref{fa2} below, which gives that the first part is less than $C(|x|^{\beta\alpha_j+2}+|x|)$. On the second part, we use the mean value theorem to get that $\big||x-y'|^{\beta\alpha_{j'}}-|y'|^{\beta\alpha_{j'}}\big|\leq C|x|$,  and the integral is less than $C|x|$ using invariance under translations. All in all
$$\iint_{B(0,\delta)^2\setminus B(0,\delta/2)^2} |x-y|^{\beta\alpha_j}\big(|x-y'|^{\beta\alpha_{j'}}-|y'|^{\beta\alpha_{j'}}\big)\frac{\dd y \dd y'}{|y-y'|^{\beta^2}}\leq C(|x|^{\beta\alpha_j+2}+|x|).$$
The first integral $\iint_{B(0,\delta/2)^2} \dots$ in \eqref{chien1}  can be dealt with using a change of variables (dilation); it is equal to $2^{-(2\beta\alpha_j+4-\beta^2)}F_\delta(2x)$. So we have obtained
$$|F_\delta(x)|\leq 2^{-(2\beta\alpha_j+4-\beta^2)}F_\delta(2x)+C(|x|^{\beta\alpha_j+2}+|x|).$$
We recall that $2\beta\alpha_j-\beta^2+4>0$. Iterating, we deduce that for $|x|\leq 2^{-n}\delta$,
$$F_\delta(x)\leq  2^{-(2\beta\alpha_j+4-\beta^2)(n-1)} F_\delta(2^{n-1}x)+C(|x|^{\beta\alpha_j+2}+|x|)\sum_{k=0}^{n-2}2^{-(2\beta\alpha_j+4-\beta^2)k} .$$
Next, for $|x|\leq \delta/4$, we can find $n$ such that $\delta 2^{-n-1}<|x|\leq \delta 2^{-n}$. The above relation then yields (using that $F_\delta$ is bounded)
  $|F_\delta(x)|\leq C(|x|^{2\beta\alpha_j+4-\beta^2}+|x|^{\beta\alpha_j+2}+|x|)$.

 which completes the proof of the lemma, up to the following estimates.

We claim for all $|x|\leq \delta/4$:
\begin{align}
 \int_{B(0,\delta) \setminus B(0,\delta/2)}&\big||y'-x|^{\beta\alpha_j}-|y'|^{\beta\alpha_j}\big|\,\dd y'\leq C|x|  ,\label{fa}\\
 \int_{B(0,\delta)} &\big||y'-x|^{\beta\alpha_j}-|y'|^{\beta\alpha_j}\big|\,\dd y'\leq C\big(|x|^{ \beta\alpha_j+2} +|x|\big). \label{fa2}
 \end{align}
 The proof of the first claim follows straightforwardly from the mean value theorem. For the second claim, let us call $G_\delta(x)$ this integral. It is plain to see that it is bounded. Next, we have
 \begin{align*}
G_\delta(x)
 \leq &\int_{ B(0,\delta/2)}  \big||y'-x|^{\beta\alpha_j}-|y'|^{\beta\alpha_j}\big|\,\dd y'+\int_{B(0,\delta)\setminus B(0,\delta/2)}  \big||y'-x|^{\beta\alpha_j}-|y'|^{\beta\alpha_j}\big|\,\dd y'\\
 \leq & 2^{-2-\beta\alpha_j}G_\delta(2x)+C|x|.
  \end{align*}
  Now, if $|x|\leq 2^{-n}\delta$, we can iterate the previous relation to get
$$G(x)\leq 2^{-(\beta\alpha_j+2)(n-1)}G(2^{n-1}x)+C|x|\sum_{k=0}^{n-2}2^{-(\beta\alpha_j+2)k}.$$
Finally, for any $|x|\leq \delta/4$, we can find $n$ such that $2^{-n-1}\delta<|x|\leq 2^{-n}\delta$. The previous relation then gives
which gives $|G(x)|\leq C(|x|^{\beta\alpha_j-2}+|x|)$ for some constant $C>0$, eventually depending on $\delta$.

\end{document}